\documentclass[11pt]{article}

\usepackage{float, amssymb, amsmath, amsthm, graphicx, bbm, multirow, siunitx, placeins, authblk, geometry, fullpage}
\usepackage[round]{natbib}

\usepackage[dvipsnames]{xcolor}
\usepackage{diagbox}
\usepackage{tikz}
\usetikzlibrary{arrows.meta}
\usetikzlibrary{decorations.pathreplacing}
\usepackage[shortlabels]{enumitem}
\setlist{nolistsep}

\usepackage{array}
\newcolumntype{C}[1]{>{\centering\arraybackslash}p{#1}}
\usepackage{algorithm}
\usepackage{algpseudocode}
\algnewcommand\algorithmicinput{\textbf{INPUT:}}
\algnewcommand\INPUT{\item[\algorithmicinput]}

\algnewcommand\algorithmicoutput{\textbf{OUTPUT:}}
\algnewcommand\OUTPUT{\item[\algorithmicoutput]}

\algnewcommand\algorithmicexploration{\textbf{Initialize:}}
\algnewcommand\Initialise{\item[\algorithmicexploration]}

\algnewcommand\algorithmicexploitation{\textbf{Exploitation:}}
\algnewcommand\Exploitation{\item[\algorithmicexploitation]}

\usepackage{booktabs}

\usepackage{xr-hyper}
\usepackage{hyperref}[]
\hypersetup{
    colorlinks=true,
    linkcolor=blue,
    filecolor=magenta,      
    urlcolor=cyan,
      citecolor=blue,
    }

\usepackage{cleveref}

\newtheorem{theorem}{Theorem}
\newtheorem{lemma}[theorem]{Lemma}
\newtheorem{proposition}[theorem]{Proposition}

\newtheorem{definition}{Definition}

\newtheorem{assumption}{Assumption}

\makeatletter
\newcommand{\neutralize}[1]{\expandafter\let\csname c@#1\endcsname\count@}
\makeatother

\allowdisplaybreaks

\DeclareMathOperator*{\argmin}{argmin}

\DeclareMathOperator*{\arginf}{arginf}

\renewcommand{\arraystretch}{1.2}

\usepackage{mathtools}

\def\E{\mathbb{E}}
\def\P{\mathbb{P}}

\def\rank{\mathrm{rank}}

\def\R{\mathbb{R}}
\def\Z{\mathbb{Z}}

\title{Online Learning for Autoregressive Multilayer Stochastic Block Models under Stationarity and Non-Stationarity}

\author[1]{Fan Wang}
\author[2]{Haotian Xu}
\author[3]{Yi Yu}

\affil[1]{School of Mathematics and Statistics, University of Melbourne}
\affil[2]{Department of Mathematics and Statistics, Auburn University}
\affil[3]{Department of Statistics, University of Warwick}

\begin{document}

\maketitle

\begin{abstract}
Dynamic multilayer networks arise in many applications where multiple types of relations among a common set of nodes evolve over time. Existing approaches often assume temporal independence, focus on single-layer networks or impose stationarity, limiting their applicability in practice.
In this paper, we introduce a first-order autoregressive multilayer stochastic block model (AR(1)-MSBM), in which edge formation and dissolution probabilities between consecutive time points are determined by latent community memberships and shared across layers. Under stationarity, we propose an online estimation procedure based on recursive updates and tensor-based spectral refinement. We establish non-asymptotic estimation rates, prove their minimax optimality and derive guarantees for community recovery.
We further consider a non-stationary setting that allows both abrupt changes and gradual shifts, and develop an adaptive windowed online algorithm that automatically adjusts to unknown structural changes. Under a quasi-stationary segmentation framework, we derive estimation and community recovery guarantees that match the stationary results when applied segmentwise. Our theoretical findings are supported by extensive numerical experiments, with code available online \url{https://github.com/HaotianXu/AR1-MSBM_Simulation}.
\end{abstract}

\section{Introduction}

Statistical network analysis studies the relationships among a set of entities. A network is typically represented by a collection of nodes and edges, corresponding to the units in a population of interest and their pairwise relationships.  Multilayer networks generalize this framework by allowing multiple types of interactions, referred to as layers, to coexist among a common set of nodes. When such multilayer network data are observed sequentially over time, they give rise to dynamic multilayer networks. 

Analyzing dynamic multilayer networks presents several key challenges.  First, the coexistence of multiple relation types requires statistical models that capture inter-layer structural relationships rather than treating each relation type in isolation. 
Second, temporal dependence is often intrinsic in many applications, calling for models that explicitly characterize the evolution of network structures over time. 
Third, many systems operate in streaming, non-stationary environments, where the underlying connectivity may drift gradually or change abruptly.
It is therefore crucial to develop adaptive online learning procedures that efficiently update estimates as new data arrive, leveraging historical information while promptly responding to substantial changes, whether abrupt or arising from accumulated drift.

Motivated by these considerations, we propose an autoregressive multilayer network process, in which edges evolve over time while the node and layer sets remain fixed. Autoregressive network models, which provide an explicit mechanism for capturing temporal dependencies between consecutive networks, have been studied for single-layer networks by \cite{jiang2023autoregressive} and \cite{chang2026autoregressive}.  We extend this framework to the multilayer setting and to non-stationary regimes, where the underlying dynamics may evolve gradually or undergo abrupt changes. We propose an adaptive online procedure that adjusts its effective look-back window in response to unknown structural shifts, thereby efficiently leveraging past information while maintaining robustness.

Such temporally dependent multilayer networks with non-stationary dynamics arise naturally in a wide range of application domains. For example, in international trade networks \citep[e.g.][]{de2015structural, jing2021community, wang2025multilayer}, countries are nodes and trade relationships form edges.
Different product categories constitute distinct layers and sequential observations over time yield dynamic multilayer data.
Temporal dependence arises because trade patterns tend to persist due to factors such as long-term contracts.
However, the overall system is rarely stationary: gradual drifts occur due to supply-chain reconfiguration or economic development, while abrupt structural shifts occur following exogenous shocks such as trade policy changes, tariff impositions, financial crises or major global events like the COVID-19 pandemic.
The ability to adaptively learn under non-stationarity is therefore essential for timely analysis of evolving networked systems.

\subsection{Related literature}

Our work lies at the intersection of three areas: dynamic networks, multilayer networks and non-stationary online learning. We review key developments in each area below.

\medskip
\noindent
\textbf{Dynamic networks.} 
A substantial body of work has focused on modeling the temporal evolution of networks, particularly through dynamic SBMs. To capture gradual structural evolution, early studies \citep[e.g.][]{yang2011detecting} modelled node community memberships as Markov chains with fixed block probabilities, while later extensions \citep[e.g.][]{xu2014dynamic, matias2016statistical, pensky2019spectral} allowed both memberships and connection probabilities to evolve over time. In contrast, networks may also undergo abrupt structural changes, in which either the community structure or the connectivity patterns shift suddenly. This regime falls under the domain of change point analysis, with works such as \cite{de2016detection}, \cite{padilla2022change}, \cite{bhattacharjee2020change} and \cite{wang2021optimal} developing methods and theoretical guarantees for detecting and localizing such structural breaks in dynamic SBMs.

It is worth noting that in all the aforementioned studies, latent community memberships are modelled through their connection probabilities. In contrast, \cite{jiang2023autoregressive} introduced a first-order autoregressive stochastic block model (AR(1)-SBM) in which the latent community memberships are characterized by their transition probabilities.  
More recently, \cite{chang2026autoregressive} proposed dependent-edge AR$(m)$ network models, where the transition probability of an edge may depend on multiple lagged observations as well as the histories of other edges.

\medskip
\noindent
\textbf{Multilayer networks.}
Multilayer network modeling has emerged as a powerful framework for representing multiple types of interactions among a common set of nodes \citep[e.g.][]{cardillo2013modeling, kivela2014multilayer, zhang2020flexible}. For multilayer SBMs, notable examples include the work of \cite{paul2020spectral}, who studied the consistency of various community detection methods, and \cite{lei2024computational}, who analyzed the computational and statistical thresholds for consistent recovery across varying network densities. Other formulations of multilayer SBMs, such as \cite{stanley2016clustering} and \cite{jing2021community}, allow heterogeneous community structures across layers.

Dynamic multilayer SBMs that incorporate temporal evolution remain comparatively less explored.
\cite{lopez2025dynamicstochasticblockmodel} proposed a dynamic multilayer SBM in which community memberships evolve according to hidden Markov chain dynamics, allowing each node’s community membership in a given layer to depend on its previous memberships across all layers. \cite{ting2020detecting} employed state-dependent Markov-switching dynamics to model temporal variations in connection probabilities while keeping node community memberships fixed over time.
In contrast, \cite{wang2025multilayer} allowed both community memberships and connection probabilities to change abruptly, focusing on detecting and localizing such structural breaks.

\medskip
\noindent
\textbf{Non-stationary online learning.}
A growing body of research investigates online learning in non-stationary environments, where the data-generating process drifts over time. In online convex optimization, studies such as \cite{besbes2015non}, \cite{jadbabaie2015online}, \cite{zhao2021improved} and \cite{pmlr-v151-baby22a} 
analyze settings in which the loss functions evolve over time and characterize dynamic regret as the performance gap between the learner and a dynamic oracle. Theoretical guarantees are typically expressed in terms of variation metrics such as the cumulative change in optimizers or loss functions.

Beyond convex optimization, non-stationary multi-armed bandit problems, where reward distributions change over time, have also been widely studied. For instance, \cite{auer2019adaptively} proposed adaptive algorithms that achieve optimal regret without prior knowledge of the number or timing of changes. Building on this, \cite{suk2022tracking} and \cite{suk2025tracking} introduced the concept of significant shifts, a refined measure of non-stationarity that focuses on substantial distributional changes requiring algorithmic adaptation, leading to sharper performance guarantees.
More recently, \cite{huang2023stability} proposed a general framework for learning under non-stationarity by introducing quasi-stationary segments and developed adaptive learning algorithms with minimax-optimal guarantees.

\subsection{List of contributions}

The main contributions of this work are summarized as follows.

Firstly, we develop a comprehensive framework for dynamic multilayer networks that explicitly captures temporal dependence through an edge-refresh mechanism, in which each edge in each layer evolves over time via formation and dissolution governed by an autoregressive structure.  Such frameworks have been considered in \cite{jiang2023autoregressive} and \cite{chang2026autoregressive} under different modeling assumptions and stationarity regimes.
In particular, we consider multilayer network sequences that may exhibit both gradual and abrupt changes in their transition probabilities.

Secondly, to address the challenges posed by non-stationarity in dynamic multilayer network data, we propose a novel adaptive online learning algorithm that automatically adjusts its temporal look-back window according to local data variability. Inspired by the stability-based principle introduced by \cite{huang2023stability}, the algorithm balances bias and stochastic error by comparing empirical losses across nested windows and selecting the largest look-back window that remains stable. This mechanism enables automatic adaptation to both smooth and abrupt structural changes without requiring prior knowledge of change frequencies or magnitudes.
Moreover, the algorithm achieves low storage and computational cost by maintaining only a small number of sufficient statistics and updating estimates recursively rather than recomputing from scratch. Building upon this adaptive framework, we further incorporate tensor-based spectral methods that enable sharp estimation of transition probabilities.

Thirdly, in the non-stationary setting, we propose a segmentation framework that partitions the entire data sequence into quasi-stationary segments. To control the bias induced by temporal drift, we assume locally dominantly monotone trajectories for the segment-level transition probabilities, allowing occasional reversals while preserving an overall trend. Under these mild conditions, our adaptive online algorithm achieves sharp estimation accuracy, as well as consistent and exact community recovery under different conditions, with estimation rates and requirement assumptions matching those of the stationary setting when applied segmentwise.

Lastly, we establish rigorous theoretical guarantees for both the stationary and non-stationary settings. In the stationary setting, we derive sharp non-asymptotic estimation bounds for the transition probabilities governing edge formation and dissolution, which we show to be minimax optimal when the time horizon $t$ is sufficiently large relative to the number of nodes and layers.  As part of this analysis, we also derive a minimax lower bound for estimating transition probabilities; to the best of our knowledge, this is the first such lower bound for dynamic multilayer network models.
We further establish both consistent and exact recovery of the latent communities under different conditions, showing an improvement upon existing results.

\subsection{Notation and organization}
For any positive integer $p$, let $[p] = \{1, \ldots, p\}$ and let $\Pi([p])$ denote the set of all permutations of $[p]$. 
We write $a_n = O(b_n)$ if there exists a constant $C > 0$, independent of $n$, such that $a_n \leq C b_n$ for all sufficiently large $n$. We write 
$a_n = \Theta(b_n)$ if both $a_n = O(b_n)$ and $b_n = O(a_n)$. 

For any matrix $A \in \R^{p_1 \times p_2}$, let $A_i$ and $A^j$ denote the $i$th row and $j$th column of $A$, respectively.  Let $A_{i, j}$ be the entry of $A$ in the $i$th row and $j$th column and let $\sigma_1(A) \geq \dots \geq  \sigma_{p_1 \wedge p_2} (A)\geq 0$ denote its singular values. Let $\|A\|$ and $\|A\|_\mathrm{F}$  be the spectral and Frobenius norms of $A$, respectively.   We define $\sigma_{\min}(A)$ as the smallest non-zero singular value of $A$ if $A \neq 0$, and as zero if $A = 0$. For any positive integer $p_1 \geq p_2 $, let $\mathbb{O}_{p_1 \times p_2} =  \{ O \in \R^{p_1\times p_2}\colon O^{\top}O = I_{p_2}\}$.

For any order-3 tensors $\mathbf{M}, \mathbf{Q}  \in \R^{p_1 \times p_2 \times p_3}$, define the inner product
\[
\langle \mathbf{M}, \mathbf{Q} \rangle = \sum_{i = 1}^{p_1} \sum_{j=1}^{p_2} \sum_{l = 1}^{p_3}   \mathbf{M} _{i, j, l } \mathbf{Q} _{i, j, l },
\]
the entrywise infinity norm $\| \mathbf{M} \|_{\infty} = \max_{i \in [p_1], j \in [p_2], l \in [p_3]} \vert \mathbf{M}_{i, j, l}\vert $ and the Frobenius norm
$\|\mathbf{M}  \|_{\mathrm{F}} =\sqrt{\langle \mathbf{M}, \mathbf{M} \rangle}$
For any $l \in [p_3]$, let $\mathbf{M}_{:,:,l} \in \R^{p_1 \times p_2}$ with $(\mathbf{M}_{:,:,l})_{i,j} = \mathbf{M}_{i,j,l}$ for any $i \in [p_1]$ and $j \in [p_2]$.
The mode-$1$ matricization of a tensor  $\mathbf{M}$  is defined as  
\[
\mathcal{M}_1(\mathbf{M}) \in \mathbb{R}^{p_{1} \times (p_{2} p_{3})} \mbox{ with entries }
\mathcal{M}_1(\mathbf{M} )_{i_1, (i_2 - 1)p_{3} +i_3 } = \mathbf{M} _{i_1, i_2, i_3}.
\]
The mode-$2$ and mode-$3$ matricizations are analogously defined as $\mathcal{M}_2(\mathbf{M}) \in \mathbb{R}^{p_2 \times (p_3 p_1)}$ and $\mathcal{M}_3(\mathbf{M}) \in \mathbb{R}^{p_3 \times (p_1 p_2)}$, respectively.
Define 
\[
\sigma_{\min}(\mathbf{M}) = \min\big\{ \sigma_{\min}\big(\mathcal{M}_1(\mathbf{M})\big), \sigma_{\min}\big(\mathcal{M}_2(\mathbf{M})\big), \sigma_{\min}\big(\mathcal{M}_3(\mathbf{M})\big)  \big\}
\]
and 
\[
\|\mathbf{M}\| = \max\big\{ \|\mathcal{M}_1(\mathbf{M})\|, \|\mathcal{M}_2(\mathbf{M})\|, \|\mathcal{M}_3(\mathbf{M})\|  \big\}.
\]
For any $s \in [3]$, let $\rank_s(\mathbf{M}) = \rank(\mathcal{M}_s(\mathbf{M}))$.
The Tucker ranks $(r_1, r_2, r_3)$ of $\mathbf{M}$ are given by 
\[
r_s = \rank_s(\mathbf{M}), \quad \forall s \in [3].
\]
For any $s \in [3]$ and matrix $U_s \in \mathbb{R}^{q_s \times p_s}$, the marginal multiplication operator $\times_1$ is defined as
\[
\mathbf{M}  \times_1 U_1 =  \bigg\{\sum_{k = 1}^{p_{1}} \mathbf{M}_{k, j, l} (U_1)_{i, k} \bigg\}_{i \in [q_1],  j \in [p_2],  l \in [p_3]} \in \mathbb{R}^{q_1 \times p_2 \times p_3}.
\]
Marginal multiplications $\times_2$ and $\times_3$ are defined analogously.

The rest of the paper is organized as follows. In \Cref{sec-stat}, we introduce the stationary AR(1)-MSBM and develop an online estimation procedure together with its theoretical guarantees. \Cref{sec-nonstat} extends the framework to the non-stationary setting, where we propose an adaptive window-based algorithm and establish its statistical properties. Extensive numerical experiments, including both simulation studies and real data analysis, are presented in \Cref{sec-num}. \Cref{sec-dis} concludes with a discussion of possible extensions. All technical proofs and auxiliary results are deferred to the appendices.

\section{Optimal transition tensor estimation and community recovery under stationarity} \label{sec-stat}

We begin with the stationary AR(1) multilayer network process, in which the transition probabilities are time-invariant. This setting induces temporal dependence among consecutive networks while preserving identical marginal distributions over time.  For completeness, we first define the stationary AR(1)-MSBM in \Cref{def-sarmsb}. A tensor-based online estimation procedure is then presented in \Cref{section-stat-alg}, followed by its theoretical guarantees in \Cref{section-stat-theory}.

\begin{definition}[Stationary first-order autoregressive multilayer stochastic block model, stationary AR(1)-MSBM]\label{def-sarmsb}
A stationary $\mbox{AR}(1)\mbox{-MSBM}$ is a sequence of adjacency tensors $\{\mathbf{A}^t \}_{t \geq 0} \subset\{0, 1\}^{n \times n \times L}$ defined recursively for $t\geq 1$ by
\[
\mathbf{A}^t_{i,j, l} = \mathbf{A}^{t-1}_{i,j,l} \mathbbm{1}\{\mathbf{E}^t_{i,j, l} = 0\} + \mathbbm{1}\{\mathbf{E}^t_{i,j, l} = 1\}, \quad \forall  1 \leq i \leq  j \leq n, l \in [L],
\]
where $\{\mathbf{E}^t_{i,j, l}\}_{1 \leq i \leq  j \leq n, l \in [L], t \geq 1} \subset \{-1, 0, 1\}$ are mutually independent random variables with 
\[
\P \{ \mathbf{E}^t_{i,j, l}  = 1 \} = \boldsymbol{\Theta}_{i, j, l} = Z_i \mathbf{W}_{:,: l} Z_j^{\top},  \quad \P \{ \mathbf{E}^t_{i,j, l}  = -1 \}  =  \boldsymbol{\Delta}_{i, j, l} = Z_i  \mathbf{M}_{:,: l} Z_j^{\top}
\]
and
\[
\P \{ \mathbf{E}^t_{i,j, l}  = 0  \}  = 1 - \boldsymbol{\Theta}_{i, j, l}  - \boldsymbol{\Delta}_{i, j, l}.
\]
The initial adjacency tensor $\mathbf{A}^0 \in \{0, 1\}^{n \times n \times L}$ is defined as
\[
 \P\{\mathbf{A}_{i,j,l}^{0} = 1\}  = \frac{\boldsymbol{\Theta}_{i, j, l}}{ \boldsymbol{\Theta}_{i, j, l} + \boldsymbol{\Delta}_{i, j, l}}, \quad \forall  1 \leq i \leq  j \leq n, l \in [L].
\]
Here, $Z \in \{0,1\}^{n \times K}$ is the community membership matrix and $\mathbf{W}, \mathbf{M} \in [c_{\min}, 1-c_{\min}]^{K \times K \times L}$ are connectivity tensors, with an absolute constant $ c_{\min} \in (0, 1/2)$. 
\end{definition}

This edge-refresh mechanism has been considered in 
\cite{jiang2023autoregressive} and \cite{chang2026autoregressive}. 
Between two consecutive time points, each edge $(i, j, l)$ may undergo formation (transition from $0$ to $1$) or dissolution (transition from $1$ to $0$), governed by the probabilities $\boldsymbol{\Theta}_{i, j, l}$ and $\boldsymbol{\Delta}_{i, j, l}$, respectively, i.e.
\begin{equation}\label{eq-tran}
 \P\{\mathbf{A}_{i,j,l}^t = 1 \mid \mathbf{A}_{i,j,l}^{t-1} = 0 \} = \boldsymbol{\Theta}_{i, j, l}  \quad \mbox{and} \quad 
 \P\{\mathbf{A}_{i,j,l}^t = 0 \mid \mathbf{A}_{i,j,l}^{t-1} = 1 \} = \boldsymbol{\Delta}_{i, j, l}.
\end{equation}
The transition probability tensors $\boldsymbol{\Theta}$ and $\boldsymbol{\Delta}$ encode a shared community structure across all layers. As we show in \Cref{prop-stat} in \Cref{sec-app-add-stat}, each edge process $\{\mathbf{A}_{i,j,l}^t\}_{t \ge 0}$ forms a two-state Markov chain that admits a unique stationary distribution.  Consequently,  the marginal edge probabilities 
\[
\boldsymbol{\Pi}_{i,j,l}^t = \P\{\mathbf{A}_{i,j,l}^t = 1\}, \quad \forall  1 \leq i \leq  j \leq n, l \in [L],
\]
also inherit this community structure. As emphasized by \cite{jiang2023autoregressive} and \cite{chang2026autoregressive}, 
leveraging both $\boldsymbol{\Theta}$ and $\boldsymbol{\Delta}$ for community detection provides substantially improved accuracy compared with relying solely on the marginal probabilities $\boldsymbol{\Pi}$. This advantage is further demonstrated by the numerical experiments in \Cref{sec:stationary}.

\subsection{An online learning algorithm}\label{section-stat-alg}

\begin{algorithm}
\caption{Online estimation for stationary $\mbox{AR}(1)\mbox{-MSBM}$}
\label{alg:fast-stat-msbm}
\begin{algorithmic}
\INPUT{Ranks $(K, r_1, r_2)$}
\Initialise{$m \leftarrow 0$, $\mathbf{N}^{0, 0, 1}_{i,j,l} \leftarrow 0$, $\mathbf{N}^{0, 1, 0}_{i,j,l} \leftarrow 0$, $\mathbf{N}^{0}_{i,j,l}  \leftarrow 0$}
\ForAll{$t \in [T]$}
\Statex \vspace{2pt}
\textbf{Stage I: Initial estimator update}
    \ForAll{$1\leq i\leq j\leq n, l\in[L]$}
        \State{ 
        $\mathbf{N}^{t,0,1}_{i,j,l} \leftarrow\mathbf{N}_{i,j,l}^{t-1,0,1}+\mathbf{A}^{t}_{i,j,l}(1-\mathbf{A}^{t-1}_{i,j,l}), \quad 
        \mathbf{N}^{t,1,0}_{i,j,l}\leftarrow \mathbf{N}_{i,j,l}^{t-1,1,0}+(1-\mathbf{A}^{t}_{i,j,l})\mathbf{A}^{t-1}_{i,j,l}$}
        \State{ $\mathbf{N}^{t}_{i,j,l} \leftarrow \mathbf{N}^{t-1}_{i,j,l} + \mathbf{A}^{t-1}_{i,j,l}$}
        \State{ ${\widehat{\boldsymbol{\Theta}}}^{t}_{i, j, l}
   \leftarrow   \mathbf{N}^{t, 0, 1}_{i, j, l}/ (t- \mathbf{N}^{t}_{i,j,l}) $, 
   $ {\widehat{\boldsymbol{\Delta}}}^{t}_{i, j, l}
   \leftarrow   \mathbf{N}^{t, 1, 0}_{i, j, l}/\mathbf{N}^{t}_{i,j,l}$
 }
    \EndFor
       \State{Form $({\widehat{\boldsymbol{\Theta}}}^{t}, \widehat{\boldsymbol{\Delta}}^{t})$ by setting $\widehat{\boldsymbol{\Theta}}^{t}_{j,i,l} = \widehat{\boldsymbol{\Theta}}^{t}_{i,j,l}$ and $\widehat{\boldsymbol{\Delta}}^{t}_{j,i,l} = \widehat{\boldsymbol{\Delta}}^{t}_{i,j,l}$ for $i < j, l \in [L]$.}
\Statex \vspace{2pt}
\textbf{Stage II: Low-rank refinement via tensor-based methods}    
   \If{$t \geq 2^m$}
   \State{$\widehat{U}^{t}_Z \leftarrow \mbox{H-PCA} \big(\mathcal{M}_1 (\widehat{\boldsymbol{\Theta}}^{t } + \widehat{\boldsymbol{\Delta}}^{t}) \mathcal{M}_1 (\widehat{\boldsymbol{\Theta}}^{t } + \widehat{\boldsymbol{\Delta}}^{t })^{\top}, K\big)$}
   \State{$\widehat{U}^{t}_W \leftarrow  \mbox{H-PCA} \big(\mathcal{M}_3 (\widehat{\boldsymbol{\Theta}}^{t } ) \mathcal{M}_3 (\widehat{\boldsymbol{\Theta}}^{t })^{\top}, r_1\big)$}
    \State{$\widehat{U}^{t}_M \leftarrow \mbox{H-PCA} \big(\mathcal{M}_3 (\widehat{\boldsymbol{\Delta}}^{t } ) \mathcal{M}_3 (\widehat{\boldsymbol{\Delta}}^{t })^{\top}, r_2\big)$}
   \State{$ m \leftarrow m+1$}
   \Else
   \State{$\widehat{U}_Z^{t} \leftarrow \widehat{U}_Z^{t-1}$,  $\widehat{U}_W^{t} \leftarrow \widehat{U}_W^{t-1}$,  $\widehat{U}_M^{t} \leftarrow \widehat{U}_M^{t-1}$}
   \EndIf
\State{$\widetilde{\boldsymbol{\Theta}}^{t} = \widehat{\boldsymbol{\Theta}}^{t} \times_1 \widehat{U}^{t}_Z (\widehat{U}^{t}_Z)^{\top}\times_2 \widehat{U}^t_Z(\widehat{U}^{t}_Z)^\top  \times_3 \widehat{U}^t_W (\widehat{U}^{t}_W)^\top$} \State{$\widetilde{\boldsymbol{\Delta}}^{t} = \widehat{\boldsymbol{\Delta}}^{t} \times_1 \widehat{U}^t_Z  (\widehat{U}^{t}_Z)^{\top} \times_2 \widehat{U}^t_Z(\widehat{U}^{t}_Z)^{\top}\times_3  \widehat{U}_M^t(\widehat{U}^{t}_M)^\top$} 

\Statex \vspace{2pt}
\textbf{Stage III: Community recovery}   
\State{$ \widetilde{U}^t 
= \argmin_{U \in \mathcal{M}_{n,K}} \|\widehat{U}_Z^{t} - U\|_{\mathrm{F}}^2. 
$}
\State{Let $\{\tilde{u}^{(t, k)}\}_{k=1}^{K}$ denote the distinct row patterns of $\widetilde{U}^t$. }
\State{$\widehat{Z}^t_{i, j} =  \mathbbm{1} \big\{\widetilde{U}_i^t = \tilde{u}^{(t, j)}\big\}, \quad \forall i \in [n], j \in [K]$
}
\EndFor
       
\OUTPUT{ 
$\{(\widetilde{\boldsymbol{\Theta}}^{t},
  \widetilde{\boldsymbol{\Delta}}^{t})\}_{t\in [T]}$ and 
$\{\widehat{Z}^t\}_{t\in [T]}$}
\end{algorithmic}
\end{algorithm}

We now present an efficient online procedure for estimating the transition probabilities and latent community structure in the stationary $\mbox{AR}(1)\mbox{-MSBM}$. 
At a high level, \Cref{alg:fast-stat-msbm} consists of three steps. First, it recursively updates edgewise sufficient statistics and forms closed-form initial estimators of the transition probabilities. Second, it refines these estimators by exploiting the low-rank Tucker structure through tensor-based spectral methods. Third, it recovers the latent communities by clustering the estimated node subspace. We describe these three components below.

\medskip
\noindent
\textbf{Likelihood formulation and initial estimators.}
At each time point $t$, conditionally on the initial state $\mathbf{A}^0$, the joint likelihood of the observed sequence
$\{\mathbf{A}^s\}_{s \in [t] }$ factorizes over time and edges as
\[
       \prod_{s=1}^{t}  \prod_{1 \leq i \leq  j \leq n, l \in [L]} \boldsymbol{\Theta}_{i, j, l}^{\mathbf{A}^s_{i,j, l} (1-\mathbf{A}^{s-1}_{i,j, l})} (1- \boldsymbol{\Theta}_{i, j, l})^{(1-\mathbf{A}^s_{i,j, l})(1-\mathbf{A}^{s-1}_{i,j, l})} \boldsymbol{\Delta}_{i, j, l}^{(1-\mathbf{A}^s_{i,j, l}) \mathbf{A}^{s-1}_{i,j, l}} (1-\boldsymbol{\Delta}_{i, j, l})^{\mathbf{A}^s_{i,j, l} \mathbf{A}^{s-1}_{i,j, l}}.
\]
Maximizing this likelihood yields the closed-form initial estimators of the transition probabilities:
\[
\widehat{\boldsymbol{\Theta}}^t_{i, j, l} = \frac{\sum_{s=1}^{t} \mathbf{A}^s_{i,j, l} (1-\mathbf{A}^{s-1}_{i,j, l})}{\sum_{s=1}^{t} (1-\mathbf{A}^{s-1}_{i,j, l})}, 
\quad
\widehat{\boldsymbol{\Delta}}^t_{i,j,l}= \frac{(1-\mathbf{A}^s_{i,j, l}) \mathbf{A}^{s-1}_{i,j, l}}{\sum_{s=1}^{t} \mathbf{A}^{s-1}_{i,j, l}}.
\]
These estimators depend only on three edgewise sufficient statistics:
the cumulative number of $0 \to 1$ transitions, 
the cumulative  number of $1 \to 0$ transitions 
and the total number of active edges, i.e.~the tensors $\mathbf{N}^{t,0,1}, \mathbf{N}^{t,1,0}, \mathbf{N}^{t} \in \R^{n \times n \times L}$, each computed separately for every edge $(i,j,l)$.
As each new observation $\mathbf{A}^t$ arrives, these statistics are updated recursively without revisiting past data, resulting in an algorithm that is both computationally and memory efficient.

\medskip
\noindent
\textbf{Low-rank refinement via tensor-based methods.}
Since both $\boldsymbol{\Theta}$ and $\boldsymbol{\Delta}$ are low-rank tensors under the Tucker decomposition (see eq.~\eqref{eq-tucker-decom} in Appendix \ref{sec-app-add-stat}), we further refine the initial estimators using a tensor-based spectral method, namely tensor heteroskedastic principal component analysis (TH-PCA) introduced by \cite{han2022optimal}. TH-PCA extends heteroskedastic principal component analysis (H-PCA), a matrix-based spectral method proposed by \cite{zhang2022heteroskedastic} and detailed in \Cref{hpca} in Appendix~\ref{section:add-algorithm}, to higher-order tensors.
To reduce computational cost, the spectral refinement is performed only at geometrically spaced time points $t = 2^m$ for integers $m \geq 0$.
At these times, the left singular subspaces of each mode matricization are estimated via H-PCA and the resulting subspaces are reassembled through tensor multiplications to produce low-rank tensor estimators for $\boldsymbol{\Theta}$ and $\boldsymbol{\Delta}$.

\medskip
\noindent
\textbf{Community recovery.}
To recover community memberships, we use $\widehat{U}^t_Z$, the estimated singular subspace corresponding to the mode-$1$ matricization (associated with the community membership matrix $Z$). We then project $\widehat{U}^t_Z$ onto the membership space
\[
\mathcal{M}_{n,K} = \{ M \in \mathbb{R}^{n \times K} : M \mbox{ has } K \mbox{ distinct rows} \}.
\]
via $K$-means clustering, yielding the estimated community membership matrix.

\medskip

In terms of the computational cost, updating the initial estimators $\widehat{\boldsymbol{\Theta}}^t$ and $\widehat{\boldsymbol{\Delta}}^t$ requires recursively updating the sufficient statistics $\mathbf{N}^{t,0,1}$, $\mathbf{N}^{t,1,0}$ and $\mathbf{N}^{t}$ at each time step,  which incurs $O(n^2L)$ operations per iteration.
The spectral procedure H-PCA is performed only at at geometrically spaced time points $t=2^m$, so up to time $t$, it is invoked only $O(\log (t))$ times. Each such update costs $O(n^2Lr)$, where $r=\max\{K,r_1,r_2\}$. 
Consequently, the total computational cost up to time $t$ is
$O(tn^2L + \log(t) n^2Lr)$. 

The storage cost is dominated by the three sufficient-statistics tensors $\mathbf{N}^{t,0,1}$, $\mathbf{N}^{t,1,0}$ and $\mathbf{N}^{t}$, together with the current adjacency tensor and the current estimators, all of dimension $n\times n\times L$. Hence, the overall memory requirement is $O(n^2L)$. The subspace estimates $\widehat{U}_Z^t$, $\widehat{U}_W^t$ and $\widehat{U}_M^t$ require only $O(nK+Lr)$ additional memory, which is lower order. Thus, the algorithm uses linear storage in the size of the current network tensor, or equivalently, constant working memory per edge.

\subsection{Theoretical guarantees of Algorithm \ref{alg:fast-stat-msbm}}\label{section-stat-theory}

We now present the main statistical guarantees for \Cref{alg:fast-stat-msbm}.  
We first establish estimation error bounds for the transition probability tensors  
$\boldsymbol{\Theta}$ and $\boldsymbol{\Delta}$ in \Cref{thm-tensor-Frobenius-stat}, followed by a minimax lower bound in \Cref{thm-minimax-lower}. We then derive conditions for consistent and exact recovery of the community structure in \Cref{thm-comunity-recovery-stat}.  To start off, we introduce the regularity conditions used throughout this subsection.

\begin{assumption}\label{ass-stat-rank}
    Let the process $\{\mathbf{A}^t \}_{t \geq 0} \subset\{0, 1\}^{n \times n \times L}$ be defined in \Cref{def-sarmsb}. Let $ \mathbf{Q}  = \mathbf{W} + \mathbf{M} \in \R^{K \times K \times L}$.
\begin{enumerate}[$(i)$]
    \item Assume that $\mathrm{rank}( \mathcal{M}_1 (\mathbf{W}) )=\mathrm{rank}( \mathcal{M}_1 (\mathbf{M}) )=\mathrm{rank}( \mathcal{M}_1 (\mathbf{Q}) ) = K$. 
    Let $r_1 = \mathrm{rank}( \mathcal{M}_3 (\mathbf{W})) $ and $ r_2 =\mathrm{rank}( \mathcal{M}_3 (\mathbf{M})) $.
    \item   Assume that $ s_{\max} \leq C_{\sigma} s_{\min} $, where $C_{\sigma} > 0$ is an absolute constant, $s_{\min} = \min_{k \in [K]} s_k$ and $s_{\max} = \max_{k \in [K]} s_k$ with for any $k \in [K]$, $s_k = \sum_{i=1}^n Z_{i, k}$.
    \item Assume that $\max\{\|\mathbf{W}\| /\sigma_{\min}(\mathbf{W}),  \|\mathbf{M}\| /\sigma_{\min}(\mathbf{M}), \|\mathbf{Q}\| /\sigma_{\min}(\mathbf{Q})\} \leq C_{\sigma}$ for some  absolute constant $ C_{\sigma}>0$.
    Further assume that $\sigma_{\min}(\mathbf{Q}) \geq \min\{\sigma_{\min}(\mathbf{M}), \sigma_{\min}(\mathbf{W})\}$. 
\end{enumerate}
\end{assumption}

\Cref{ass-stat-rank}$(i)$ imposes rank constraints on the mode-$1$ matricizations of the connectivity tensors, commonly used in tensor-based formulations of multilayer networks \citep[e.g.,][]{jing2021community, wang2025multilayer}.   \Cref{ass-stat-rank}$(ii)$ is a standard community-balance condition, while \Cref{ass-stat-rank}$(iii)$ controls the condition numbers of the connectivity tensors.
When either of these conditions is violated, the estimation rate in \eqref{thm-tensor-Frobenius-1} deteriorates, with additional factors depending on $s_{\max}/s_{\min}$ and $\max\{\|\mathbf{W}\| /\sigma_{\min}(\mathbf{W}), \|\mathbf{M}\| /\sigma_{\min}(\mathbf{M}), \|\mathbf{Q}\| /\sigma_{\min}(\mathbf{Q})\}$, respectively.

We next show the estimation rate of the refined tensor-based estimators obtained in \textbf{Stage II} of \Cref{alg:fast-stat-msbm}.  

\begin{theorem}\label{thm-tensor-Frobenius-stat}
Let the process $\{\mathbf{A}^t \}_{t \geq 0} \subset\{0, 1\}^{n \times n \times L}$ be defined in \Cref{def-sarmsb}. Suppose that \Cref{ass-stat-rank} holds.  For any $t \in [T]$, under \Cref{alg:fast-stat-msbm}, with probability at least $1- (n \vee L \vee T)^{-c}$,
\begin{align}\label{thm-tensor-Frobenius-1}
 \| \widetilde{\boldsymbol{\Theta}}^t -\boldsymbol{\Theta} \|_{\mathrm{F}} +  \| \widetilde{\boldsymbol{\Delta}}^t -\boldsymbol{\Delta}\|_{\mathrm{F}}  \leq & C    \sigma^{-1}  \bigg( \frac{  K\sqrt{(K\vee r)L} \log^2(n \vee L \vee T)}{t}  + \frac{K Ln}{t^2} \bigg)   \nonumber\\
  & \hspace{0.5cm}  +   C    \bigg( \sqrt{\frac{ (nK + Lr  +K^2r)  \log(n \vee L \vee T)}{t}} + \frac{n\sqrt{L}}{t} \bigg), 
\end{align}
where $\sigma = \min\{ \sigma_{\min}(\mathbf{W}), \sigma_{\min}(\mathbf{M})\} $, $r = \max\{r_1, r_2\}$ and $C, c >0$ are absolute constants.    
\end{theorem}

When  $\sigma \gtrsim  \max \{\xi_1 \wedge \xi_2,  \xi_3 \wedge \xi_4\}$, 
where $\xi_1 =  K \sqrt{K \vee r}\log^2(n \vee L \vee T)/n$, $\xi_2 = K \sqrt{(K \vee r)L} $ $\log^{3/2}(n \vee L \vee T)/ \sqrt{t(nK + Lr  +K^2r)}$, $\xi_3 =   K \sqrt{L}$  and  $\xi_4 =   K Ln /( t\sqrt{t(nK + Lr  +K^2r)})$, the upper bound in \eqref{thm-tensor-Frobenius-1} simplifies, up to poly-logarithmic factors, to
\begin{equation}\label{eq-stat-tensor-final}
    \sqrt{\frac{nK + Lr + rK^2}{t}} + \frac{n\sqrt{L}}{t}.
\end{equation}
The first term in \eqref{eq-stat-tensor-final} reflects the stochastic fluctuation in estimating the low-rank tensor structure, while the second captures the bias inherited from the initial transition-probability MLEs. 
This decomposition is consistent with the standard autoregressive literature, where the model complexity $n$ is usually considered as small or even fixed, and the bias quickly becomes negligible when $t$ grows.  In the single-layer AR(1)-SBM \citep[Theorem 10 in][]{jiang2023autoregressive}, the error bound similarly separates into stochastic and bias components, with the former dominating when $t \gg n$.  An analogous phenomenon occurs here: whenever
\[
t \gg \min\bigg\{\frac{nL}{K},\ \frac{n^2}{r},\ \frac{n^2L}{K^2 r}\bigg\},
\]
the stochastic term dominates the bias term in \eqref{eq-stat-tensor-final}. This condition can be viewed as a natural extension of the single-layer result of \cite{jiang2023autoregressive} to the multilayer setting.

We would like to highlight that our results benefit from a refined tensor estimation. Prior to \textbf{Stage~II}, as shown in \Cref{lemma-sub-Gaussian} in Appendix \ref{sec-app-add-stat}, each entry of the initial estimators incurs an error of order $t^{-1/2}$ up to logarithmic factors. Aggregating over $O(n^2L)$ entries yields a Frobenius loss of order
\[
    \frac{n\sqrt{L}}{\sqrt{t}}.
\]
If one instead unfolds the tensor into a matrix and applies a matrix spectral refinement, this procedure fails to exploit the low-rank structure along the third (layer) mode. As a result, the stochastic error scales as
\[
\sqrt{\frac{nLK}{t}},
\]
up to logarithmic factors. In contrast, TH-PCA fully leverages the joint low-rank structure across both node and layer modes, effectively reducing the intrinsic dimensionality. This leads to a sharper stochastic error of order
\[
\sqrt{\frac{nK + Lr + K^2r}{t}},
\]
again up to logarithmic factors.

\subsubsection{Minimax lower bounds on the transition probability tensors estimation}

To assess the optimality of the estimator in \Cref{alg:fast-stat-msbm} for estimating the transition probability tensors $(\boldsymbol{\Theta}, \boldsymbol{\Delta})$, we derive minimax lower bounds for this problem.  The parameter space considered is defined as
\begin{align}\label{def-par-space-new}
\mathcal{P}= \Big\{ &  (\boldsymbol{\Theta}, \boldsymbol{\Delta}) 
\in (0, 1)^{n \times n \times L} \colon
\boldsymbol{\Theta}_{:, :, l} =  Z \mathbf{W}_{:, :, l} Z^{\top},
\boldsymbol{\Delta}_{:, :, l} =  Z \mathbf{M}_{:, :, l} Z^{\top},  \forall l \in [L], \mbox{ where }  Z \in \{0, 1\}^{n \times K}
\nonumber\\
& \mbox{ is a community membership matrix satisfying } \max_{j \in [K]} \sum_{i=1}^n Z_{i, j} \leq C_{\sigma} \min_{j \in [K]}\sum_{i=1}^n Z_{i, j}, \mbox{and }
\nonumber\\
& 
\mathbf{W}, \mathbf{M} \in [c_{\min}, 1 - c_{\min}]^{K \times K \times L} \mbox{ satisfy }
\max \big\{\mbox{rank}_3(\mathbf{W}),  \mbox{rank}_3(\mathbf{M}) \big\}\leq r\Big\}.
\end{align}

\begin{proposition}\label{thm-minimax-lower}
Let the sequence of adjacency tensors $\{\mathbf{A}^t\}_{t \in [T] \cup \{0\}}$ 
be generated according to \Cref{def-sarmsb} 
with parameters $(\boldsymbol{\Theta}, \boldsymbol{\Delta}) \in \mathcal{P}$,
where the parameter space $\mathcal{P}$ is defined in \eqref{def-par-space-new}. 
Then there exists an absolute constant $C> 0$ such that
\[
\inf_{(\widehat{\boldsymbol{\Theta}}, \widehat{\boldsymbol{\Delta}})}\sup_{(\boldsymbol{\Theta}, \boldsymbol{\Delta}) \in \mathcal{P}}
\mathbb{E}_{\boldsymbol{\Theta}, \boldsymbol{\Delta}}\Big\{\big\| \widehat{\boldsymbol{\Theta}} -\boldsymbol{\Theta}\big\|_{\mathrm{F}}^2 + \big\| \widehat{\boldsymbol{\Delta}} -\boldsymbol{\Delta}\big\|_{\mathrm{F}}^2 \Big\}
\geq  C \frac{ r K^2 + n\log(K) + Lr }{ T }. 
\]
where $r = \max \{r_1, r_2\}$.
\end{proposition}

As shown in \eqref{eq-tucker-decom}  in Appendix \ref{sec-app-add-stat}, the transition probability tensor admits the Tucker decomposition 
\begin{equation}\label{eq-lowe-tucker}
\boldsymbol{\Theta} = \mathbf{R} \times_1  U_Z \times_2 U_Z \times_3 U_W.
\end{equation}
The three terms in the lower bound correspond to distinct sources of statistical difficulty.
 The term $rK^2/T$ corresponds to estimating the core tensor $\mathbf{R} \in \mathbb{R}^{K \times K \times r_1}$ when the singular subspaces $(U_Z, U_W)$ for node communities and layer loadings are known.
The second term, $n\log(K) / T$, encodes the difficulty of recovering the unknown community membership matrix $Z \in \{0, 1\}^{n \times K}$. The third term, $Lr / T$, arises from estimating the layer-loading matrix $U_W \in \mathbb{R}^{L \times r}$.

To better interpret \Cref{thm-minimax-lower}, we compare it with existing minimax lower bounds for several related models, including the single-layer SBM, dynamic SBMs and tensor denoising. For ease of comparison, we rewrite these results using our notation. 
We begin with the single-layer SBM
\[
A_{i,j} \sim \mathrm{Bernoulli}\big(P_{i, j}\big), \quad P =Z BZ^{\top} \quad 1 \le i \leq  j \le n,
\]
where $Z \in \{0, 1\}^{n \times K}$ is the community membership matrix and $B \in [0,1]^{K \times K}$ is the connectivity matrix. \cite{gaorate-optimal2015} established that the minimax rate for estimating $P$ is
\[
  K^2 + n\log(K),
\]
where the first term corresponds to estimating the block probabilities and the second term captures the combinatorial complexity of recovering the community assignments.
We next consider dynamic single-layer SBMs:
\[
A^{t}_{i,j} \sim \mathrm{Bernoulli}\big(P^{t}_{i, j}\big),  \quad P =Z B^{t}Z^{\top}, \quad  \quad 1 \le i \leq  j \le n, t\in[T],
\]
where the community membership matrix $Z$ is fixed over time, while the connectivity matrices $B^t \in [0,1]^{K \times K}$ varies with $t$. In this setting, \cite{pensky2019dynamic} established that the minimax rate
\[
\inf_{\{\widehat P^t\}_{t \in [T]}} \sup_{\{P^t\}_{t \in [T]}}
\mathbb{E}\bigg\{ \sum_{t=1}^T \|\widehat P^t-P^t\|_{\mathrm{F}}^2 \bigg\}
\asymp
s\log(TK^2/s)+n\log(K),
\]
where $s$ quantifies the temporal complexity of $\{B^t\}_{t \in [T]}$ (with $s \asymp K^2$ in the time-invariant case and $s \asymp TK^2$ in the fully time-varying case).
This formulation involves estimating $T$ matrices $\{P^t\}_{t \in [T]}$, while in contrast, our setting focuses on estimating a single set of time-invariant transition parameters under temporal dependence, yielding a $1/T$-scaled rate. \Cref{thm-minimax-lower} complements these results by capturing both the autoregressive temporal dependence pattern and the multilayer structure.

For the tensor denoising problem, \cite{zhang2018tensor} consider the model
\[
\mathbf{Y} = \mathbf{R} \times_1 U_Z \times_2 U_Z \times_3 U_W + \mathbf{Z} 
\]
where $\mathbf{Y}, \mathbf{Z} \in \R^{n \times n \times L}$,   $\mathbf{R} \in \R^{K \times K \times r}$, $U_Z  \in \mathbb{O}_{n \times K}$, $U_W  \in \mathbb{O}_{n \times K}$ and $\{\mathbf{Z}_{i, j, l}\}_{i, j \in [n], l \in [L]} \stackrel{\text{i.i.d.}}{\sim} \mathcal{N}(0, 1)$.
They established a minimax rate of order
\[
nK + Lr.
\]
Our setting refines this rate by including an additional $K^2r$ term, representing the contribution of the core tensor in the Tucker decomposition. The difference between $nK$ and $n\log (K)$ reflecting the combinatorial complexity of recovering discrete community memberships rather than continuous low-rank factors.

Comparing \Cref{thm-minimax-lower} with the upper bound in \eqref{eq-stat-tensor-final}, we observe three sources of mismatch. First, the clustering term in the upper bound is $nK/T$, whereas the lower bound yields the sharper term $n\log (K)/T$. This difference corresponds to the computational-statistical gap in community detection: from an information-theoretic perspective, recovering the community assignments costs $n \log (K)$, while polynomial-time procedures typically exhibit a linear dependence on $K$ \citep{xu2018rates}. 

Second, the upper bound contains an additional term of order $n^2L/T^2$. This term is specific to the autoregressive setting, which arises from the bias of the maximum likelihood estimator for the transition parameters in the autoregressive setting. Such bias is well documented for autoregressive estimators \citep{shaman1988bias}.  In particular, whenever
\[
T \gtrsim \min\{nL/K,\ n^2/r,\ n^2L/(K^2r)\},
\]
the bias term $n^2L/T^2$ is dominated by the stochastic term $(rK^2+nK+Lr)/T$ and \Cref{alg:fast-msbm} attains the minimax rate up to the $K/\log K$ computational gap. Bias-corrected and median-unbiased estimators have been proposed in the literature \citep{andrews1993exactly,andrews1994approximately,patterson2000bias}; incorporating them into our framework is left for future work. 

Third, the upper bound is established under \Cref{ass-stat-rank}(iii), which requires the condition numbers of the transition probability tensors $\mathbf{W}$, $\mathbf{M}$, and $\mathbf{Q}$ to be bounded by an absolute constant, whereas the lower bound is proved over a larger parameter class without this restriction. This mismatch is mainly technical. 
The lower bound constructions are designed to ensure sufficient Frobenius separation under bounded-entry and low-rank constraints, but they do not preserve uniform conditioning. In particular, to obtain the $rK^2/T$ term, we use the Varshamov--Gilbert lemma to construct core tensors $R$ in \eqref{eq-lowe-tucker}, which guarantees large Hamming separation but does not control the singular spectrum of the associated matricizations.
As a result, the upper and lower bounds are derived over slightly different parameter classes. Nevertheless, \Cref{ass-stat-rank}(iii) can be relaxed in the upper bound analysis, at the cost of an additional multiplicative factor polynomial in the condition number in \Cref{thm-tensor-Frobenius-stat}.

\subsubsection{Theoretical guarantees for community recovery}

Having established estimation error bounds for the transition probability tensors $(\boldsymbol{\Theta}, \boldsymbol{\Delta})$, we now turn to their downstream statistical implications, namely the estimation of singular subspaces and, in particular,  community recovery.

\begin{proposition}\label{thm-comunity-recovery-stat}

Let the process $\{\mathbf{A}^t \}_{t \geq 0} \subset\{0, 1\}^{n \times n \times L}$ be defined in \Cref{def-sarmsb} and suppose that \Cref{ass-stat-rank} holds.  For  any $t \in [T]$, let 
\begin{align}\label{eq-sin-theta-stat-new}
\mathcal{E}_t =  C_\mathcal{E}  \sigma_Q^{-2} \bigg( \frac{ K^2 \sqrt{L} \log^2(n \vee L \vee T)}{tn}  + \frac{K^2L}{t^2}\bigg) +  C\sigma_Q^{-1}  \bigg(K\sqrt{\frac{  \log(n \vee L \vee T)}{tn}} + \frac{K\sqrt{L}}{t} \bigg), 
\end{align}
where $\sigma_Q = \sigma_{\min}(\mathbf{Q})$ and $C_\mathcal{E}  >0$ is an absolute constant.  
\begin{enumerate}
\item 
Define the normalized Hamming loss between the estimated and true community membership matrices by
\begin{equation}\label{eq-def-hamming}
\mathcal{L}(\widehat{Z}^t, Z)
=\frac{1}{n} \min_{\tau \in \Pi([K])}
\sum_{i=1}^n \sum_{k=1}^K  \mathbbm{1} \big\{ 
  \widehat{Z}^t_{i, \tau(k)} \neq Z_{i, k} \big\}.
\end{equation}
Then for any $t \in [T]$, if
\begin{equation}\label{eq-ass-singular-1}
    C K  \mathcal{E}_t^2 \leq 1,
\end{equation}
under \Cref{alg:fast-stat-msbm}, with probability at least $1- (n \vee L \vee T)^{-c}$, 
\begin{equation}\label{thm-com-1}
\mathcal{L}(\widehat{Z}^t, Z) \leq C' \mathcal{E}_t^2, 
\end{equation}
where  $C>0$ is a sufficiently large constant and $C', c>0$ are absolute constants. 

\item For any $t \in [T]$, if
\begin{equation}\label{eq-ass-singular}
     C n  \mathcal{E}_t^2   \leq  1,
\end{equation}
then under \Cref{alg:fast-stat-msbm},  with probability at least $1- (n \vee L \vee T)^{-c}$,   
\begin{equation}\label{thm-com-2}
\widehat{Z}_i^t = \widehat{Z}_j^t  \quad \mbox{if and only if} \quad Z_i = Z_j, \quad \forall i, j \in [n], 
\end{equation}
where $C>0$ is a sufficiently large constant and $c>0$ is an absolute constant.
\end{enumerate}
\end{proposition}

\Cref{thm-comunity-recovery-stat} establishes two types of statistical consistency under distinct conditions. The first yields Hamming-consistent clustering under $K\mathcal{E}_t^2 \lesssim 1$, while the second guarantees exact recovery (up to label permutation) under the stronger condition $n\mathcal{E}_t^2 \lesssim 1$. 

 We note that $\mathcal{E}_t$ characterizes the estimation error of the subspace estimator $\widehat{U}_Z^t$; see \Cref{prop-sin-theta-stat} in Appendix \ref{sec-app-add-stat}. When $K = O(1)$  and $\sigma_Q \asymp \sqrt{L}$, $\mathcal{E}_t$ simplifies (up to logarithmic factors) to
\begin{equation}\label{eq-stat-rate}
    \frac{1}{\sqrt{tnL}} + \frac{1}{t}.
\end{equation}
The first term reflects stochastic variability, while the second captures the bias from the autoregressive MLE.
Consequently, Hamming consistency holds whenever
\[
\frac{1}{tLn} + \frac{1}{t^2} \ll 1,
\]
with the corresponding Hamming error rate of order
\begin{equation}\label{eq-stat-rate-hamming}
\frac{1}{tnL} + \frac{1}{t^2}.
\end{equation}
Exact recovery is guaranteed under the stronger condition
\begin{equation}\label{thm-com-1-K}
\frac{1}{tL} + \frac{n}{t^2} \ll 1,
\end{equation}
which is strictly stronger than that for Hamming consistency.

A closely related line of work is \cite{jiang2023autoregressive}, which studies a single-layer AR(1)-SBM and establishes only an exact-recovery guarantee,  without providing a separate Hamming-loss rate. Translated into our notation (with $K = O(1)$), their result requires
\[
 \frac{1}{t} + \frac{n}{t^2} + \frac{1}{n} \ll 1,
\]
which is strictly stronger than \eqref{thm-com-1-K}. 
Our result improves upon this in two aspects. First, the multilayer structure  reduces the stochastic term from $1/t$ to $1/(tL)$  by leveraging information across $L$ layers. Second, our estimator avoids the additional $1/n$ bias term, since it performs spectral estimation directly on the initial estimators rather than on a normalized Laplacian matrix as in  \cite{jiang2023autoregressive}, thereby eliminating the bias induced by diagonal normalization.

More broadly, community detection in dynamic SBMs has been extensively studied under independent or conditionally independent network snapshots with evolving structures. 
For example, \cite{pensky2019spectral} consider H\"older-continuous connection probabilities in time (exponent $\beta>0$) with at most $s$ nodes switching communities between consecutive times. They establish a Hamming error rate which, in our notation and under $K=O(1)$, takes the form
\[
   \frac{1}{n} \Big(\frac{n}{t^{2\beta}}\Big)^{\frac{1}{2\beta +1}} + \sqrt{s}.
\]
In the static case without temporal variation, this reduces to $1/(tn)$. Compared with \eqref{eq-stat-rate-hamming}, our result improves the stochastic term to $1/(tnL)$ by incorporating multilayer structure, while introducing an additional bias term $1/t^2$ due to the autoregressive dependence. Thus, our result is complementary to this literature, as it explicitly exploits both temporal dependence and multilayer structure.

\section{Transition tensor estimation and community recovery under non-stationarity}\label{sec-nonstat}

Building on the stationary analysis in \Cref{sec-stat}, we now consider a non-stationary $\mbox{AR}(1)\mbox{–MSBM}$, where the edge transition probability tensors evolve over time. This setting introduces additional challenges due to temporal variation, requiring methods that adapt to both gradual drifts and abrupt changes.
To address this, we develop a memory- and compute-efficient online algorithm with adaptive windowing in \Cref{section-nonstat-alg}, and establish its theoretical guarantees in \Cref{section-nonstat-theory}.

We begin by extending the stationary $\mbox{AR}(1)\mbox{-MSBM}$  in \Cref{def-sarmsb} to a non-stationary setting in which the transition probability tensors vary over time.

\begin{definition}[Non-stationary first-order autoregressive multilayer stochastic block models, non-stationary $\mbox{AR}(1)\mbox{-MSBM}$]\label{def-armsb}
A non-stationary $\mbox{AR}(1)\mbox{-MSBM}$
is a adjacency tensor sequence $\{\mathbf{A}^t \}_{t \geq 0}$ $\subset\{0, 1\}^{n \times n \times L}$ defined recursively for $t\geq 1$ by
\[
\mathbf{A}^t_{i,j, l} = \mathbf{A}^{t-1}_{i,j,l} \mathbbm{1}\{\mathbf{E}^t_{i,j, l} = 0\} + \mathbbm{1}\{\mathbf{E}^t_{i,j, l} = 1\}, \quad \forall  1 \leq i \leq  j \leq n, \, l \in [L],
\]
where $\{\mathbf{E}^t_{i,j, l}\}_{1 \leq i \leq  j \leq n, l \in [L], t \geq 1} \subset \{-1, 0, 1\}$ are mutually independent random variables with 
\[
\P \{ \mathbf{E}^t_{i,j, l}  = 1 \} = \boldsymbol{\Theta}^t_{i, j, l} = Z_i \mathbf{W}^t_{:, :, l} Z_j,  \quad \P \{ \mathbf{E}^t_{i,j, l}  = -1 \}  =  \boldsymbol{\Delta}^t_{i, j, l} = Z_i \mathbf{M}^t_{:, :, l} Z_j
\]
and
\[
\P \{ \mathbf{E}^t_{i,j, l}  = 0  \}  = 1 - \boldsymbol{\Theta}_{i, j, l}^t  - \boldsymbol{\Delta}_{i, j, l}^t.
\]
The initial adjacency tensor $\mathbf{A}^0 \in \{0, 1\}^{n \times n \times L}$ is defined so that
\[
 c_{\min} \leq \mathbb{E}\{\mathbf{A}_{i,j,l}^{0}\}  \leq 1-c_{\min} \quad \forall  1 \leq i \leq  j \leq n, \, l \in [L].
\]
Here, $Z \in \{0,1\}^{n \times K}$ denotes the community membership matrix, while $\{\mathbf{W}^t\}_{t \geq 1}, \{\mathbf{M}^t\}_{t \geq 1} \subset [c_{\min},1 - c_{\min}]^{K \times K \times L}$ are time-varying connectivity tensors, with an absolute constant $ c_{\min} \in (0, 1/2)$. 
\end{definition}

In \Cref{def-armsb}, the temporal variation in $\boldsymbol{\Theta}^t$ and $\boldsymbol{\Delta}^t$ arises solely from the time-varying connectivity tensors $\mathbf{W}^t$ and $\mathbf{M}^t$, while the community memberships remain fixed.  This framework can be extended to accommodate time-varying community memberships; see Appendix \ref{app:latent} for details.

\subsection{An online adaptive learning algorithm}\label{section-nonstat-alg}

\begin{algorithm}
\caption{Adaptive windowed online estimation for non-stationary $\mbox{AR}(1)\mbox{-MSBM}$}
\label{alg:fast-msbm}
\begin{algorithmic}
\INPUT{Ranks $K$, $\{r^{t, k}\}_{k \in [t], t \geq 1}$, $\{\tilde{r}^{t, k}\}_{k \in [t], t \geq 1}$, dynamic grid map $\{\mathcal {G}^{(t)}\}_{t\geq 1}$ and  tolerance level $\{\tau(t)\}_{t\geq 1}$.}
\Initialise{$m \leftarrow 0$, $\mathbf{N}^{0, 0, 1}_{i,j,l} \leftarrow 0$, $\mathbf{N}^{0, 1, 0}_{i,j,l} \leftarrow 0$, $\mathbf{N}^{0}_{i,j,l}  \leftarrow 0$}
\ForAll{$t \in [T]$}
\Statex \vspace{2pt}
\textbf{Stage I: Initial estimator update}

\Statex \vspace{2pt}
\textbf{Stage I.1: Cumulative transition count updates and memory management}
        \State{$\mathbf{N}^{t,0,1} \leftarrow\mathbf{N}^{t-1,0,1}+\mathbf{A}^{t}(1-\mathbf{A}^{t-1}), \quad 
        \mathbf{N}^{t,1,0}\leftarrow \mathbf{N}^{t-1,1,0}+(1-\mathbf{A}^{t})\mathbf{A}^{t-1}$}
        \State{$\mathbf{N}^{t} \leftarrow \mathbf{N}^{t-1} + \mathbf{A}^{t-1}$}
 
    \State{Discard all $\{\mathbf{N}^{s,0,1}, \mathbf{N}^{s,1,0}, \mathbf{N}^{s}\}$ for $s \notin \{t-k \colon k\in\mathcal{G}^{(t)}\}$}
    \State{Store $\{\mathbf{N}^{t,0,1}, \mathbf{N}^{t,1,0}, \mathbf{N}^{t}\}$}
    \State{Let $\mathcal{G}^{(t)}= \big\{k_{(1)}, k_{(2)}, \cdots, k_{(h)}\big\}$ with $h=|\mathcal{G}^{(t)}|$ and $k_{(1)}<k_{(2)}<\cdots<k_{(h)}$.} 
    
\Statex \vspace{2pt}
\textbf{Stage I.2: Windowed MLE computation}
    \ForAll{$k\in\mathcal{G}^{(t)}$}
        \State{$\widehat{\boldsymbol{\Theta}}^{t, k}=    (\mathbf{N}^{t, 0, 1}- \mathbf{N}^{t - k, 0, 1})/ (k -(\mathbf{N}^{t}- \mathbf{N}^{t - k}) )$}
     \State{$\widehat{\boldsymbol{\Delta}}^{t, k} = 
(\mathbf{N}^{t, 1, 0}- \mathbf{N}^{t - k, 1, 0})/(\mathbf{N}^{t}- \mathbf{N}^{t - k} )
$}
     \EndFor
\Statex \vspace{2pt}

\textbf{Stage I.3: Adaptive window selection} 
\State{$q \leftarrow 1$, $\textbf{brk}\leftarrow 0$}    
\While{$\textbf{brk}=0$ and $q\leq h$}
\If{$ \big\|f^{t,k_{(u)}}\big(\widehat{\boldsymbol{\Theta}}^{t,k_{(q)}}, \widehat{\boldsymbol{\Delta}}^{t,k_{(q)}} \big) - f^{t,k_{(u)}}\big(\widehat{\boldsymbol{\Theta}}^{t,k_{(u)}}, \widehat{\boldsymbol{\Delta}}^{t,k_{(u)}}\big)\big\|_{\infty} \leq \tau(k_{(u)}),  \forall u \in [q]$} 
     \State{$q\leftarrow q+1$}
     \Else
          \State{$\textbf{brk}\leftarrow1$}
      \EndIf
\EndWhile
\State {$ \hat{k}_t \leftarrow
       \begin{cases}
         k_{(q-1)}, & \text{if }\textbf{brk}=1,\\
         k_{(h)},   & \text{if }\textbf{brk}=0.
       \end{cases}$}
\EndFor
\Statex \vspace{2pt}
\textbf{Stage II: Low-rank refinement via tensor-based methods}       \State{$\widehat{U}^{t}_Z \leftarrow $H-PCA$(\mathcal{M}_1 (\widehat{\boldsymbol{\Theta}}^{t, \hat k_t } + \widehat{\boldsymbol{\Delta}}^{t,\hat k_t }) \mathcal{M}_1 (\widehat{\boldsymbol{\Theta}}^{t, \hat k_t } + \widehat{\boldsymbol{\Delta}}^{t, \hat k_t})^{\top}, K)$}
   \State{$\widehat{U}^{t}_W \leftarrow $H-PCA$\big(\mathcal{M}_3 (\widehat{\boldsymbol{\Theta}}^{t, \hat k_t } ) \mathcal{M}_3 (\widehat{\boldsymbol{\Theta}}^{t, \hat k_t })^{\top}, r^{t, \hat{k}_t} \big)$}
   \State{$\widehat{U}^{t}_M \leftarrow $H-PCA$\big(\mathcal{M}_3 (\widehat{\boldsymbol{\Delta}}^{t, \hat{k}_t} ) \mathcal{M}_3 (\widehat{\boldsymbol{\Delta}}^{t, \hat{k}_t})^{\top}, \tilde{r}^{t, \hat{k}_t}\big)$}
\State{$\widetilde{\boldsymbol{\Theta}}^{t, \hat{k}_t} = \widehat{\boldsymbol{\Theta}}^{t, \hat{k}_t} \times_1 \widehat{U}^{t}_Z (\widehat{U}^{t}_Z)^{\top}\times_2 \widehat{U}^t_Z(\widehat{U}^{t}_Z)^\top  \times_3 \widehat{U}^t_W (\widehat{U}^{t}_W)^\top$} \State{$\widetilde{\boldsymbol{\Delta}}^{t, \hat{k}_t} = \widehat{\boldsymbol{\Delta}}^{t, \hat{k}_t} \times_1 \widehat{U}^t_Z  (\widehat{U}^{t}_Z)^{\top} \times_2 \widehat{U}^t_Z(\widehat{U}^{t}_Z)^{\top}\times_3  \widehat{U}_M^t(\widehat{U}^{t}_M)^\top$}
\Statex \vspace{2pt}
\textbf{Stage III: Community recovery} 
\State{$ \widetilde{U}^t 
= \argmin_{U \in \mathcal{M}_{n,K}} \|\widehat{U}_Z^{t} - U\|_{\mathrm{F}}^2. 
$}
\State{Let $\{\tilde{u}^{(t, k)}\}_{k=1}^{K}$ denote the distinct row patterns of $\widetilde{U}_Z^t$. }
\State{$\widehat{Z}^t_{i, j} =  \mathbbm{1} \big\{\widetilde{U}_i^t = \tilde{u}^{(t, j)}\big\}, \quad \forall i \in [n], j \in [K]$
}
\OUTPUT{ 
$\{(\widetilde{\boldsymbol{\Theta}}^{t},
  \widetilde{\boldsymbol{\Delta}}^{t})\}_{t\in [T]}$ and 
$\{\widehat{Z}^{t}\}_{t\in [T]}$}
\end{algorithmic}
\end{algorithm}

We present an adaptive online algorithm in \Cref{alg:fast-msbm} for estimation under non-stationarity. Similar to the stationary procedure in \Cref{alg:fast-stat-msbm}, the algorithm consists of three stages: constructing initial estimators, refining them via tensor-based spectral methods and recovering communities via subspace clustering.
The key difference lies in the construction of the initial estimators, which must adapt to temporal variation. To this end, we employ an adaptive windowing scheme that balances bias and variance by selecting a suitable look-back window based on local stability.

Specifically, \textbf{Stage I} consists of three components: (i) maintaining cumulative transition counts on a dynamic geometric grid for efficiency, (ii) computing windowed maximum likelihood estimators over candidate windows, and (iii) selecting the window adaptively via a stability criterion. We describe \textbf{Stage I} in detail below, while \textbf{Stages II} and \textbf{III} proceed as their counterparts in the stationary case in \Cref{alg:fast-stat-msbm}.

\medskip
\noindent
\textbf{Cumulative transition count updates and memory management.}
At any time $t \in [T]$ and for each look-back window $k \in [t]$, the sufficient statistics consist of
the cumulative number of $0 \to 1$ transitions, $1 \to 0$ transitions and active edges between times $t-k$ and $t$, namely
\[
\mathbf{N}^{(t, k),0,1} = \mathbf{N}^{t,0,1} - \mathbf{N}^{t-k,0,1}_{i,j,l}, \quad \mathbf{N}^{(t, k),1,0} = \mathbf{N}^{t,1,0}- \mathbf{N}^{t-k,1,0}, \quad \mathbf{N}^{(t, k)} = \mathbf{N}^{t} - \mathbf{N}^{t-k,1,0}.
\]
To evaluate all possible window sizes $k \in [t]$, one would need to store the entire collection of sufficient statistics 
\[
\{ ( \mathbf{N}^{t-k,0,1},  \mathbf{N}^{t-k,1,0}, \mathbf{N}^{t-k}) \}_{ k  \in [t]} \cup \{ ( \mathbf{N}^{t,0,1},  \mathbf{N}^{t,1,0}, \mathbf{N}^{t}) \},
\]
which incurs a storage cost of order $O(tn^2L)$, and computing all 
\[
\{(\mathbf{N}^{(t, k),0,1},  \mathbf{N}^{(t, k),1,0}, \mathbf{N}^{(t, k)}) \}_{k \in [t]},
\]
incurs a computational cost of order $O(tn^2L)$. 

To substantially reduce both memory usage and computational cost, we restrict attention to a dynamic geometric grid $\mathcal{G}^{(t)}$ that retains only a logarithmic subset of candidate window sizes.
This idea, introduced by \cite{moen2025general}, defines  $\mathcal{G}^{(t)}$ for $t \geq 2$ as
\begin{equation}\label{def-dynamic-grid} 
\mathcal{G}^{(t)} = \{1\} \cup \bigcup_{j=1}^{\lfloor \log_2\{(t-1)/3\} \rfloor + 1} \big\{ \mathcal{G}^{(t)}_{\mathrm{L},j} \big\} \cup \bigcup_{j=1}^{\lfloor \log_2(t-1) \rfloor - 1} \big\{ \mathcal{G}^{(t)}_{\mathrm{R},j} \big\},
\end{equation}
where 
\[
\mathcal{G}^{(t)}_{\mathrm{L},j} =  2^j +  \{(t-1) \bmod 2^{j-1}\} \quad \mbox{and} \quad \mathcal{G}^{(t)}_{\mathrm{R},j}  =   \mathcal{G}^{(t)}_{\text{L},j} + 2^{j-1}.
\]
For each $t\geq 2$, we store
\[
 \{ (\mathbf{N}^{t-k, 0, 1} , \mathbf{N}^{t-k, 1, 0}, \mathbf{N}^{t-k} ) \}_{ k \in  \mathcal{G}^{(t)}} \cup 
 \{ (\mathbf{N}^{t, 0, 1} , \mathbf{N}^{t, 1, 0}, \mathbf{N}^{t})\},
\]
requiring only $O(\log_2(t)n^2L)$ memory.

When a new observation $\mathbf{A}^{t+1}$ arrives,the following inclusion holds:
\[
\{t+1-k: k\in\mathcal{G}^{(t+1)}\} \subseteq \{t-k \colon k \in \mathcal{G}^{(t)}\} \cup\{t\}, 
\]
which implies that updating the sufficient statistics requires computing only
\[
  ( \mathbf{N}^{t+1, 0, 1}, \mathbf{N}^{t+1, 1, 0}, \mathbf{N}^{t+1}),
\]
at a computational cost of order $O(n^2L)$. The updated collection is then utilized to compute 
\[
   \{ (\mathbf{N}^{(t+1, k), 0, 1}, \mathbf{N}^{(t+1, k), 1, 0}, \mathbf{N}^{(t+1, k)})\}_{k \in \mathcal{G}^{(t+1)}},
\]
which incurs only $O(\log(t)n^2L)$ additional computational cost.

\medskip
\noindent
\textbf{Windowed MLE computation.}
At time $t$ and for a look-back window $k \in [t]$, conditioning on the state $\mathbf{A}^{t-k}$, the joint likelihood of the segment
$\{\mathbf{A}^s\}_{s \in [t] \backslash[t-k] }$ factorizes over time and edges as
\[
       \prod_{1 \leq i \leq  j \leq n, l \in [L]} \prod_{s=t-k+1}^{t}   \boldsymbol{\Theta}_{i, j, l}^{\mathbf{A}^s_{i,j, l} (1-\mathbf{A}^{s-1}_{i,j, l})} (1- \boldsymbol{\Theta}_{i, j, l})^{(1-\mathbf{A}^s_{i,j, l})(1-\mathbf{A}^{s-1}_{i,j, l})} \boldsymbol{\Delta}_{i, j, l}^{(1-\mathbf{A}^s_{i,j, l}) \mathbf{A}^{s-1}_{i,j, l}} (1-\boldsymbol{\Delta}_{i, j, l})^{\mathbf{A}^s_{i,j, l} \mathbf{A}^{s-1}_{i,j, l}}.
\]
Motivated by this factorization, we introduce the averaged negative log-likelihood tensor-valued function
\[
f^{t,k} \colon [0, 1]^{n \times n \times L} \times [0, 1]^{n \times n \times L} \to  \R^{n \times n \times L}
\]
which maps a candidate pair $(\boldsymbol{\Theta}', \boldsymbol{\Delta}' ) $ to a tensor whose $(i,j,l)$-th entry is 
\begin{align}
& \hspace{1cm} f^{t,k}_{i, j, l} (\boldsymbol{\Theta}'_{i, j, l}, \boldsymbol{\Delta}'_{i, j, l})
 = -\frac{1}{k} \sum_{u=t-k+1}^{t}  \Big\{ \mathbf{A}^u_{i, j, l} (1-\mathbf{A}^{u-1}_{i, j, l}) \log(\boldsymbol{\Theta}'_{i, j, l})  \nonumber  \\
 &  + (1-\mathbf{A}^u_{i, j, l}) (1-\mathbf{A}^{u-1}_{i, j, l}) \log (1-\boldsymbol{\Theta}'_{i, j, l}) + (1- \mathbf{A}^u_{i, j, l}) \mathbf{A}^{u-1}_{i, j, l}\log(\boldsymbol{\Delta}'_{i, j, l})  + \mathbf{A}^u_{i, j, l} \mathbf{A}^{u-1}_{i, j, l}\log (1-\boldsymbol{\Delta}'_{i, j, l})  \Big\}. \nonumber
\end{align}
The windowed MLEs are defined as minimizers
\[
(\widehat{\boldsymbol{\Theta}}^{t, k}, \widehat{\boldsymbol{\Delta}}^{t, k} )
=\argmin_{(\boldsymbol{\Theta}, \boldsymbol{\Delta}) \in [0, 1]^{n \times n \times L} \times [0, 1]^{n \times n \times L} } f^{t, k} (\boldsymbol{\Theta}, \boldsymbol{\Delta}),
\]
which admit the closed forms
\begin{align}
\widehat{\boldsymbol{\Theta}}^{t, k}
   =     \frac{\sum_{u=t-k+1}^{t} \mathbf{A}^u (1-\mathbf{A}^{u-1})}{\sum_{u=t-k+1}^{t} (1-\mathbf{A}^{u-1})} \quad \mbox{and} \quad
\widehat{\boldsymbol{\Delta}}^{t, k} = \frac{\sum_{u=t-k+1}^{t} (1- \mathbf{A}^u) \mathbf{A}^{u-1}}{\sum_{u=t-k+1}^{t}  \mathbf{A}^{u-1}}. \nonumber
\end{align}
In terms of cumulative transition counts, equivalently, we have that
\[
\widehat{\boldsymbol{\Theta}}^{t, k}
   =     
   \frac{\mathbf{N}^{t, 0, 1}- \mathbf{N}^{t - k, 0, 1}}{ k -(\mathbf{N}^{t}- \mathbf{N}^{t - k}) } \quad \mbox{and} \quad
\widehat{\boldsymbol{\Delta}}^{t, k} =   
   \frac{\mathbf{N}^{t, 1, 0}- \mathbf{N}^{t - k, 1, 0}}{\mathbf{N}^{t}- \mathbf{N}^{t - k}}. 
\]

\medskip
\noindent
\textbf{Adaptive window selection.} Motivated by \cite{huang2023stability}, we adopt the following window selection procedure. At any time $t \in [T]$, the candidate look-back windows are $k_{(1)} < k_{(2)} < \cdots < k_{(h)}$ taken from the dynamic grid $\mathcal{G}^{(t)}$. For each candidate $k_{(q)}$,
we compare its windowed MLEs $(\widehat{\boldsymbol{\Theta}}^{t,k_{(q)}}, \widehat{\boldsymbol{\Delta}}^{t,k_{(q)}})$ against all shorter windows $(\widehat{\boldsymbol{\Theta}}^{t,k_{(u)}}, \widehat{\boldsymbol{\Delta}}^{t,k_{(u)}})$, $u \in [q]$, using the averaged loss $f^{t,k_{(u)}}(\cdot,\cdot)$. We accept $k_{(q)}$ if the worst per-edge deterioration in loss satisfies
\[ \big\|f^{t,k_{(u)}}\big(\widehat{\boldsymbol{\Theta}}^{t,k_{(q)}}, \widehat{\boldsymbol{\Delta}}^{t,k_{(q)}} \big) - f^{t,k_{(u)}}\big(\widehat{\boldsymbol{\Theta}}^{t,k_{(u)}}, \widehat{\boldsymbol{\Delta}}^{t,k_{(u)}}\big)\big\|_{\infty} \leq \tau(k_{(u)}),  \quad \forall u \in [q].
\]
If the condition fails, we stop and set
$\hat{k}_t = k_{q-1}$, the largest window that still satisfies all tolerance constraints; if it never fails, we set  $\hat{k}_t = k_{(h)}$.
An illustration of the selection procedure for $t=8$ and $\mathcal{G}^{(t)}=\{1,2,3,5\}$ is shown in \Cref{Fig-alg}.

\begin{figure}[t] 
        \centering
        \includegraphics[width=0.85 \textwidth]{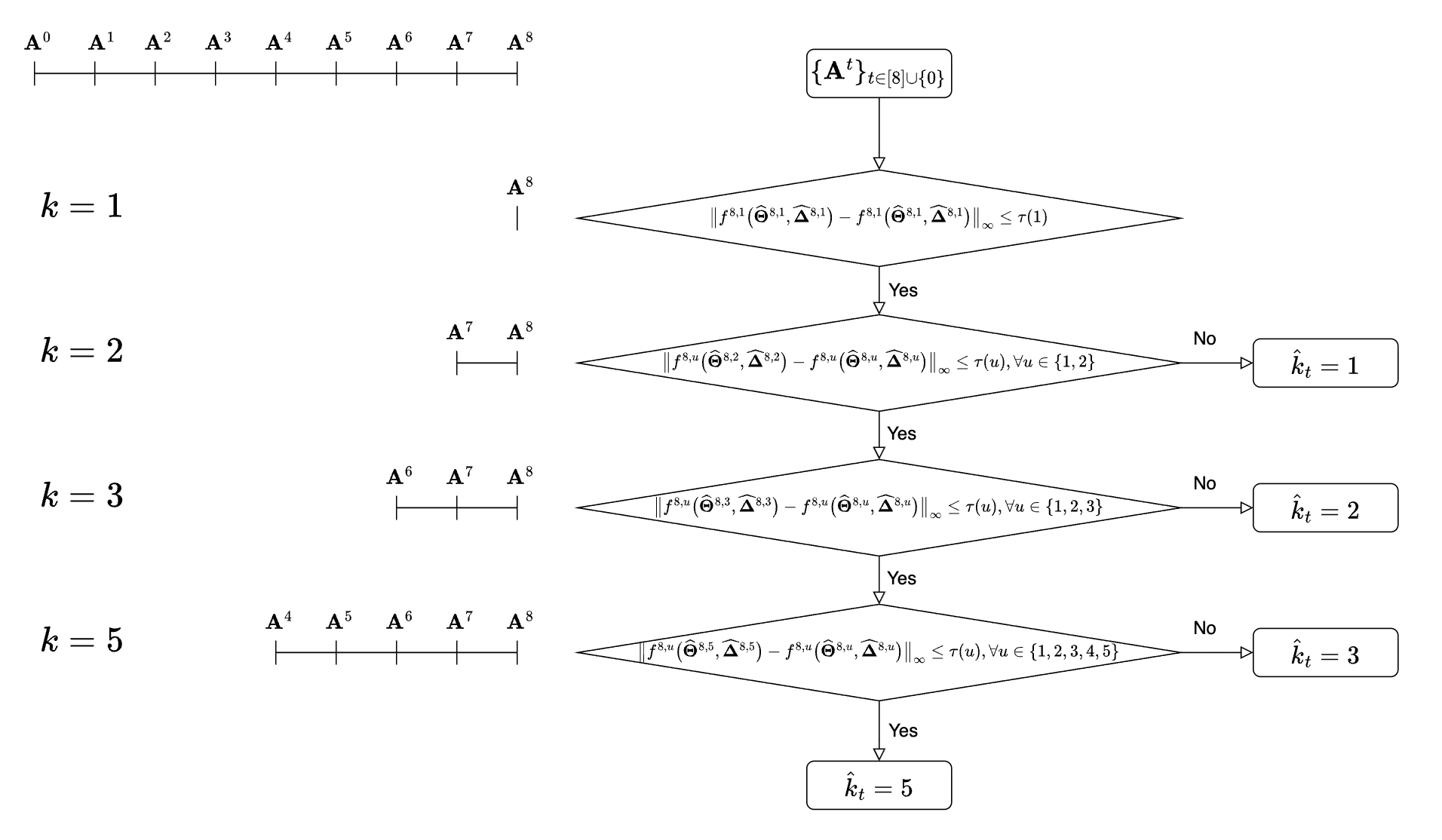}
        \caption{Illustration of the adaptive window selection procedure (\textbf{Stage I.3} in \Cref{alg:fast-msbm}) at time $t=8$, where the dynamic grid of candidate windows is $\mathcal{G}^{(t)}=\{1,2,3,5\}$.}
    \label{Fig-alg}
\end{figure}

\subsection{Theoretical guarantees of Algorithm \ref{alg:fast-msbm}}\label{section-nonstat-theory}

To establish the statistical guarantees of Algorithm \ref{alg:fast-msbm} under non-stationarity, the main challenge is that the data may exhibit both abrupt structural changes and smaller local drifts.  
A naive approach is to classify each local variation as either large or small according to a fixed threshold. 
However, this is insufficient, as it ignores cumulative effects: individually small variations may accumulate into a substantial shift that a local thresholding rule fails to detect.  

We instead adopt a segmentation view and introduce the notion of quasi-stationary segments. 
In this framework, both substantial changes and accumulated small variations determine segment boundaries, while smaller fluctuations are allowed within each segment. 
This is formalized in \Cref{def-seg}, based on the concepts of dominant and local dominant monotonicity defined in \Cref{def-local-dom}.

\begin{figure}[t] 
        \centering
        \includegraphics[width=1 \textwidth]{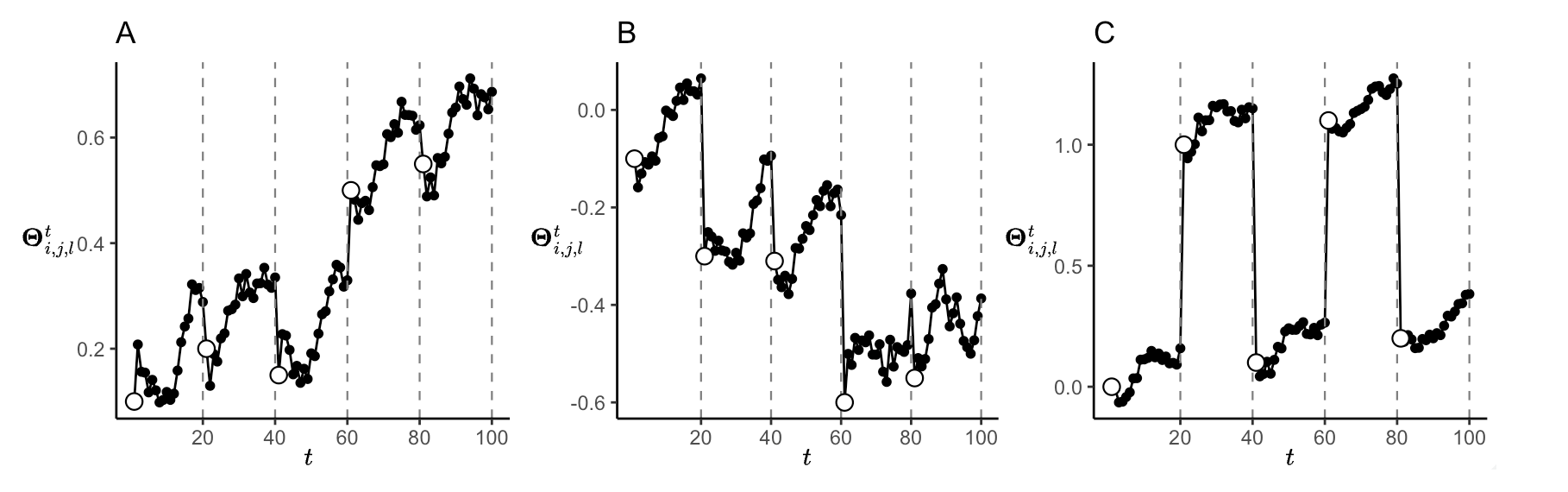}
        \caption{Examples of the sequence $\{\boldsymbol{\Theta}_{i, j, l}^{t}\}_{t \in [T]}$ with $T =100$, $V(t) = t^{-1/2}$, $G = 5$,  $\{T_1, T_2, T_3, T_4\} = \{20, 40, 60, 80\}$. Circles  mark the segment starting values$\{\boldsymbol{\Theta}_{i, j, l}^{ T_{{g-1}+1}}\}_{g \in [G]}$. }
        \label{Fig-locally}
\end{figure}

\begin{definition}[Dominantly monotone and locally dominantly monotone sequences]\label{def-local-dom}
Let $\{u_t\}_{t=1}^p \subset \mathbb{R}$ be a real-valued sequence. For any $1 \leq s < e \leq p$, let
\[
p_{s,e} = \sum_{t=s}^{e-1} (u_{t+1}-u_t)_+ \quad \mbox{and} \quad 
n_{s,e} = \sum_{t=s}^{e-1} (u_{t+1}-u_t)_-,
\]
where $(u)_+ = \max\{u,0\}$ and $(u)_- = \max\{-u,0\}$.
The sequence is dominantly monotone if there exists a constant $C_{\alpha} \in [0,1)$ such that
\[
n_{1,p} \leq C_{\alpha} p_{1,p}
\quad \mbox{or} \quad  p_{1,p} \leq C_{\alpha} n_{1,p}.
\]
It is locally dominantly monotone if there exists a constant $C_{\alpha} \in \big[0,c_{\min} (1-c_{\min}^{-1}) \big)$ such that for every $1 \leq s < e \leq p$, 
\[
n_{s,e} \leq C_{\alpha} p_{s,e}
\quad \mbox{or} \quad  p_{s,e} \leq C_{\alpha} n_{s,e}.  
\]
\end{definition}

\begin{definition}[Segmentation]\label{def-seg}
Let the process $\{\mathbf{A}^t \}_{t \geq 0} \subset\{0, 1\}^{n \times n \times L}$ be defined in \Cref{def-armsb}. The sequence $\{(\boldsymbol{\Theta}^t, \boldsymbol{\Delta}^t)\}_{t=1}^{T}$ is said to have $G$ $V$-quasi-stationary segments for some $G \in \Z^{+}$ and non-increasing function $V\colon \R^+ \to \R^+$ if  there exist time points $0 = T_0 < T_1 < \cdots < T_G = T$ such that 
\[
   \max_{s, u \in [T_g] \backslash [T_{g-1}]} \left\{\| \boldsymbol{\Theta}^{s} - \boldsymbol{\Theta}^{u} \|_{\infty}+  \| \boldsymbol{\Delta}^{s} - \boldsymbol{\Delta}^{u} \|_{\infty} \right\}
\leq V(T_g - T_{g-1}), \quad \forall g \in [G],
\]
the sequences of segment-starting values $\{\boldsymbol{\Theta}_{i, j, l}^{T_{g-1}+1}\}_{g \in [G]}$ and  $\{\boldsymbol{\Delta}_{i, j, l}^{T_{g-1}+1}\}_{g \in [G]}$ are locally dominantly monotone; and within each segment, the trajectories $\{\boldsymbol{\Theta}_{i, j, l}^{t}\}_{t \in [T_g]\backslash[T_{g-1}]}$ and $\{\boldsymbol{\Theta}_{i, j, l}^{t}\}_{t \in [T_g]\backslash[T_{g-1}]}$ are dominantly monotone, all as defined in \Cref{def-local-dom}.
\end{definition}

\Cref{def-seg} introduces a flexible framework through a general drift function $V(t)$, which can be seen as a generalization of the fixed rate $t^{-1/2}$ used in \cite{huang2023stability}. The function $V(t)$ regulates the maximal within-segment variation: abrupt or accumulated changes are captured by the segmentation points $T_1,\ldots,T_{G-1}$, while smaller fluctuations are allowed within each segment. This framework accommodates piecewise-stationary trajectories, smoothly drifting trajectories and mixtures of these two. 

To further characterize the temporal structure, \Cref{def-seg} incorporates dominant and local dominant monotonicity (see \Cref{def-local-dom}). This is to say, we allow that within and across segments, the evolution is predominantly in one direction, preventing cancellation between positive and negative variations.  
This condition is illustrated in \Cref{Fig-locally}.
Panels A and B satisfy the monotonicity conditions both across and within segments, whereas Panel~C does not. In Panel~C, although each segment exhibits an upward trend, the increments cancel across segments, so no dominant direction emerges.
This condition is crucial for controlling the bias of the windowed MLEs in \textbf{Stage I} of \Cref{alg:fast-msbm}.  In the stationary setting, the entrywise MLE bias is of order $t^{-1}$.  
 Without monotonicity, the within-segment bias can be as large as
\[
(t - T_{g-1})^{-1/2} + V(t-T_{g-1}) 
\]
whereas under \Cref{def-seg} it improves to 
\[
(t - T_{g-1})^{-1} + \{ V(t-T_{g-1}) \}^2,
\]
as shown in \Cref{prop-bias-final} in Appendix \ref{app-sec-non-proof}. Here $V(t-T_{g-1})$ captures local fluctuations and in the regime of interest, it is small, so squaring it yields a genuine improvement. 

In the stationary case, $(\boldsymbol{\Theta}^t, \boldsymbol{\Delta}^t)$ are time-invariant and thus $G=1$.  At the other extreme, if the sequence is highly non-stationary, one may take $T_j = j$, yielding $G = T$. Therefore, for a fixed~$V$, a larger value of $G$ corresponds to a higher degree of non-stationarity.

\begin{assumption}\label{ass-bias}
Let the process $\{\mathbf{A}^t \}_{t=1}^T \subset\{0, 1\}^{n \times n \times L}$ be defined in \Cref{def-armsb} and let $G$ be the number of $V$-quasi-stationary segments in $\{(\boldsymbol{\Theta}^t, \boldsymbol{\Delta}^t)\}_{t=1}^{T}$ as defined in \Cref{def-seg}. 
\begin{enumerate}[$(i)$]
    \item  Assume that the drift function $V$ is $\log$–$\log$ Lipschitz continuous, i.e.~there exists an absolute constant $C_{\mathrm{Lip}} > 0$ such that
\[
\big\vert \log \big( V(t_1) \big)  - \log \big( V(t_2) \big) \big\vert \leq C_{\mathrm{Lip}} \vert \log(t_1) - \log(t_2) \vert, \quad  \forall t_1, t_2 > 0,  
\quad   
\]
\item Let $T_{\min} = \min_{g \in [G]} (T_g - T_{g-1})$. Assume $ T_{\min} = \Theta(T)$.
    \item  For all $g \in [G]$, $1 \leq i\leq j \leq n$ and $l\in [L] $, assume that
\begin{align}
\big\vert \boldsymbol{\Pi}^{T_{g-1}+2}_{i,j,l}-\boldsymbol{\Pi}^{T_{g-1}+1}_{i,j,l} \big\vert   \leq C_{\pi} \big( \big\vert \boldsymbol{\Theta}^{T_{g-1}+2}_{i,j,l} -   \boldsymbol{\Theta}^{T_{g-1}+1}_{i,j,l}  \big\vert  +  \big\vert  \boldsymbol{\Delta}^{T_{g-1}+2}_{i,j,l} - \boldsymbol{\Delta}^{T_{g-1}+1}_{i,j,l}  \big\vert  \big),\nonumber
\end{align}
where $\boldsymbol{\Pi}^{t}  = \E\{ \mathbf{A}^t\}$ and  $C_{\pi} \geq 1$ is an absolute constant.   
\item   Assume $ s_{\max} \leq C_{\sigma} s_{\min} $, where $C_{\sigma} > 0$ is an absolute constant, $s_{\min} = \min_{k \in [K]} s_k$ and $s_{\max} = \max_{k \in [K]} s_k$ with, for any $k \in [K]$, $s_k = \sum_{i=1}^n Z_{i, k}$. For any $t \in [T]$ and $s \in [t]$, assume that \[
\mathrm{rank} (\mathcal{M}_1 ( \mathbf{W}^{t, s}))= \mathrm{rank} (\mathcal{M}_1 ( \mathbf{M}^{t, s}))= \mathrm{rank} (\mathcal{M}_1 ( \mathbf{Q}^{t, s}))= K,
\]
\[
\max\{\|\mathbf{W}^{t, s}\| /\sigma_{\min}(\mathbf{W}^{t, s}),  \|\mathbf{M}^{t, s}\| /\sigma_{\min}(\mathbf{M}^{t, s}), \|\mathbf{Q}^{t, s}\| /\sigma_{\min}(\mathbf{Q}^{t, s})\} \leq C_{\sigma},
\]
and $\sigma_{\min}(\mathbf{Q}^{t, s}) \geq \min\{\sigma_{\min}(\mathbf{M}^{t, s}), \sigma_{\min}(\mathbf{W}^{t, s})\}$   for some  absolute constant $ C_{\sigma}>0$.
Here $\mathbf{W}^{t,s} = s^{-1}\sum_{u={t-s+1}}^t  \mathbf{W}^s$, $\mathbf{M}^{t,s} = s^{-1}\sum_{u={t-s+1}}^t     \mathbf{M}^s$ and $ \mathbf{Q}^{t,s} = s^{-1}\sum_{u={t-s+1}}^t  ( \mathbf{W}^s +  \mathbf{M}^s )$, and we define $r^{t, s} = \mathrm{rank}( \mathcal{M}_3 (\mathbf{W}^{t, s})) $ and $ \tilde{r}^{t, s} =\mathrm{rank}( \mathcal{M}_3 (\mathbf{M}^{t, s}))$.
\end{enumerate}
\end{assumption}

\Cref{ass-bias}$(i)$ specifies our measurement of the magnitude of temporal variation: changes exceeding the scale $V$ are regarded as substantial and separated by segment boundaries.
Together with \Cref{ass-bias}$(ii)$, this implies that only $O(1)$ such substantial changes occur over the full time horizon.  
Assumptions of this type are commonly used in change-point analysis for complex models \citep[e.g.][]{padilla2022change,wang2025multilayer}.    
\Cref{ass-bias}$(iii)$ links the marginal probabilities~$\boldsymbol{\Pi}^t$ to the transition probabilities $(\boldsymbol{\Theta}^t,\boldsymbol{\Delta}^t)$ at segment boundaries and rules out pathological cases in which the marginals fluctuate much more erratically than the underlying transition parameters.  
Finally, \Cref{ass-bias}$(iv)$ imposes community balance and well-conditioned low-rank structure for the window-averaged connectivity tensors. This parallels \Cref{ass-stat-rank} in the stationary setting, which ensures a low-rank Tucker structure and enables the TH-PCA refinement step in \textbf{Stage II} of \Cref{alg:fast-msbm}.

\begin{theorem}\label{thm-tensor-Frobenius}
Let the process $\{\mathbf{A}^t \}_{t \geq 0} \subset\{0, 1\}^{n \times n \times L}$ be defined in \Cref{def-armsb}  and let $G$ be the number of $V$-quasi-stationary segments in $\{(\boldsymbol{\Theta}^t, \boldsymbol{\Delta}^t)\}_{t=1}^{T}$ as defined in \Cref{def-seg}. Suppose that \Cref{ass-bias} holds.
Define
\begin{equation}\label{def-tau-new}
\tau(t) =  C_{\tau}\max \Big\{t^{-1} \log(n \vee L \vee T) \log(T) \log\log(T), \big(V(t)\big)^2 \Big\}, \quad \forall t \in [T],
\end{equation}
where $ C_{\tau} >0 $ is a sufficiently large constant. 
For any $g \in [G]$ and any $t \in [T_g]\backslash[T_{g-1}]$, under \Cref{alg:fast-msbm},  with probability at least $1- (n \vee L \vee T)^{-c}$,
\begin{align}\label{thm-tensor-Frobenius-1-non}
 \| \widetilde{\boldsymbol{\Theta}}^t -\boldsymbol{\Theta}^t \|_{\mathrm{F}} + \| \widetilde{\boldsymbol{\Delta}}^t -\boldsymbol{\Delta}^t\|_{\mathrm{F}} \leq  & C    \sigma^{-1}  \big(  K\sqrt{(K \vee r)L}  \epsilon_t^2
  + K Ln \epsilon_t^4    \big)  \nonumber\\
& \hspace{0.5cm}+   C   \big(  \sqrt{nK + Lr +K^2r} \epsilon_t 
+ n\sqrt{L}  \epsilon_t^2 \big),
\end{align}
where 
\begin{equation}\label{def-epsilon-t-new}
 \epsilon_{t} = \max \bigg\{\sqrt{\frac{ \log(n \vee L \vee T)\log(T) \log\log(T)}{t-T_{g-1}}}, V(t-T_{g-1})\bigg\},
\end{equation}
$\sigma = \min_{t \in [T], s \in [t]} \min \{ \sigma_{\min}(\mathbf{W}^{t, s}), \sigma_{\min}(\mathbf{M}^{t, s})\}$, $r = \max_{t \in [T], s \in [t]} \max\{ r^{t, s}, $ $\tilde{r}^{t, s}\}$ and $C, c >0$ are absolute constants.
\end{theorem}

To interpret \Cref{thm-tensor-Frobenius}, we decompose the upper bound in \eqref{thm-tensor-Frobenius-1-non} into three main components.  
The term $\sqrt{nK + Lr + K^2 r}\epsilon_t$
captures the stochastic error from estimating the low-rank tensor structure, while 
$n\sqrt{L}\epsilon_t^2$
corresponds to the bias inherited from the autoregressive MLEs.  
The remaining terms $\sigma^{-1} ( K\sqrt{(K \vee r)L}\epsilon_t^2 $ $+ KLn\epsilon_t^4 \big)$
arise from subspace estimation errors in the tensor refinement step.

When  $\sigma \gtrsim  \max \{\tilde{\xi}_1 \wedge \tilde{\xi}_2,  \tilde{\xi}_3 \wedge \tilde{\xi}_4\}$, 
where $\tilde{\xi}_1 =  K \sqrt{K \vee r}/n$, $\tilde{\xi}_2 = \epsilon_t K \sqrt{(K \vee r)L} (nK + Lr  +K^2r)^{-1/2}$, $\tilde{\xi}_3 =   K \sqrt{L}$  and  $\tilde{\xi}_4 =   K Ln \epsilon_t^3 (nK + Lr  +K^2r)^{-1/2}$, these higher-order terms become negligible and the bound in \eqref{thm-tensor-Frobenius-1-non} simplifies to
\[
\sqrt{nK + Lr + K^2 r}\epsilon_t + n\sqrt{L}\epsilon_t^2,
\]
which has the same stochastic–bias decomposition as in the stationary setting.

Compared with the stationary rate in \eqref{eq-stat-tensor-final}, the global factor $t^{-1/2}$ is replaced by the local difficulty parameter $\epsilon_t$.  
Consequently, the effective sample size is determined by the segment length $t - T_{g-1}$, while temporal drift enters through $V(t - T_{g-1})$.
This dependence on $V$ reflects a bias–variance trade-off: faster decay of $V$ leads to shorter segments and larger stochastic error, whereas slower decay reduces stochastic variability but increases the drift-induced bias.

In particular, if $V(t) = t^{-1/2}$, then, up to logarithmic factors, the bound in \eqref{thm-tensor-Frobenius-1-non} reduces to
\[
\sqrt{\frac{nK + Lr + K^2 r}{t - T_{g-1}}}
+
\frac{n\sqrt{L}}{t - T_{g-1}},
\]
which matches the stationary rate in \eqref{eq-stat-tensor-final}, when restricted to each quasi-stationary segment.  
This shows that, under this regime, the statistical cost of non-stationarity is limited to replacing the global sample size $t$ by the local segment length $t - T_{g-1}$.  
It also highlights the adaptive nature of \Cref{alg:fast-msbm}: the effective look-back window shortens in the presence of substantial drift or abrupt changes and expands when the process is locally stable.

\begin{proposition}\label{thm-comunity-recovery-non}
Let the process $\{\mathbf{A}^t \}_{t \geq 0} \subset\{0, 1\}^{n \times n \times L}$ be defined in \Cref{def-armsb}  and let $G$ be the number of $V$-quasi-stationary segments in $\{(\boldsymbol{\Theta}^t, \boldsymbol{\Delta}^t)\}_{t=1}^{T}$ as defined in \Cref{def-seg}. Suppose that \Cref{ass-bias} holds.  Let $\{\tau(t)\}_{t \in [T]}$ be defined in \eqref{def-tau-new}. For any $t \in [T]$, define
\begin{align}\label{eq-sin-theta-new}
\widetilde{\mathcal{E}}_t =  C_{\mathcal{E}}  \sigma_Q^{-2}  \bigg( \frac{K^2\sqrt{L}}{n}  \epsilon_{t}^2 + K^2L\epsilon_{t}^4 \bigg)   
+  C\sigma_Q^{-1}  \bigg( \frac{ K}{\sqrt{n}}   \epsilon_{t}  +     K  \sqrt{L} \epsilon_{t}^2  \bigg), 
\end{align}
where $\sigma_Q =\min_{t \in [T], s \in [t]}\sigma_{\min}(\mathbf{Q}^{t, s})$, $\epsilon_{t}$ is defined in \eqref{def-epsilon-t-new} and $ C_{\mathcal{E}} >0 $ is an absolute constant. 

\begin{enumerate}
 \item For any $g \in [G]$ and any $t \in [T_g]\backslash[T_{g-1}]$, if 
\begin{equation}\label{eq-ass-singular-non-1}
    C K  \widetilde{\mathcal{E}}_t^2 \leq 1,
\end{equation}
then under \Cref{alg:fast-msbm},  with probability at least $1- (n \vee L \vee T)^{-c}$, 
\begin{equation}\label{thm-com-1-non}
\mathcal{L}(\widehat{Z}^t, Z) \leq C' \widetilde{\mathcal{E}}_t^2, 
\end{equation}
where $\mathcal{L}(\widehat{Z}^t, Z)$ is the normalized Hamming loss defined in \eqref{eq-def-hamming},  $\widetilde{\mathcal{E}}_t$ is defined \eqref{eq-sin-theta-new}, $C>0$ is a sufficiently large constant and $C', c>0$ are absolute constants. 

\item  For any $g \in [G]$ and any $t \in [T_g]\backslash[T_{g-1}]$, if 
\begin{equation}\label{eq-ass-singular-non}
     C n  \widetilde{\mathcal{E}}_t^2   \leq  1,
\end{equation}
then under \Cref{alg:fast-msbm}, with probability at least $1- (n \vee L \vee T)^{-c}$, 
\begin{equation}\label{thm-com-2-non}
\widehat{Z}_i^t = \widehat{Z}_j^t  \quad \mbox{if and only if} \quad Z_i = Z_j, \quad \forall i, j \in [n], 
\end{equation}
where $\widetilde{\mathcal{E}}_t$ is as in \eqref{eq-sin-theta-new}, $C>0$ is a sufficiently large constant and $c>0$ is an absolute constant. 
\end{enumerate}
\end{proposition}

\Cref{thm-comunity-recovery-non} establishes two levels of community recovery under non-stationarity. When $K\widetilde{\mathcal{E}}_t^2 \lesssim 1$, the estimated partition is Hamming-consistent, with error of order $\widetilde{\mathcal{E}}_t^2$. Under the stronger condition $n\widetilde{\mathcal{E}}_t^2 \lesssim 1$, the procedure further achieves exact recovery.
This result is a non-stationary analogue of \Cref{thm-comunity-recovery-stat}. The key difference lies in the definition of $\widetilde{\mathcal{E}}_t$, which now depends on the local difficulty parameter $\epsilon_t$.  Consequently, recovery performance is governed by both the effective sample size determined by the segment length $t - T_{g-1}$ and the temporal drift $V$ within the current segment.
In particular if $K = O(1)$, $\sigma_Q \asymp \sqrt{L}$ and
\[
V(t) = t^{-1/2}, \quad \forall t > 0,
\]
then, up to logarithmic factors, both the conditions \eqref{eq-ass-singular-non-1} and \eqref{eq-ass-singular-non}, as well as the resulting guarantees, reduce to their stationary counterparts applied segmentwise.

To the best of our knowledge, these are the first results achieving both Hamming-consistent and exact community recovery for non-stationary, temporally dependent dynamic network models. Existing community detection results for dynamic networks are typically developed under independence assumptions, where network snapshots are independent or conditionally independent over time \citep[e.g.][]{han2015consistent, pensky2019spectral} and therefore do not cover the non-stationary autoregressive setting considered here.

\section{Numerical experiments}\label{sec-num}

In this section, we evaluate the empirical performance of the proposed methods. In \Cref{sec-simu}, we present results based on simulated data, while in \Cref{sec:realdata_air}, we analyze a real-world data set on U.S.~air transportation. The code and datasets used in all experiments are available online \footnote{\url{https://github.com/HaotianXu/AR1-MSBM_Simulation}}.

\subsection{Simulation studies}\label{sec-simu}

\subsubsection{Stationary setting}\label{sec:stationary}

We evaluate \Cref{alg:fast-stat-msbm} under the stationary AR(1)-MSBM in \Cref{def-sarmsb}. We consider four simulation scenarios that isolate the effects of temporal dependence and layer heterogeneity, allowing a systematic comparison with baselines.

\medskip
\noindent
\textbf{Simulation setup.}
We simulate stationary AR(1)-MSBM networks following \Cref{def-sarmsb}, with $n=100$ nodes, $L=2$ layers, $K=2$ balanced communities and time horizon $T=300$. 
Let $Z\in\{0,1\}^{n\times K}$ denote the community membership matrix. The connectivity tensors $\mathbf{W}, \mathbf{M} \in [0,1]^{K \times K \times L}$ induce the transition probability tensors $\boldsymbol{\Theta}, \boldsymbol{\Delta} \in [0,1]^{n \times n \times L}$ and the adjacency sequence $\{\mathbf{A}^t\}_{t \in [T] \setminus \{1\}} \subset \{0, 1\}^{n \times n \times L}$.

For reference, the stationary marginal probabilities at the community level are
\[
\widetilde{\boldsymbol{\Pi}}
=
\frac{\mathbf{W}}{\mathbf{W}+\mathbf{M}},
\]
with node-level probabilities
\[
\boldsymbol{\Pi}_{i,j,l}
=
Z_i \widetilde{\boldsymbol{\Pi}}_{:,:,l} Z_j^\top,
\quad 1 \le i \le j \le n, \, l \in [L].
\]
These capture edge densities but not temporal dependence. We initialize $\mathbf{A}^1$ from $\boldsymbol{\Pi}$ to start near stationarity. All networks are symmetric with zero diagonals unless stated otherwise.

\medskip
\noindent
\textbf{Methods.}
We compare three methods, including the proposed approach (\Cref{alg:fast-stat-msbm}) and two baselines.
\begin{itemize}
\item \textbf{Online estimation for stationary AR(1)-MSBM (\Cref{alg:fast-stat-msbm}).}
We implement \Cref{alg:fast-stat-msbm}, hereafter referred to as \textbf{Proposed}, with ranks $r_1 = r_2 = L$. 

 \item \textbf{Static cumulative MSBM estimator.}
For any time $t \in [T]$, we form the cumulative-average tensor
\[
\widetilde{\mathbf{A}}^t = \frac{1}{t} \sum_{\tau=1}^{t} \mathbf{A}^\tau.
\]
This baseline estimates communities from $\widetilde{\mathbf{A}}^t$ using the same tensor spectral clustering procedure as \textbf{Stages II} and \textbf{III} of \Cref{alg:fast-stat-msbm}, but ignores temporal dependence and relies only on stationary marginals. We refer to this method as \textbf{Static}.

\item \textbf{Layer-aggregated AR(1)-SBM estimator} \citep{jiang2023autoregressive}.
For any time $t \in [T]$, we average the layers,
\[
\bar{A}^t = \frac{1}{L} \sum_{l=1}^{L} \mathbf{A}^t_{:,:, l},
\]
and fit the single-layer AR(1)-SBM using the method of \cite{jiang2023autoregressive} to $\{\bar{A}^t\}_{t=1}^{T} \subset \{0, 1\}^{n \times n}$. This baseline exploits temporal dependence but discards layer-specific structure. We refer to this method as \textbf{Aggregated}.
\end{itemize}

\medskip
\noindent
\textbf{Scenarios.}
We consider four scenarios, each designed to highlight a different challenge in community recovery.

\begin{itemize}
\item  \textbf{Scenario I: Weak signal in both transition probabilities and marginals.}
For each layer $l \in [2]$, we set
\[
(\mathbf{W}_{1,1,l}, \mathbf{M}_{1,1,l})=(0.10,0.20),\quad
(\mathbf{W}_{2,2,l},\mathbf{M}_{2,2,l})=(0.08,0.20),\quad
(\mathbf{W}_{1,2,l},\mathbf{M}_{1,2,l})=(0.05,0.20).
\]
This yields
\[
\widetilde{\boldsymbol{\Pi}}_{1,1,l}\approx 0.33,\qquad
\widetilde{\boldsymbol{\Pi}}_{2,2,l}\approx 0.29,\qquad
\widetilde{\boldsymbol{\Pi}}_{1,2,l}=0.20.
\]
Both the stationary marginals and the transition probabilities encode the two-community structure, but the corresponding second singular values are small. As a result, the signal is weak and accurate recovery requires a sufficiently long time horizon.

\item  \textbf{Scenario II: Signal only in transition probabilities.} For each layer $l \in [2]$, we set
\[
(\mathbf{W}_{1, 1, l},\mathbf{M}_{1, 1, l})=(0.4, 0.4),\quad
(\mathbf{W}_{2, 2, l},\mathbf{M}_{2, 2, l})=(0.2, 0.2),\quad
(\mathbf{W}_{1, 2, l},\mathbf{M}_{1, 2, l})=(0.3, 0.3),
\]
which yields
\[
\widetilde{\boldsymbol{\Pi}}_{1,1,l}=\widetilde{\boldsymbol{\Pi}}_{2,2,l}=\widetilde{\boldsymbol{\Pi}}_{1,2,l}=0.5.
\]
Thus, the stationary marginals contain no community information and the community structure is encoded in the transition probabilities.

\item \textbf{Scenario III: Opposing layer structures.}
For Layer $1$, we set 
\[
(\mathbf{W}_{1,1,1},\mathbf{M}_{1,1,1})=(0.4,0.4),\quad
(\mathbf{W}_{2,2,1},\mathbf{M}_{2,2,1})=(0.4,0.4),\quad
(\mathbf{W}_{1,2,1},\mathbf{M}_{1,2,1})=(0.05,0.45),
\]
which yields
\[
\widetilde{\boldsymbol{\Pi}}_{1,1,1}=\widetilde{\boldsymbol{\Pi}}_{2,2,1}=0.5,\quad \widetilde{\boldsymbol{\Pi}}_{1,2,1}=0.1.
\]
For Layer $2$, we reverse the pattern:
\[
(\mathbf{W}_{1,1,2},\mathbf{M}_{1,1,2})=(0.05,0.45),\quad
(\mathbf{W}_{2,2,2},\mathbf{M}_{2,2,2})=(0.05,0.45),\quad
(\mathbf{W}_{1,2,2},\mathbf{M}_{1,2,2})=(0.4,0.4),
\]
which yields
\[
\widetilde{\boldsymbol{\Pi}}_{1,1,2}=\widetilde{\boldsymbol{\Pi}}_{2,2,2}=0.1,\quad \widetilde{\boldsymbol{\Pi}}_{1,2,2}=0.5.
\]
Each layer individually contains clear community structure, but in opposite directions. As a result, aggregating the layers cancels the signal.

\item \textbf{Scenario IV: Density imbalance across layers.}
For Layer $1$, we set
\[
(\mathbf{W}_{1,1,1}, \mathbf{M}_{1,1,1}) = (\mathbf{W}_{2,2,1}, \mathbf{M}_{2,2,1}) = (0.015, 0.09),\quad
(\mathbf{W}_{1,2,1}, \mathbf{M}_{1,2,1}) = (0.010, 0.10),
\]
which yields
\[
\widetilde{\boldsymbol{\Pi}}_{1,1,1} = \widetilde{\boldsymbol{\Pi}}_{2,2,1} \approx 0.14,\quad
\widetilde{\boldsymbol{\Pi}}_{1,2,1} \approx 0.09.
\]
Thus, Layer $1$ is sparse but contains clear community structure. 
For Layer $2$, we set
\[
(\mathbf{W}_{a,b,2}, \mathbf{M}_{a,b,2}) = (0.15, 0.35) \quad \forall a,b \in [2],
\]
which yields
\[
\boldsymbol{\Pi}_{a,b, 2} \approx 0.3, \quad \forall a,b \in [2],
\]
Thus, Layer $2$ is dense and does not contain community information. 
As a result, the informative signal from the sparse layer is dominated by the dense layer under aggregation over the layers.
\end{itemize}

\medskip
\noindent\textbf{Evaluation criteria.}
At each time $t \in [T]$, let $\widehat Z^{t}$ and $Z$ denote the estimated and true community membership matrices, respectively. The clustering error is measured by
\[
\mathrm{Err}_Z(t)=1-\mathrm{ARI}\big(\widehat Z^{t}, Z\big),
\]
where ARI denotes the Adjusted Rand Index \citep{hubert1985comparing}. To be specific, for two partitions $U$ and $V$, let $n_{ab}=|\{i:U_i=a,V_i=b\}|$, $n_{a:}=\sum_b n_{ab}$, $n_{:b}=\sum_a n_{ab}$, and $N_2=\binom{n}{2}$. 
Then
\[
\mathrm{ARI}(U,V)=\frac{\mathcal A-\mathcal B\mathcal C/N_2}{(\mathcal B+\mathcal C)/2-\mathcal B\mathcal C/N_2},
\]
where 
\[
\mathcal A=\sum_{a,b}\binom{n_{ab}}{2},\quad
\mathcal B=\sum_a\binom{n_{a:}}{2},\quad
\mathcal C=\sum_b\binom{n_{: b}}{2}.
\]
The ARI equals $1$ for identical partitions, is close to $0$ under chance-level agreement and may be negative.

\begin{figure}[t]
    \centering
    \includegraphics[width=0.8\textwidth]{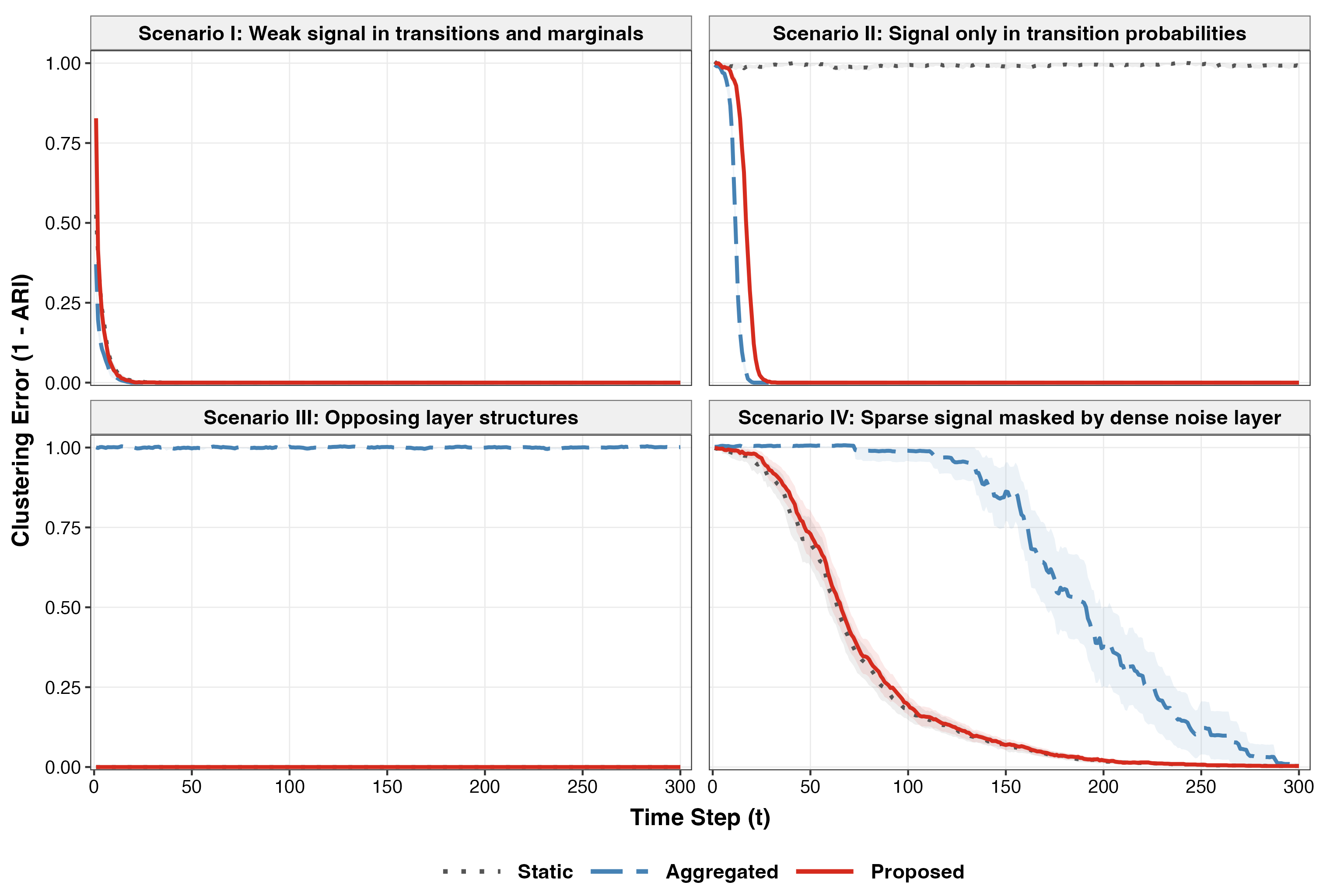}
    \caption{\small Stationary simulation results: mean clustering error ($1 - \mathrm{ARI}$) over time with pointwise $95\%$ Monte Carlo confidence bands. Top-Left: Scenario I (weak signal in both transition probabilities and marginals). Top-Right: Scenario II (signal only in transition probabilities). Bottom-Left: Scenario III (opposing layer structures; signals cancel under aggregation). Bottom-Right: Scenario IV (sparse signal masked by a dense noise layer).}
    \label{fig:sim_results}
\end{figure}

\medskip
\noindent
\textbf{Results.}
Figure~\ref{fig:sim_results} reports the averaged clustering error trajectories over $R=50$ Monte Carlo trials, with pointwise $95\%$ Monte Carlo confidence bands.
In \textbf{Scenario I}, all methods converge and eventually recover the communities.
In \textbf{Scenario II}, \textbf{Static} fails because marginals contain no signal, while \textbf{Aggregated} and \textbf{Proposed} recover the communities by exploiting transition probabilities.
In \textbf{Scenario III}, \textbf{Aggregated} fails because cross-layer averaging cancels the signal, whereas the other methods recover the communities by preserving layer-specific information. 
In \textbf{Scenario IV}, \textbf{Aggregated} converges more slowly, as the dense noise layer dominates under aggregation, while the other methods exhibit faster and nearly identical convergence.
Overall, \textbf{Proposed} performs consistently well across all four scenarios, highlighting its ability to jointly capture temporal dependence and multilayer structure.

\subsubsection{Non-stationary setting}\label{sec:nonstationary}  
We evaluate \Cref{alg:fast-msbm} under a non-stationary  AR(1)-MSBM with fixed memberships and time-varying transition probabilities. We consider four simulation scenarios that capture different types of temporal variation, enabling a systematic comparison with full-history and fixed-window baselines.

\medskip
\noindent
\textbf{Simulation setup.}
We simulate non-stationary AR(1)-MSBM networks following \Cref{def-armsb}, with $n=100$ nodes, $K=2$ balanced communities, $L=2$ layers and horizon $T=175$. The community membership matrix $Z\in\{0,1\}^{n\times K}$ is fixed over time. Time-varying connectivity tensors $\{\mathbf{W}^t\}_{t\in [T]}, \{\mathbf{M}^t\}_{t\in [T]}\subset[0,1]^{K\times K\times L}$ which will be specified in each scenario. The induced transition probability tensors are denoted by $\{\boldsymbol{\Theta}^t\}_{t\in [T]}$ and $\{\boldsymbol{\Delta}^t\}_{t\in [T]}$, and the adjacency sequence is denoted by $\{\mathbf{A}^t\}_{t \in [T] \setminus \{1\}} $.
We initialize $\mathbf{A}^1$ from $\boldsymbol{\Pi}$, where
\[
\boldsymbol{\Pi}_{i,j,l}
=
Z_i \widetilde{\boldsymbol{\Pi}}_{:,:,l} Z_j^\top,
\quad 1 \le i \le j \le n,\ l \in [L],  
\]
with $\widetilde{\boldsymbol{\Pi}}
= \mathbf{W}^1 / (\mathbf{W}^1+\mathbf{M}^1).$
All networks are symmetric with zero diagonals unless stated otherwise.

We consider the following connectivity tensors to construct different scenarios for each $l \in [L]$:
\[
\mathbf{W}^{(1)}_{:, :, l}=\begin{bmatrix}
0.10 & 0.05\\
0.05 & 0.10
\end{bmatrix},\quad
\mathbf{M}^{(1)}_{:, :, l}=\begin{bmatrix}
0.20 & 0.05\\
0.05 & 0.20
\end{bmatrix},
\]
\[
\mathbf{W}^{(2)}_{:, :, l}=\begin{bmatrix}
0.15 & 0.05\\
0.05 & 0.15
\end{bmatrix},\quad
\mathbf{M}^{(2)}_{:, :, l}=\begin{bmatrix}
0.25 & 0.05\\
0.05 & 0.25
\end{bmatrix},
\]
\[
\mathbf{W}^{(3)}_{:, :, l}=\begin{bmatrix}
0.35 & 0.05\\
0.05 & 0.35
\end{bmatrix},\quad
\mathbf{M}^{(3)}_{:, :, l}=\begin{bmatrix}
0.45 & 0.05\\
0.05 & 0.45
\end{bmatrix},
\]
\[
\mathbf{W}^{(4)}_{:, :, l}=\begin{bmatrix}
0.50 & 0.05\\
0.05 & 0.50
\end{bmatrix},\quad
\mathbf{M}^{(4)}_{:, :, l}=\begin{bmatrix}
0.50 & 0.05\\
0.05 & 0.50
\end{bmatrix}.
\]

\medskip
\noindent
\textbf{Methods.} We compare four methods, including the proposed approach (\Cref{alg:fast-msbm}) and three
baselines. All methods use the same low-rank refinement and clustering stages (\textbf{Stages II} and \textbf{III} of \Cref{alg:fast-msbm}) and differ only in the window selection step in \textbf{Stage I}.

\begin{itemize}
\item \textbf{Adaptive method (\Cref{alg:fast-msbm}).}   
We run \Cref{alg:fast-msbm} with $r^{t,k}=\tilde r^{t,k}=L$, the dynamic grid $\{\mathcal{G}^{(t)}\}_{t\in[T]}$ in \eqref{def-dynamic-grid}, and the tolerance $\tau(t)$ in \eqref{def-tau-new} with $V(t)=t^{-1/2}$. The constant $C_\tau$ is calibrated by \Cref{alg:tune-Ctau} (Appendix~\ref{section:add-algorithm}) on $50$ trajectories simulated from the stationary AR(1)-MSBM (\Cref{def-sarmsb}) with $(\mathbf{W}^{(1)},\mathbf{M}^{(1)})$, taking the smallest value for which the probability of selecting a shorter window than the full history is at most $\alpha=0.05$ under stationarity.  
\item \textbf{Full-history method.} At time $t$, use all past observations ($\hat{k}_t=t$).
\item \textbf{Fixed-window ($\hat{k}_t=30$) method.} Use a fixed window of length $30$. 
\item \textbf{Fixed-window ($\hat{k}_t=20$) method.} Use a fixed window of length $20$. 
\end{itemize}

\medskip
\noindent
\textbf{Scenarios.} 
We consider four scenarios, each representing a different type of temporal variation.
\begin{itemize}
    \item \textbf{Scenario I: Abrupt changes with a larger second shift.} We define two abrupt changes with a larger shift occurring at the second change:
\[
(\mathbf{W}^t,\mathbf{M}^t)=
\begin{cases}
(\mathbf{W}^{(1)}, \mathbf{M}^{(1)}), & 1\le t\le 50,\\
(\mathbf{W}^{(2)},\mathbf{M}^{(2)}), & 51\le t\le 100,\\
(\mathbf{W}^{(4)},\mathbf{M}^{(4)}), & 101\le t\le 175.
\end{cases}
\]
   \item \textbf{Scenario II: Abrupt changes with a larger first shift.} We define two abrupt changes with a larger shift occurring at the first change:
\[
(\mathbf{W}^t,\mathbf{M}^t)=
\begin{cases}
(\mathbf{W}^{(1)},\mathbf{M}^{(1)}), & 1\le t\le 50,\\
(\mathbf{W}^{(3)},\mathbf{M}^{(3)}), & 51\le t\le 100,\\
(\mathbf{W}^{(4)},\mathbf{M}^{(4)}), & 101\le t\le 175.
\end{cases}
\]
\item \textbf{Scenario III: Gradual changes.} We construct a smooth transition between two regimes in the middle segment:
\[
(\mathbf{W}^t,\mathbf{M}^t)=
\begin{cases}
(\mathbf{W}^{(1)},\mathbf{M}^{(1)}), & 1\le t\le 50,\\
\big((1-s_t)\mathbf{W}^{(1)}+s_t\mathbf{W}^{(4)},\ (1-s_t)\mathbf{M}^{(1)}+s_t\mathbf{M}^{(4)}\big), & 51\le t\le 100,\\
(\mathbf{W}^{(4)},\mathbf{M}^{(4)}), & 101\le t\le 175,
\end{cases}
\]
where $s_t=(t-51)/49$ for $t \in [100] \setminus [50]$.
\item   \textbf{Scenario IV: Rapidly alternating changes.} We alternate between two regimes in the middle segment at a high frequency:
\[
(\mathbf{W}^t,\mathbf{M}^t)=
\begin{cases}
(\mathbf{W}^{(1)},\mathbf{M}^{(1)}), & 1\le t\le 50,\\
(\mathbf{W}^{(2)},\mathbf{M}^{(2)}), & 51\le t\le 100,\ \left\lfloor\frac{t-51}{2}\right\rfloor\text{ even},\\
(\mathbf{W}^{(3)},\mathbf{M}^{(3)}), & 51\le t\le 100,\ \left\lfloor\frac{t-51}{2}\right\rfloor\text{ odd},\\
(\mathbf{W}^{(4)},\mathbf{M}^{(4)}), & 101\le t\le 175.
\end{cases}
\]
\end{itemize}

\medskip
\noindent\textbf{Evaluation criteria.}
At each time $t \in [T]$, let $\widehat{\boldsymbol{\Theta}}^t$ and $\boldsymbol{\Theta}^t$, and $\widehat{\boldsymbol{\Delta}}^t$ and $\boldsymbol{\Delta}^t$, denote the estimated and true transition probability tensors, respectively, and let $\widehat Z^{t}$ and $Z$ denote the estimated and true community membership matrices. We evaluate the estimation errors
\[
\mathrm{Err}_{\boldsymbol{\Theta}}(t)=\frac{\|\widehat{\boldsymbol{\Theta}}^t-\boldsymbol{\Theta}^t\|_{\mathrm F}}{n\sqrt{L}},\qquad
\mathrm{Err}_{\boldsymbol{\Delta}}(t)=\frac{\|\widehat{\boldsymbol{\Delta}}^t-\boldsymbol{\Delta}^t\|_{\mathrm F}}{n\sqrt{L}},
\]
and the clustering error
\[
\mathrm{Err}_{Z}(t)=1-\mathrm{ARI}\big(\widehat Z^{t}, Z\big).
\]
Here ARI is defined as in Section~\ref{sec:stationary}. We report averages over the burn-in period $t>15$. For each metric $M \in \{\mathrm{Err}_{\boldsymbol{\Theta}}, \mathrm{Err}_{\boldsymbol{\Delta}}, \mathrm{Err}_{Z}\}$, method $m$, scenario $s$ and replication $r$, we define the post-burn-in average
\[
\bar M^{(r)}_{m,s}
=
\frac{1}{160}\sum_{t=16}^{175} M^{(r)}_{m,s}(t),
\]
The mean and standard deviation over $R=50$ replications are
\begin{equation}\label{eq-mean-sd}
\mu_{m,s}(M)=\frac{1}{R}\sum_{r=1}^{R}\bar M^{(r)}_{m,s}, \qquad
\sigma_{m,s}(M)=\sqrt{\frac{1}{R-1}\sum_{r=1}^{R}\big(\bar M^{(r)}_{m,s}-\mu_{m,s}(M)\big)^2},
\end{equation}

\begin{table}[t]
\centering
\caption{Non-stationary simulation results for Scenarios I--IV: mean (standard deviation) of post-burn-in averages over $R =50$ Monte Carlo trials as defined in \eqref{eq-mean-sd}}
\label{tab:nonstat_advantage}
\renewcommand{\arraystretch}{1.15}
\setlength{\tabcolsep}{4pt}
\footnotesize
\begin{tabular}{clccc}
\toprule
\textbf{Scenario} & \textbf{Method} & \textbf{Clustering Error} & \textbf{Theta Error} & \textbf{Delta Error} \\
\midrule
\multirow{4}{*}{1} & Adaptive & 0.0000 (0.0000) & \textbf{0.0446 (0.0008)} & \textbf{0.0395 (0.0009)} \\
& Full-history & 0.0000 (0.0000) & 0.1126 (0.0003) & 0.0765 (0.0005) \\
 & Fixed-window ($\hat{k}_t=30$) & 0.0000 (0.0000) & 0.0498 (0.0004) & 0.0471 (0.0005) \\
 & Fixed-window ($\hat{k}_t=20$) & 0.0000 (0.0000) & 0.0456 (0.0004) & 0.0482 (0.0005) \\
\midrule
\multirow{4}{*}{2} & Adaptive & 0.0000 (0.0000) & 0.0481 (0.0015) & \textbf{0.0391 (0.0007)} \\
& Full-history & 0.0000 (0.0000) & 0.1204 (0.0003) & 0.0757 (0.0005) \\
 & Fixed-window ($\hat{k}_t=30$) & 0.0000 (0.0000) & 0.0483 (0.0003) & 0.0429 (0.0005) \\
 & Fixed-window ($\hat{k}_t=20$) & 0.0000 (0.0000) & \textbf{0.0438 (0.0003)} & 0.0425 (0.0005) \\
\midrule
\multirow{4}{*}{3}  & Adaptive & 0.0000 (0.0000) & 0.0438 (0.0013) & \textbf{0.0389 (0.0008)} \\
 & Full-history & 0.0000 (0.0000) & 0.1245 (0.0003) & 0.0803 (0.0005) \\
 & Fixed-window ($\hat{k}_t=30$) & 0.0000 (0.0000) & 0.0454 (0.0004) & 0.0419 (0.0005) \\
 & Fixed-window ($\hat{k}_t=20$) & 0.0000 (0.0000) & \textbf{0.0389 (0.0004)} & 0.0404 (0.0005) \\
\midrule
\multirow{4}{*}{4}  & Adaptive & 0.0000 (0.0000) & 0.0566 (0.0014) & \textbf{0.0515 (0.0008)} \\
  & Full-history & 0.0000 (0.0000) & 0.1175 (0.0003) & 0.0822 (0.0004) \\
 & Fixed-window ($\hat{k}_t=30$) & 0.0000 (0.0000) & 0.0582 (0.0003) & 0.0543 (0.0005) \\
 & Fixed-window ($\hat{k}_t=20$) & 0.0000 (0.0000) & \textbf{0.0552 (0.0003)} & 0.0544 (0.0004) \\
\bottomrule
\end{tabular}
\normalsize
\end{table}

\begin{figure}[t]
    \centering
    \includegraphics[width=0.92\textwidth]{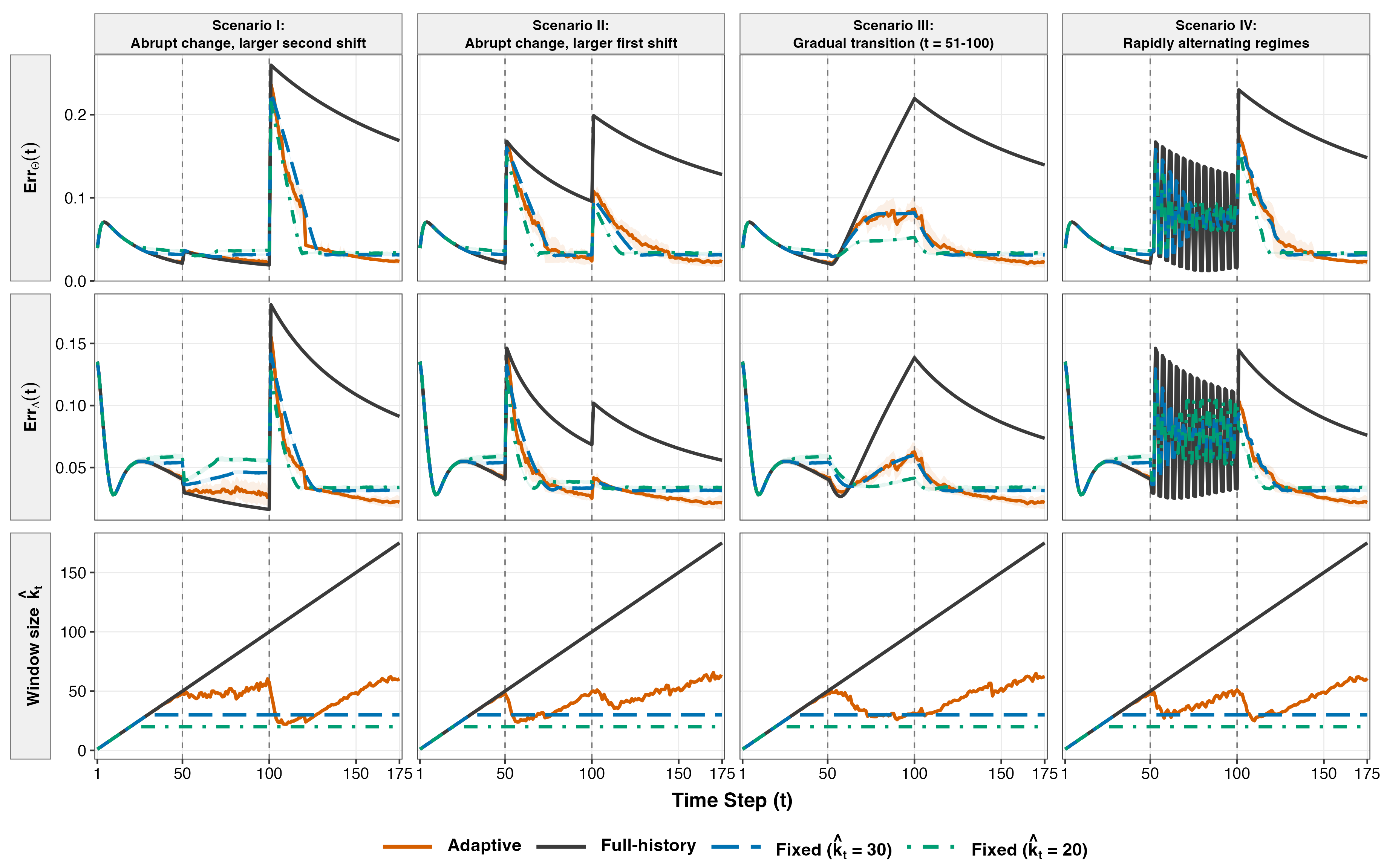}
    \caption{\small Non-stationary simulation results.  mean over time with  Top row: $\mathrm{Err}_{\Theta}(t)$. Middle row: $\mathrm{Err}_{\Delta}(t)$. Bottom row: window size $\hat{k}_t$.  Columns (left to right): Scenarios I--IV. Curves show Adaptive, Full-history, Fixed-window ($\hat{k}_t=30$) and Fixed-window ($\hat{k}_t=20$); shaded bands denote mean $\pm$ standard deviation (error rows only). Dashed lines indicate phase boundaries at $t=50$ and $t=100$.}
    \label{fig:nonstationary_adaptive}
\end{figure}

\medskip
\noindent
\textbf{Results.}
Table~\ref{tab:nonstat_advantage} reports post-burn-in means and standard deviations for the four scenarios. Since the communities are strongly separated in every phase, the clustering error is essentially zero for all methods; the comparison is therefore driven by the estimation of $\boldsymbol{\Theta}$ and $\boldsymbol{\Delta}$. Across all scenarios, \textbf{Full-history} consistently has the largest estimation errors. \textbf{Adaptive} achieves the smallest error for $\boldsymbol{\Delta}$ in all scenarios. For $\boldsymbol{\Theta}$, \textbf{Adaptive} performs best in \textbf{Scenario I}, while \textbf{Fixed-window ($\hat{k}_t=20$)} is slightly better in \textbf{Scenarios II--IV}; however, the differences are small. 

Figure~\ref{fig:nonstationary_adaptive} shows the trajectories of $\mathrm{Err}_{\boldsymbol{\Theta}}(t)$, $\mathrm{Err}_{\boldsymbol{\Delta}}(t)$ and the selected window size $\hat{k}_t$. \textbf{Scenarios I} and \textbf{II} involve abrupt changes, with the larger shift occurring at $t=100$ in \textbf{Scenario I} and at $t=50$ in \textbf{Scenario II}. \textbf{Scenario III} evolves smoothly over $t\in [100] \setminus [50]$, whereas \textbf{Scenario IV} alternates rapidly between two regimes in the middle phase. For \textbf{Full-history}, the window grows over time, leading to slow response to regime changes and large post-change errors due to increased bias. \textbf{Fixed-window} baselines respond more quickly, but their performance depends on how well the chosen window length aligns with the local regime structure. By contrast, \textbf{Adaptive} adjusts the window size according to the local regime: it shortens the window after abrupt changes, keeps it relatively short during the alternating phase, and increases it once the process stabilizes. As a result, in the final stationary phase after $t=100$, \textbf{Adaptive} typically outperforms \textbf{Fixed-window} baselines across all scenarios.

\subsection{Real data analysis}
\label{sec:realdata_air}

We analyze the U.S.~air transportation network \citep{bts2026t100} using monthly data from January 2015 to November 2025 ($T=131$).  We construct dynamic multilayer networks by retaining the four most active carriers (Southwest, United, Delta, and SkyWest) as layers ($L =4$) and selecting the top $50$ airports by frequency of appearance as origins or destinations as nodes. For each month $t$, origin $i$, destination $j$ and carrier $l$, we define a directed edge
\[
\mathbf{A}^t_{i,j,l}=\mathbbm{1}\{\text{at least one flight from $i$ to $j$ by carrier $l$ in month $t$}\},
\]
and remove self-loops by setting $\mathbf{A}^t_{i, i, l}=0$. This yields a sequence $\{\mathbf{A}^t\}_{t=1}^{T}\subset\{0,1\}^{50\times 50\times 4}$.
Data from January 2015 to December 2017 ($T_{\mathrm{train}} =36$) are used as the training set and the remaining observations form the test set. We apply the adaptive estimator (\Cref{alg:fast-msbm}) with $K=3$ and select the remaining tuning parameters as in \Cref{sec:nonstationary}.

For community recovery, let $\widehat Z^{(t)}$ denote the estimated community membership matrix at time $t \in [T]$. We quantify temporal variation using the community sizes
\[
\hat{s}_k^t  = \frac{1}{n}\sum_{i=1}^n \mathbbm 1\{\widehat Z_{i,k}^{(t)} = 1\}, \quad \forall k \in [K],
\]
and the switch rate
\[
\text{SwitchRate}(t) = \frac{1}{n}\sum_{i=1}^n \mathbbm 1\{\widehat Z_{i}^{(t)} \neq \widehat Z_{i}^{(t-1)}\},
\]
which measures the proportion of nodes whose estimated community membership changes between consecutive months (after label alignment).

The results are shown in Figure~\ref{fig:air_comm_threepanel}. The top panel indicates that one dominant community consistently contains most nodes, with the rest forming small peripheral groups. The middle panel shows that community switching is concentrated around major events, with a prominent spike in mid-2020 corresponding to the COVID disruption; the larger fluctuations in the early period are mainly due to the limited amount of data. After 2022, switching becomes infrequent, suggesting gradual stabilization. The bottom panel shows that the selected window size $\hat{k}_t$ expands during stable periods and contracts around structural changes, aligning with spikes in the switch rate and reflecting adaptive adjustment to network dynamics.

Estimation results for formation ($\boldsymbol{\Theta}^t$), dissolution ($\boldsymbol{\Delta}^t$), persistence $\boldsymbol{\Pi}^t$ and turnover $\boldsymbol{\Theta}^t+\boldsymbol{\Delta}^t$ are reported in Table~\ref{tab:air_transition_probs}. Panel A summarizes the yearly trends. From 2015 to 2025, formation decreases steadily while dissolution increases consistently. Turnover remains relatively stable, whereas persistence declines, indicating reduced network connectivity.
Panel B compares pre- and post-COVID periods by carrier. All carriers exhibit lower formation and higher dissolution after 2020, reflecting reduced flight activity and a contraction of the route network. Delta shows the highest post-COVID turnover, whereas SkyWest has the lowest, indicating greater stability. These patterns are consistent with disruption during COVID-19 followed by gradual network reconfiguration.

\begin{figure}[t]
    \centering
    \includegraphics[width=0.88\textwidth]{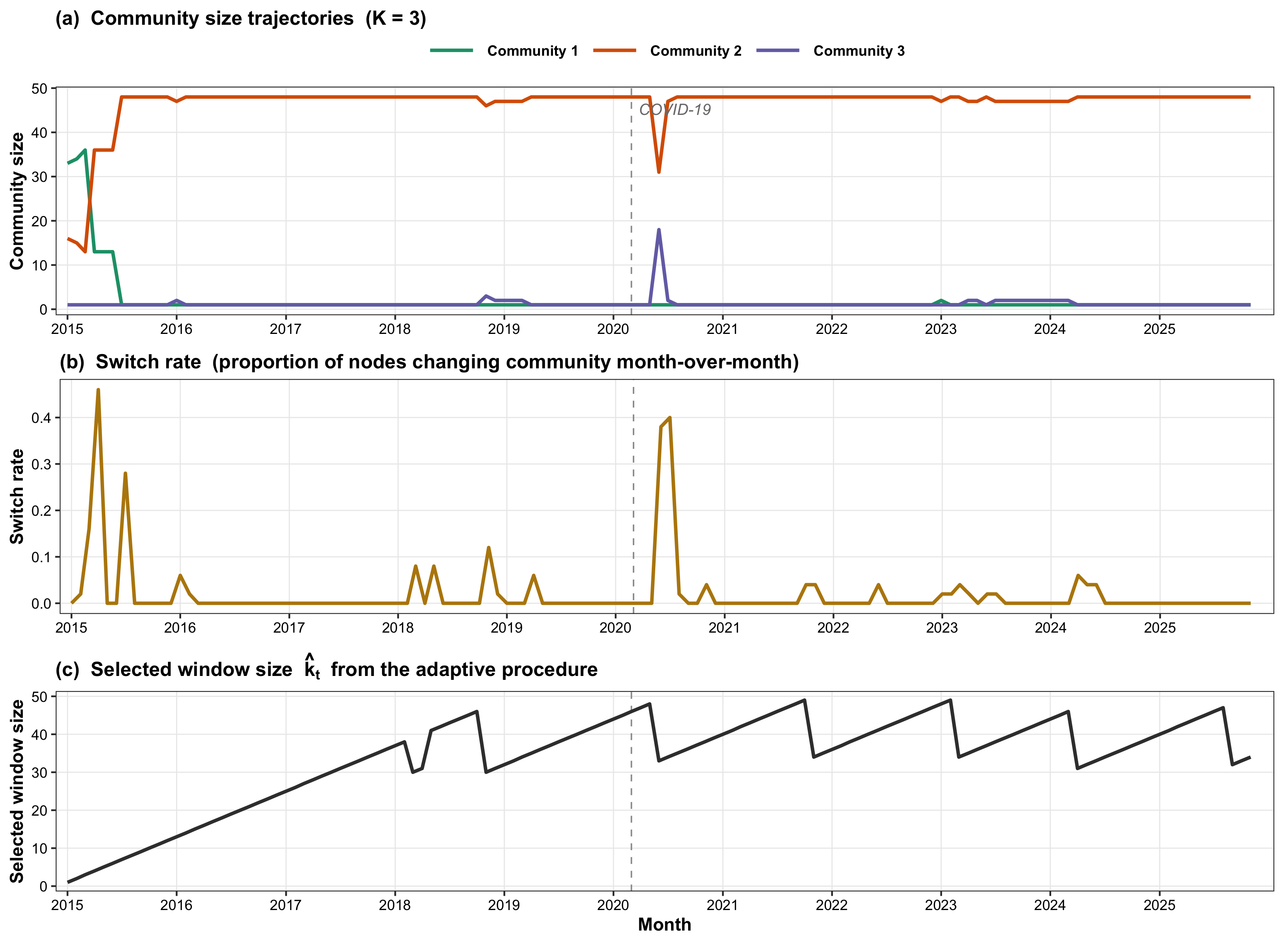}
    \caption{\small Estimated community and adaptation dynamics for the U.S.~air transportation network. Top: community size trajectories ($K=3$). Middle: switch rate. Bottom: selected window size $\hat{k}(t)$.}
    \label{fig:air_comm_threepanel}
\end{figure}

\begin{table}[t]
\centering
\caption{Estimated results for the U.S.~air transportation network. Panel A reports yearly network averages. Panel B reports carrier-level averages before and after COVID (Pre-COVID: January 2018–February 2020; Post-COVID: March 2020–November 2025).}
\label{tab:air_transition_probs}
\renewcommand{\arraystretch}{1.08}
\setlength{\tabcolsep}{7pt}
\small
\begin{tabular}{llcccc}
\toprule
\textbf{Carrier} & \textbf{Period} & $\widehat{\boldsymbol{\Theta}}^t$ & $\widehat{\boldsymbol{\Delta}}^t$ & $\widehat{\boldsymbol{\Pi}}^t$ & $\widehat{\boldsymbol{\Theta}}^t+\widehat{\boldsymbol{\Delta}}^t$ \\
\midrule
\multicolumn{6}{l}{\textbf{Panel A:  Yearly network averages}} \\[3pt]
\multirow{11}{*}{All carriers}
 & 2015 & 0.296 & 0.103 & 0.726 & 0.399 \\
 & 2016 & 0.263 & 0.177 & 0.597 & 0.440 \\
 & 2017 & 0.251 & 0.207 & 0.545 & 0.458 \\
 & 2018 & 0.195 & 0.222 & 0.445 & 0.417 \\
 & 2019 & 0.089 & 0.225 & 0.294 & 0.314 \\
 & 2020 & 0.090 & 0.231 & 0.291 & 0.321 \\
 & 2021 & 0.098 & 0.234 & 0.301 & 0.332 \\
 & 2022 & 0.089 & 0.242 & 0.279 & 0.331 \\
 & 2023 & 0.083 & 0.243 & 0.261 & 0.326 \\
 & 2024 & 0.069 & 0.262 & 0.211 & 0.331 \\
 & 2025 & 0.062 & 0.274 & 0.186 & 0.336 \\
\midrule
\multicolumn{6}{l}{\textbf{Panel B: arrier-level averages (pre- and post-COVID)}}  \\[3pt]
\multirow{2}{*}{Southwest} & Pre-COVID  & 0.228 & 0.164 & 0.545 & 0.392 \\
 & Post-COVID & 0.100 & 0.210 & 0.315 & 0.310 \\
\multirow{2}{*}{United}    & Pre-COVID  & 0.103 & 0.294 & 0.243 & 0.397 \\
   & Post-COVID & 0.076 & 0.270 & 0.221 & 0.346 \\
\multirow{2}{*}{Delta}     & Pre-COVID  & 0.104 & 0.280 & 0.256 & 0.384 \\
 & Post-COVID & 0.074 & 0.329 & 0.182 & 0.403 \\
\multirow{2}{*}{SkyWest}   & Pre-COVID  & 0.117 & 0.160 & 0.409 & 0.277 \\
& Post-COVID & 0.077 & 0.181 & 0.301 & 0.258 \\
\bottomrule
\end{tabular}
\normalsize
\end{table}

\section{Discussion}\label{sec-dis}

In this paper, we introduced the AR(1)-MSBM for dynamic multilayer networks. The model captures temporal dependence through edge formation and dissolution probabilities, while allowing a common set of nodes to interact across multiple layers. We developed online estimation procedures for both stationary and non-stationary settings. In the stationary case, the proposed method recursively updates sufficient statistics and refines the initial estimators via tensor-based spectral methods. In the non-stationary case, we proposed an adaptive windowing procedure that automatically selects the appropriate amount of historical information, enabling the algorithm to respond to both abrupt structural changes and gradual temporal drift. We established non-asymptotic estimation guarantees, minimax lower bounds and community recovery results, and demonstrated the empirical performance of the proposed methods. 

Several extensions remain open for future research.
Firstly, the bias arises from the estimation procedure in both the stationary and non-stationary settings, since the initial estimators are based on maximum likelihood, which is known to be inherently biased in autoregressive models. In classical low-dimensional autoregressive settings, this bias is typically negligible compared to stochastic fluctuations when the time horizon is moderately large. In our high-dimensional tensor setting, however, the situation is more complex. While the tensor refinement step substantially reduces stochastic error by exploiting low-rank structure, the bias persists from the initial entrywise estimators. As a result, when the time horizon in the stationary setting or the adaptive window size in the non-stationary setting is small, the bias can become non-negligible and may even dominate the overall error. A possible direction is to construct bias-corrected versions of the initial transition-probability estimators, inspired by bias-corrected and median-unbiased methods in the autoregressive literature \citep{andrews1993exactly,andrews1994approximately,patterson2000bias}.

Secondly, our current non-stationary theoretical guarantees assume that the number of quasi-stationary segments is of constant order. This assumption simplifies the analysis of the bias, as shown in the proof of \Cref{prop-bias-final} in Appendix \ref{sec-pro-bias}. We conjecture that incorporating bias-corrected versions of the initial estimators may help relax this assumption and we leave this direction for future work.

Thirdly, the current non-stationary framework assumes fixed community memberships over time, but it can be extended to allow time-varying memberships; see Appendix~\ref{app:latent} for the formulation and corresponding theoretical adjustments.

Lastly, our framework only allows dependence across snapshots for edges with the same indices $(i,j,l)$ and focuses on first-order (AR(1)) temporal dynamics. 
An important direction for future work is to extend the model to incorporate cross-edge dependence and higher-order autoregressive dynamics. For single-layer networks, \cite{chang2026autoregressive} showed that allowing edge transition probabilities to depend on multiple lagged observations and on the histories of other edges can capture important features such as transitivity, degree heterogeneity and persistence.  Extending this to multilayer networks requires more involved tensor analysis and we will leave this as future work.

\bibliographystyle{plainnat}
\bibliography{references.bib}

\appendix

\section*{Appendices}

All technical details of this paper are provided in the appendices. In particular, additional algorithms are presented in Appendix~\ref{section:add-algorithm}. Extensions of the non-stationary framework in \Cref{sec-nonstat} to allow time-varying community memberships are discussed in Appendix~\ref{app:latent}. Proofs of the results in \Cref{sec-stat} are given in Appendix~\ref{app-proof-sec-stat}, while those for \Cref{sec-nonstat} are provided in Appendix~\ref{app-proof-sec-nonstat}.

\section{Additional algorithms}\label{section:add-algorithm} 

We present the heteroskedastic principal component analysis (H-PCA) algorithm proposed by \cite{zhang2022heteroskedastic} in \Cref{hpca}, which is used for tensor spectral estimation in \textbf{Stage II} of both Algorithms~\ref{alg:fast-stat-msbm} and \ref{alg:fast-msbm}.

To implement \Cref{alg:fast-msbm} in practice, we require a data-driven choice of the tolerance level $\tau(t)$ defined in \eqref{def-tau-new}. The constant $C_\tau$ is calibrated by \Cref{alg:tune-Ctau} on $50$ trajectories simulated from the stationary AR(1)-MSBM (\Cref{def-sarmsb}) with $(\mathbf{W}^{(1)}, \mathbf{M}^{(1)})$, taking the smallest value for which the probability of selecting a shorter window than the full history is at most $\alpha=0.05$ under stationarity.

\begin{algorithm}[ht] 
\caption{H-PCA$(\Sigma, r)$} \label{hpca}
\begin{algorithmic}
    \INPUT{Matrix $\Sigma \in \mathbb{R}^{n \times n}$ and rank $r$}
   \Initialise{$\widehat{\Sigma}^{(0)} \leftarrow \Sigma$, $\mathrm{diag}\big(\widehat{\Sigma}^{(0)}\big) \leftarrow 0$, $\tilde{t} \leftarrow  5 \log\{\sigma_{\min}(\Sigma)/n\} $.}
    \For{$t \in \{0\} \cup [\tilde{t}-1]$}
        \State{Compute the singular value decomposition on $\widehat{\Sigma}^{(t)}$:
        \[
        \widehat{\Sigma}^{(t)} = \sum_{i=1}^n\sigma^{i, (t)} \mathbf{u}^{i,(t)} ( \mathbf{v}^{i, (t)})^{\top},  \quad  \sigma^{1, (t)} \geq \cdots \geq  \sigma^{n, (t)}  \geq 0
        \]}
        \State{Form the rank-$r$ approximation:
        \[
        \widetilde{\Sigma}^{(t)} \leftarrow \sum_{i = 1}^r \sigma^{i, (t)} \mathbf{u}^{i, (t)} \big(\mathbf{v}^{i, (t)}\big)^{\top}
        \]}
        \State{Update the diagonal:
        \[
        \widehat{\Sigma}^{(t+1)} \leftarrow \widehat{\Sigma}^{(t)},  \quad \mathrm{diag}\big(\widehat{\Sigma}^{(t+1)}\big) \leftarrow \mathrm{diag}\big(\widetilde{\Sigma}^{(t)}\big)
        \]}
    \EndFor
    \State{Let $\{\mathbf{u}^i\}_{i = 1}^r$ be the top-$r$ left singular vectors of $\widehat{\Sigma}^{(\tilde{t})}$.}
    \OUTPUT{$U\leftarrow (\mathbf{u}^1, \ldots, \mathbf{u}^r) \in \mathbb{R}^{n \times r}$}
\end{algorithmic}
\end{algorithm}

\begin{algorithm}[ht]
\caption{Calibration of $C_\tau$ via Stationary Training Data}
\label{alg:tune-Ctau}
\begin{algorithmic}
\INPUT{Training sequence $\{\mathbf{A}^t\}_{t=1}^{T_{\mathrm{train}}}$, grid $\mathcal{C}=\{c_1<\cdots<c_m\}$, burn-in time $t_0$, level $\alpha$, bootstrap size $B$}

\Initialise{$\mathbf{A}^0 \leftarrow 0$}

\Statex \vspace{2pt}
\textbf{Stage I: Estimation of stationary transition parameters}
\State{$\mathbf{N}^{0,1} \leftarrow \sum_{t=1}^{T_{\mathrm{train}}} \mathbf{A}^t \odot (1-\mathbf{A}^{t-1})$, 
$\mathbf{N}^{1,0} \leftarrow \sum_{t=1}^{T_{\mathrm{train}}} (1-\mathbf{A}^t)\odot \mathbf{A}^{t-1}$, 
$\mathbf{N} \leftarrow \sum_{t=1}^{T_{\mathrm{train}}} \mathbf{A}^{t-1}$}
\State{$\widehat{\boldsymbol{\Theta}} \leftarrow \mathbf{N}^{0,1}/(T_{\mathrm{train}}-\mathbf{N})$,  
$\widehat{\boldsymbol{\Delta}} \leftarrow \mathbf{N}^{1,0}/\mathbf{N}$}

\Statex \vspace{2pt}
\textbf{Stage II: Parametric bootstrap generation}

\ForAll{$b \in [B]$}
    \State{$\mathbf{A}^{0,(b)} \leftarrow 0$}
    \ForAll{$t \in [T_{\mathrm{train}}]$}
        \ForAll{$1 \leq i \leq j \leq n, l \in [L]$}
        \State{
        $\mathbb{P}\big(\mathbf{A}^{t,(b)}_{i,j,l}=1 \mid \mathbf{A}^{t-1,(b)}_{i,j,l}=0\big)
        = \widehat{\boldsymbol{\Theta}}_{i,j,l},$
        }
        \State{
        $\mathbb{P}\big(\mathbf{A}^{t,(b)}_{i,j,l}=0 \mid \mathbf{A}^{t-1,(b)}_{i,j,l}=1\big)
        = \widehat{\boldsymbol{\Delta}}_{i,j,l}
        $}
        \EndFor
        \State{Form $\mathbf{A}^{t,(b)}$ by setting $\mathbf{A}^{t,(b)}_{j,i,l}=\mathbf{A}^{t,(b)}_{i,j,l}$ for $i<j$, $l\in[L]$}
    \EndFor
\EndFor

\Statex \vspace{2pt}
\textbf{Stage III: Threshold family construction}

\State{For each $c \in \mathcal{C}$, define}
\[
\tau_c(t) \leftarrow 
\frac{c}{t}
\log(n\vee L\vee T_{\mathrm{train}})
\log(T_{\mathrm{train}})
\log\log(T_{\mathrm{train}}),
\quad \forall t \in [T_{\mathrm{train}}]
\]

\Statex \vspace{2pt}
\textbf{Stage IV: Selection of optimal threshold constant}

\ForAll{$c \in \mathcal{C}$}
    \ForAll{$b \in [B]$}
        \State{Run Algorithm~\ref{alg:fast-msbm} with $\tau_c(\cdot)$ on $\{\mathbf{A}^{t,(b)}\}_{t\in[T_{\mathrm{train}}]}$ to obtain $\{\widehat{k}^{(b)}_t(c)\}_{t \in [T_{\mathrm{train}}]}$}
        \State{
        $a^{(b)}(c) \leftarrow 
        \mathbbm{1}\big\{
        \widehat{k}^{(b)}_t(c)=t, \ \forall t \in [t_0,T_{\mathrm{train}}]
        \big\}.
        $}
    \EndFor
    \State{
    $\bar a(c) \leftarrow \frac{1}{B}\sum_{b=1}^B a^{(b)}(c).
    $}
    \If{$\bar a(c)\ge 1-\alpha$}
        \State{\textbf{Return} $\widehat{C}_\tau \leftarrow c$}
    \EndIf
\EndFor

\State{If no $c \in \mathcal{C}$ satisfies the constraint, enlarge $\mathcal{C}$ and repeat \textbf{Stage IV}.}
\OUTPUT{$\widehat{C}_\tau$}
\end{algorithmic}
\end{algorithm}

\FloatBarrier

\section{Extension to time-varying community memberships}\label{app:latent}

In \Cref{sec-nonstat}, we consider a non-stationary $\mbox{AR}(1)\mbox{-MSBM}$ in which temporal variation arises through the connectivity tensors, while the community membership matrix remains fixed. 
We now extend this framework to a more general setting where both the community memberships and the connectivity tensors are allowed to evolve over time. 

Specifically, we consider the following non-stationary $\mbox{AR}(1)\mbox{-MSBM}$.

\begin{definition}[Non-stationary first-order autoregressive multilayer stochastic block models (non-stationary $\mbox{AR}(1)\mbox{-MSBM}$)]\label{def-armsb-ext}
A non-stationary $\mbox{AR}(1)\mbox{-MSBM}$ is a sequence of adjacency tensors $\{\mathbf{A}^t \}_{t \geq 0} \subset\{0, 1\}^{n \times n \times L}$ defined recursively for $t\geq 1$ by
\[
\mathbf{A}^t_{i,j, l} = \mathbf{A}^{t-1}_{i,j,l} \mathbbm{1}\{\mathbf{E}^t_{i,j, l} = 0\} + \mathbbm{1}\{\mathbf{E}^t_{i,j, l} = 1\}, \quad \forall  1 \leq i \leq  j \leq n, \ l \in [L],
\]
where $\{\mathbf{E}^t_{i,j, l}\} \subset \{-1, 0, 1\}$ are mutually independent with
\[
\P \{ \mathbf{E}^t_{i,j, l}  = 1 \} = \boldsymbol{\Theta}^t_{i, j, l} = Z_i^t \mathbf{W}^t_{:, :, l} Z_j^t, \quad 
\P \{ \mathbf{E}^t_{i,j, l}  = -1 \} = \boldsymbol{\Delta}^t_{i, j, l} = Z_i^t \mathbf{M}^t_{:, :, l} Z_j^t
\]
and
\[
\P \{ \mathbf{E}^t_{i,j, l}  = 0  \} = 1 - \boldsymbol{\Theta}_{i, j, l}^t - \boldsymbol{\Delta}_{i, j, l}^t.
\]
The initial tensor $\mathbf{A}^0 \in \{0, 1\}^{n \times n \times L}$ is defined so that
\[
 c_{\min} \leq \mathbb{E}\{\mathbf{A}_{i,j,l}^{0}\} \leq 1-c_{\min}, \quad \forall  1 \leq i \leq  j \leq n, \ l \in [L],
\]
Here, $\{Z^t\}_{t\geq1} \subset \{0, 1\}^{n \times K}$ are time-varying memberships and $\{\mathbf{W}^t\}_{t\geq1}, \{\mathbf{M}^t\}_{t\geq1} \subset [c_{\min},1-c_{\min}]^{K\times K\times L}$ are time-varying connectivity tensors, with an absolute constant $c_{\min}\in(0,1/2)$.
\end{definition}

The definition of quasi-stationary segmentation remains unchanged (see \Cref{def-seg}) and we denote by $G$ the number of $V$-quasi-stationary segments in $\{(\boldsymbol{\Theta}^t, \boldsymbol{\Delta}^t)\}_{t=1}^{T}$.

Assumptions~\ref{ass-bias}$(i)$--$(iii)$ extend directly to this setting. To ensure the validity of the tensor-based refinement step (TH-PCA) in \Cref{alg:fast-msbm}, we strengthen the low-rank condition in \Cref{ass-bias}$(iv)$. Specifically, we replace \Cref{ass-bias}$(iv)$ with the following assumption.

\begin{assumption}\label{ass-bias-new}
For any $t \in [T]$ and $s \in [t]$, define the window-averaged tensors
\[
\mathbf{\Theta}^{t,s} = \frac{1}{s}\sum_{u=t-s+1}^t \boldsymbol{\Theta}^u, 
\quad
\mathbf{\Delta}^{t,s} = \frac{1}{s}\sum_{u=t-s+1}^t \boldsymbol{\Delta}^u.
\]
Assume for any $t \in [T]$,  $s_{t, \max}\le C_\sigma s_{t, \min}$, where $C_{\sigma} > 0$ is an absolute constant, $s_{t, \min}=\min_{k\in[K]} s_{t, k}$, $s_{t, \max}=\max_{k\in[K]} s_{t, k}$, and for any $k \in [K]$, $s_{t,k}=\sum_{i=1}^n Z_{i,k}^t$.Moreover, for all $t \in [T]$ and $s \in [t]$, assume that
 \[
    \mathrm{rank}\big(\mathcal M_1(\mathbf{\Theta}^{t,s})\big)
    =
     \mathrm{rank}\big(\mathcal M_1(\mathbf{\Delta}^{t,s})\big)
    =
    \mathrm{rank}\big(\mathcal M_1(\mathbf{\Theta}^{t,s}+\mathbf{\Delta}^{t,s})\big)
    = K.
\]
Define
    \[
    r^{t,s} = \rank(\mathcal M_3(\mathbf{\Theta}^{t,s})), \quad 
    \tilde r^{t,s} = \rank(\mathcal M_3(\mathbf{\Delta}^{t,s})).
    \]
Finally, for all $t \in [T]$ and $s \in [t]$, assume that 
    \[
    \max\bigg\{
    \frac{\|\mathbf{\Theta}^{t,s}\|}{\sigma_{\min}(\mathbf{\Theta}^{t,s})},
    \frac{\|\mathbf{\Delta}^{t,s}\|}{\sigma_{\min}(\mathbf{\Delta}^{t,s})},
    \frac{\|\mathbf{\Theta}^{t,s}+\mathbf{\Delta}^{t,s}\|}{\sigma_{\min}(\mathbf{\Theta}^{t,s}+\mathbf{\Delta}^{t,s})}
    \bigg\}
    \le C_\sigma.
    \]
\end{assumption}

Compared with \Cref{ass-bias}$(iv)$ in the main text, these conditions are imposed directly on the transition probability tensors $(\boldsymbol{\Theta}^{t,s}, \boldsymbol{\Delta}^{t,s})$ rather than on the lower-dimensional connectivity tensors $(\mathbf{W}^{t,s}, \mathbf{M}^{t,s})$. 
This reflects the additional complexity introduced by time-varying community memberships and therefore constitutes a stronger assumption.

Under Assumptions~\ref{ass-bias} $(i)$--$(iii)$ together with the \Cref{ass-bias-new}, the theoretical guarantees established in the main text continue to hold. 
In particular, \Cref{alg:fast-msbm} achieves the same estimation rates for the transition probability tensors as in \Cref{thm-tensor-Frobenius} and the same community recovery guarantees as in \Cref{thm-comunity-recovery-non}.
This establishes that our framework accommodates dynamic multilayer networks with time-varying community structures, without requiring any modification to the core algorithms or the corresponding statistical guarantees.

\section{Technical details for results in Section \ref{sec-stat}}\label{app-proof-sec-stat}

All necessary auxiliary results for the stationary setting are collected in Appendix~\ref{sec-app-add-stat}, with proofs provided in Appendix~\ref{app-proof-app-stat}. The proofs of \Cref{thm-tensor-Frobenius-stat}, \Cref{thm-minimax-lower} and \Cref{thm-comunity-recovery-stat} are given in Appendices~\ref{app-proof-thm-tensor-Frobenius-stat}, \ref{app-proof-thm-minimax-lower} and \ref{app-proof-thm-comunity-recovery-stat}, respectively.

\subsection{Additional results}\label{sec-app-add-stat}

We first establish several basic properties of the stationary AR(1)–MSBM, which will be used throughout the analysis. 
\begin{proposition}\label{prop-stat}
Let the process $\{\mathbf{A}^t \}_{t \geq 0} \subset\{0, 1\}^{n \times n \times L}$ be defined in \Cref{def-sarmsb}. Then it is a strictly stationary process. Moreover, for any $ 1 \leq i \leq j \leq n$, $ l \in [L]$ and $t,s \geq 0$, we have
\[
    \boldsymbol{\Pi}_{i,j,l}^t = \mathbb{P} \{ \mathbf{A}_{i,j, l}^t = 1\}  =\frac{\boldsymbol{\Theta}_{i, j, l}}{ \boldsymbol{\Theta}_{i, j, l} + \boldsymbol{\Delta}_{i, j, l}}, \quad
    \mathrm{Var}(\mathbf{A}_{i, j, l}^t) =  \frac{\boldsymbol{\Theta}_{i, j, l} \boldsymbol{\Delta}_{i, j, l}}{ (\boldsymbol{\Theta}_{i, j, l} + \boldsymbol{\Delta}_{i, j, l})^2},
\]
and
\[
    \rho_{i,j, l}(|t-s|) = \mathrm{Corr}(\mathbf{A}_{i,j, l}^t, \mathbf{A}_{i,j, l}^s) = (1 -  \boldsymbol{\Theta}_{i, j, l} - \boldsymbol{\Delta}_{i, j, l} )^{|t-s|}.
\]
\end{proposition}

The proof follows directly from the same argument as Proposition 2 in \cite{jiang2023autoregressive} and is omitted.

\medskip

To analyze temporal dependence, we recall the notion of strong mixing, which quantifies the decay of dependence across time. The following lemma establishes that each edgewise process in the stationary AR(1)-MSBM is strongly mixing.

\begin{definition}[Strongly mixing]\label{def-mixing}
Let $\{X_t\}_{t \in \mathbb{N}}$ be a stochastic process. 
The $\alpha$-mixing coefficient at lag $k \in \mathbb{Z}^+$ is defined as 
\[
\alpha(k) = \sup_{t \in \mathbb{Z}^+} \sup_{\mathcal{A} \in \mathcal{F}_0^t, \mathcal{B} \in \mathcal{F}_{t+k}^\infty} |\mathbb{P}(\mathcal{A} \cap \mathcal{B}) - \mathbb{P}(\mathcal{A})\mathbb{P}(\mathcal{B})|,
\]
where $\mathcal{F}_{0}^{t} = \sigma(X_0, X_1, \cdots, X_t)$ is the $\sigma$-algebra generated by the process up to time $t$ and $\mathcal{F}_{t+k}^{\infty} = \sigma(X_{t+k}, X_{t+k+1}, \cdots)$ is the $\sigma$-algebra generated by the process from time $t+k$ onward. The process $\{X_t\}_{t \in \mathbb{N}}$ is is said to be strongly mixing if
\[
  \alpha(k) \to 0 \quad \mbox{as} \quad k \to \infty.
\]
\end{definition}

\begin{lemma}\label{lemma-mixing-stat}
Let the process $\{\mathbf{A}^t \}_{t \geq 0} \subset\{0, 1\}^{n \times n \times L}$ be defined in \Cref{def-sarmsb}. For any $ 1 \leq i \leq j \leq n$ and $l \in [L]$, the univariate edge process $\{\mathbf{A}^t_{i, j, l} \}_{t \geq 0} \subset\{0, 1\}^{n \times n \times L}$ is  strongly mixing. 
\end{lemma}

The proof of \Cref{lemma-mixing-stat} follows the same argument as that of \Cref{lemma-mixing-nonstat} and is in fact simpler; we omit the details for brevity.

\medskip

We next establish theoretical guarantees for the initial estimators of the transition probabilities, $\widehat{\boldsymbol{\Theta}}^t$ and $\widehat{\boldsymbol{\Delta}}^t$, corresponding to \textbf{Stage I} of \Cref{alg:fast-stat-msbm}.  This result extends the single-layer analysis in Proposition 6 of \cite{jiang2023autoregressive} to the multilayer setting.

\begin{proposition}\label{lemma-sub-Gaussian}
Let the process $\{\mathbf{A}^t \}_{t \geq 0} \subset\{0, 1\}^{n \times n \times L}$ be defined in \Cref{def-sarmsb}.  For any $t \in  [T]$ with $t \{\log(t) \log\log(t)\}^{-2} \gtrsim \log(n \vee L \vee T )$, there exist absolute constants $C, c >0$ such that 
\begin{align*}
    \P \Big\{ &  \big\vert  \widehat{\boldsymbol{\Theta}}_{i,j, l}^t -  \boldsymbol{\Theta}_{i,j, l}  \big\vert +  \big\vert  \widehat{\boldsymbol{\Delta}}_{i,j, l}^t - \boldsymbol{\Delta}_{i,j, l} \big\vert >  C\sqrt{ \frac{ \log(n \vee L \vee T)}{t}},    \forall 1 \leq i \leq j \leq n, l\in [L] \bigg\}
\leq (n\vee L \vee T)^{-c}.
\end{align*}
Moreover, for any $1 \leq i \leq j \leq n$, $l \in [L]$ and $t \in [T]$, both $ \widehat{\boldsymbol{\Theta}}_{i,j, l}^t -  \boldsymbol{\Theta}_{i,j, l}$ and $\widehat{\boldsymbol{\Delta}}_{i,j, l}^t -  \boldsymbol{\Delta}_{i,j, l}$ are $C^* t^{-1/2}$-sub-Gaussian distributed for some absolute constant $C^* >0$.
\end{proposition}

We next study spectral properties of the transition probability tensors. Define
\begin{equation}\label{def-U-Z}
    U_Z = Z D_Z^{-1}, \quad \mbox{where } D_Z = \mbox{diag}(\sqrt{s_1}, \dots, \sqrt{s_K}),
\end{equation}
and let the singular value decomposition of   $\mathcal{M}_3(\mathbf{W})$ be 
$\mathcal{M}_3 (\mathbf{W}) = U_W D_W V_W^{\top}$,
with $U_W \in \mathbb{O}_{L \times r_1}$, $V_W \mathbb{O}_{n^2 \times r_1} $ and  $D_W \in \R^{r_1 \times r_1}$. Let $\mathbf{S} \in \mathbb{R}^{K \times K \times r_1}$ satisfy
$
\mathcal{M}_3(\mathbf{S}) = D_W V_W^{\top}$,
and define
\[
\mathbf{R} = \mathbf{S} \times_1 D_Z \times_2 D_Z \in \mathbb{R}^{K \times K \times r_1}.
\]
Under \Cref{ass-stat-rank}$(i)$, the transition probability tensor $\boldsymbol{\Theta}$ admits the Tucker decomposition
\begin{equation}\label{eq-tucker-decom}
   \boldsymbol{\Theta} =  \mathbf{R} \times_1  U_Z \times_2 U_Z \times_3 U_W \quad \mbox{with Tucker ranks }(K, K, r_1),
\end{equation}
and an analogous representation holds for $\boldsymbol{\Delta}$.

\begin{lemma}\label{lemma-singular-values-stat}
Let $\mathbf{Q} \in \mathbb{R}^{K\times K \times L}$ be defined in \Cref{ass-stat-rank} and let $\boldsymbol{\Omega} \in \mathbb{R}^{n \times n \times L}$ be given by $\boldsymbol{\Omega} = \boldsymbol{\Theta} + \boldsymbol{\Delta}$. Assume that $\mathrm{rank}\big( \mathcal{M}_1 (\mathbf{Q}) \big) = K$ and let $r = \mathrm{rank}\big( \mathcal{M}_3 (\mathbf{Q}) \big)$.   Then, we have that 
\[
\mathrm{rank} \big(\mathcal{M}_1(\mathbf{\Omega})\big)= \mathrm{rank} \big(\mathcal{M}_2(\mathbf{\Omega})\big)= K,
\quad \mathrm{rank} \big(\mathcal{M}_3(\mathbf{\Omega})\big)= r, 
\]
and
\[ s_{\min} \sigma_{\min} (\mathbf{Q}) \leq  \sigma_{\min}\big(\mathcal{M}_s(\boldsymbol{\Omega})\big)  \leq \big\| \mathcal{M}_s(\boldsymbol{\Omega}) \big\| \leq   s_{\max} \| \mathbf{Q}\|, \quad \forall s \in [3],
\] 
where $s_{\min} = \min_{k \in [K]} s_k$ and $s_{\max} = \max_{k \in [K]} s_k$ with for any $k \in [K]$, $s_k = \sum_{i=1}^{n} Z_{i, k}$.
\end{lemma}

Under \Cref{ass-stat-rank}, the estimated subspace $\widehat{U}_Z^t$ obtained via H-PCA (\textbf{Stage II} of \Cref{alg:fast-stat-msbm}) satisfies the following bound.

\begin{proposition}\label{prop-sin-theta-stat}
Let the process $\{\mathbf{A}^t \}_{t \geq 0} \subset\{0, 1\}^{n \times n \times L}$ be defined in \Cref{def-sarmsb}. Suppose that \Cref{ass-stat-rank} holds. 
Let $U_Z$ be defined in \eqref{def-U-Z}. For any $t \in [T]$, under \Cref{alg:fast-stat-msbm}, with probability at least $1- (n \vee L \vee T)^{-c}$,
\[
\inf_{O \in \mathbb{O}_{K \times K}}\| \widehat{U}^t_Z - U_Z O   \| \leq   \mathcal{E}_t
\]
where
\begin{align}\label{eq-sin-theta-stat}
\mathcal{E}_t =  C  \sigma_Q^{-2} \bigg( \frac{ K^2 \sqrt{L} \log^2(n \vee L \vee T)}{tn}  + \frac{K^2L}{t^2}\bigg) +  C\sigma_Q^{-1}  \bigg(K\sqrt{\frac{  \log(n \vee L \vee T)}{tn}} + \frac{K\sqrt{L}}{t} \bigg), 
\end{align}
with $\sigma_Q = \sigma_{\min}(\mathbf{Q})$ and absolute constants $C, c >0$. 
\end{proposition}

When $K$ is treated as a constant and $\sigma_Q \asymp \sqrt{L}$, the bound in \eqref{eq-sin-theta-stat} simplifies (up to logarithmic factors) to  
\[
    \frac{1}{\sqrt{tnL}} + \frac{1}{t}.
\]
The first term reflects stochastic variability, while the second captures the bias inherent to the initial estimators, which are maximum likelihood estimators (MLEs) of the autoregressive model. 
For comparison, the single-layer result in  Theorem 10 od \cite{jiang2023autoregressive} achieved an error rate of order
\[
    \frac{1}{\sqrt{tn}} + \frac{1}{t} + \frac{1}{n}.
\]
Our multilayer setting yields a stochastic term that scales as $(t n L)^{-1/2}$, reflecting the additional information aggregated across $L$ layers. Moreover, our approach does not incur the extra $n^{-1}$ bias term present in \cite{jiang2023autoregressive}, since we perform spectral estimation directly on the initial estimators instead on a normalized Laplacian matrix, thereby avoiding the bias induced by diagonal normalization.

Finally, a key ingredient for proving the minimax lower bound in \Cref{thm-minimax-lower} is a packing argument over the space of community labels. To construct the required packing set, we establish \Cref{lemma-packing}, whose proof relies on \Cref{lemma-binomial}.

\begin{lemma}\label{lemma-packing}
Let $k \geq 2$ and $n \in \Z^{+}$. A labeling is a function $Z : [n] \to [k]$ that assigns one of the $k$ labels to each of the $n$ indices. 
Let $q = \lfloor n/k\rfloor$ and $r = n -kq$. Define $\mathcal{Z}$ to be the collection of all balanced labelings, i.e.
\begin{align}
\mathcal{Z} =  \Big\{ Z \colon [n] \to [k]  \Big\vert  &   \big\vert\{ i \in [n] : Z(i) = a \} \big\vert  \in \{ q, q+1 \},   \forall a \in [k], \nonumber\\
& \mbox{ and }  \big\vert \big\{a \in [k]  \colon  \vert \{ i \in [n] : Z(i) = a \}  \vert  = q+1 \big\} \big\vert =r  \Big\}. \nonumber
\end{align}
Then there exist absolute constants $c,  c' > 0$  
and a subset $\mathcal{Z}^{*} \subset \mathcal{Z}$ such that 
\[\vert \mathcal{Z}^{*} \vert  \geq \exp\{c n \log (k)\}
\quad  \mbox{and} \quad
  d_H(Z, Z') \geq c' n, \quad \forall Z \neq Z' \in \mathcal{Z}^{*},
\]
where $d_H(Z, Z') = \vert \{ i \in [n] : Z(i) \neq Z'(i) \}\vert $ is the Hamming distance.
\end{lemma}

\begin{lemma}\label{lemma-binomial}
For any $n \in \Z^+$ and  $ 0 < \delta < 1/2 $, it holds that  
\[
\sum_{s=0}^{\lfloor \delta n \rfloor} \binom{n}{s} \leq \exp \big\{n H(\delta)\big\}, 
\quad \mbox{where } H(\delta) = -\delta \log \delta - (1-\delta) \log(1-\delta).
\]
\end{lemma}

\subsection{Proof of Theorem \ref{thm-tensor-Frobenius-stat}}\label{app-proof-thm-tensor-Frobenius-stat}
\begin{proof}[Proof of Theorem \Cref{thm-tensor-Frobenius-stat}]

Recall that 
\[
\widetilde{\boldsymbol{\Theta}}^t =  \widehat{\boldsymbol{\Theta}}^{t} \times_1 \widehat{U}^{t}_Z (\widehat{U}^{t}_Z)^{\top}\times_2 \widehat{U}^t_Z(\widehat{U}^{t}_Z)^\top  \times_3 \widehat{U}^t_W (\widehat{U}^{t}_W)^\top.
\]
Hence, 
\begin{align}\label{eq-Frobenius-0}
\|\widetilde{\boldsymbol{\Theta}}^t-\boldsymbol{\Theta} \|_{\mathrm{F}} \leq &  
\big\|\boldsymbol{\Theta}\times_1 \widehat{U}^{t}_Z (\widehat{U}^{t}_Z)^{\top}\times_2 \widehat{U}^t_Z(\widehat{U}^{t}_Z)^\top  \times_3 \widehat{U}^t_W (\widehat{U}^{t}_W)^\top - \boldsymbol{\Theta}\big\|_{\mathrm{F}}
\nonumber\\
& \hspace{0.5cm}
+ \big\|(\widehat{\boldsymbol{\Theta}}^{t} - \boldsymbol{\Theta} )\times_1 \widehat{U}^{t}_Z (\widehat{U}^{t}_Z)^{\top}\times_2 \widehat{U}^t_Z(\widehat{U}^{t}_Z)^\top  \times_3 \widehat{U}^t_W (\widehat{U}^{t}_W)^\top \big\|_{\mathrm{F}} 
\nonumber\\
= & (I) +(II).
\end{align}

\medskip
\noindent
\textbf{Step 1.} In this step, we first consider the term $(I)$ in \eqref{eq-Frobenius-0}. Observe that 
\begin{align}\label{eq-Frobenius-1}
(I) = & \bigl\|\boldsymbol{\Theta}-\boldsymbol{\Theta}\times_1 P_{\widehat{U}^{t}_Z}\times_2 P_{\widehat{U}^{t}_Z}\times_3 P_{\widehat{U}^{t}_W}\bigr\|_{\mathrm{F}} \nonumber\\
=&  \bigl\| \boldsymbol{\Theta}\times_1 P_{\widehat{U}^{t}_{Z\perp}}
+ \boldsymbol{\Theta}\times_1 P_{\widehat{U}^{t}_Z}\times_2 P_{\widehat{U}^{t}_{Z\perp}} +
\boldsymbol{\Theta}\times_1 P_{\widehat{U}^{t}_Z}\times_2 P_{\widehat{U}^{t}_Z}\times_3 P_{\widehat{U}^{t}_{W\perp}}
\bigr\|_{\mathrm{F}} \nonumber\\
\leq & \big\|\boldsymbol{\Theta}\times_1 P_{\widehat{U}^{t}_{Z\perp}} \big\|_{\mathrm{F}}
+\big\|\boldsymbol{\Theta}\times_2P_{\widehat{U}^{t}_{Z\perp}} \big\|_{\mathrm{F}} +\big\|\boldsymbol{\Theta}\times_3 P_{\widehat{U}^{t}_{W\perp}}\big\|_{\mathrm{F}} \nonumber\\
=&\big\| \widehat{U}^{t}_{Z\perp}\mathcal{M}_1(\boldsymbol{\Theta})\big\|_{\mathrm{F}} + \big\|\widehat{U}^{t}_{Z\perp}\mathcal{M}_2(\boldsymbol{\Theta})\big\|_{\mathrm{F}} + \big\|\widehat{U}^{t}_{W\perp}\mathcal{M}_1(\boldsymbol{\Theta})\big\|_{\mathrm{F}} \nonumber\\
= & (I.1) + (I.2) + (I,3).
\end{align}
For the term $(I.1)$, we have that
\begin{align}\label{eq-Frobenius-1.1}
(I.1)  = &  \big\|(\widehat{U}^{t}_{Z\perp})^{\top} U_ZU_Z^{\top}\mathcal{M}_1(\boldsymbol{\Theta})\big\|_{\mathrm{F}} \leq \big\|\widehat{U}^{t}_{Z\perp} U_Z\big\|_{\mathrm{F}} \big\|U_Z^{\top}\mathcal{M}_1(\boldsymbol{\Theta})\big\| \nonumber\\
= & \big\| \sin\Theta( \widehat{U}^{t}_{Z}, U_Z^{\top}) \big\|_{\mathrm{F}} \big\|\mathcal{M}_1(\boldsymbol{\Theta})\big\| 
\leq  s_{\max} \big\|  \sin\Theta( \widehat{U}^{t}_{Z}, U_Z^{\top}) \big\|_{\mathrm{F}}  \|\mathbf{W}\|,
\end{align}
where the first inequity follows from \Cref{lemma-singular-values-stat}.
Similarly, we obtain that
\begin{align}\label{eq-Frobenius-1.2}
(I.2) \leq  s_{\max} \big\| \sin\Theta( \widehat{U}^{t}_{Z}, U_Z^{\top}) \big\|_{\mathrm{F}} \|\mathbf{W}\| \quad \mbox{and} \quad
 (I.3) \leq  s_{\max}  \big\| \sin\Theta( \widehat{U}^{t}_{W}, U_W^{\top}) \big\|_{\mathrm{F}} \|\mathbf{W}\|.
\end{align}

Under \Cref{ass-stat-rank}, by \Cref{prop-sin-theta-stat} we have that 
\begin{equation}\label{eq-Frobenius-1.4}
   \P \big\{ \big \| \sin \Theta \big(\widehat{U}^t_Z, U_Z )  \big\|_{\mathrm{F}} >  \mathcal{E}_t\big\} \leq (n \vee L \vee T)^{-c_1},
\end{equation}
where
\[
\mathcal{E}_t =  C_1  \sigma^{-2} \bigg( \frac{ K^2 \sqrt{K} \sqrt{L} \log^2(n \vee L \vee T)}{tn}  + \frac{K^2L}{t^2}\bigg) + C_1 \sigma^{-1}  \bigg(K\sqrt{\frac{ K \log(n \vee L \vee T)}{tn}} + \frac{K\sqrt{L}}{t} \bigg), 
\]
with absolute constants $C_1, c_1 >0$.
Similarly,   under \Cref{ass-stat-rank}, we have that 
\begin{equation}\label{eq-Frobenius-1.5}
   \P \big\{ \big \| \sin \Theta \big(\widehat{U}^t_W, U_W )  \big\|_{\mathrm{F}}  >   \mathcal{E}_t'\big\} \geq (n \vee L \vee T)^{-c_1},
\end{equation}
where
\[
\mathcal{E}_t' =  C_1  \sigma^{-2} \bigg( \frac{K^2  \sqrt{Lr} \log^2(n \vee L \vee T)}{tn}  + \frac{K^2L}{t^2} \bigg) +   C_1 \sigma^{-1}   \bigg(K\sqrt{\frac{ L r\log(n \vee L \vee T)}{tn^2}} + \frac{K\sqrt{L}}{t} \bigg).  
\]
Combining \eqref{eq-Frobenius-1}, \eqref{eq-Frobenius-1.1}, \eqref{eq-Frobenius-1.2}, \eqref{eq-Frobenius-1.4} and \eqref{eq-Frobenius-1.5} yields that 
\begin{align}\label{eq-Frobenius-1.6}
   & \P \bigg\{ (I) >    C_2  \sigma^{-1}  \bigg( \frac{  K\sqrt{(K\vee r)L} \log^2(n \vee L \vee T)}{t}  + \frac{K  Ln}{t^2} \bigg)   \nonumber\\
  & \hspace{0.5cm}  +   C_2    \bigg( \sqrt{\frac{ (Lr+ nK)  \log(n \vee L \vee T)}{t}} + \frac{\sqrt{L}n}{t} \bigg)\bigg\} \geq  2(n \vee L \vee T)^{c},  
\end{align}
where $C_2 >0$ is an absolute constant. 

\medskip
\noindent
\textbf{Step 2.}
In this step, we focus on the term $(II)$ in \eqref{eq-Frobenius-0}. Note that 
\begin{align}\label{eq-Frobenius-2.0}
(II) = & \sup_{\substack{\mathbf{V}\in\R^{n\times n \times L}\\ \|\mathbf{V}\|_{\mathrm{F}}\leq 1}}
\big\langle \widehat{\boldsymbol{\Theta}}^{t} - \boldsymbol{\Theta}, \mathbf{V} \times_1 P_{\widehat{U}^t_Z} \times_2 P_{\widehat{U}^t_Z} \times_3 P_{\widehat{U}^t_W} \big\rangle \nonumber\\
\leq & \sup_{\substack{\mathbf{V}\in\R^{n\times n \times L}\\ \|\mathbf{V}\|_{\mathrm{F}}\leq 1}}
\big\langle \widehat{\boldsymbol{\Theta}}^{t} - \boldsymbol{\Theta} - \E \{\widehat{\boldsymbol{\Theta}}^{t}\} + \boldsymbol{\Theta}, \mathbf{V} \times_1 P_{\widehat{U}^t_Z} \times_2 P_{\widehat{U}^t_Z} \times_3 P_{\widehat{U}^t_W} \big\rangle \nonumber\\
& \hspace{0.5cm} + \sup_{\substack{\mathbf{V}\in\R^{n\times n \times L}\\ \|\mathbf{V}\|_{\mathrm{F}}\leq 1}}
\big\langle  \E \{\widehat{\boldsymbol{\Theta}}^{t}\} - \boldsymbol{\Theta}, \mathbf{V} \times_1 P_{\widehat{U}^t_Z} \times_2 P_{\widehat{U}^t_Z} \times_3 P_{\widehat{U}^t_W} \big\rangle \nonumber\\
= & (II.1) + (II.2).
\end{align}
For the term $(II.1)$ in \eqref{eq-Frobenius-2.0}, by 
Lemma E.5 in the Supplement of \cite{han2022optimal}, we have that
\begin{align}\label{eq-Frobenius-2.1}
    \P \bigg\{(II.1) > C_3  \sqrt{\frac{( K^2r_1 +nK+ Lr_1 ) \log(n \vee L \vee T)}{t}} \bigg\} \leq (n \vee L \vee T)^{-c_3},
\end{align}
where $C_3, c_3 >0$ are absolute constants.
For the term $(II.2)$ in \eqref{eq-Frobenius-2.0}, we can derive for some absolute constant $C_4 >0$ that
\begin{align}\label{eq-Frobenius-2.2}
    (II.2) \leq & \big\| \E \{\widehat{\boldsymbol{\Theta}}^{t}\} - \boldsymbol{\Theta} \big\|_{\mathrm{F}} \sup_{\substack{\mathbf{V}\in\R^{n\times n \times L}\\ \|\mathbf{V}\|_{\mathrm{F}}\leq 1}}\big\|\mathbf{V} \times_1 P_{\widehat{U}^t_Z} \times_2 P_{\widehat{U}^t_Z} \times_3 P_{\widehat{U}^t_W} \big\|_{\mathrm{F}} \leq C_4 \frac{n\sqrt{L}}{t}, 
\end{align}
where the last inequality follows from \eqref{prop-sin-1.5}.
Combining \eqref{eq-Frobenius-2.0}, \eqref{eq-Frobenius-2.1}, \eqref{eq-Frobenius-2.2} yields that
\begin{align}\label{eq-Frobenius-2.3}
  \P\bigg\{ (II) >    C_3  \sqrt{\frac{( K^2r_1 +nK+ Lr_1 ) \log(n \vee L \vee T)}{t}}  + C_4 \frac{n\sqrt{L}}{t}   \bigg\}  \leq   (n \vee L \vee T)^{-c}.
\end{align}

\medskip
\noindent
\textbf{Step 3.}
Combining \eqref{eq-Frobenius-0}, \eqref{eq-Frobenius-1.6} and \eqref{eq-Frobenius-2.3}, we obtain for some absolute constants $C_5, c_5 >0$ that 
\begin{align}\label{eq-for-final-1}
       & \P \big\{ \|\widetilde{\boldsymbol{\Theta}}^t-\boldsymbol{\Theta} \|_{\mathrm{F}} >    C_5  \sigma^{-1} \bigg( \frac{  K\sqrt{(K\vee r)L} \log^2(n \vee L \vee T)}{t}  + \frac{K Ln}{t^2} \bigg)   \nonumber\\
  & \hspace{0.5cm}  +   C_5    \bigg( \sqrt{\frac{ (nK + Lr +K^2r)  \log(n \vee L \vee T)}{t}} + \frac{\sqrt{L}n}{t} \bigg)\big\} \geq  (n \vee L \vee T)^{c_5}.  
\end{align}
Similarly, we can derive that
\begin{align}\label{eq-for-final-2}
       & \P \big\{ \|\widetilde{\boldsymbol{\Delta}}^t-\boldsymbol{\Delta} \|_{\mathrm{F}} >    C_5  \sigma^{-1} \bigg( \frac{  K\sqrt{(K\vee r)L} \log^2(n \vee L \vee T)}{t}  + \frac{K Ln}{t^2} \bigg)   \nonumber\\
  & \hspace{0.5cm}  +   C_5    \bigg( \sqrt{\frac{ (nK + Lr +K^2r)  \log(n \vee L \vee T)}{t}} + \frac{\sqrt{L}n}{t} \bigg)\big\} \geq  (n \vee L \vee T)^{c_5}.  
\end{align}
Combining \eqref{eq-for-final-1} and \eqref{eq-for-final-2}, we complete the proof.

\end{proof}

\subsection{Proof of Proposition \ref{thm-minimax-lower}}\label{app-proof-thm-minimax-lower}

\begin{proof}
    
The proof proceeds in five main steps. Throughout, for any $(\boldsymbol{\Theta},\boldsymbol{\Delta}) \in \mathcal P$, let $\mathcal{P}_{(\boldsymbol{\Theta}, \boldsymbol{\Delta})}$ denote the path distribution of the adjacency tensor sequence $\{\mathbf{A}^t\}_{t \in [T] \cup \{0\}}$ generated by \Cref{def-sarmsb} with transition probabilities $(\boldsymbol{\Theta},\boldsymbol{\Delta})$.

\medskip
\noindent
\textbf{Step 1: Kullback--Leibler(KL) divergence control.} Fix $(\boldsymbol{\Theta},\boldsymbol{\Delta}),(\boldsymbol{\Theta}',\boldsymbol{\Delta}')\in\mathcal P$.
In this step, we compute the KL divergence between the path distributions induced by $(\boldsymbol{\Theta}, \boldsymbol{\Delta})$ and  $(\boldsymbol{\Theta}', \boldsymbol{\Delta}')$.

The chain rule for KL divergence yields that 
\begin{align}\label{proof-lower-step1-0}
& \mathrm{KL}\big(\mathcal{P}_{(\boldsymbol{\Theta}, \boldsymbol{\Delta})}\| \mathcal{P}_{(\boldsymbol{\Theta}', \boldsymbol{\Delta}')} \big) \nonumber\\
= & \mathrm{KL}\big( \pi_{(\boldsymbol{\Theta}, \boldsymbol{\Delta})}(\mathbf{A}^0) \big\| \pi_{(\boldsymbol{\Theta}', \boldsymbol{\Delta}')}(\mathbf{A}^0) \big)
  + \sum_{t=1}^{T} 
\Big\{ \mathrm{KL}\big(\pi_{(\boldsymbol{\Theta}, \boldsymbol{\Delta})}(\mathbf{A}^t\mid \mathbf{A}^{t-1})
               \big\| \pi_{(\boldsymbol{\Theta}', \boldsymbol{\Delta}')}(\mathbf{A}^t\mid \mathbf{A}^{t-1})\big) \Big\} 
               \nonumber\\
= & \sum_{1 \leq i \leq j \leq n, l \in [L]} \mathrm{KL}\Big( \pi_{ (\boldsymbol{\Theta}_{i, j, l}, \boldsymbol{\Delta}_{i, j, l})} \big(\mathbf{A}_{i, j, l}^0\big) \Big\| \pi_{(\boldsymbol{\Theta}_{i, j, l}', \boldsymbol{\Delta}_{i, j, l}')}\big(\mathbf{A}_{i, j, l}^0 \big) \Big)      \nonumber\\
&  \hspace{0.5cm} + \sum_{1 \leq i \leq j \leq n, l \in [L]}  \sum_{t=1}^{T} 
\Big\{ \mathrm{KL}\Big(\pi_{(\boldsymbol{\Theta}_{i, j, l}, \boldsymbol{\Delta}_{i, j, l})} \big(\mathbf{A}_{i, j, l}^t\mid \mathbf{A}_{i, j, l}^{t-1} \big)
               \Big\| \pi_{(\boldsymbol{\Theta}_{i, j, l}', \boldsymbol{\Delta}_{i, j, l}')} \big(\mathbf{A}_{i, j, l}^t\mid \mathbf{A}_{i, j, l}^{t-1} \big)\Big) \Big\} 
               \nonumber\\
= & (I.1) + (I.2),               
\end{align}
where
\begin{itemize}
    \item $\pi_{(\boldsymbol{\Theta},\boldsymbol{\Delta})}(\mathbf A^0)$ is the marginal distribution
    of $\mathbf A^0$ under the parameters $(\boldsymbol{\Theta},\boldsymbol{\Delta})$;
    \item $\pi_{(\boldsymbol{\Theta},\boldsymbol{\Delta})}(\mathbf A^t \mid \mathbf A^{t-1})$ is the conditional distribution of $\mathbf A^t$ given $\mathbf A^{t-1}$ under the parameters $(\boldsymbol{\Theta},\boldsymbol{\Delta})$;
    \item $\pi_{(\boldsymbol{\Theta}_{i,j,l},\boldsymbol{\Delta}_{i,j,l})}(\mathbf A^0_{i,j,l})$ is the marginal distribution
    of the edge variable $\mathbf A^0_{i,j,l}$ under the parameters $(\boldsymbol{\Theta}_{i,j,l},\boldsymbol{\Delta}_{i,j,l})$; and
    \item $\pi_{(\boldsymbol{\Theta}_{i,j,l},\boldsymbol{\Delta}_{i,j,l})}(\mathbf A^t_{i,j,l} \mid \mathbf A^{t-1}_{i,j,l})$
    is the conditional distribution of $\mathbf A^t_{i,j,l}$ given its previous state $\mathbf A^{t-1}_{i,j,l}$ under the parameters $(\boldsymbol{\Theta}_{i,j,l},\boldsymbol{\Delta}_{i,j,l})$.
\end{itemize}

\medskip
\noindent
\textbf{Step 1.1.} In this sub-step, we analyze the term $(I.1)$ in \eqref{proof-lower-step1-0}.

By Lemma 15 in \cite{cai2024transfer}, we have that 
\begin{align}\label{proof-lower-step1-1.1}
(I.1)& = \sum_{1 \leq i \leq j \leq n, l \in [L]}  \mathrm{KL}\Big( \pi_{ (\boldsymbol{\Theta}_{i, j, l}, \boldsymbol{\Delta}_{i, j, l})} \big(\mathbf{A}_{i, j, l}^0\big) \Big\| \pi_{(\boldsymbol{\Theta}_{i, j, l}', \boldsymbol{\Delta}_{i, j, l}')}\big(\mathbf{A}_{i, j, l}^0 \big) \Big)   \nonumber\\
= & \sum_{1 \leq i \leq j \leq n, l \in [L]}   \mathrm{KL}\bigg( \mbox{Bernoulli} \bigg(\frac{\boldsymbol{\Theta}_{i, j, l}}{ \boldsymbol{\Theta}_{i, j, l} + \boldsymbol{\Delta}_{i, j, l}}\bigg)\Big\|  \mbox{Bernoulli} \bigg(\frac{\boldsymbol{\Theta}_{i, j, l}'}{ \boldsymbol{\Theta}_{i, j, l}' + \boldsymbol{\Delta}_{i, j, l}'} \bigg) \bigg)   \nonumber\\
\leq &   \sum_{1 \leq i \leq j \leq n, l \in [L]}   \frac{( \boldsymbol{\Theta}_{i, j, l}' + \boldsymbol{\Delta}_{i, j, l}' )^2 } {\boldsymbol{\Theta}_{i, j, l}' \boldsymbol{\Delta}_{i, j, l}'}   
\bigg( \frac{\boldsymbol{\Theta}_{i, j, l}}{ \boldsymbol{\Theta}_{i, j, l} + \boldsymbol{\Delta}_{i, j, l}}  -   \frac{\boldsymbol{\Theta}_{i, j, l}'}{ \boldsymbol{\Theta}_{i, j, l}' + \boldsymbol{\Delta}_{i, j, l}'} \bigg)^2 \nonumber\\
\leq  &  \sum_{1 \leq i \leq j \leq n, l \in [L]}  \frac{1}{c_{\min}^2} \bigg( \frac{\boldsymbol{\Theta}_{i, j, l}}{ \boldsymbol{\Theta}_{i, j, l} + \boldsymbol{\Delta}_{i, j, l}}  -   \frac{\boldsymbol{\Theta}_{i, j, l}'}{ \boldsymbol{\Theta}_{i, j, l}' + \boldsymbol{\Delta}_{i, j, l}'} \bigg)^2,
\end{align}
where the last inequality follows from $\boldsymbol{\Theta}, \boldsymbol{\Delta}, \boldsymbol{\Theta}', \boldsymbol{\Delta}' \in [c_{\min}, 1-c_{\min}]^{n \times n \times L}$.

Let $g: [c_{\min}, 1 - c_{\min}]^2 \to [0,1]$ de defined by  
\[
   g (\theta, \delta) = \frac{\theta}{ \theta + \delta}.
\]
It follows that 
\begin{equation}\label{proof-lower-step1-1.2}
 \| \nabla g (\theta, \delta)\|_2^2 = \bigg\| \bigg(\frac{\delta}{(\theta+\delta)^2}, -\frac{\theta}{(\theta+\delta)^2} \bigg)\bigg\|_2^2 =  \frac{\delta^2}{(\theta+\delta)^4} + \frac{\theta^2}{(\theta+\delta)^4} \leq \frac{2}{c_{\min}^4}. 
\end{equation}
Fix $1 \leq i \leq j \leq n$ and $l \in [L]$. By the mean value theorem, there exists a point 
$(\widetilde{\boldsymbol{\Theta}}_{i, j, l}, \widetilde{\boldsymbol{\Delta}}_{i, j, l})$
on the line segment between 
$(\boldsymbol{\Theta}_{i, j, l}, \boldsymbol{\Delta}_{i, j, l})$ and 
$(\boldsymbol{\Theta}_{i, j, l}', \boldsymbol{\Delta}_{i, j, l}')$ such that 
\begin{align}\label{proof-lower-step1-1.3}
    \big \vert g (\boldsymbol{\Theta}_{i, j, l}, \boldsymbol{\Delta}_{i, j, l}) - g (\boldsymbol{\Theta}_{i, j, l}', \boldsymbol{\Delta}_{i, j, l}')    \big \vert = & \nabla g(\widetilde{\boldsymbol{\Theta}}_{i, j, l}, \widetilde{\boldsymbol{\Delta}}_{i, j, l})^{\top} \big( \boldsymbol{\Theta}_{i, j, l} - \boldsymbol{\Theta}_{i, j, l}', \boldsymbol{\Delta}_{i, j, l} - \boldsymbol{\Delta}_{i, j, l}'\big)  \nonumber\\
     \leq &  \| g(\widetilde{\boldsymbol{\Theta}}_{i, j, l}, \widetilde{\boldsymbol{\Delta}}_{i, j, l}) \|_2 
     \sqrt{ (\boldsymbol{\Theta}_{i, j, l} - \boldsymbol{\Theta}_{i, j, l}')^2 + (\boldsymbol{\Delta}_{i, j, l} - \boldsymbol{\Delta}_{i, j, l}')^2} \nonumber\\
     \leq &   \frac{\sqrt{2}}{c_{\min}^2}  \sqrt{ (\boldsymbol{\Theta}_{i, j, l} - \boldsymbol{\Theta}_{i, j, l}')^2 + (\boldsymbol{\Delta}_{i, j, l} - \boldsymbol{\Delta}_{i, j, l}')^2},
\end{align}
where the first inequality follows from Cauchy--Schwartz inequality and the last inequality follows from \eqref{proof-lower-step1-1.2}.
Combining \eqref{proof-lower-step1-1.1} and \eqref{proof-lower-step1-1.3} yields that 
\begin{align}\label{proof-lower-step1-1.4}
   (I.1) \leq  2c_{\min}^{-6}  \Big( \| \boldsymbol{\Theta} - \boldsymbol{\Theta}'\|_{\mathrm{F}}^2  + \big\| \boldsymbol{\Delta} - \boldsymbol{\Delta}'\big\|_{\mathrm{F}}^2 \Big).
\end{align}

\medskip
\noindent
\textbf{Step 1.2.} In this sub-step, we analyze the term $(I.2)$ in \eqref{proof-lower-step1-0}.

By  \eqref{eq-tran}, we obtain that 
\begin{align}\label{proof-lower-step1-2.1}
(I.2) = &\sum_{1 \leq i \leq j \leq n, l \in [L]}  \sum_{t=1}^{T}   \Big\{ \pi_{ (\boldsymbol{\Theta}_{i, j, l}, \boldsymbol{\Delta}_{i, j, l})} \big(\mathbf{A}_{i, j, l}^{t-1} = 1 \big)  \mathrm{KL}\Big( \mbox{Bernoulli} \big(1 -   \boldsymbol{\Delta}_{i, j, l} \big)\Big\|  \mbox{Bernoulli} \big(1 -   \boldsymbol{\Delta}_{i, j, l}' \big)\Big)  
\nonumber\\
& +  \pi_{ (\boldsymbol{\Theta}_{i, j, l}, \boldsymbol{\Delta}_{i, j, l})} \big(\mathbf{A}_{i, j, l}^{t-1} = 0 \big)  \mathrm{KL}\Big( \mbox{Bernoulli} \big(\boldsymbol{\Theta}_{i, j, l} \big)\Big\|  \mbox{Bernoulli} \big(  \boldsymbol{\Theta}_{i, j, l}' \big)\Big) \Big\}
\nonumber\\
\leq & \sum_{1 \leq i \leq j \leq n, l \in [L]}  \sum_{t=1}^{T}   \bigg\{  \frac{\boldsymbol{\Theta}_{i, j, l}}{ \boldsymbol{\Theta}_{i, j, l} + \boldsymbol{\Delta}_{i, j, l}} \frac{\big(\boldsymbol{\Delta}_{i, j, l} - \boldsymbol{\Delta}_{i, j, l}' \big)^2}{(1- \boldsymbol{\Delta}_{i, j, l}')\boldsymbol{\Delta}_{i, j, l}' } +  \frac{\boldsymbol{\Delta}_{i, j, l}}{ \boldsymbol{\Theta}_{i, j, l} + \boldsymbol{\Delta}_{i, j, l}} \frac{\big(\boldsymbol{\Theta}_{i, j, l} - \boldsymbol{\Theta}_{i, j, l}' \big)^2}{(1- \boldsymbol{\Theta}_{i, j, l}')\boldsymbol{\Theta}_{i, j, l}' } \bigg\}  \nonumber\\
\leq &  \frac{T}{c_{\min}^3} \Big( \| \boldsymbol{\Theta} - \boldsymbol{\Theta}'\|_{\mathrm{F}}^2  + \big\| \boldsymbol{\Delta}- \boldsymbol{\Delta}'\big\|_{\mathrm{F}}^2 \Big),
\end{align}
where the first inequality follows from Lemma 15 in \cite{cai2024transfer} and  the last inequality follows from $\boldsymbol{\Theta}, \boldsymbol{\Delta}, \boldsymbol{\Theta}', \boldsymbol{\Delta}' \in [c_{\min}, 1-c_{\min}]^{n \times n \times L}$.

\medskip
\noindent
\textbf{Step 1.3.} Combining \eqref{proof-lower-step1-0}, \eqref{proof-lower-step1-1.4} and  \eqref{proof-lower-step1-2.1},  we can derive that 
\begin{align}\label{proof-lower-step1-3}
& \mathrm{KL}\big(\mathcal{P}_{(\boldsymbol{\Theta}, \boldsymbol{\Delta})}\| \mathcal{P}_{(\boldsymbol{\Theta}', \boldsymbol{\Delta}')} \big)  \leq      \big(2c_{\min}^{-6} + Tc_{\min}^{-3}\big) \Big( \| \boldsymbol{\Theta} - \boldsymbol{\Theta}'\|_{\mathrm{F}}^2  + \big\| \boldsymbol{\Delta} - \boldsymbol{\Delta}'\big\|_{\mathrm{F}}^2 \Big).
\end{align}

\medskip
\noindent
\textbf{Step 2: Connectivity probability construction.}
In this step, we construct a lower bound by fixing a balanced community membership matrix $Z$ and varying the connectivity tensors $\mathbf{W}, \mathbf{M} \in [c_{\min}, 1-c_{\min}]^{K \times K \times L}$ with  $\max\{\mbox{rank}_3(\mathbf{W}),\mbox{rank}_3(\mathbf{M}) \} \leq r$.

By the Varshamov--Gilbert Lemma, there exists a collection of tensors $\{\mathbf{S}^{(m)}\}_{m=1}^{M} \subset \{0,1\}^{K \times K \times r}$ such that for some absolute constants $c_1, c_2 >0$, 
\begin{equation}\label{proof-lower-step2-1}
\sum_{ 1 \leq i \leq j \leq K, l \in [r]}\mathbbm{1}
\Big\{\big(\mathbf{S}^{(m)}\big)_{i, j, l} \neq \big(\mathbf{S}^{(m')}\big)_{i, j, l}\Big\} \geq c_1 rK^2, \quad \forall m \neq m' \in [M] 
\end{equation}
and
\begin{equation}\label{proof-lower-step2-2}
 M \geq 2^{c_2K^2r}.
\end{equation}

For each $m \in [M]$, define
\begin{equation}\label{proof-lower-step2-3}
\mathbf{H}^{(m)} = \mathbf{H}^{(0)} + \delta \mathbf{S}^{(m)} \in [0, 1]^{K \times K \times r}, \quad \mbox{where } \big(\mathbf{H}^{(0)} \big)_{i, j, l}= \frac{1}{2},  \quad \forall  1 \leq i \leq j \leq K, l \in [r].
\end{equation}

Next, partition $[L]$ into $r$ disjoint subsets $G_1, \ldots, G_r$ with
\[
 \big\vert \vert G_s\vert - L/r  \big\vert \leq 1, \quad \forall s \in [r].
\]
 Define $U \in \mathbb{R}^{L \times r}$ by
\[
    U_{l, s} =
    \begin{cases}
        1, & l \in G_s,\\
        0, & \text{otherwise.}
    \end{cases} \quad
    \forall l \in [L], s \in [r].
 \]
Then
\begin{equation}\label{proof-lower-step2-4}
U^{\top}U = \mathrm{diag} \big( \vert G_1\vert,  \ldots, \vert G_r\vert \big) \quad \mbox{and} \quad 
\lfloor L/r \rfloor I_r  \preceq U^{\top}U\preceq \lceil L/r \rceil I_r. 
\end{equation}

Now define
\begin{equation}\label{proof-lower-step2-5}
\mathbf{W}^{(m)}= \mathbf{H}^{(m)} \times_3 U, \quad 
\mathbf{M}^{(m)}= \mathbf{H}^{(m)} \times_3 U, 
\end{equation}
and
\begin{equation}\label{proof-lower-step2-6}
\boldsymbol{\Theta}^{(m)}=\mathbf{W}^{(m)} \times_1 Z \times_2 Z, \quad 
\boldsymbol{\Delta}^{(m)}=\mathbf{M}^{(m)} \times_1 Z \times_2 Z.
\end{equation}

Let $s_k = \sum_{i=1}^n Z_{i,k}$ be the size of community $k \in [K]$.  
We then have
\begin{align}
\|\boldsymbol{\Theta}^{(m)}-\boldsymbol{\Theta}^{(m')}\|_{\mathrm{F}}^2
= &   \sum_{i, j \in [n], l \in [L]} \Big(\boldsymbol{\Theta}^{(m)}_{i, j, l}- \boldsymbol{\Theta}^{(m')}_{i, j, l}  \Big)^2
=  \sum_{i, j \in [n], l \in [L]}  \Big(  Z_i^{\top} \big(\mathbf{W}_{:, :, l}^{(m)} - \mathbf{W}_{:, :, l}^{(m')}  \big)   Z_j\Big)^2
\nonumber\\
= & \sum_{a, b \in [K], l \in [L]}    s_a  s_b \big(\mathbf{W}_{a, b, l}^{(m)} - \mathbf{W}_{a, b, l}^{(m')}\big)^2
\nonumber\\
= &  \sum_{a, b \in [K], l \in [L]}    s_a  s_b \bigg\{\sum_{u=1}^r U_{l, u} \ \big(\mathbf{H}_{a, b, u}^{(m)}  -\mathbf{H}_{a, b, u}^{(m')} \big)  \bigg\}^2 \nonumber\\
=& \sum_{a, b \in [K]}   s_a  s_b \delta^2 \sum_{ l \in [L]}  \bigg\{  \sum_{u=1}^r U_{l, u}   \big( \mathbf{S}^{(m)}_{a, b, u} - \mathbf{S}^{(m')}_{a, b, u}  \big) \bigg\}^2 \nonumber
\end{align}
where  the second equality follows from \eqref{proof-lower-step2-6},
the fourth equality follows from \eqref{proof-lower-step2-5} and
the last equality follows from \eqref{proof-lower-step2-3}.

Since \eqref{proof-lower-step2-4} and $\max_{k \in [K]} s_k \leq C_{\sigma} \min_{k \in [K]} s_k$, for any $m\neq m' \in [M]$,  we have that
\begin{align}\label{proof-lower-step2-7}
    \|\boldsymbol{\Theta}^{(m)}-\boldsymbol{\Theta}^{(m')}\|_{\mathrm{F}}^2
    \leq  C_{\sigma}^2 \frac{n^2}{K^2} \lceil L/r \rceil \delta^2  \sum_{a, b \in [K]}  \big\|    \mathbf{S}^{(m)}_{a, b, :} - \mathbf{S}^{(m')}_{a, b, :} \big\|_2^2
    \leq   C_{\sigma}^2 n^2 L \delta^2,
\end{align}
and 
\begin{align}\label{proof-lower-step2-8}
    \|\boldsymbol{\Theta}^{(m)}-\boldsymbol{\Theta}^{(m')}\|_{\mathrm{F}}^2
    \geq  C_{\sigma}^{-2} \frac{n^2}{K^2} \lfloor L/r \rfloor  \delta^2  \sum_{a, b \in [K]}  \big\|   \mathbf{S}^{(m)}_{a, b, :} - \mathbf{S}^{(m')}_{a, b, :} \big\|^2
    \geq   C_{\sigma}'  n^2 L \delta^2,
\end{align}
where the last inequality follows from \eqref{proof-lower-step2-1} and $ C_{\sigma}' > 0$ is an absolute constant. 
An identical argument yields that 
\begin{align}\label{proof-lower-step2-9}
    C_{\sigma}' n^2 L \delta^2  \leq  \|\boldsymbol{\Delta}^{(m)}-\boldsymbol{\Delta}^{(m')}\|_{\mathrm{F}}^2
    \leq  C_{\sigma}^2 n^2 L \delta^2, \quad \forall m\neq m' \ \in [M]
\end{align}

For any $m \neq m' \in [M]$, combining  \eqref{proof-lower-step1-3}, \eqref{proof-lower-step2-7} and \eqref{proof-lower-step2-9},  the KL divergence between $\mathcal{P}_{(\boldsymbol{\Theta}^{(m)}, \boldsymbol{\Delta}^{(m)})}$ and $\mathcal{P}_{(\boldsymbol{\Theta}^{(m')}, \boldsymbol{\Delta}^{(m')})}$ is bounded by
\[
\mathrm{KL}\big(\mathcal{P}_{(\boldsymbol{\Theta}^{(m)}, \boldsymbol{\Delta}^{(m)})}\| \mathcal{P}_{(\boldsymbol{\Theta}^{(m')}, \boldsymbol{\Delta}^{(m')})} \big) 
\leq   2  \big(2c_{\min}^{-6} + Tc_{\min}^{-3}\big) C_{\sigma}^2 n^2 L \delta^2
\leq  C_1 T n^2 L \delta^2,
\]
where $C_1 =  2(2c_{\min}^{-6} +c_{\min}^{-3}) C_{\sigma}^2$. 

Since $\log(M) \geq  c_2 \log(2) rK^2$, choosing 
\[
\delta^2 = \frac{ c_2 \log(2) rK^2 }{ 8 C_1 T  Ln^2}, 
\]
ensures that
 \[
\frac{1}{M^2}\sum_{m, m' \in [M]} \mathrm{KL}\big(\mathcal{P}_{(\boldsymbol{\Theta}^{(m)}, \boldsymbol{\Delta}^{(0)})}\| \mathcal{P}_{(\boldsymbol{\Theta}^{(m')},\boldsymbol{\Delta}^{(0)})} \big)  \le \log(M)/8,
 \]
By Fano’s lemma, it follows that
\begin{align}\label{proof-lower-step2-10}
\inf_{(\widehat{\boldsymbol{\Theta}}, \widehat{\boldsymbol{\Delta}})}\sup_{(\boldsymbol{\Theta}, \boldsymbol{\Delta}) \in \mathcal{P}}
\mathbb{E}_{\boldsymbol{\Theta}, \boldsymbol{\Delta}}\Big\{\big\| \widehat{\boldsymbol{\Theta}} -\boldsymbol{\Theta}\big\|_{\mathrm{F}}^2 + \big\| \widehat{\boldsymbol{\Delta}} -\boldsymbol{\Delta}\big\|_{\mathrm{F}}^2 \Big\}
\geq  C_2 \frac{rK^2}{T},
\end{align}
for some constant $C_2 > 0$.

\medskip
\noindent
\textbf{Step 3: Node community construction.}
In this step, we construct a lower bound by varying the balanced community membership matrix $Z$, while fixing the connectivity tensors $\mathbf{W}, \mathbf{M} \in [c_{\min}, 1-c_{\min}]^{K \times K \times L}$ with $\mbox{rank}_3(\mathbf{W})= \mbox{rank}_3(\mathbf{M})= 1 \leq r$.

Fix all layers to share identical two-level connectivity structures
\[
\mathbf{W}_{a, a, l} = p_1, \quad \forall a \in [K],  l \in [L],
\quad \mathbf{W}_{a, b, l}= p_2, \quad \forall a \neq b \in [K],  l \in [L] 
\]
and
\[
\mathbf{M}_{a, a, l} = p_3, \quad \forall a \in [K],  l \in [L],
\quad \mathbf{M}_{a, b, l}= p_4, \quad \forall a \neq b \in [K],  l \in [L], 
\]
where $p_1,p_2,p_3,p_4 \in [c_{\min}, 1-c_{\min}]$ and $ p_1 \neq p_2$, $p_3\neq p_4$.
Since $\mathbf{W}$ and $\mathbf{M}$ do not vary across layers,  the mode-3 matricizations $\mathcal{M}_3(\mathbf{W})$ and $\mathcal{M}_3(\mathbf{M})$ both have rank $1 \le r$, 
satisfying the imposed rank constraint.

Let $\{Z^{(1)}, \ldots, Z^{(M)}  \}$ denote a collection of community membership matrices such that: 
\begin{itemize}
    \item  For the first $\lceil n/2 \rceil$ nodes, the community assignments are identical across all $\{Z^{(m)}\}_{m \in [M]}$, with each community containing either $\lfloor n/(2K) \rfloor$ or $\lceil n/(2K) \rceil$   nodes.
    \item For the remaining $\lfloor n/2\rfloor$ nodes, the community assignments are also balanced and for any $ m \neq m' \in [M]$, 
\begin{equation} \label{proof-lower-step3-1}
  d_H(Z^{(m)}, Z^{(m')}) \geq c_3 n/2, 
\end{equation}
where  $d_H(Z, Z') = \vert \{ i \in [n] : Z_i^{(m)} \neq Z^{(m')}_i \}\vert $  denotes the Hamming distance between community assignments and $c_3 >0 $ is an absolute constant. 
\end{itemize}

By \Cref{lemma-packing}, there exists an absolute constant $c_4 >0$ such that 
\[
M  \geq \exp \big\{2^{-1}c_4 n \log (K)\big\}. 
\] 

For each $m \in [M]$, define the corresponding connection probability tensors as
\[ 
\boldsymbol{\Theta}^{(m)} = \mathbf{W} \times_1 Z^{(m)} \times_2 Z^{(m)} \quad \mbox{and} \quad
\boldsymbol{\Delta}^{(m)} = \mathbf{M} \times_1  Z^{(m)} \times_2  Z^{(m)}.
\]

Fix $m \neq m' \in [M]$ and define
\[
  \mathcal{S} = \big\{ i \in [n] \colon Z_{i}^{(m)} \neq Z_{i}^{(m')} \big\}.
\]
For any $k \in [K]$, let
\[
\mathcal{S}^{(m)}_k =  \big\{ i \in \mathcal{S} \colon Z_{i, k}^{(m)} = 1 \big\}, \quad
\mathcal{S}^{(m')}_k = \big\{ i \in \mathcal{S} \colon Z_{i, k}^{(m')} = 1 \big\}, \quad
\mathcal{N}_{k} = \big\{ i \notin \mathcal{S} \colon Z_{i, k}^{(m)} = 1 \big\}.
\]
By the construction of  $\{Z^{(m)}\}_{m \in [M]}$ and \eqref{proof-lower-step3-1}, there exist absolute constants $c_4' \geq c_3' >0$ such that
\begin{equation} \label{proof-lower-step3-2-1}
   c_3 n/2   \leq \vert  \mathcal{S} \vert  \leq   \lfloor n / 2  \rfloor,  \quad \max\big\{ \vert  \mathcal{S}^{(m)}_k\vert, \vert \mathcal{S}^{(m')}_k\vert  \big\} \leq   c_4' \frac{n}{K},  \quad\forall k \in [K],
\end{equation}   
and
\begin{equation} \label{proof-lower-step3-2-2}   
\quad c_3' \frac{n}{2K}   \leq \vert \mathcal{N}_{k}  \vert<  c_4' \frac{n}{K}, \quad\forall k \in [K].
\end{equation}

Define the functions $a, b\colon [n] \to [k]$ by
\[
  a(i) = \sum_{k \in [K]} k Z_{i, k}^{(m)} \quad \mbox{and} \quad
  b(i) = \sum_{k \in [K]} k Z_{i, k}^{(m')}.
\]
Then we have that  
\begin{align}\label{proof-lower-step3-3} 
& \big\| \boldsymbol{\Theta}^{(m)} -  \boldsymbol{\Theta}^{(m')}  \big\|_{\mathrm{F}}^2 \nonumber\\
= & \sum_{l\in [L]} \sum_{i,j \in [n]}  \big(  \mathbf{W}_{a(i), a(j), l} - \mathbf{W}_{b(i), b(j), l} \big)^2  \nonumber\\
= & \sum_{l\in [L]} \sum_{i \in  \mathcal{S} } \sum_{j \notin  \mathcal{S} }  \big(  \mathbf{W}_{a(i), a(j), l} - \mathbf{W}_{b(i), b(j), l} \big)^2  +  \sum_{l=1} \sum_{i \notin  \mathcal{S} } \sum_{j \in  \mathcal{S} }  \big(  \mathbf{W}_{a(i), a(j), l} - \mathbf{W}_{b(i), b(j), l} \big)^2 
\nonumber\\ 
& + \sum_{l\in [L]} \sum_{i \in  \mathcal{S} } \sum_{j \in  \mathcal{S} }  \big(  \mathbf{W}_{a(i), a(j), l} - \mathbf{W}_{b(i), b(j), l} \big)^2  + \sum_{l=1} \sum_{i \notin  \mathcal{S} } \sum_{j \notin  \mathcal{S} }  \big(  \mathbf{W}_{a(i), a(j), l} - \mathbf{W}_{b(i), b(j), l} \big)^2  \nonumber\\
=  & \sum_{l\in [L]} \sum_{i \in  \mathcal{S} } \sum_{j \notin  \mathcal{S} }  \big(  \mathbf{W}_{a(i), a(j), l} - \mathbf{W}_{b(i), b(j), l} \big)^2  +  \sum_{l=1} \sum_{i \notin  \mathcal{S} } \sum_{j \in  \mathcal{S} }  \big(  \mathbf{W}_{a(i), a(j), l} - \mathbf{W}_{b(i), b(j), l} \big)^2 
\nonumber\\ 
& + \sum_{l=1} \sum_{i \in  \mathcal{S} } \sum_{j \in  \mathcal{S} }  \big(  \mathbf{W}_{a(i), a(j), l} - \mathbf{W}_{b(i), b(j), l} \big)^2  \nonumber\\
& = (II.1) + (II.2) + (II.3).
\end{align}

For the term $(II.1)$ in \eqref{proof-lower-step3-3}, we have that
\begin{align}\label{proof-lower-step3-4} 
    (II.1) = &  \sum_{l\in [L]} \sum_{i \in  \mathcal{S} }\sum_{k=1}^{K} \sum_{j \in  \mathcal{N}_k}  \big(  \mathbf{W}_{a(i), k, l} - \mathbf{W}_{b(i), k, l} \big)^2 \nonumber\\
    = & \sum_{l\in [L]}\sum_{i \in  \mathcal{S} } \Big\{ \vert \mathcal{N}_{a(i)} \vert \big(  \mathbf{W}_{a(i), a(i), l} - \mathbf{W}_{b(i),  a(i), l} \big)^2 +   \vert \mathcal{N}_{b(i)} \vert  \big(  \mathbf{W}_{a(i), b(i), l} - \mathbf{W}_{b(i),  b(i), l} \big)^2  \big\} \nonumber\\
    = &  \sum_{l\in [L]} \sum_{i \in  \mathcal{S} } \Big\{ \vert \big(
    \mathcal{N}_{a(i)} \vert  +   \vert \mathcal{N}_{b(i)} \vert  \big) ( p_1  - p_2 )^2  \big\}.
\end{align}
Combining \eqref{proof-lower-step3-2-1}, \eqref{proof-lower-step3-2-2} and \eqref{proof-lower-step3-4}, we can derive that 
\begin{align} \label{proof-lower-step3-5} 
(II.1)  \leq &  \sum_{l=1}^L \sum_{i \in  \mathcal{S} } \frac{2c_4'n}{K} ( p_1  - p_2 )^2  =  L  \vert \mathcal{S} \vert \frac{2c_4'n}{K} ( p_1  - p_2 )^2   \leq  \frac{ c_4' Ln^2 ( p_1  - p_2 )^2 }{K},
 \end{align}
 and 
\begin{align}\label{proof-lower-step3-6} 
(II.1)  \geq &  \sum_{l=1}^L \sum_{i \in  \mathcal{S} } \frac{c_3'n}{K} ( p_1  - p_2 )^2  =  L  \vert \mathcal{S} \vert \frac{c_3'n}{K} ( p_1  - p_2 )^2 
\geq    \frac{ c_3 c_3' Ln^2 ( p_1  - p_2 )^2 }{2K}.
 \end{align}
 Similarly, for the term $(II.2)$ in \eqref{proof-lower-step3-3}, we have that
\begin{align}\label{proof-lower-step3-7} 
  \frac{ c_3 c_3' Ln^2 ( p_1  - p_2 )^2 }{2K} \leq (II.2) \leq  \frac{ c_4' Ln^2 ( p_1  - p_2 )^2 }{K}
 \end{align}

For the term  $(II.3)$ in \eqref{proof-lower-step3-3}, we have that
\begin{align}\label{proof-lower-step3-8} 
   0 \leq (II.3)  \leq &    \sum_{l\in [L]} \sum_{i \in  \mathcal{S} } \sum_{j \in  \mathcal{S} }\Big\{ \mathbbm{1} \{ j\in  \mathcal{S}_{a(i)}^{(m)},  j \notin  \mathcal{S}_{b(i)}^{(m')} \big\}  + \mathbbm{1} \{ j\notin  \mathcal{S}_{a(i)}^{(m)},  j \in  \mathcal{S}_{b(i)}^{(m')} \big\}   \Big\}
  ( p_1  - p_2 )^2 \nonumber\\
  \leq &   \sum_{l\in [L]} \sum_{i \in  \mathcal{S} } \big( \big\vert \mathcal{S}_{a(i)}^{(m)} \big\vert+  \big\vert  \mathcal{S}_{b(i)}^{(m')} \big\vert  \big) ( p_1  - p_2 )^2  \nonumber\\
   \leq &   \sum_{l\in [L]} \sum_{i \in  \mathcal{S} }  \frac{2 c_4' n}{K} ( p_1  - p_2 )^2  \leq \frac{L c_4' n^2}{K} ( p_1  - p_2 )^2,
\end{align}
where the third and last inequalities follow from \eqref{proof-lower-step3-2-1}. 

Combining \eqref{proof-lower-step3-3},  \eqref{proof-lower-step3-5}, \eqref{proof-lower-step3-6}, \eqref{proof-lower-step3-7} and \eqref{proof-lower-step3-8}, we have that for any $ m \neq m' \in [M]$, 
\begin{equation}\label{proof-lower-step3-9}
   C_3' \frac{L  n^2}{K} ( p_1  - p_2 )^2 \leq   \big\| \boldsymbol{\Theta}^{(m)} -  \boldsymbol{\Theta}^{(m')}  \big\|_{\mathrm{F}}^2 \leq C_3 \frac{L  n^2}{K} ( p_1  - p_2 )^2,
\end{equation}
where $C_3 \geq C_3' >0$ are absolute constants.
Applying the same arguments, we have that for any $ m \neq m' \in [M]$, 
\begin{equation}\label{proof-lower-step3-10}
   C_3' \frac{L  n^2}{K} ( p_3  - p_4 )^2 \leq  \big\| \boldsymbol{\Delta}^{(m)} -  \boldsymbol{\Delta}^{(m')}  \big\|_{\mathrm{F}}^2 \leq C_3 \frac{L  n^2}{K} ( p_3  - p_4 )^2.
\end{equation}

Combining \eqref{proof-lower-step1-3}, \eqref{proof-lower-step3-9} and \eqref{proof-lower-step3-10},  the corresponding KL divergence between $\mathcal{P}_{(\boldsymbol{\Theta}^{(m)}, \boldsymbol{\Delta}^{(m)})}$ and $\mathcal{P}_{(\boldsymbol{\Theta}^{(m')}, \boldsymbol{\Delta}^{(m')})}$ is bounded by
\begin{align}
\mathrm{KL}\big(\mathcal{P}_{(\boldsymbol{\Theta}^{(m)}, \boldsymbol{\Delta}^{(m)}}\| \mathcal{P}_{(\boldsymbol{\Theta}^{(m')}, \boldsymbol{\Delta}^{(m')})} \big)  \leq      C_4   \frac{TLn^2}{K}\big\{(p_1 - p_2)^2 + (p_3 - p_4)^2 \big\}, \nonumber
\end{align}
where $C_4 >0$ is an absolute constant

Choosing 
\[
(p_1 - p_2)^2 = (p_3 - p_4)^2=  \frac{  c_4 n K \log (K)  }{ 16 C_4 T  Ln^2},
\]
yields
 \[
\frac{1}{M^2}\sum_{m, m' \in [M]} \mathrm{KL}\big(\mathcal{P}_{(\boldsymbol{\Theta}^{(m)}, \boldsymbol{\Delta}^{(m)})}\| \mathcal{P}_{(\boldsymbol{\Theta}^{(m')},\boldsymbol{\Delta}^{(m')})} \big)  \le \log(M)/8.
 \]
By Fano’s lemma, it follows that
\begin{align}\label{proof-lower-step3-11}
\inf_{(\widehat{\boldsymbol{\Theta}}, \widehat{\boldsymbol{\Delta}})}\sup_{(\boldsymbol{\Theta}, \boldsymbol{\Delta}) \in \mathcal{P}}
\mathbb{E}_{\boldsymbol{\Theta}, \boldsymbol{\Delta}}\Big\{\big\| \widehat{\boldsymbol{\Theta}} -\boldsymbol{\Theta}\big\|_{\mathrm{F}}^2 + \big\| \widehat{\boldsymbol{\Delta}} -\boldsymbol{\Delta}\big\|_{\mathrm{F}}^2 \Big\}
\geq  C_4' \frac{n \log(K)}{T}.
\end{align}
for some absolute constant $C_4' > 0$. 

\medskip
\noindent
\textbf{Step 4: Layer construction.}
In this step, we construct a lower bound by varying the continuous layer loadings while fixing a balanced community membership matrix $Z$.
Throughout we ensure that $\max\{\mathrm{rank}_3(\mathbf W),\mathrm{rank}_3(\mathbf M)\}\le r$ and $\mathbf{W}, \mathbf{M} \in [c_{\min},1-c_{\min}]^{K \times K \times L}$.

We choose $r$ orthonormal columns from the $K^2$-dimensional Euclidean space $\R^{K^2}$ and reshape them into matrices $\{S^{(s)}\}_{s=1}^r \subset \R^{K\times K}$. By construction, we have that 
\begin{equation}\label{proof-lower-step4-0p}
\langle S^{(s)},S^{(t)}\rangle = \sum_{a,b \in [K]} S^{(s)}_{a,b} S^{(t)}_{a,b} = \mathbbm{1} \{s=t\},  \quad \forall s, t \in[r], 
\quad \|S^{(s)}\|_{\infty}\leq 1, \quad \forall s \in[r].
\end{equation}

For any $s \in [r]$, let 
\[
A^{(s)} = Z S^{(s)} Z^{\top}
\]
Since 
\[
\max_{j \in [K]} \sum_{i=1}^n Z_{i, j} \leq C_{\sigma} \min_{j \in [K]}\sum_{i=1}^n Z_{i, j},
\]
we obtain that
\begin{equation}\label{proof-lower-step4-1p}
\langle A^{(s)}, A^{(t)}\rangle = 0, \quad \forall s \neq t \in[r],  
\quad 
\frac{n^2}{C_\sigma^{2} K^2} \leq \|A^{(s)}\|_{\mathrm F}^2 \leq\ \frac{C_\sigma^{2} n^2}{K^2}, \quad  \forall s \in[r] .
\end{equation}

By the Varshamov--Gilbert lemma, there exist absolute constants $c_5,c_6>0$ and a collection $\{ V^{(1)},\ldots,V^{(M)} \} \subset \{0, 1\}^{L\times r}$
such that
\begin{equation}\label{proof-lower-step4-4p}
\big\| V^{(m)} - V^{(m')}\big\|_{\mathrm F}^2 \geq c_5  L r, \quad  \forall m\neq m' \in [M], \quad  \mbox{and} \quad M \geq 2^{c_6 L r}. 
\end{equation}

Let $\alpha>0$ be specified later. For any $m \in [M]$, define
\[
\mathbf{W}^{(m)}_{:,:,l} =  \frac{1}{2}\mathbf{1}_K\mathbf{1}_K^{\top} + \alpha \sum_{s=1}^r V_{l,s}^{(m)} S^{(s)}
\quad  \mbox{and} \quad
\mathbf{M}^{(m)}_{:,:,l} = \frac{1}{2}\mathbf{1}_K\mathbf{1}_K^{\top} + \alpha \sum_{s=1}^r V_{l,s}^{(m)} S^{(s)}, \quad \forall l\in[L].
\]
It follow that  $\mathrm{rank}_3(\mathbf W^{(m)})\leq r$ and $\mathrm{rank}_3(\mathbf M^{(m)})\leq r$.
Moreover,  for each $a, b\in [K]$ and $l \in [L]$
\[
 \frac{1}{2} - \alpha r    \leq  \mathbf{W}^{(m)}_{a,b,l}\leq \frac{1}{2} + \alpha r
\quad \mbox{and} \quad
 \frac{1}{2} - \alpha r    \leq  \mathbf{M}^{(m)}_{a,b,l}  \leq \frac{1}{2} + \alpha r, 
\]
which implies that
\begin{equation}\label{proof-lower-step4-2}
\mathbf{W}^{(m)}, \mathbf{M}^{(m)} \in [c_{\min}, 1-c_{\min}]^{K \times K \times L} \quad \mbox{if} \quad \alpha r   \leq   \frac{1}{2} - c_{\min}.
\end{equation}
For any $m \in [M]$, we then define
\[
\boldsymbol{\Theta}^{(m)} = \mathbf W^{(m)} \times_1 Z \times_2 Z
\quad  \mbox{and} \quad
\boldsymbol{\Delta}^{(m)} = \mathbf M^{(m)} \times_1 Z \times_2 Z,
\]
For any $m \in [M]$ and $l\in[L]$, we have that 
\begin{equation}\label{proof-lower-step4-3p}
\boldsymbol{\Theta}^{(m)}_{:,:,l} 
= \frac{1}{2} \mathbf{1}_n\mathbf{1}_n^{\top} + \alpha \sum_{s=1}^r V_{l,s}^{(m)} A^{(s)},
\quad  \mbox{and} \quad
\boldsymbol{\Delta}^{(m)}_{:,:,l} 
= \frac{1}{2} \mathbf{1}_n\mathbf{1}_n^{\top}  + \alpha \sum_{s=1}^r V_{l,s}^{(m)} A^{(s)}.
\end{equation}

For any $m \neq m' \in [M] $, we can derive that 
\begin{align}\label{proof-lower-step4-5p}
\big\| \boldsymbol{\Theta}^{(m)} - \boldsymbol{\Theta}^{(m')} \big\|_{\mathrm F}^2
= &  \sum_{l=1}^L \bigg\| \alpha \sum_{s=1}^r \big( V^{(m)}_{l,s} - V^{(m')}_{l,s} \big) A^{(s)} \bigg\|_{\mathrm F}^2 \nonumber\\
= & \alpha^2 \sum_{l=1}^L \sum_{s=1}^r \big( V^{(m)}_{l,s} - V^{(m')}_{l,s} \big)^2  \big\|A^{(s)}\big\|_{\mathrm F}^2 
\end{align}
where the first equality follows from \eqref{proof-lower-step4-3p} and  the final equality follows from \eqref{proof-lower-step4-1p}.
Combining \eqref{proof-lower-step4-1p} and \eqref{proof-lower-step4-5p} yields that for any $m \neq m' \in [M] $,
\begin{align}\label{proof-lower-step4-6p-1}
\big\| \boldsymbol{\Theta}^{(m)} - \boldsymbol{\Theta}^{(m')} \big\|_{\mathrm F}^2  \geq \frac{ \alpha^2 n^2}{C_\sigma^{2} K^2}   \big \| V^{(m)} - V^{(m')} \big\|_{\mathrm F}^2 
\geq \frac{ c_5 \alpha^2 n^2 Lr }{C_\sigma^{2} K^2},
\end{align}
where the last inequality follows from \eqref{proof-lower-step4-4p}.
Similarly, for any $m \neq m' \in [M] $, we have that
\begin{align}\label{proof-lower-step4-6p-2}
\big\| \boldsymbol{\Theta}^{(m)} - \boldsymbol{\Theta}^{(m')} \big\|_{\mathrm F}^2  \leq \frac{ C_\sigma^{2} \alpha^2 n^2}{ K^2}   \big \| V^{(m)} - V^{(m')} \big\|_{\mathrm F}^2 
\leq \frac{ C_\sigma^{2} \alpha^2 n^2 Lr }{ K^2}. 
\end{align}

Applying the same arguments, for any $m \neq m' \in [M] $, we obtain that 
\begin{align}\label{proof-lower-step4-7p}
\frac{ c_5 \alpha^2 n^2 Lr }{C_\sigma^{2} K^2} \leq \big\| \boldsymbol{\Delta}^{(m)} - \boldsymbol{\Delta
}^{(m')} \big\|_{\mathrm F}^2 \leq \frac{ C_\sigma^{2} \alpha^2 n^2 Lr }{ K^2}. 
\end{align}

Combining \eqref{proof-lower-step1-3}, \eqref{proof-lower-step4-6p-2} and \eqref{proof-lower-step4-7p},  the corresponding KL divergence between $\mathcal{P}_{(\boldsymbol{\Theta}^{(m)}, \boldsymbol{\Delta}^{(m)})}$ and $\mathcal{P}_{(\boldsymbol{\Theta}^{(m')}, \boldsymbol{\Delta}^{(m')})}$ is bounded by 
\[
\mathrm{KL}\big(\mathcal P_{(\boldsymbol{\Theta}^{(m)},\boldsymbol{\Delta}^{(m)})}\big\| \mathcal P_{(\boldsymbol{\Theta}^{(m')},\boldsymbol{\Delta}^{(m')})}\big) \leq C_5 \frac{ \alpha^2 T n^2 Lr }{ K^2}.
\]
for some absolute constant $C_5 >0$.

Choosing 
\[
 \alpha^2 =  \frac{  c_6 \log(2) K^2  }{ 8 C_5 Tn^2   },
\]
and using \eqref{proof-lower-step4-2} together with the assumption $rK\leq c\sqrt{T} n$ for some sufficiently small constant $c >0$, for any $m \in [M]$, we have that  
\[
\mathbf{W}^{(m)},  \mathbf{M}^{(m)} \in [c_{\min}, 1-c_{\min}]^{K \times K \times L}.
\]
Moreover, 
 \[
\frac{1}{M^2}\sum_{m, m' \in [M]} \mathrm{KL}\big(\mathcal{P}_{(\boldsymbol{\Theta}^{(m)}, \boldsymbol{\Delta}^{(m)})}\| \mathcal{P}_{(\boldsymbol{\Theta}^{(m')},\boldsymbol{\Delta}^{(m')})} \big)  \leq \log(M)/8.
 \]
By Fano’s lemma, it follows that
\begin{align}\label{proof-lower-step4-8p}
\inf_{(\widehat{\boldsymbol{\Theta}}, \widehat{\boldsymbol{\Delta}})}\sup_{(\boldsymbol{\Theta}, \boldsymbol{\Delta}) \in \mathcal{P}}
\mathbb{E}_{\boldsymbol{\Theta}, \boldsymbol{\Delta}}\Big\{\big\| \widehat{\boldsymbol{\Theta}} -\boldsymbol{\Theta}\big\|_{\mathrm{F}}^2 + \big\| \widehat{\boldsymbol{\Delta}} -\boldsymbol{\Delta}\big\|_{\mathrm{F}}^2 \Big\}
\geq  C_6 \frac{Lr}{T}.
\end{align}
for some absolute constant $C_6 > 0$.

\medskip
\noindent
\textbf{Step 5.} Combining \eqref{proof-lower-step2-10},  \eqref{proof-lower-step3-11} and \eqref{proof-lower-step4-8p}, we finally obtain that  there exists an absolute constant $C' >0$ such that 
\[
\inf_{(\widehat{\boldsymbol{\Theta}}, \widehat{\boldsymbol{\Delta}})}\sup_{(\boldsymbol{\Theta}, \boldsymbol{\Delta}) \in \mathcal{P}}
\mathbb{E}_{\boldsymbol{\Theta}, \boldsymbol{\Delta}}\Big\{\big\| \widehat{\boldsymbol{\Theta}} -\boldsymbol{\Theta}\big\|_{\mathrm{F}}^2 + \big\| \widehat{\boldsymbol{\Delta}} -\boldsymbol{\Delta}\big\|_{\mathrm{F}}^2 \Big\}
\geq   C' \frac{ rK^2+  n\log(K)+Lr}{ T },
\]
which completes the proof. 

\end{proof}

\subsection{Proof of Proposition \ref{thm-comunity-recovery-stat}}\label{app-proof-thm-comunity-recovery-stat}

\begin{proof}[Proof of  \Cref{thm-comunity-recovery-stat}]

Let 
\[
O_{K}^t = \arginf_{O \in \mathbb{O}_{K \times K}}\| \widehat{U}^t_Z - U_Z O   \|. 
\]
Recall that $U_Z$ is defined \eqref{def-U-Z}. For any $i\in[n]$, let $\gamma_i^t$ and $\widehat{\gamma}_i^t$ denote the $i$th rows of
$U_Z O_{K}^t$ and $\widehat{U}_Z^t$, respectively.
Let $\{\tilde{u}^{(t,k)}\}_{k=1}^K$ denote the $K$ distinct row patterns of $\widetilde{U}^t$. For any node $i \in [n]$, denote its true community by $z(i) \in [K]$ and its assigned  community  at time $t \in [T]$ by $\hat{z}^t(i) \in [K]$.  For any $k \in [K]$, let $\mathcal{V}_k = \{i \in [n] \colon z(i) = k \}$.  

For any $i, j \in [n]$, we have that 
\[
\|\gamma_i^t - \gamma_j^t\|_2 = \|Z_i D_Z^{-1} O_{K}^t - Z_j D_Z^{-1} O_{K}^t\|_2 = \|Z_i D_Z^{-1} - Z_j D_Z^{-1}\|_2 =  
\begin{cases}
    0, & \mbox{if } Z_i = Z_j, \\
    \geq \sqrt{\frac{2}{s_{\max}}}, & \mbox{otherwise}.
\end{cases}
\]
Thus, there are exactly $K$ distinct rows in $U_Z O_{K}^t$, say
$\{\gamma_k^{t,\ast}\}_{k=1}^K$, and $U_Z O_{K}^t \in \mathcal{M}_{n,K}$. By \Cref{ass-stat-rank}$(ii)$, for any $k \neq l \in [K]$,  we have that
\begin{equation}\label{theorem-community-stat-1} 
\|\gamma_k^{t,\ast}  - \gamma_l^{t,\ast}\|_2 \geq \sqrt{\frac{2}{s_{\max}}} \geq \sqrt{\frac{2K}{C_{\sigma} n}}.
\end{equation}

Define the event
\begin{equation}\label{stat-U-diff-bound}
\mathcal{A} =\Big\{ \| \widehat{U}^t_Z - U_Z O_{K}^t \|_{\mathrm{F}} \leq    \sqrt{K}\mathcal{E}_t,  \forall t \in [T]\Big\},
\end{equation}
where $ \mathcal{E}_t$ is defined in \eqref{eq-sin-theta-stat}. By \Cref{prop-sin-theta-stat}, $\mathcal{A}$ holds with probability at least $1- (n \vee L \vee T)^{-c}$. From now on, we assume the event $\mathcal{A}$ holds.  

\medskip
\noindent
\textbf{Step 1: Misclassification error bound.} 
Define
\begin{equation}\label{def-J-stat}
\mathcal{J} = \bigg\{ i\in[n] \colon \|\widehat{\gamma}_i^t - \gamma_i^t\|_2 \leq  \frac{1}{9}\sqrt{ \frac{2K}{C_{\sigma} n}} \bigg\}.
\end{equation}
Then, from  \eqref{stat-U-diff-bound},
\[
\vert \mathcal{J}^c \vert  \frac{2K}{81C_{\sigma} n}   
\leq \sum_{i\in  \mathcal{J}^c} \|\widehat{\gamma}_i^t - \gamma_i^t\|_2^2
\leq \| \widehat{U}^t_Z - U_Z O_{K}^t \|_{\mathrm{F}}^2 \leq  K \mathcal{E}_t^2,
\]
Hence, it holds that 
\begin{equation}\label{bound-Jc-stat}
\vert \mathcal{J}^c \vert  \leq \frac{81 C_{\sigma}n  \mathcal{E}_t^2 }{2}.
\end{equation}

\medskip
\noindent
\textbf{Step 1.1: Existence and uniqueness of nearest cluster centers.}
In this sub-step, we aim to show that for each $k \in [K]$, there exists a unique   $k' \in [K]$ and   
\[
\big\| \tilde{u}^{(t,k')} -\gamma_{k}^\ast \big\|_2 \leq \frac{1}{3} \sqrt{\frac {2K}{C_{\sigma} n}}.
\]

We first prove the existence. Suppose instead that for some $k \in [K]$, 
\[
\big\| \tilde{u}^{(t,k')} - \gamma_k^{t,\ast} \big\|_2
> \frac{1}{3} \sqrt{\frac{2K}{C_{\sigma} n}}
\quad \mbox{for all } k' \in [K].
\]
In particular, for any $j \in \mathcal{V}_k \cap \mathcal{J}$, the center assigned to $j$ is $\tilde{u}^{(t,\hat{z}^t(j))}$, and therefore
\begin{align*}
\|\widehat{\gamma}_j^t - \tilde{u}^{(t,\hat{z}^t(j))}\|_2
\geq & \|\tilde{u}^{(t,\hat{z}^t(j))} - \gamma_k^{t,\ast}\|_2 - \|\widehat{\gamma}_j^t - \gamma_k^{t,\ast}\|_2  \nonumber\\
= & \|\tilde{u}^{(t,\hat{z}^t(j))} - \gamma_k^{t,\ast}\|_2 - \|\widehat{\gamma}_j^t - \gamma_j\|_2   \\
> & \frac{1}{3}\sqrt{\frac{2K}{C_{\sigma} n}} - \frac{1}{9}\sqrt{\frac{2K}{C_{\sigma} n}}
= \frac{2}{9}\sqrt{\frac{2K}{C_{\sigma} n}},
\end{align*}
where the first equality follows from $j \in \mathcal{V}_k$, and second inequality follows from $j \in \mathcal{J}$ and \eqref{def-J-stat}.

By \Cref{ass-stat-rank}$(ii)$,  each community $k \in [K]$ satisfies
$ \vert \mathcal{V}_k\vert \geq C_{\sigma}^{-1}  n/K$. By \eqref{bound-Jc-stat}, we obtain that 
\[
\vert  \mathcal{V}_k \cap \mathcal{J} \vert 
= \vert\mathcal{V}_k\vert - \vert\mathcal{V}_k \cap \mathcal{J}^c\vert \geq
  \frac{n}{ C_{\sigma}K} - \frac{81 C_{\sigma}n  \mathcal{E}_t^2 }{2}  \geq   \frac{n}{2C_{\sigma}K}. 
\]
where the last inequality follows from \eqref{eq-ass-singular-1}. 
Then,  by the definition of $\widetilde{U}^t$, we have that 
\begin{align}\label{theorem-community-stat-2}
\| \widehat{U}^t_Z - U_Z O_{K}^t \|_{\mathrm{F}}^2 
\geq & \| \widehat{U}^t_Z - \widetilde{U}^t  \|_{\mathrm{F}}^2   \geq 
\vert  \mathcal{V}_k \cap \mathcal{J} \vert    \cdot
\|\widehat{\gamma}_j^t - \tilde{u}^{(t,\hat{z}^t(j))}\|_2^2 
\nonumber\\
> &   \frac{n}{ 2 C_{\sigma}K}  
\frac{8K}{81 C_{\sigma}n  }
=  \frac{4}{81C_{\sigma}^2}. 
\end{align}
Combining \eqref{stat-U-diff-bound} and \eqref{theorem-community-stat-2}  yields that
\[
    \frac{4}{81C_{\sigma}^2} \leq \| \widehat{U}^t_Z - U_Z O_{K}^t \|_{\mathrm{F}}^2 \leq K\mathcal{E}_t^2   
\] 
This contradicts \eqref{eq-ass-singular-1}. 
Thus, we prove the existence.

We now prove the uniqueness.  By \eqref{theorem-community-stat-1}, any two distinct population centers $\gamma_{k}^\ast$ and $\gamma_{l}^\ast$ with $k \neq l$ are separated by at least $\sqrt{ 2K/(C_{\sigma} n)}$.  
Therefore, no single empirical center $\tilde{u}^{(t,k')}$ can lie within distance $\sqrt{ 2K/(C_{\sigma} n)}/3$ of two different population centers simultaneously.
This establishes the uniqueness.

\medskip
\noindent
\textbf{Step 1.2: Correct classification of nodes in $\mathcal{J}$.}
For each $k\in [K]$, we relabel $\tilde{u}^{(t,k)}$ according to its nearest population center $\gamma_{k'}^\ast$ established in \textbf{Step 1.1}. 

For any $j \in \mathcal{V}_k \cap J$, we have that 
\[
\|\widehat{\gamma}_j^t - \tilde{u}^{(t,k))} \|_2 \leq 
\|\widehat{\gamma}_j^t - \gamma_k^{t,\ast}\|_2
  + \|\gamma_k^{t,\ast} - \tilde{u}^{(t,k)}\|_2
\leq \frac{1}{9}\sqrt{\frac{2K}{C_{\sigma} n}} +  \frac{1}{3} \sqrt{\frac{2K}{C_{\sigma} n}} \leq 
\frac{4}{9}\sqrt{\frac{2K}{C_{\sigma} n}},
\]
where the second inequality follows from $j \in \mathcal{J}$ and \eqref{def-J-stat}. 
For any $l \neq k$, we have that 
\[
\| \widehat{\gamma}_j^t - \tilde{u}^{(t,l)} \|
\geq 
\|\gamma_k^{t,\ast} - \gamma_l^{t,\ast}\|
  - \|\gamma_l^\ast - \tilde{u}^{(t,l)} \|
\geq  \sqrt{\frac{2K}{C_{\sigma} n}}  - \frac{1}{3} \sqrt{\frac{2K}{C_{\sigma} n}}
\geq \frac{2}{3} \sqrt{\frac{2K}{C_{\sigma} n}}.
\]

Since
\[
\frac{4}{9} < \frac{2}{3},
\]
each $\widehat{\gamma}_j^t$ with $j \in \mathcal{V}_k \cap \mathcal{J}$ is closer to  $\tilde{u}^{(t,k)}$ than to any $\tilde{u}^{(t,l)}$ for $l\neq k$.  
Thus all nodes in $\mathcal{J}$ are correctly clustered, and any misclassification can only occur in 
$\mathcal{J}^c$.

Therefore, using \eqref{bound-Jc-stat}, 
\[
\mathcal{L}(\widehat{Z}^t, Z)
= n^{-1} \vert \mathcal{J}^c \vert \leq  n^{-1} \frac{81 C_{\sigma}n  \mathcal{E}_t^2 }{2}  \leq \frac{81 C_{\sigma} }{2} \mathcal{E}_t^2,
\]
which proves \eqref{thm-com-1}.

\medskip
\noindent
\textbf{Step 2. Exact recovery.} 

\medskip
\noindent
\textbf{Step 2.1.}
We first show that $Z_i \neq Z_j$ implies $\widehat{Z}_i^t \neq \widehat{Z}_j^t$.

Suppose, for contradiction, that there exist $i, j \in [n]$ such that $Z_i \neq Z_j$ but $\widehat{Z}_i^t = \widehat{Z}_j^t$, which implies that $\widetilde{U}^t_i = \widetilde{U}^t_j$.
By the definition of $\widetilde{U}^t$ and $U_Z O_K^t \in \mathcal{M}_{n ,K}$,
\begin{equation}\label{theorem-community-stat-3}
\|U_Z O_K^t - \widetilde{U}^t\|_F^2 
\leq \|\widehat{U}_Z^t - \widetilde{U}^t\|_F^2 + \|\widehat{U}_Z^t - U_Z O_K^t\|_F^2 \\
\leq 2 \|\widehat{U}_Z^t - U_Z O_K^t \|_F^2 \leq 2 K\mathcal{E}_t^2.
\end{equation}

On the other hand, we have that
\begin{equation}\label{theorem-community-stat-4}
\|U_Z O_K^t - \widetilde{U}^t\|_F^2 
\geq \|\gamma_i^t - \widetilde{U}_i^t\|_2^2 + \|\gamma_j^t- \widetilde{U}_j^t\|_2^2 
\geq \|\gamma_i^t - \gamma_j^t\|_2^2. 
\end{equation}
Combining \eqref{theorem-community-stat-1},  \eqref{theorem-community-stat-3} and  \eqref{theorem-community-stat-4} yields that
\[
\sqrt{\frac{2K}{C_{\sigma} n}}
\leq \|U_Z O_K - \widetilde{U}^t\|_F 
\leq \sqrt{2}\|\widehat{U}_Z^t - U_Z O_K\|_F  \leq     \sqrt{2K}\mathcal{E}_t, 
\]
contradicting \eqref{eq-ass-singular}.  
Hence $Z_i\neq Z_j$ implies $\widehat{Z}_i^t\neq \widehat{Z}_j^t$.

\medskip
\noindent
\textbf{Step 2.2.}
Next we show that $Z_i = Z_j$ implies $\widehat{Z}_i^t = \widehat{Z}_j^t$.  

Assume instead that there exist $ i \neq j \in [n]$ such that $Z_i = Z_j$ but $\widehat{Z}_i^t \neq \widehat{Z}_j^t$, which implies
$\widetilde{U}_i^t \neq \widetilde{U}_j^t$. 
From \textbf{Step 1}, $\widetilde{U}^t$ has exactly $K$ distinct rows, each corresponding to a true community.  
Thus, because  $\widetilde{U}_i^t \neq \widetilde{U}_j^t$ but $Z_i = Z_j$, there must exist some $k \neq i,j$ such that $Z_i = Z_j \neq Z_k$ and $\widetilde{U}_j^t = \widetilde{U}_k^t$.

Construct $\widetilde{U}^{t,*}$ by replacing the $j$th row of $\widetilde{U}^t$
with $\widetilde{U}_i^t$.  Then,
\begin{align*}
\|\widehat{U}_Z^t - \widetilde{U}^{t, *}\|_F^2 - \|\widehat{U}_Z^t - \widetilde{U}^t\|_F^2
&= \|\widehat{\gamma}_j^t - \widetilde{U}_i^t\|_2^2 - \|\widehat{\gamma}_j^t - \widetilde{U}_k^t\|_2^2 \\
&= \|\widehat{\gamma}_j^t - \gamma_j^t + \gamma_i^t - \widetilde{U}_i^t\|_2^2 
   - \|\widehat{\gamma}_j^t - \gamma_j^t + \gamma_i^t - \gamma_k^t + \gamma_k^t - \widetilde{U}_k^t\|_2^2 \\
&\leq \|\widehat{\gamma}_j^t - \gamma_j^t + \gamma_i^t - \widetilde{U}_i^t\|_2^2 +
\|\widehat{\gamma}_j^t - \gamma_j^t  +\gamma_k^t - \widetilde{U}_k^t\|_2^2 
- \| \gamma_i^t - \gamma_k^t \|_2^2 \\
&\leq 2\|\widehat{U}_Z^t - U_Z O_K^t\|_F^2 + 2\|U_Z O_K^t - \widetilde{U}^t\|_F^2 - \frac{2K}{C_{\sigma} n} \\
&\leq 4 K\mathcal{E}_t^2
   - \frac{2K}{C_{\sigma} n} \\
&< 0,
\end{align*}
a contradiction. Therefore $Z_i = Z_j$ implies $\widehat{Z}_i^t = \widehat{Z}_j^t$.

Combining \textbf{Steps 2.1} and \textbf{2.2}, we conclude that under \eqref{eq-ass-singular}, $\widehat{Z}^t$ exactly recovers the clustering structure encoded by $Z$, completing the proof of \eqref{thm-com-2}.

\end{proof}

\subsection{Proofs for Appendix~\ref{sec-app-add-stat}}\label{app-proof-app-stat}

This section provides proofs of the results stated in Appendix~\ref{sec-app-add-stat}. The proofs are organized as follows: Proposition~\ref{lemma-sub-Gaussian} is proved in Appendix~\ref{sec-app-1}, Lemma~\ref{lemma-singular-values-stat} in Appendix~\ref{sec-app-2}, \Cref{prop-sin-theta-stat} in Appendix~\ref{sec-app-3}, Lemma~\ref{lemma-packing} in Appendix~\ref{sec-app-4} and Lemma~\ref{lemma-binomial} in Appendix~\ref{sec-app-5}.

\subsubsection{Proof of Proposition \ref{lemma-sub-Gaussian}}
\label{sec-app-1}
\begin{proof}[Proof of Proposition  \ref{lemma-sub-Gaussian}]
By \Cref{lemma-mixing-stat},  for any $1 \leq i \leq j \leq n$ and $l \in [L]$, the sequence $\{\mathbf{A}_{i,j,l}^{t}\}_{t \geq 0}$ is strongly mixing. Since strong mixing is preserved under measurable transformations, any process constructed from finitely many such variables is also strongly mixing. In particular, the process
\[
\big\{\mathbf{A}_{i,j,l}^{t}(1-\mathbf{A}_{i,j,l}^{t-1})\big\}_{t\geq 1}
\]
is strongly mixing.
For any $t \in [T]$, define
\[
\mathbf{X}_{i, j, l}^{t} = t^{-1}\sum_{u=1}^{t} \mathbf{A}_{i,j,l}^{u-1}  \quad \mbox{and} \quad \mathbf{Y}_{i, j, l}^{t} =  t^{-1}\sum_{u=1}^{t}\mathbf{A}_{i,j,l}^{u}(1-\mathbf{A}_{i,j,l}^{u-1}).
\]
Let 
\[
\boldsymbol{\Pi}_{i, j, l} =  \E\{\mathbf{A}_{i, j, l}^0\}, \quad \forall  1 \leq i \leq j  \leq n, l \in [L].
\]
By stationarity (\Cref{prop-stat}),
\begin{align}
 \E\big\{\mathbf{X}_{i, j, l}^{t}\big\} =   t^{-1}\sum_{u=1}^{t} \E\big\{\mathbf{A}_{i,j,l}^{u-1} \big\} = \boldsymbol{\Pi}_{i, j, l},
 \nonumber
\end{align}
and
\begin{align}
 \E\big\{\mathbf{Y}_{i, j, l}^{t}\big\} = &  t^{-1}\sum_{u=1}^{t} \E\big\{\mathbf{A}_{i,j,l}^{u}(1-\mathbf{A}_{i,j,l}^{u-1}) \big\} \Big\} =  ( 1-\boldsymbol{\Pi}_{i, j, l}) \boldsymbol{\Theta}_{i, j, l}. 
 \nonumber
\end{align}

Applying the Bernstein inequality for $\alpha$-mixing sequences \citep[e.g.~Theorem 1 in][]{merlevede2009bernstein}, for any $\varepsilon >0$, we obtain  that
\begin{align}\label{eq-prop1-1}
\P\big\{ \big\vert  \mathbf{Y}_{i, j, l}^{t}  -  \E\big\{\mathbf{Y}_{i, j, l}^{t}\big\}  \big\vert \geq  \varepsilon \big\} \leq \exp\bigg\{- \frac{ c_{0} t\varepsilon ^{2}}
         {1 + \varepsilon \log (t)\log\log (t)}\bigg\},
\end{align}
and similarly
\begin{align}\label{eq-prop1-2}
\P\big\{  \big\vert  \mathbf{X}_{i, j, l}^{t}  -  \E\big\{\mathbf{X}_{i, j, l}^{t, k}\big\}  \big\vert \geq \varepsilon \big\} \leq \exp\bigg\{- \frac{ c_{0} t \varepsilon ^{2}}
         {1 +  \varepsilon  \log (t)\log\log (t)}\bigg\},
\end{align}
where $C_0, c_0 >0$ are absolute constants. 

Let 
\[
\mathcal{A} = \Big\{ \max\big\{ \big\vert  \mathbf{X}_{i, j, l}^{t, k}  -  \E\big\{\mathbf{X}_{i, j, l}^{t, k}\big\}  \big\vert, \big\vert  \mathbf{Y}_{i, j, l}^{t, k}  -  \E\big\{\mathbf{Y}_{i, j, l}^{t, k}\big\}  \big\vert  \big\} \leq  \varepsilon  \Big\}.
\]
Using a union bound argument with \eqref{eq-prop1-1} and \eqref{eq-prop1-2}, we can derive that
\begin{equation}\label{eq-prop1-3}
\P \{  \mathcal{A} \}  \geq 1 - 2\exp\bigg\{- \frac{ c_{0} \varepsilon^{2}}
         {1 + \varepsilon\log (t)\log\log (t)}\bigg\}.
\end{equation}
Note that 
\[
\widehat{\boldsymbol{\Theta}}_{i,j, l}^{t}    = \frac{t^{-1}\sum_{u=1}^t \mathbf{A}_{i, j, l}^u (1 - \mathbf{A}_{i, j, l}^{u-1} )}{t^{-1}\sum_{u=1}^t (1 - \mathbf{A}_{i,j, l}^{u-1} )} =   \frac{\mathbf{Y}_{i,j, l}^{t}}{ 1 - \mathbf{X}_{i,j, l}^{t}}.
\]
Under $\mathcal{A}$, it holds that
\begin{align}\label{eq-prop1-4}
   \frac{ \E\big\{\mathbf{Y}_{i, j, l}^{t}\big\} - \varepsilon } { 1 - \E\big\{\mathbf{X}_{i, j, l}^{t}\big\} + \varepsilon}
 \leq \widehat{\boldsymbol{\Theta}}_{i,j, l}^t \leq \frac{\E\big\{\mathbf{Y}_{i, j, l}^{t}\big\} + \varepsilon } { 1 - \E\big\{\mathbf{X}_{i, j, l}^{t}\big\}-\varepsilon}.
\end{align}
Next, still under the event $\mathcal{A}$, we obtain that 
\begin{align}\label{eq-prop1-5}
   \boldsymbol{\Theta}_{i,j, l}  -  \frac{ \E\big\{\mathbf{Y}_{i, j, l}^{t}\big\} - \varepsilon } { 1 - \E\big\{\mathbf{X}_{i, j, l}^{t, k}\big\} + \varepsilon} 
= & \boldsymbol{\Theta}_{i,j, l}  -  \frac{  (1- \boldsymbol{\Pi}_{i, j, l})\boldsymbol{\Theta}_{i, j, l} - \varepsilon}{   (1- \boldsymbol{\Pi}_{i, j, l})+\varepsilon}
=  \frac{  \varepsilon  (1+ \boldsymbol{\Theta}_{i,j,l} ) }
       {1-\boldsymbol{\Pi}_{i, j, l}  + \varepsilon}
\leq   \varepsilon \frac{  (1+c_{\min}) }{c_{\min}}.
\end{align}
Similarly, for any $\varepsilon \leq c_{\min}/2 $, 
\begin{align}\label{eq-prop1-6}
 \frac{\E\big\{\mathbf{Y}_{i, j, l}\big\} + \varepsilon } { 1 - \E\big\{\mathbf{X}_{i, j, l}\big\}-\varepsilon} -   \boldsymbol{\Theta}_{i,j, l}   
=  & \frac{  (1- \boldsymbol{\Pi}_{i, j, l})\boldsymbol{\Theta}_{i, j, l} +\varepsilon}{ 1-   \boldsymbol{\Pi}_{i, j, l}  - \varepsilon} - \boldsymbol{\Theta}_{i,j, l} 
=   \frac{  \varepsilon  (1+ \boldsymbol{\Theta}_{i,j,l} )}{1-\boldsymbol{\Pi}_{i, j, l}  - \varepsilon}\nonumber\\ 
\leq  &   \frac{  \varepsilon  (1+ \boldsymbol{\Theta}_{i,j,l} ) }{1-\boldsymbol{\Pi}_{i, j, l} - c_{\min}/2}
\leq    \varepsilon  \frac{ 2 (1+c_{\min})}{c_{\min}}.
\end{align}

Combining \eqref{eq-prop1-3}, \eqref{eq-prop1-4}, \eqref{eq-prop1-5} and \eqref{eq-prop1-6},  we conclude that 
\begin{equation}\label{eq-Theta-sub-begin}
\P \Big\{ \big\vert  \widehat{\boldsymbol{\Theta}}_{i,j, l}^t -  \boldsymbol{\Theta}_{i,j, l}  \big\vert  \geq  C_0 \varepsilon  \Big\} \leq  2 \exp\bigg\{- \frac{ c_{0} t\varepsilon^{2}}
         {1 + \varepsilon\log (t)\log\log (t)}\bigg\},
\end{equation}
where $C_0 >0$ is an absolute constant.
Similarly, we have that 
\[
\P \Big\{ \big\vert  \widehat{\boldsymbol{\Delta}}_{i,j, l}^t -  \boldsymbol{\Delta}_{i,j, l}  \big\vert  \geq  C_0 \varepsilon  \Big\} \leq  2 \exp\bigg\{- \frac{ c_{0} t\varepsilon^{2}}
         {1 + \varepsilon\log (t)\log\log (t)}\bigg\},
\]
Taking 
\[
  \varepsilon = \sqrt{ \frac{ \log(n \vee L \vee T)}{t}},
\]
and by a union bound argument, we have that 
\[
\P \Big\{ \big\vert  \widehat{\boldsymbol{\Theta}}_{i,j, l}^t -  \boldsymbol{\Theta}_{i,j, l}  \big\vert +  \big\vert  \widehat{\boldsymbol{\Delta}}_{i,j, l}^t -  \boldsymbol{\Delta}_{i,j, l}  \big\vert  \geq  2C_0  \sqrt{ \frac{ \log(n \vee L \vee T)}{t}}  \Big\} \leq  4 \exp\bigg\{- \frac{ c_{0}    \log(n \vee L \vee T) }
         {1 +  \sqrt{ \frac{ \log(n \vee L \vee T)}{t}}\log\log (t)}\bigg\}
\]
Since $t \{\log(t) \log\log(t)\}^{-2} \gtrsim \log(n \vee L \vee T )$,  we obtain that
\[
\P \Big\{ \big\vert  \widehat{\boldsymbol{\Theta}}_{i,j, l}^t -  \boldsymbol{\Theta}_{i,j, l}  \big\vert +  \big\vert  \widehat{\boldsymbol{\Delta}}_{i,j, l}^t -  \boldsymbol{\Delta}_{i,j, l}  \big\vert  \geq  C_1  \sqrt{ \frac{ \log(n \vee L \vee T)}{t}}  \Big\} \leq    (n \vee L \vee T)^{-c_1}
\]
where $C_1 > 0$ and $c_1 > 3$. By a union bound argument, there exists an absolute constant $c_2 >0$ such that 
\[
\P \Big\{ \big\vert  \widehat{\boldsymbol{\Theta}}_{i,j, l}^t -  \boldsymbol{\Theta}_{i,j, l}  \big\vert +  \big\vert  \widehat{\boldsymbol{\Delta}}_{i,j, l}^t -  \boldsymbol{\Delta}_{i,j, l}  \big\vert >   C_1  \sqrt{ \frac{ \log(n \vee L \vee T)}{t}},   \forall 1 \leq i \leq j \leq n, l\in [L] \bigg\}  \leq    (n \vee L \vee T)^{-c_2}.
\]

Let 
\[
u_0=\frac{C_0}{\log(t) \log\log (t)}.
\]
For any $q \geq 1$, by integration by parts, we obtain that 
\begin{align}\label{eq-Theta-sub-1}
 & \E \big\{ \big\vert  \widehat{\boldsymbol{\Theta}}_{i,j, l}^t -  \boldsymbol{\Theta}_{i,j, l}  \big\vert^q \big\}  \nonumber\\
 = & \int_{0}^{\infty} q u^{q-1}\P \big\{\big\vert  \widehat{\boldsymbol{\Theta}}_{i,j, l}^t -  \boldsymbol{\Theta}_{i,j, l}  \big\vert \geq u\big\}du \nonumber\\
 = & \int_{0}^{u_0} q u^{q-1} \P\big\{ \big\vert  \widehat{\boldsymbol{\Theta}}_{i,j, l}^t -  \boldsymbol{\Theta}_{i,j, l}  \big\vert \geq u\big \}du + \int_{u_0}^{\infty} q u^{q-1} \P \big\{\big\vert  \widehat{\boldsymbol{\Theta}}_{i,j, l}^t -  \boldsymbol{\Theta}_{i,j, l}  \big\vert \geq u \big\}du \nonumber\\
 \leq & 2\int_{0}^{u_0} q u^{q-1}\exp \Big\{- \frac{ c_0 t}{2C_0^2}u^2\Big\} du+  2\int_{u_0}^{\infty} q u^{q-1}\exp\Big\{-\frac{c_0 t}{2C_0 \log(t) \log\log(t)}u\Big\}du
 \nonumber\\
 =& (I.1) + (I.2),
\end{align}
where the first inequality follows from \eqref{eq-Theta-sub-begin}. 

For the term $(I.1)$ in \eqref{eq-Theta-sub-1}, let $\kappa =\sqrt{ c_0 t/(2C_0^2)}$ and substitute $s = \kappa u$. Then
\begin{align}\label{eq-Theta-sub-1.1}
(I.1)= & 2q \kappa^{-q}\int_{0}^{\kappa u_0}s^{q-1}e^{-s^2}ds \leq 2q \kappa^{-q}\int_{0}^{\infty}s^{q-1}e^{-s^2}ds = q \kappa^{-q} \Gamma(q/2)  \nonumber\\
= & q \bigg(\frac{2C_0^2}{c_0 t} \bigg)^{q/2} \Gamma(q/2) \leq \bigg(C_2\sqrt{\frac{q}{t}}\bigg)^{q}, 
\end{align}
for some absolute constant $C_2 >0$. 
For the term $(I.2)$ in \eqref{eq-Theta-sub-1}, let  $\lambda = c_0t/\{ 2C_0 \log(t) \log\log(t)\}$ and substitute  $s =\lambda u$. Then 
\begin{align}\label{eq-Theta-sub-1.2}
(I.2) = & 2q \lambda^{-q}\int_{\lambda u_0}^{\infty}s^{q-1}e^{-s} ds = 2q \lambda^{-q}\int_{0}^{\infty}s^{q-1}e^{-s} ds =2q\lambda^{-q}\Gamma(q) \nonumber\\
= &  2q  \bigg(\frac{2C_0 \log(t) \log\log(t)}{c_0t}\bigg)^{q}
\Gamma(q) \leq \bigg(C_3\frac{q \log(t) \log\log(t)}{C_0t}\bigg)^{q},
\end{align}
for some absolute constant $C_3 >0$. 
Combining \eqref{eq-Theta-sub-1}, \eqref{eq-Theta-sub-1.1} and \eqref{eq-Theta-sub-1.1} yields that 
\begin{align}
\big[\mathbb E \big\{ \vert \boldsymbol{\Theta}_{i, k, l} - \widehat{\boldsymbol{\Theta}}_{i, k, l} \vert^q \big\}   \big]^{1/q}
\leq  C_4 \bigg(\sqrt{\frac{q}{t}}+\frac{q\log(t) \log\log(t)}{C_0 t}\bigg), \nonumber
\end{align}
for some absolute constant $C_4 >0$.
Consequently,
\begin{align}
\sup_{q \geq 1}\frac{\big[\E\big\{ \vert \boldsymbol{\Theta}_{i, k, l} - \widehat{\boldsymbol{\Theta}}_{i, k, l} \vert^q \big\}\big]^{1/q}}{q}
 \leq  C_4 \bigg(\sqrt{\frac{1}{t}}+\frac{\log(t) \log\log(t)}{C_0 t}\bigg) \leq C_5 t^{-1/2},
\end{align}
where $C_5>0$ is an absolute constant and the last inequality uses the fact $C't \geq  \{\log(t) \log\log(t)\}^{2} $ for any $t\geq 1$ and some absolute constant $C' >0$. By Proposition 2.6.1 in \cite{vershynin2018high}, it follows that $  \widehat{\boldsymbol{\Theta}}_{i,j, l}^t -  \boldsymbol{\Theta}_{i,j, l}$ is $C_6 t^{-1/2}$-sub-Gaussian distributed for some absolute constant $C_6>0$. An analogous argument applies to $\widehat{\boldsymbol{\Delta}}_{i,j,l}^t - \boldsymbol{\Delta}_{i,j,l}$.  
These complete the proof.
\end{proof}

\subsubsection{Proof of Lemma \ref{lemma-singular-values-stat}}\label{sec-app-2}
\begin{proof}[Proof of \Cref{lemma-singular-values-stat}]

Each row of $Z$ contains exactly one $1$ and each column $k \in[K]$ contains $s_k$ ones, so
\[
Z^\top Z = \mathrm{diag}(s_1,\dots,s_K).
\]
Therefore $\mathrm{rank}(Z) = K$, $\sigma_1(Z) = \sqrt{s_{\max}}$ and 
$\sigma_K(Z) = \sqrt{s_{\min}}$.

By the definition and basic properties of tensor matricisation \citep[e.g. Lemma 4 in ][]{zhang2018tensor}, we see that 
\[
    \mathcal{M}_1(\mathbf{\Omega}) = Z \mathcal{M}_1 (\mathbf{Q})(Z \otimes I_{L})^{\top}.
\]
We aim to show that $\mathrm{rank}(\mathcal{M}_1(\mathbf{\Omega})) = \mathrm{rank}(\mathcal{M}_1(\mathbf{Q}))$.  According to Lemma S3 in the Supplement of \cite{wang2025multilayer}, it suffices to verify that $Z \otimes I_{L} \in \mathbb{R}^{nL \times KL}$ has full column rank, i.e.~$\mathrm{rank}(Z \otimes I_{L}) = KL$. 
Using the property of Kronecker products, we obtain $\mathrm{rank}(Z \otimes I_{L}) = \mathrm{rank}(Z) \mathrm{rank}( I_{L}) = LK$.

Next, we have that
\[
    \|\mathcal{M}_1(\mathbf{\Omega})\| \leq \|Z\| \|\mathcal{M}_1(\mathbf{Q})\| \|Z \otimes I_L\| = \sigma_1^2(Z) \| \mathcal{M}_1(\mathbf{Q}) \| \leq s_{\max} \|\mathbf{Q}\|.
\]
It follows from Lemma S4 in the Supplement of \cite{wang2025multilayer}  that 
\[
    \sigma_{\min} (\mathcal{M}_1(\mathbf{\Omega})) \geq \sigma_{\min} (Z) \sigma_{\min}(\mathcal{M}_1(\mathbf{Q})) \sigma_{\min} (Z \otimes I_L) = \sigma_K^2(Z)  \sigma_{K} (\mathcal{M}_1(\mathbf{Q})) \geq s_{\min} \sigma_{\min} (\mathbf{Q}).
\]
Applying the same arguments, we obtain for $s \in [3]$ that 
\[
\mathrm{rank} \big( \mathcal{M}_s(\mathbf{\Omega} \big)
= \mathrm{rank}(\mathcal{M}_s(\mathbf{Q}))
\]
and
\[ s_{\min} \sigma_{\min} (\mathbf{Q}) \leq  \sigma_{\min}\big(\mathcal{M}_s(\boldsymbol{\Omega})\big)  \leq \big\| \mathcal{M}_s(\boldsymbol{\Omega}) \big\| \leq   s_{\max} \| \mathbf{Q}\|.
\]  
We complete the proof.

\end{proof}

\subsubsection{Proof of Proposition \ref{prop-sin-theta-stat}}\label{sec-app-3}

\begin{proof}[Proof of Proposition \ref{prop-sin-theta-stat}]
By Lemma 1 in \cite{cai2017rate-optimal}, it holds that
\begin{align}\label{prop-sin-0}
\inf_{O \in \mathbb{O}_{K \times K}}\| \widehat{U}^t_Z - U_Z O   \| \leq \sqrt{2} \| \sin\Theta(\widehat{U}^t_Z,  U_Z)   \|. 
\end{align}

For any $t \in [T]$ and $k \in [t]$, define
\begin{equation}\label{eq-def-Omega-stat}
\widehat{\boldsymbol{\Omega}}^{t} = \widehat{\boldsymbol{\Theta}}^{t} + \widehat{\boldsymbol{\Delta}}^{t} \quad \mbox{and} \quad 
\boldsymbol{\Omega} = \boldsymbol{\Theta} + \boldsymbol{\Delta}.
\end{equation}
Let $\mathbf{E}^{\tilde{t}} = \mathbf{\widehat{\Omega}}^{\tilde{t}} - \mathbf{\Omega}$, where $\tilde{t} = 2^{\lfloor \log_2 (t) \rfloor}$. Note that $\tilde{t} \leq t < 2\tilde{t}$.

By Theorem 3 in \cite{zhang2022heteroskedastic} and \Cref{lemma-singular-values-stat},   we obtain that 
\begin{align}\label{prop-sin-1}
 &  \big\| \sin\Theta \big(\widehat{U}_Z^{t},  U_Z \big)   \big\|  \nonumber\\
\leq    &  C_0 \frac{\big\| \mathcal{M}_1(\mathbf{\widehat{\Omega}}^{\tilde{t}})\mathcal{M}_1(\mathbf{\widehat{\Omega}}^{\tilde{t}})^{\top} - \mathcal{M}_1(\mathbf{\Omega})\mathcal{M}_1(\mathbf{\Omega})^{\top}  -  \mathrm{diag} \big(\mathcal{M}_1(\mathbf{\widehat{\Omega}}^{\tilde{t}})\mathcal{M}_1(\mathbf{\widehat{\Omega}}^{\tilde{t}})^{\top} - \mathcal{M}_1(\mathbf{\Omega})\mathcal{M}_1(\mathbf{\Omega})^{\top} \big) \big\| }{\sigma_K^2\big(\mathcal{M}_1(\mathbf{\Omega})\big)} \nonumber\\ 
\leq & C_0 \frac{\big\| \mathcal{M}_1(\mathbf{\widehat{\Omega}}^{\tilde{t}})\mathcal{M}_1(\mathbf{\widehat{\Omega}}^{\tilde{t}})^{\top} - \mathcal{M}_1(\mathbf{\Omega})\mathcal{M}_1(\mathbf{\Omega})^{\top}  -  \mathrm{diag} \big(\mathcal{M}_1(\mathbf{\widehat{\Omega}}^{\tilde{t}})\mathcal{M}_1(\mathbf{\widehat{\Omega}}^{\tilde{t}})^{\top} - \mathcal{M}_1(\mathbf{\Omega})\mathcal{M}_1(\mathbf{\Omega})^{\top} \big) \big\| }{ s_{\min}^2 \sigma_{K}^2\big( \mathcal{M}_1(\mathbf{Q}) \big) },
\end{align}
where $\mathbf{Q}$ is defined in \Cref{ass-stat-rank} and the last inequality follows from \cref{lemma-singular-values-stat}.

Moreover,
\[ \mathcal{M}_1(\mathbf{\widehat{\Omega}}^{\tilde{t}})\mathcal{M}_1(\mathbf{\widehat{\Omega}}^{\tilde{t}})^{\top} - \mathcal{M}_1(\mathbf{\Omega})\mathcal{M}_1(\mathbf{\Omega})^{\top} =  \mathcal{M}_1( \mathbf{E}^{\tilde{t}})  \mathcal{M}_1( \mathbf{E}^{\tilde{t}})^{\top} + \mathcal{M}_1( \mathbf{E}^{\tilde{t}}) \mathcal{M}_1(\mathbf{\Omega})^{\top} + \mathcal{M}_1(\mathbf{\Omega}) \mathcal{M}_1( \mathbf{E}^{\tilde{t}})^{\top}.
\]
Hence,
\begin{align}\label{prop-sin-2}
   &  \big\| \mathcal{M}_1(\mathbf{\widehat{\Omega}}^{\tilde{t}})\mathcal{M}_1(\mathbf{\widehat{\Omega}}^{\tilde{t}})^{\top} - \mathcal{M}_1(\mathbf{\Omega})\mathcal{M}_1(\mathbf{\Omega})^{\top}  -  \mathrm{diag} \big(\mathcal{M}_1(\mathbf{\widehat{\Omega}}^{\tilde{t}})\mathcal{M}_1(\mathbf{\widehat{\Omega}}^{\tilde{t}})^{\top} - \mathcal{M}_1(\mathbf{\Omega})\mathcal{M}_1(\mathbf{\Omega})^{\top} \big) \big\| \nonumber\\
   \leq  & \big\| \mathcal{M}_1( \mathbf{E}^{\tilde{t}})  \mathcal{M}_1( \mathbf{E}^{\tilde{t}})^{\top}  - \mathrm{diag} \big( \mathcal{M}_1( \mathbf{E}^{\tilde{t}})  \mathcal{M}_1( \mathbf{E}^{\tilde{t}})^{\top} ) \big) \big\| + 2\big\|  \mathcal{M}_1( \mathbf{E}^{\tilde{t}}) \mathcal{M}_1(\mathbf{\Omega})^{\top} -  \mathrm{diag} \big( \mathcal{M}_1( \mathbf{E}^{\tilde{t}}) \mathcal{M}_1(\mathbf{\Omega})^{\top} \big) \big\| \nonumber\\
   \leq & \big\| \mathcal{M}_1( \mathbf{E}^{\tilde{t}})  \mathcal{M}_1( \mathbf{E}^{\tilde{t}})^{\top}  - \mathrm{diag} \big( \mathcal{M}_1( \mathbf{E}^{\tilde{t}})  \mathcal{M}_1( \mathbf{E}^{\tilde{t}})^{\top} ) \big) \big\| + 4\big\|  \mathcal{M}_1( \mathbf{E}^{\tilde{t}}) \mathcal{M}_1(\mathbf{\Omega})^{\top}  \big\|  \nonumber\\
   = & (I) + (II), 
\end{align}
where the last inequality follows from Lemma 3 in the Supplement of \cite{zhang2022heteroskedastic}.

\medskip
\noindent
\textbf{Step 1.} In this step, we aim to prove that \begin{align}\label{prop-sin-1.0}
   \P \bigg\{ (I) >  C'\bigg( \frac{ n\sqrt{L} \log^2(n \vee L \vee T)}{ t}  + \frac{n^2 L}{t^2}  \bigg) \bigg\}  \leq ( n \vee L \vee t )^{-c'}, 
\end{align}
where $C', c' > 0$ are absolute constants.

Note that $\widehat{\boldsymbol{\Omega}}^{\tilde{t}} = \widehat{\boldsymbol{\Theta}}^{\tilde{t}} + \widehat{\boldsymbol{\Delta}}^{\tilde{t}}$ and $
\boldsymbol{\Omega} = \boldsymbol{\Theta} + \boldsymbol{\Delta}$. By the triangle inequality, we obtain that 
\begin{align}\label{prop-sin-1.1}
(I) \leq & \bigg\| \{\mathcal{M}_1( \widehat{\boldsymbol{\Theta}}^{\tilde{t}}) - \mathcal{M}_1( \boldsymbol{\Theta})\} \{\mathcal{M}_1( \widehat{\boldsymbol{\Theta}}^{\tilde{t}}) - \mathcal{M}_1( \boldsymbol{\Theta})\}^{\top} \nonumber\\
& \hspace{0.5cm}-  \mathrm{diag} \big( ( \mathcal{M}_1( \widehat{\boldsymbol{\Theta}}^{\tilde{t}}) - \mathcal{M}_1( \boldsymbol{\Theta}) )(\mathcal{M}_1( \widehat{\boldsymbol{\Theta}}^{\tilde{t}}) - \mathcal{M}_1( \boldsymbol{\Theta}))^{\top}\big)\bigg\| \nonumber\\
& \hspace{0.5cm} + \bigg\| \{\mathcal{M}_1( \widehat{\boldsymbol{\Delta}}^{\tilde{t}}) - \mathcal{M}_1( \boldsymbol{\Delta})\} \{\mathcal{M}_1( \widehat{\boldsymbol{\Delta}}^{\tilde{t}}) - \mathcal{M}_1( \boldsymbol{\Delta})\}^{\top} \nonumber\\
& \hspace{0.5cm}-  \mathrm{diag} \big( ( \mathcal{M}_1( \widehat{\boldsymbol{\Delta}}^{\tilde{t}}) - \mathcal{M}_1( \boldsymbol{\Delta}) )(\mathcal{M}_1( \widehat{\boldsymbol{\Delta}}^{\tilde{t}}) - \mathcal{M}_1( \boldsymbol{\Delta}))^{\top}\big)\bigg\| \nonumber\\
& \hspace{0.5cm} + \bigg\| \{\mathcal{M}_1( \widehat{\boldsymbol{\Theta}}^{\tilde{t}}) - \mathcal{M}_1( \boldsymbol{\Theta})\} \{\mathcal{M}_1( \widehat{\boldsymbol{\Delta}}^{\tilde{t}}) - \mathcal{M}_1( \boldsymbol{\Delta})\}^{\top} \nonumber\\
& \hspace{0.5cm}-  \mathrm{diag} \big( ( \mathcal{M}_1( \widehat{\boldsymbol{\Theta}}^{\tilde{t}}) - \mathcal{M}_1( \boldsymbol{\Theta}) )(\mathcal{M}_1( \widehat{\boldsymbol{\Delta}}^{\tilde{t}}) - \mathcal{M}_1( \boldsymbol{\Delta}))^{\top}\big)\bigg\| \nonumber\\
& \hspace{0.5cm} + \bigg\| \{\mathcal{M}_1( \widehat{\boldsymbol{\Delta}}^{\tilde{t}}) - \mathcal{M}_1( \boldsymbol{\Delta})\} \{\mathcal{M}_1( \widehat{\boldsymbol{\Theta}}^{\tilde{t}}) - \mathcal{M}_1( \boldsymbol{\Theta})\}^{\top} \nonumber\\
& \hspace{0.5cm}-  \mathrm{diag} \big( ( \mathcal{M}_1( \widehat{\boldsymbol{\Delta}}^{\tilde{t}}) - \mathcal{M}_1( \boldsymbol{\Delta}) )(\mathcal{M}_1( \widehat{\boldsymbol{\Theta}}^{\tilde{t}}) - \mathcal{M}_1( \boldsymbol{\Theta}))^{\top}\big)\bigg\| \nonumber\\
= & (I.1) + (I.2) + (I.3) + (I.4).
\end{align}

We now focus on the term $(I.1)$ in \eqref{prop-sin-1.1}.  Let 
\[
X^{\tilde{t}} = \mathcal{M}_1( \widehat{\boldsymbol{\Theta}}^{\tilde{t}}) - \mathcal{M}_1( \boldsymbol{\Theta}) \in \R^{n \times nL}.
\]
By \Cref{lemma-sub-Gaussian}, each entry of $X^{\tilde{t}}$  is $C^* t^{-1/2}$-sub-Gaussian distributed for some absolute constant $C^* >0$.

Then we have that
\begin{align}\label{prop-sin-1.2}
(I.1) =  & \bigg\| \sum_{j=1}^{nL} \Big\{(X^{\tilde{t}})^{j} ((X^{\tilde{t}})^{j})^{\top} -  \mathrm{diag}\big((X^{\tilde{t}})^{j} ((X^{\tilde{t}})^{j})^{\top}  \big) \Big\} \bigg\|  \nonumber\\
= &\bigg\| \sum_{j=1}^{nL} \Big\{ (X^{\tilde{t}})^{j} ((X^{\tilde{t}})^{j})^{\top} -  \E\big\{(X^{\tilde{t}})^{j} ((X^{\tilde{t}})^{j})^{\top}  \big\} - \mathrm{diag}\Big((X^{\tilde{t}})^{j} ((X^{\tilde{t}})^{j})^{\top} -  \E\big\{(X^{\tilde{t}})^{j} ((X^{\tilde{t}})^{j})^{\top}  \big\}\Big) \nonumber\\
& \hspace{0.5cm} +  E\big\{(X^{\tilde{t}})^{j} ((X^{\tilde{t}})^{j})^{\top}  \big\} - \mathrm{diag}\Big(\E\big\{(X^{\tilde{t}})^{j} ((X^{\tilde{t}})^{j})^{\top} \big\}\Big) \Big\}\bigg\| \nonumber\\
\leq  &  \bigg\|  \sum_{j=1}^{nL}\Big\{ (X^{\tilde{t}})^{j} ((X^{\tilde{t}})^{j})^{\top} -  \E\big\{(X^{\tilde{t}})^{j} ((X^{\tilde{t}})^{j})^{\top}  \big\} - \mathrm{diag}\Big((X^{\tilde{t}})^{j} ((X^{\tilde{t}})^{j})^{\top} -  \E\big\{(X^{\tilde{t}})^{j} ((X^{\tilde{t}})^{j})^{\top}  \big\}\Big) \Big\} \bigg\|\nonumber\\
& \hspace{0.5cm} +  \bigg\| \sum_{j=1}^{nL}  \Big\{ \E\big\{(X^{\tilde{t}})^{j} ((X^{\tilde{t}})^{j})^{\top}  \big\} - \mathrm{diag}\Big(\E\big\{(X^{\tilde{t}})^{j} ((X^{\tilde{t}})^{j})^{\top}  \big\}\Big) \Big\} \bigg\| \nonumber\\
\leq  & 2 \bigg\| \sum_{j=1}^{nL} \Big\{ (X^{\tilde{t}})^{j} ((X^{\tilde{t}})^{j})^{\top} -  \E\big\{(X^{\tilde{t}})^{j} ((X^{\tilde{t}})^{j})^{\top} \big\} \Big\} \bigg\| + \bigg\|\sum_{j=1}^{nL} \Big\{ \E\big\{(X^{\tilde{t}})^{j} ((X^{\tilde{t}})^{j})^{\top}  \big\}\nonumber\\
& \hspace{0.5cm} - \mathrm{diag}\Big(\E\big\{(X^{\tilde{t}})^{j} ((X^{\tilde{t}})^{j})^{\top}  \big\}\Big) \Big\} \bigg\| \nonumber\\
= & (I.1.1) + (I.1.2),
\end{align}
where the last inequality follows from  Lemma 3 in the Supplement of \cite{zhang2022heteroskedastic}.

For the term $(I.1.1)$ in \eqref{prop-sin-1.2}, the matrix Bernstein inequality \citep[e.g.~Theorem 5.4.1 in][]{vershynin2018high} implies that
\begin{align}\label{prop-sin-1.3}
\P\bigg\{(I. 1.1)  > 
 C_1 \max\bigg\{\frac{n \sqrt{L \log(n \vee L \vee T)}}{t}, \frac{ n\log^2(n \vee L \vee T)}{t}  \bigg\} \leq  (n \vee L \vee t)^{-c_1},
\end{align}
where $C_1, c_1 > 0$ are absolute constants.
For the term $(I. 1.2)$ in \eqref{prop-sin-1.2}, we have that 
\begin{align}\label{prop-sin-1.4}
(I.1.2) \leq&  \sum_{j=1}^{nL} \Big\| \E\big\{ (X^{\tilde{t}})^{j} \{(X^{\tilde{t}})^{j}\}^{\top} - \mathrm{diag}\Big(\E\big\{(X^{\tilde{t}})^{j} ((X^{\tilde{t}})^{j})^{\top}  \big\}\Big)\Big\|_{\mathrm{F}}   \nonumber\\
= & \sum_{j=1}^{nL}  \sqrt{\sum_{i, k =1}^{n} \Big\{\E\{ X^{\tilde{t}}_{i, j} X^{\tilde{t}}_{k, j} \}\Big\}^2 - \sum_{i =1}^{n} \Big\{\E\{ X^{\tilde{t}}_{i, j} X^{\tilde{t}}_{i, j} \}\Big\}^2 }.
\end{align}
By Lemma 18  in \cite{jiang2023autoregressive},  there exists an absolute constant $C_2 > 0$ such that  
\begin{equation}\label{prop-sin-1.5}
 \max\big\{ \vert \E \{ \widehat{\boldsymbol{\Theta}}_{i,j, l}^{t} \} - \boldsymbol{\Theta}_{i,j, l} \vert,  \vert \E \{ \widehat{\boldsymbol{\Delta}}_{i,j, l}^{t} \} - \boldsymbol{\Delta}_{i,j, l} \vert \big\}  \leq C_2  t^{-1}, \forall i, j \in [n], l \in [L], t \in [T].
 \end{equation}
Combining  \eqref{prop-sin-1.4} and \eqref{prop-sin-1.5} yields that 
\begin{align}\label{prop-sin-1.6}
(I.1.2) \leq \frac{C_2^2 n^2L}{t^2}. 
\end{align}
Putting together \eqref{prop-sin-1.2}, \eqref{prop-sin-1.3} and \eqref{prop-sin-1.6}, we obtain that 
\begin{align}
    \P \bigg\{ (I.1) > C_1 \frac{n \sqrt{L} \log^2(n \vee L \vee T)}{t}  + C_2^2\frac{n^2 L}{t^2}   \bigg\}  \leq  ( n \vee L \vee T )^{-c_{1}}. \nonumber 
\end{align}
Similar arguments apply to $(I.2)$, $(I.3)$ and $(I.4)$ in \eqref{prop-sin-1.2}.
Therefore, by a union bound argument, we derive that 
\begin{align}
   \P \bigg\{ (I) > 4C_1 \frac{n \sqrt{ L} \log^2(n \vee L \vee T)}{t}  + 4C_2^2\frac{n^2L}{t^2} \bigg\}     \leq  4( n \vee L \vee T )^{-c_{1}},
\end{align}
which proves \eqref{prop-sin-1.0}.

\medskip
\noindent
\textbf{Step 2}. In this step, we aim to prove that 
\begin{align}\label{prop-sin-2.0}
   \P \bigg\{ (II) >  C'' \big\| \mathcal{M}_1(\mathbf{\Omega}) \big\|   \bigg( \sqrt{\frac{ n \log(n \vee L \vee T)}{ t}} +\frac{\sqrt{n^2L}}{t} \bigg) \bigg\} \leq ( n \vee L \vee t )^{-c''} 
\end{align}
where $C'', c'' > 0$ are absolute constants.

Note that 
\begin{align}\label{prop-sin-2.1}
    (II) =  & 4\big\| \big\{ \mathcal{M}_1( \widehat{\boldsymbol{\Theta}}^{\tilde{t}}) - \mathcal{M}_1( \boldsymbol{\Theta} ) + \mathcal{M}_1( \widehat{\boldsymbol{\Delta}}^{\tilde{t}}) - \mathcal{M}_1( \boldsymbol{\Delta}) 
    \big\}\mathcal{M}_1(\mathbf{\Omega})^{\top}  \big\| \nonumber\\
    \leq &  4\Big\| \Big\{ \mathcal{M}_1( \widehat{\boldsymbol{\Theta}}^{\tilde{t}}) - \mathcal{M}_1( \boldsymbol{\Theta}) 
    -\E \big\{\mathcal{M}_1( \widehat{\boldsymbol{\Theta}}^{\tilde{t}}) - \mathcal{M}_1( \boldsymbol{\Theta})\big\} \Big\}
    \mathcal{M}_1(\mathbf{\Omega})^{\top}  \Big\| \nonumber\\
    & \hspace{0.5cm}
+ 4\Big\| \Big\{ \mathcal{M}_1( \widehat{\boldsymbol{\Delta}}^{\tilde{t}}) - \mathcal{M}_1( \boldsymbol{\Delta}) 
    -\E \big\{\mathcal{M}_1( \widehat{\boldsymbol{\Delta}}^{ \tilde{t}}) - \mathcal{M}_1( \boldsymbol{\Delta})\big\} \Big\}
    \mathcal{M}_1(\mathbf{\Omega})^{\top}  \Big\| \nonumber\\
    & \hspace{0.5cm} + 4\Big\| \E \big\{\mathcal{M}_1( \widehat{\boldsymbol{\Theta}}^{\tilde{t}}) - \mathcal{M}_1( \boldsymbol{\Theta})\big\}
    \mathcal{M}_1(\mathbf{\Omega})^{\top}  \Big\| \nonumber\\
    & \hspace{0.5cm} + 4\Big\|  \E \big\{\mathcal{M}_1( \widehat{\boldsymbol{\Delta}}^{\tilde{t}}) - \mathcal{M}_1( \boldsymbol{\Delta})\big\}
    \mathcal{M}_1(\mathbf{\Omega})^{\top}  \Big\| \nonumber\\
    = & (II.1) + (II.2) + (II.3) + (II.4).
    \end{align}
By Lemma 2 in the Supplement of \cite{zhang2022heteroskedastic} and \Cref{lemma-sub-Gaussian}, we have that 
\begin{align}\label{prop-sin-2.2}
       \P \bigg\{ (II.1) + (II.2) >  C_3 \big\| \mathcal{M}_1(\mathbf{\Omega}) \big\|   \sqrt{\frac{ n \log(n \vee L \vee T)}{ t}}\bigg\} \leq ( n \vee L \vee t )^{-c_3},
\end{align} 
where $C_3, c_3 > 0$ are absolute constants. We also have that
\begin{align}\label{prop-sin-2.3}
(II.3) + (II.4) \leq & 4 \Big[\Big\| \E \big\{\mathcal{M}_1( \widehat{\boldsymbol{\Theta}}^{\tilde{t}}) - \mathcal{M}_1( \boldsymbol{\Theta})\big\} \Big\| 
 + \Big\| \E \big\{\mathcal{M}_1( \widehat{\boldsymbol{\Delta}}^{\tilde{t}}) - \mathcal{M}_1( \boldsymbol{\Delta})\big\} \Big\| \Big]
 \big\|    \mathcal{M}_1(\mathbf{\Omega})^{\top}  \big\|  \nonumber\\
 \leq &  8C_2\frac{n\sqrt{L}}{t}\big\|    \mathcal{M}_1(\mathbf{\Omega})^{\top}  \big\|  
\end{align}
where the last inequality follows from  \eqref{prop-sin-1.5}.
Combining \eqref{prop-sin-2.1},  \eqref{prop-sin-2.2} and  \eqref{prop-sin-2.3} yields  for some absolute constant $C_4 >0$ that
\[
   \P \bigg\{ (II) >  C_4  \big\| \mathcal{M}_1(\mathbf{\Omega}) \big\|  \bigg\{ \sqrt{\frac{ n \log(n \vee L \vee T)}{ t}} +\frac{\sqrt{n^2L}}{t}\bigg\} \leq ( n \vee L \vee t )^{-c_3} 
\]
which proves \eqref{prop-sin-2.0}.

\medskip
\noindent
\textbf{Step 3.}
Combining \eqref{prop-sin-0}, \eqref{prop-sin-1}, \eqref{prop-sin-2}, \eqref{prop-sin-1.0} and \eqref{prop-sin-2.0} yields that
\begin{align}
\P\bigg[ & \inf_{O \in \mathbb{O}_{K \times K}}\| \widehat{U}_Z^t - U_Z O   \| >    C_4 \bigg\{ \frac{ n\sqrt{L} \log^2(n \vee L \vee T)}{ts_{\min}^2 \sigma_{K}^2\big( \mathcal{M}_1(\mathbf{Q}) \big)}  + \frac{n^2 L}{t^2s_{\min}^2 \sigma_{K}^2\big( \mathcal{M}_1(\mathbf{Q}) \big)}  \nonumber\\ 
& \hspace{0.5cm} +  \frac{ s_{\max}\sigma_1(\mathcal{M}_1 (\mathbf{Q}))  }{ s_{\min}^2 \sigma_K^2(\mathcal{M}_1 (\mathbf{Q})) } \bigg(\sqrt{\frac{ n \log(n \vee L \vee T)}{ t}}+ \frac{n \sqrt{L}}{t} \bigg) \bigg\} \leq 2 ( n \vee L \vee t )^{-c_4} , \nonumber 
\end{align}
where $C_4, c_4 > 0$ are absolute constants. Under \Cref{ass-stat-rank}, we finally obtain that
\begin{align}
\P\bigg[ & \inf_{O \in \mathbb{O}_{K \times K}}\| \widehat{U}_Z^t - U_Z O   \| >    C_5 \sigma_Q^{-2}\bigg( \frac{ K^2\sqrt{L} \log^2(n \vee L \vee T)}{tn}  + \frac{K^2 L}{t^2} \bigg)  \nonumber\\ 
& \hspace{0.5cm} +C_5  \sigma_Q^{-1} \bigg(K\sqrt{\frac{  \log(n \vee L \vee T)}{ tn}}+ \frac{ K\sqrt{L}}{t} \bigg) \bigg\} \leq 2 ( n \vee L \vee t )^{-c_4} , \nonumber 
\end{align}
where $C_5 >0$ is an absolute constant. We complete the proof. 
\end{proof}

\subsubsection{Proof of Lemma \ref{lemma-packing}}\label{sec-app-4}
\begin{proof}[Proof of Lemma \ref{lemma-packing}]
Observe that 
\[
 \vert \mathcal {Z} \vert = \frac{n!}{ \big( (q+1)! \big)^r (q!)^{k-r} }.
\]
By the bounds of \cite{robbins1955remark}, for any $m \in \Z^{+}$, 
\[
\sqrt{2\pi} m^{m+1/2} \exp\{ -m +1/(12m+1)\} <   m! <  \sqrt{2\pi} m^{m+1/2} \exp\{ -m +1/(12m)\}.
\]
Applying these to $n!$, $(q+1)!$ and $q!$, we obtain that
\begin{equation}\label{eq-packing-1}
\log  \big( \vert \mathcal {Z} \vert   \big) \geq  n \log(k) - C_1 k \log(n),
\end{equation}
for some absolute constant $C_1 >0$.

Fix any labeling $Z' \in \mathcal Z$.  
For $0 \leq s \leq n$, the number of labelings that differ from $Z'$ in exactly $s$ positions is at most
\[
\binom{n}{s} (k-1)^s,
\]
since we choose the $s$ changed positions and assign to each one of the $k-1$ other labels.
Hence, the number of labelings within Hamming distance at most $t$ of $Z'$ satisfies
\[
V_k(n,t) = \sum_{s=0}^{t} \binom{n}{s} (k-1)^s.
\]

Let $0 < \delta < 1/2$ and set $t = \lfloor \delta n \rfloor$.  
By \Cref{lemma-binomial}, we can derive that 
\begin{equation}\label{eq-packing-2}
\sum_{s=0}^{\lfloor \delta n \rfloor}\binom{n}{s} \le \exp \big\{n H(\delta)\big\},   \quad
\mbox{where } H(\delta) = -\delta \log(\delta) - (1-\delta) \log(1-\delta).
\end{equation}
Thus, we have that 
\begin{align}\label{eq-packing-3}
   V_k(n, \lfloor \delta n \rfloor) \leq  &  \sum_{s=0}^{\lfloor \delta n \rfloor} \binom{n}{s} (k-1)^{ \delta n }  \leq (k-1)^{ \delta n } \exp \big\{n H(\delta)\big\} \nonumber\\
   = &  \exp\Big\{ n \big\{H(\delta) + \delta \log(k-1)\big\} \Big\},
\end{align}
where the second inequality follows from \eqref{eq-packing-2}.

We now construct $\mathcal{Z}^*$ via the usual greedy packing procedure
\begin{itemize}
  \item Initialize $\mathcal{R} = \mathcal{Z}$ and $\mathcal{Z}^* = \emptyset$.
  \item While $\mathcal{R} \neq \emptyset$: pick any $Z \in \mathcal{R}$, add $Z$ to $\mathcal{Z}^*$ and remove from $\mathcal{R}$ all labelings within Hamming distance at most  $\lfloor \delta n \rfloor$ of $Z$.
\end{itemize}
Then for any two distinct $Z, Z' \in \mathcal{Z}^*$, we have  $   d_H(Z, Z') > \lfloor \delta n \rfloor$.
Since each step removes at most $V_k(n, \lfloor \delta n \rfloor)$ labelings, we have that 
\begin{equation}\label{eq-packing-4}
\vert \mathcal{Z}^*\vert   \geq \frac{|\mathcal Z|}{V_k(n, \lfloor \delta n \rfloor)}. 
\end{equation}

Combining \eqref{eq-packing-1}, \eqref{eq-packing-3} and \eqref{eq-packing-4} yields that
\begin{equation}\label{eq-packing-5}
\log \big( \vert \mathcal{Z}^*\vert  \big)\geq n \log(k) - C_1k \log (n) - n \big\{H(\delta) + \delta \log(k-1)\big\} .
\end{equation}
Choose $\delta = 0.1$,  we have $H(0.1) \leq  0.3251$ and for any  $k \geq 2$,
\begin{equation}\label{eq-packing-6}
H(\delta) + \delta \log(k-1) = H(0.1) + 0.1 \log(k-1) \leq 0.3251 + 0.1 \log (k) \leq 2^{-1} \log(k).
\end{equation}
Combining \eqref{eq-packing-5} and \eqref{eq-packing-6}, there exists absolute constant $C_2 > 0$ such that
\[
\log \big( \vert \mathcal{Z}^*\vert  \big) \geq  n \log(k) - C_1k \log (n)  -  2^{-1} n  \log (k) = 2^{-1}n \log(k) - C_1k \log (n)
\geq C_2 n \log(k),
\]
which completes the proof. 
\end{proof}

\subsubsection{Proof of Lemma \ref{lemma-binomial}}\label{sec-app-5}
\begin{proof}[Proof of Lemma \ref{lemma-binomial}]
Let $ X \sim \mathrm{Bin}(n, 1/2) $. Then
\begin{equation}\label{eq-binomial-1}
    \mathbb{P}(X = s) = 2^{-n}\binom{n}{s}
\quad \mbox{and} \quad
\sum_{s=0}^{\lfloor \delta n \rfloor}\binom{n}{s}
= 2^n  \mathbb{P}\{X \leq \lfloor \delta n \rfloor\}
\leq 2^n  \mathbb{P}\{X \leq \delta n\}.
\end{equation}
For any $t < 0$, by Markov's inequality, we have that
\begin{equation}\label{eq-binomial-2}
\mathbb{P} \{X \leq \delta n\}
= \mathbb{P} \Big\{e^{tX} \geq e^{t \delta n}\Big\}
\leq e^{-t \delta n}  \mathbb{E} \big[\exp\{tX\}\big].
\end{equation}

Since $ X = \sum_{i=1}^n Y_i $ with $ Y_i \overset{\mbox{i.i.d.}}{\sim} \mathrm{Bernoulli}(1/2)$, we can derive that 
\begin{equation}\label{eq-binomial-3}
 \mathbb{E} \big[\exp\{tX\}\big]
= \prod_{i=1}^n \mathbb{E}\big[\exp\{tY_i\}\big]
= \Big[ \mathbb{E}\big[\exp\{tY_1\} \big] \Big]^n
= \big\{ (1 + e^t)/2 \big\}^{n}.
\end{equation}
Combining \eqref{eq-binomial-2} and \eqref{eq-binomial-3} yields that 
\begin{equation}\label{eq-binomial-4}
\mathbb{P} \{ X \leq \delta n\}
\leq \exp\Big[ n \big\{ -t \delta + \log \big((1 + e^t)/2 \big) \big\} \Big].
\end{equation}

Define
\[
\phi(t) = -t\delta + \log \big((1 + e^t)/2 \big), \quad \forall t \in \R.
\]
Then, we have that 
\[
\phi'(t) = -\delta + \frac{e^t}{1 + e^t}, \quad 
\phi''(t) = \frac{e^t}{(1 + e^t)^2} > 0, \quad \forall t \in \R.
\]
Hence, $\phi$ is convex and its minimizer satisfies $\phi'(t) = 0$, giving  
\[
t = \log \bigg(\frac{\delta}{1-\delta}\bigg) < 0,
\]
where the last inequality follows from $0 < \delta < 1/2$.

Substitute this $t$ into $\phi(t)$ gives
\begin{align}\label{eq-binomial-5}
\phi(t) \geq &  \phi\bigg( \log \bigg(\frac{\delta}{1-\delta}\bigg)\bigg) 
= -\delta \log \bigg( \frac{\delta}{1-\delta} \bigg)
  + \log \bigg( 2^{-1}\bigg(1 + \frac{\delta}{1-\delta}\bigg) \bigg) \nonumber\\
=&  -\delta \log \delta - (1-\delta)\log(1-\delta) - \log 2
= H(\delta) - \log 2.
\end{align}

Combining \eqref{eq-binomial-1}, \eqref{eq-binomial-4} and \eqref{eq-binomial-5}, we finally obtain that 
\[
\sum_{s=0}^{\lfloor \delta n \rfloor}\binom{n}{s} 
\leq 2^n \exp\Big\{n \big\{H(\delta) - \log (2) \big\}\Big\} = \exp\big\{nH(\delta) \big\},
\]
which completes the proof. 
\end{proof}

\section{Technical details for results in Section \ref{sec-nonstat}}\label{app-proof-sec-nonstat}

All auxiliary results for the non-stationary setting are collected in Appendix~\ref{app-non-stat-add}, with their proofs provided in Appendix~\ref{app-non-stat-add-proof}. The proofs of the main results in Section~\ref{sec-nonstat}, namely \Cref{thm-tensor-Frobenius} and \Cref{thm-comunity-recovery-non}, are given in Appendix~\ref{app-sec-non-proof}.

\subsection{Additional results}\label{app-non-stat-add}

We collect several auxiliary results for the non-stationary AR(1)-MSBM defined in \Cref{def-armsb}. These results characterize temporal dependence, bias of local estimators and spectral properties of window-averaged tensors, which are key ingredients for the analysis in Section~\ref{sec-nonstat}.

We first establish a strong mixing property for edgewise processes, extending \Cref{lemma-mixing-stat} to the non-stationary setting.

\begin{lemma}\label{lemma-mixing-nonstat}
Let the process $\{\mathbf{A}^t \}_{t \geq 0} \subset\{0, 1\}^{n \times n \times L}$ be defined in \Cref{def-armsb}. For any $ 1 \leq i \leq j \leq n$ and $l \in [L]$, the process $\{\mathbf{A}^t_{i, j, l} \}_{t \geq 0} \subset\{0, 1\}^{n \times n \times L}$ is  strongly mixing. 
\end{lemma}

We next quantify the temporal dependence of the edgewise processes. The following lemma shows that the covariance between edge variables is non-negative and decays geometrically with the time lag.

\begin{lemma}\label{lem-temporal-cov}
    Let the process $\{\mathbf{A}^t \}_{t \geq 0} \subset\{0, 1\}^{n \times n \times L}$ be defined in \Cref{def-armsb}. For any  $1\leq i\leq j\leq n$, $l \in [L]$ and  $ r, s \in [T]$, it holds that 
    \[
    0 \leq \mathrm{Cov}(\mathbf{A}_{i,j,l}^{r},\mathbf{A}_{i,j,l}^{s})
\leq  \frac{1}{4} (1-2c_{\min})^{\vert s-r \vert },
    \]
\end{lemma}

We now analyze the bias of the local look-back estimators for the transition probabilities.

\begin{proposition}\label{prop-bias}
Let the process $\{\mathbf{A}^t \}_{t \geq 0} \subset\{0, 1\}^{n \times n \times L}$ be defined in \Cref{def-armsb} and let $G$ denote the number of $V$-quasi-stationary segments in $\{(\boldsymbol{\Theta}^t, \boldsymbol{\Delta}^t)\}_{t=1}^{T}$ as defined in \Cref{def-seg}. Suppose \Cref{ass-bias}$(i)$-$(iii)$ hold.  
Then for any $t \in [T]$ and $k \in [t]$,   there exists an absolute constant $C>0$ such that
\begin{align}\label{eq-bias-1}
\Big\vert\mathbb{E}\big\{\widehat{\boldsymbol{\Theta}}_{i,j,l}^{t,k} \big\} - \boldsymbol{\Theta}_{i,j, l}^{t, k} \Big\vert \leq   & \big\vert  \mathcal{E}^{t,k}_{i, j, l} \big\vert  + \frac{C}{k}, \quad \forall 1 \leq i \leq j \leq n, l \in [L], 
\end{align}
and
\begin{align}\label{eq-bias-2}
\Big\vert \mathbb{E} \big\{ \widehat{\boldsymbol{\Delta}}_{i,j, l}^{t, k} \big\} - \boldsymbol{\Delta}_{i,j, l}^{t, k} \Big\vert  \leq &  \big\vert \widetilde{\mathcal{E}}^{t,k}_{i, j, l} \big\vert  + \frac{C}{k}, \quad \forall 1 \leq i \leq j \leq n, l \in [L].
\end{align}
Here $\boldsymbol{\Theta}^{t, k}$ and $\boldsymbol{\Delta}^{t, k}$ are defined in \eqref{def-Theta-t-k}.  If $k=1$, then
\[
\mathcal{E}^{t,k}_{i, j, l} =  \widetilde{\mathcal{E}}^{t,k}_{i, j, l} = 0;
\]
and if $t -k  \geq T_g +1$ for some $g \in [G]$, then
\[
\big\vert  \mathcal{E}^{t,k}_{i, j, l} \big\vert \leq   \frac{1}{2c_{\min}^2}  \max_{s, u \in [t]\backslash[t-k]} \vert \boldsymbol{\Theta}^{s}_{i, j, l} -  \boldsymbol{\Theta}^{u}_{i, j, l} \vert
\sum_{u' = T_g +2}^{t}  \big( \big\vert \boldsymbol{\Theta}^{u'}_{i,j,l} -   \boldsymbol{\Theta}^{u'-1}_{i,j,l}  \big\vert  +  \big\vert  \boldsymbol{\Delta}^{u'}_{i,j,l} - \boldsymbol{\Delta}^{u'-1}_{i,j,l}  \big\vert\big),  
\]
and 
\[
 \big\vert   \widetilde{\mathcal{E}}^{t,k}_{i, j, l}\big\vert 
\leq   \frac{1}{2c_{\min}^2}  \max_{s, u \in [t]\backslash[t-k]} \vert \boldsymbol{\Delta}^{s}_{i, j, l} -  \boldsymbol{\Delta}^{u}_{i, j, l} \vert
\sum_{u' = T_g +2}^{t}  \big( \big\vert \boldsymbol{\Theta}^{u'}_{i,j,l} -   \boldsymbol{\Theta}^{u'-1}_{i,j,l}  \big\vert  +  \big\vert  \boldsymbol{\Delta}^{u'}_{i,j,l} - \boldsymbol{\Delta}^{u'-1}_{i,j,l}  \big\vert \big). 
\]

\end{proposition}

We now establish spectral properties of the window-averaged transition probability tensors in the non-stationary regime. The next lemma shows that their Tucker ranks and singular values behave analogously to the stationary case described in \Cref{lemma-singular-values-stat}.

\begin{lemma}\label{lemma-singular-values-non-stat}
For any $t \in [T]$ and any look-back window size $k \in [t]$, let $\mathbf{Q}^{t, k} \in \R^{K\times K \times L}$ be the tensor defined in \Cref{ass-bias}$(iv)$. Assume that $\mathrm{rank}\big( \mathcal{M}_1 (\mathbf{Q}^{t, k}) \big) = K$ and denote $\bar{r}^{t, k} = \mathrm{rank}\big( \mathcal{M}_3 (\mathbf{Q}^{t, k}) \big)$. Define
\[
\boldsymbol{\Omega}^{t, k} =  \frac{1}{k}\sum_{u=t-k+1}^{t} \ (\boldsymbol{\Theta}^{u} + \boldsymbol{\Delta}^{u} ). 
\]
Then, for any $t \in [T]$ and $k \in [t]$, we have that 
\[
\mathrm{rank} \big(\mathcal{M}_1(\mathbf{\Omega}^{t, k})\big)= \mathrm{rank} \big(\mathcal{M}_2(\mathbf{\Omega}^{t, k})\big)= K,
\quad \mathrm{rank} \big(\mathcal{M}_3(\mathbf{\Omega}^{t, k})\big)= \bar{r}^{t, k}, 
\]
and
\[ 
s_{\min} \sigma_{\min} (\mathbf{Q}^{t, k}) \leq  \sigma_{\min}\big(\mathcal{M}_s(\boldsymbol{\Omega}^{t, k})\big)  \leq \big\| \mathcal{M}_s(\boldsymbol{\Omega}^{t, k}) \big\| \leq   s_{\max} \| \mathbf{Q}^{t, k}\|, \quad \forall s \in [3],
\] 
where $s_{\min} = \min_{j \in [K]} s_j$ and $s_{\max} = \max_{j \in [K]} s_j$ with for any $j \in [K]$, $s_j = \sum_{i=1}^{n} Z_{i, j}$. 
\end{lemma}

Finally, we establish an error bound for the subspace estimator obtained via H-PCA (\textbf{Stage II} of \Cref{alg:fast-msbm}).

\begin{proposition}\label{prop-sin-theta}
Let the process $\{\mathbf{A}^t \}_{t \geq 0} \subset\{0, 1\}^{n \times n \times L}$ be defined in \Cref{def-armsb}  and let $G$ be the number of $V$-quasi-stationary segments in $\{(\boldsymbol{\Theta}^t, \boldsymbol{\Delta}^t)\}_{t=1}^{T}$ as defined in \Cref{def-seg}. Suppose that \Cref{ass-bias} holds.
Define
\begin{equation}\label{def-tau}
\tau(t) =  C_{\tau}\max \Big\{t^{-1} \log(n \vee L \vee T) \log(T) \log\log(T), \big(V(t)\big)^2 \Big\}, \quad \forall t \in [T],
\end{equation}
where $ C_{\tau} >0 $ is a sufficiently large constant.  For any $g \in [G]$ and any $t \in [T_g]\backslash[T_{g-1}]$, under \Cref{alg:fast-msbm}, with probability at least $1- (n \vee L \vee T)^{-c}$,
\[
\inf_{O \in \mathbb{O}_{K \times K}}\| \widehat{U}^t_Z - U_Z O   \| \leq   \widetilde{\mathcal{E}}_t,
\]
where
\begin{align}\label{eq-sin-theta}
\widetilde{\mathcal{E}}_t =  C  \sigma_Q^{-2}  \bigg( \frac{K^2\sqrt{L}}{n}  \epsilon_{t}^2 + K^2L\epsilon_{t}^4 \bigg)   
+  C\sigma_Q^{-1}  \bigg( \frac{ K}{\sqrt{n}}   \epsilon_{t}  +     K  \sqrt{L} \epsilon_{t}^2  \bigg), 
\end{align}
with
\begin{equation}\label{def-epsilon-t}
 \epsilon_{t} = \max \bigg\{\sqrt{\frac{ \log(n \vee L \vee T)\log(T) \log\log(T)}{t-T_{g-1}}}, V(t-T_{g-1})\bigg\}, \quad \sigma_Q =\min_{t \in [T], s \in [t]}\sigma_{\min}(\mathbf{Q}^{t, s}),
\end{equation}
and absolute constants $C, c >0$. 
\end{proposition}

When $K$ are treated as a constant and $\sigma_Q \asymp \sqrt{L}$, the bound in \eqref{eq-sin-theta} simplifies to
\begin{equation}\label{eq-E_t-sim}
             \frac{ \epsilon_t }{\sqrt{nL}}+ \epsilon_t^2 + \epsilon_t^4.
\end{equation}
Thus, a smaller $\epsilon_t$ leads to a sharper bound.
However, $\epsilon_t$ itself reflects a trade-off between the segment length $(t-T_{g-1})$ and the drift envelope $V(\cdot)$:
a fast-decaying $V(\cdot)$ creates more and shorter segments and increases the stochastic term $(t-T_{g-1})^{-1/2}$, whereas a slow-decaying $V(\cdot)$ yields longer segments and reduces stochastic variability, but at the cost of a larger drift term $V(t-T_{g-1})$. For example, if
\[
    V(t) = t^{-1/2},  \quad  \forall t>0,
\]
then up to logarithmic factors,
\[
    \widetilde{\mathcal{E}}_t  \lesssim  \frac{1}{\sqrt{(t-T_{g-1})nL}}+ \frac{1}{t-T_{g-1}}
\]
which matches the stationary-rate bound in \eqref{eq-stat-rate},  when restricted to each quasi-stationary segment. 
This shows that \Cref{alg:fast-msbm} automatically adapts its look-back window: it contracts the window in the presence of substantial cumulative drift or abrupt changes, while expanding it to the largest statistically reliable size when the process remains quasi-stationary.

\subsection{Proofs for Section \ref{sec-nonstat}}\label{app-sec-non-proof}

To establish the theoretical results for \Cref{alg:fast-msbm}, we first analyze the windowed estimators of the transition probabilities, $\widehat{\boldsymbol{\Theta}}^{t, \hat{k}_t}$ and $\widehat{\boldsymbol{\Delta}}^{t, \hat{k}_t}$. 
This analysis relies on two key ingredients: (i) a concentration inequality (\Cref{main-theorem-1}, with proof provided in Appendix~\ref{app-proof-main-theorem-1}), and (ii) a bias control result (\Cref{prop-bias-final}, with proof provided in Appendix~\ref{sec-pro-bias}).

\begin{theorem}\label{main-theorem-1}
Let the process $\{\mathbf{A}^t \}_{t \geq 0} \subset\{0, 1\}^{n \times n \times L}$ be defined as in \Cref{def-armsb}. Let $G$ be the number of $V$-quasi-stationary segments in $\{(\boldsymbol{\Theta}^t, \boldsymbol{\Delta}^t)\}_{t=1}^{T}$ as defined in \Cref{def-seg}. Suppose that \Cref{ass-bias}$(i)$ holds. 
For \Cref{alg:fast-msbm}, define
\[
\tau(t) =  C_{\tau}\max \Big\{t^{-1} \log(n \vee L \vee T) \log(T) \log\log(T), \big(V(t)\big)^2 \Big\}, \quad \forall t \in [T]
\]
where $C_{\tau} > 0$ is a sufficiently large constant. 
Then 
\begin{align}
 \P \bigg\{  & \big \vert \widehat{\boldsymbol{\Theta}}^{t,\hat{k}_t} _{i, j, l}  -  \boldsymbol{\Theta}^{t}_{i, j, l}  \big \vert^2  + \big \vert \widehat{\boldsymbol{\Delta}}^{t,\hat{k}_t}_{i, j, l}  -  \boldsymbol{\Delta}^{t} _{i, j, l} \big \vert^2    
 \leq  C \max \bigg\{ \frac{ \log(n \vee L \vee T)  \log(T) \log\log(T)}{ t - T_{g-1}},   \big(V(t-T_{g-1})\big)^2 \bigg\},  \nonumber\\
 & \hspace{0.5cm} \forall
  1 \leq i\leq j \leq n, l\in [L], \mbox{ for any } t \in [T_{g}]/ [T_{g-1}]  \mbox{ with } g \in [G] \bigg\} \geq 1- (n \vee L \vee T)^{-c},
\end{align}
for some absolute constants $C, c> 0$.

\end{theorem}

\begin{proposition}\label{prop-bias-final}
Let the process $\{\mathbf{A}^t \}_{t \geq 0} \subset\{0, 1\}^{n \times n \times L}$ be defined in \Cref{def-armsb}, and let $G$ be the number of $V$-quasi-stationary segments in $\{(\boldsymbol{\Theta}^t, \boldsymbol{\Delta}^t)\}_{t=1}^{T}$ as defined in \Cref{def-seg}. Suppose that \Cref{ass-bias} holds. For \Cref{alg:fast-msbm}, let 
\[
\tau(t) =  C_{\tau}\max \Big\{t^{-1} \log(n \vee L \vee T) \log(T) \log\log(T), \big(V(t)\big)^2 \Big\}, \quad \forall t \in [T],
\]
where $C_{\tau} > 0$ is a sufficiently large constant. Then 
\begin{align}
& \P \bigg[ \max \Big\{ \Big\vert\mathbb{E}\big\{\widehat{\boldsymbol{\Theta}}_{i,j,l}^{t,\hat{k}_t} \big\} - \boldsymbol{\Theta}_{i,j, l}^{t, \hat{k}_t} \Big\vert, \Big\vert \mathbb{E} \big\{ \widehat{\boldsymbol{\Delta}}_{i,j, l}^{t, \hat{k}_t} \big\} - \boldsymbol{\Delta}_{i,j, l}^{t, \hat{k}_t} \Big\vert  \Big\}   >   C \max \bigg\{ \frac{ \log(n \vee L \vee T)  \log(T) \log\log(T)}{ t - T_{g-1}}, 
\nonumber \\
& \hspace{0.5cm} \big(V(t-T_{g-1})\big)^2 \bigg\},   1 \leq i\leq j \leq n, l\in [L],
\mbox{ for any } t \in [T_{g}]/ [T_{g-1}]  \mbox{ with } g \in [G]
\bigg] 
\leq  (n\vee L \vee T)^{-c}, \nonumber
\end{align}
for some absolute constants $C, c> 0$. Here, for any $t \in [T]$ and $k \in [t]$, we define
\begin{equation}\label{def-Theta-t-k}
    \boldsymbol{\Theta}^{t, k}  = \frac{1}{k} \sum_{u=t-k+1}^t    \boldsymbol{\Theta}^{u} \quad \mbox{and} \quad 
     \boldsymbol{\Delta}^{t, k}  = \frac{1}{k} \sum_{u=t-k+1}^t    \boldsymbol{\Delta}^{u}.
\end{equation}
\end{proposition}

Based on \Cref{main-theorem-1}, \Cref{prop-bias-final} and the spectral properties of the window-averaged transition probability tensors established in \Cref{lemma-singular-values-non-stat}, the proofs of \Cref{thm-tensor-Frobenius},  \Cref{thm-comunity-recovery-non} and \Cref{prop-sin-theta}
follow the same structure as their stationary counterparts
(\Cref{thm-tensor-Frobenius-stat}, \Cref{thm-comunity-recovery-stat} and \Cref{prop-sin-theta-stat}).
We therefore omit the details.

\subsubsection{Proof of Theorem \ref{main-theorem-1}}
\label{app-proof-main-theorem-1}

The following notion of functional similarity was introduced by \cite{huang2023stability} and we copy it here for completeness.

\begin{definition}[Closeness]
Let $f, g \colon \mathcal{X} \to \mathbb{R}$ be lower bounded functions and $\varepsilon, \delta \geq 0$. We say functions $f$ and $g$ are $(\varepsilon, \delta)$-close, if for all $x \in \mathcal{X}$, 
\[
g(x) - \inf_{x' \in \mathcal{X}} g(x') \leq e^\varepsilon \Big( f(x) - \inf_{x' \in \mathcal{X}}  f(x')  + \delta \Big),
\]
and
\[
f(x) - \inf_{x' \in \mathcal{X}}  f(x')  \leq e^\varepsilon \Big( g(x) - \inf_{x' \in \mathcal{X}} g(x')  + \delta \Big).
\]
\end{definition}

We now introduce the tensor-valued population loss  function
\[
F^{t} \colon [0, 1]^{n \times n \times L} \times [0, 1]^{n \times n \times L} \to  \R^{n \times n \times L}
\]
which maps a candidate pair $(\boldsymbol{\Theta}', \boldsymbol{\Delta}')$ to a tensor whose $(i, j, l)$-th entry is
\begin{align}
 F^t_{i, j, l}(\boldsymbol{\Theta}'_{i, j, l}, \boldsymbol{\Delta}'_{i, j, l} ) =  &  - \big[ \big\{\boldsymbol{\Theta}^{t}_{i, j, l} \log(\boldsymbol{\Theta}'_{i, j, l})  + (1- \boldsymbol{\Theta}_{i, j, l}^{t}) \log (1- \boldsymbol{\Theta}'_{i, j, l}) \big\} ( 1 - \boldsymbol{\Pi}_{i,j,l}^{t-1} ) \nonumber \\
& \hspace{0.5cm} + \big\{ \boldsymbol{\Delta}^{t}_{i, j, l} \log(\boldsymbol{\Delta}'_{i, j, l})  + (1- \boldsymbol{\Delta}_{i, j, l}^{t}) \log (1-\boldsymbol{\Delta}'_{i, j, l}) \big\}  \boldsymbol{\Pi}_{i,j,l}^{t-1}   \big], \nonumber
\end{align}
where for any $t \in [T]$, $\boldsymbol{\Pi}^{t} = \mathbb{E}\{\mathbf{A}^{t}\}$.
For any $t \in [T]$ and any window length $k \in [t]$, we define the averaged population losses as
\[
F^{t,k} = \frac{1}{k} \sum_{u=t-k+1}^{t} F^u.
\]

\begin{proof}[Proof of Theorem \ref{main-theorem-1}]

For any $t \in [T]$ and $k \in [t]$, define
\[
\boldsymbol{\Pi}^{t} = \mathbb{E}\{\mathbf{A}^{t}\} \quad \mbox{and} \quad
\boldsymbol{\Pi}^{t,k} = \frac{1}{k} \sum_{s=t-k+1}^t \boldsymbol{\Pi}^{s-1}.
\]

Note that for any $t \in [T]$, $1 \leq i \leq j \leq n$ and $l \in [L]$, since $\boldsymbol{\Theta}^t_{i, j, l}, \boldsymbol{\Delta}^t_{i, j, l} \in [c_{\min}, 1-c_{\min}]$, we can derive the following bounds 
\begin{align*}
\boldsymbol{\Pi}_{i,j,l}^{t}  = & \P \big\{ \mathbf{A}^t_{i,j, l} =1 \mid \mathbf{A}^{t-1}_{i,j, l}  =1\big\} \P \big\{ \mathbf{A}^{t-1}_{i,j, l}  =1\big\} +  \P \big\{ \mathbf{A}^t_{i,j, l} =1 \mid \mathbf{A}^{t-1}_{i,j, l}  =0 \big\} \P \big\{ \mathbf{A}^{t-1}_{i,j, l}  =0 \big\} \nonumber\\
= & \big( 1- \boldsymbol{\Delta}^t_{i, j, l} \big) \boldsymbol{\Pi}_{i,j,l}^{t-1}   + \boldsymbol{\Theta}^t_{i, j, l}  \big( 1 -  \boldsymbol{\Pi}_{i,j,l}^{t-1}  \big) 
\geq  c_{\min} \boldsymbol{\Pi}_{i,j,l}^{t-1} +c_{\min}  \big( 1 -  \boldsymbol{\Pi}_{i,j,l}^{t-1}\big)  = c_{\min}. 
\end{align*}
Similarly,
\begin{align*}
\boldsymbol{\Pi}_{i,j,l}^{t} = 
& \big( 1- \boldsymbol{\Delta}^t_{i, j, l} \big) \boldsymbol{\Pi}_{i,j,l}^{t-1} + \boldsymbol{\Theta}^t_{i, j, l}  \big( 1 - \boldsymbol{\Pi}_{i,j,l}^{t-1} \big) \nonumber\\
\leq & (1-c_{\min}) \boldsymbol{\Pi}_{i,j,l}^{t-1} + (1-c_{\min})  \big( 1 -  \boldsymbol{\Pi}_{i,j,l}^{t-1} \big)  \leq 1-c_{\min}. 
\end{align*}
Thus, by \Cref{def-armsb},  we can conclude that 
\begin{equation}\label{eq-pi-upper-lower}
c_{\min} \leq     \boldsymbol{\Pi}_{i,j,l}^{t} \leq 1-c_{\min}, \quad \forall t \in [T] \cup\{0\}.
\end{equation}

\medskip
\noindent \textbf{Step 1.}
In this step, we will show that  with high probability for any $1 \leq i \leq j \leq n$, $l \in [L]$, $t \in [T]$ and $k \in [t]$, 
\begin{equation}\label{eq-step1-goal}
f^{t,k}_{i, j, l}  \mbox{ and } F^{t,k}_{i, j, l}  \mbox{ are } \big( \log(2), \psi(k) \big)\mbox{-close}.
\end{equation}
where $ \psi(k) = C_{\psi}\log(n \vee L \vee T) \log(T) \log\log(T) /k $ with a constant $C_{\psi} > 0$.

\medskip
\noindent \textbf{Step 1.1.} In this sub-step,  we will derive a concentration inequality for $\nabla f^{t,k}_{i, j, l} (\boldsymbol{\Theta}_{i, j, l}, \boldsymbol{\Delta}_{i, j, l}) \in \R^{2}$ around $\nabla F^{t,k}_{i, j, l} (\boldsymbol{\Theta}_{i, j, l}, \boldsymbol{\Delta}_{i, j, l}) \in \R^{2}$ for any $t \in [T]$, $k \in [t]$ and $(\boldsymbol{\Theta}_{i, j, l}, \boldsymbol{\Delta}_{i, j, l})  \in \mathcal{Q}'$, where $\mathcal{Q}' = \{ (\theta, \delta)^{\top} \in [c_{\min}, 1-c_{\min}]^{2} \}$.

By \Cref{lemma-mixing-nonstat},  for any $1 \leq i \leq j \leq n$ and $l \in [L]$, the sequence $\{\mathbf{A}_{i,j,l}^{t}\}_{t \geq 0}$ is strongly mixing.
Since the mixing property is hereditary, any process defined as a measurable function of finitely many $\mathbf{A}_{i,j,l}^{t}$’s is also strongly mixing. In particular, the processes
\[
\big\{\mathbf{A}_{i,j,l}^{t}(1-\mathbf{A}_{i,j,l}^{t-1})\big\}_{t\geq 1},\quad \big\{(1-\mathbf{A}_{i,j,l}^{t})\mathbf{A}_{i,j,l}^{t-1}\big\}_{t\geq 1} \quad \mbox{and} \quad \big\{(1-\mathbf{A}_{i,j,l}^{t})(1-\mathbf{A}_{i,j,l}^{t-1}) \big\}_{t\geq 1}
\]
are also strongly mixing.
For any  $1 \leq i \leq j \leq n$,  $l \in [L]$, $t \in [T]$ and $k \in [t]$, let
\[
 \mathbf{Y}_{i, j, l}^{t, k} = \Big \vert k^{-1}\sum_{u=t-(k-1)}^{t} \Big\{ \mathbf{A}_{i,j,l}^{u}(1-\mathbf{A}_{i,j,l}^{u-1})  - \E\big\{\mathbf{A}_{i,j,l}^{u}(1-\mathbf{A}_{i,j,l}^{u-1}) \big\} \Big\} \Big \vert, 
\]
\[
 \widetilde{\mathbf{Y}}_{i, j, l}^{t, k} = \Big \vert k^{-1}\sum_{u=t-(k-1)}^{t} \Big\{  (1-\mathbf{A}_{i,j,l}^{u})(1-\mathbf{A}_{i,j,l}^{u-1})  - \E\big\{(1-\mathbf{A}_{i,j,l}^{u})(1-\mathbf{A}_{i,j,l}^{u-1}) \big\} \Big\} \Big \vert, 
\]
\[
 \mathbf{X}_{i, j, l}^{t, k}  = \Big \vert  k^{-1}\sum_{u=t-(k-1)}^{t}\Big\{(1-\mathbf{A}_{i,j,l}^{u})\mathbf{A}_{i,j,l}^{u-1}  - \E\big\{(1-\mathbf{A}_{i,j,l}^{u})\mathbf{A}_{i,j,l}^{u-1} \big\} \Big\} \Big \vert, 
\]
and 
\[
 \widetilde{\mathbf{X}}_{i, j, l}^{t, k} = \Big \vert  k^{-1}\sum_{u=t-(k-1)}^{t}\Big\{\mathbf{A}_{i,j,l}^{u}\mathbf{A}_{i,j,l}^{u-1}  - \E\big\{\mathbf{A}_{i,j,l}^{u}\mathbf{A}_{i,j,l}^{u-1} \big\}\Big\} \Big\vert. 
\]
By applying the Bernstein inequality for $\alpha$-mixing sequences \citep[e.g.~Theorem 1 in][]{merlevede2009bernstein}, we obtain that  
\[
\P\big\{  \mathbf{Y}_{i, j, l}^{t, k} \geq  C_0 x \big\} \leq \exp\bigg\{- \frac{ c_{0} kx^{2}}
         {1 + x\log (k)\log\log (k)}\bigg\}.
\]
where $C_0 >0$ is sufficiently large constant and $c_0 > 10$ is an absolute constant. 
Setting
\[
x =\max\bigg\{\sqrt{ \frac{ \log(n \vee L\vee  T)}{k}}, \frac{\log(n \vee L \vee T) \log (k)\log\log (k)}{k} \bigg\},
\]
yields
\[
\P\bigg\{  \mathbf{Y}_{i, j, l}^{t, k} \geq  C_0\max\bigg\{\sqrt{ \frac{ \log(n \vee L \vee T)}{k}}, \frac{\log(n \vee L \vee T) \log (k)\log\log (k)}{k} \bigg\} \leq (n\vee L \vee T)^{-c_0/2}.
\] 
Similarly, we can conclude that
\begin{align}\label{def-event-A-prob}
 & \P \bigg\{  \max \Big\{ \mathbf{Y}_{i, j, l}^{t, k}, \widetilde{\mathbf{Y}}_{i, j, l}^{t, k}, \mathbf{X}_{i, j, l}^{t, k}, \widetilde{\mathbf{X}}_{i, j, l}^{t, k}  \Big\}  \leq C_0\max\bigg( \sqrt{ \frac{ \log(n \vee L \vee T)}{k}}, \nonumber\\
& \hspace{5cm} \frac{\log(n \vee L \vee T) \log (k)\log\log (k)}{k}\bigg) \bigg\} \geq 1- 4 (n\vee L \vee T)^{-c_0/2}.
\end{align}

For any $1  \leq i \leq j \leq n$,  $l \in [L]$, $t\in [T]$, $k \in [t]$ and $\theta, \delta \in [c, 1-c]$, define 
\[
  \tilde{f}^{t, k}_{i, j, l} (\theta) =  - k^{-1}\sum_{u=t-(k-1)}^{t} \bigg\{ \frac{  \mathbf{A}^u_{i, j, l} (1-\mathbf{A}^{u-1}_{i, j, l}) }{\theta}  -  \frac{(1-\mathbf{A}^u_{i, j, l}) (1-\mathbf{A}^{u-1}_{i, j, l}) }{1-\theta}\bigg\},
\]
\[
\bar{f}^{t, k}_{i, j, l} (\delta)= - k^{-1} \sum_{u=t-(k-1)}^{t} 
       \bigg\{  \frac{ (1- \mathbf{A}^u_{i, j, l}) \mathbf{A}^{u-1}_{i, j, l}}{\delta}  - \frac{\mathbf{A}^u_{i, j, l} \mathbf{A}^{u-1}_{i, j, l}}{1-\delta} \bigg\},
\]
\[
\tilde{F}^{t, k}_{i, j, l} (\theta) =     - k^{-1}\sum_{u=t-(k-1)}^{t}    \bigg\{   \frac{  \E\{\mathbf{A}^u_{i, j, l}  (1-\mathbf{A}^{u-1}_{i, j, l})\} }{\theta}  -  \frac{ \E\{(1-\mathbf{A}^u_{i, j, l}) (1-\mathbf{A}^{u-1}_{i, j, l}) \}}{1-\theta}   \bigg\}, 
\]
\[
\bar{F}^{t, k}_{i, j, l} (\delta)  =     - k^{-1} \sum_{u=t-(k-1)}^{t}   \bigg\{   \frac{  \E\{(1- \mathbf{A}^u_{i, j, l}) \mathbf{A}^{u-1}_{i, j, l}\}}{\delta}  - \frac{\E\{\mathbf{A}^u_{i, j, l} \mathbf{A}^{u-1}_{i, j, l}\}}{1-\delta}   \bigg\}.
\]
Then we have
\[
\nabla_{\boldsymbol{\Theta}_{i, j, l}} f^{t, k}_{i, j, l} (\boldsymbol{\Theta}_{i, j, l} , \boldsymbol{\Delta}_{i, j, l} ) 
= \tilde{f}^{t, k}_{i, j, l} \big( \boldsymbol{\Theta}_{i, j, l} \big),  \quad
\nabla_{\boldsymbol{\Delta}_{i, j, l}} f^{t, k}_{i, j, l} (\boldsymbol{\Theta}_{i, j, l} , \boldsymbol{\Delta}_{i, j, l} )=   
   \bar{f}^{t, k}_{i, j, l}  \big( \boldsymbol{\Delta}_{i, j, l} \big),
\]
\[
\nabla_{\boldsymbol{\Theta}_{i, j, l}} F^{t, k}_{i, j, l} (\boldsymbol{\Theta}_{i, j, l} , \boldsymbol{\Delta}_{i, j, l} ) =  \tilde{F}^{t, k}_{i, j, l} \big( \boldsymbol{\Theta}_{i, j, l} \big), \quad 
    \nabla_{\boldsymbol{\Delta}_{i, j, l}}    F^{t, k}_{i, j, l}  (\boldsymbol{\Theta}_{i, j, l} , \boldsymbol{\Delta}_{i, j, l} )  =\bar{F}^{t, k}_{i, j, l} \big( \boldsymbol{\Delta}_{i, j, l} \big).
\]
Therefore,
\begin{align}
  \| \nabla f^{t, k}_{i, j, l}(\boldsymbol{\Theta}_{i, j, l}, \boldsymbol{\Delta}_{i, j, l}) - \nabla F^{t, k}_{i, j, l}(\boldsymbol{\Theta}_{i, j, l}, \boldsymbol{\Delta}_{i, j, l})\|_2^2
= & \Big\{ \tilde{f}^{t, k}_{i, j, l} \big( \boldsymbol{\Theta}_{i, j, l} \big) - \tilde{F}^{t, k}_{i, j, l} \big( \boldsymbol{\Theta}_{i, j, l} \big)  \Big\}^2 \nonumber\\
& \hspace{1cm} +  \Big\{\bar{f}^{t, k}_{i, j, l}  \big( \boldsymbol{\Delta}_{i, j, l} \big) - \bar{F}^{t, k}_{i, j, l} \big( \boldsymbol{\Delta}_{i, j, l} \big) \Big\}^2. 
\nonumber
\end{align}
 Then
\begin{align}\label{eq-main-step-1}
& \sup_{(\boldsymbol{\Theta}_{i, j ,l }, \boldsymbol{\Delta}_{i, j, l})^{\top} \in \mathcal{Q}'} \big\| \nabla f^{t, k}_{i, j, l}(\boldsymbol{\Theta}_{i, j, l}, \boldsymbol{\Delta}_{i, j, l}) -  \nabla F^{t, k}_{i, j, l}(\boldsymbol{\Theta}_{i, j, l}, \boldsymbol{\Delta}_{i, j, l})\big\|_2^2   \nonumber\\
 = &    \sup_{(\boldsymbol{\Theta}_{i, j ,l }, \boldsymbol{\Delta}_{i, j, l})^{\top} \in \mathcal{Q}'}  \Big[ \Big\{ \tilde{f}^{t, k}_{i, j, l} \big( \boldsymbol{\Theta}_{i, j, l} \big)  
- \tilde{F}^{t, k}_{i, j, l} \big( \boldsymbol{\Theta}_{i, j, l} \big)  \Big\}^2 +  \Big\{\bar{f}^{t, k}_{i, j, l}  \big( \boldsymbol{\Delta}_{i, j, l} \big) - \bar{F}^{t, k}_{i, j, l} \big( \boldsymbol{\Delta}_{i, j, l} \big) \Big\}^2 \Big].
\end{align}
Now we focus on the term
\[
\sup_{(\boldsymbol{\Theta}_{i, j ,l }, \boldsymbol{\Delta}_{i, j, l})^{\top} \in \mathcal{Q}'}  \Big[ \Big\{ \tilde{f}^{t, k}_{i, j, l} \big( \boldsymbol{\Theta}_{i, j, l} \big) - 
 \tilde{F}^{t, k}_{i, j, l} \big( \boldsymbol{\Theta}_{i, j, l} \big)  \Big\}^2 +  \Big\{\bar{f}^{t, k}_{i, j, l}  \big( \boldsymbol{\Delta}_{i, j, l} \big) - \bar{F}^{t, k}_{i, j, l} \big( \boldsymbol{\Delta}_{i, j, l} \big) \Big\}^2 \Big].
\]
Let  $\mathcal{N}$ be an $\varepsilon$-net of $\mathcal{Q}'$, then we have that
\begin{equation}\label{eq-varepsilon-net}
|\mathcal{N}| \le (1 + 2\varepsilon^{-1})^{2},
\end{equation}
and for every $ \big(\boldsymbol{\Theta}_{i, j, l}, \boldsymbol{\Delta}_{i, j, l} \big)^{\top} \in \mathcal{Q}'$, there exists $\big(\pi(\boldsymbol{\Theta}_{i, j, l}), \pi({\boldsymbol{\Delta}}_{i, j, l}) \big)^{\top}  \in \mathcal{N}$ such that
\[
\big\| \big(\boldsymbol{\Theta}_{i, j, l}, \boldsymbol{\Delta}_{i, j, l} \big)^{\top}  -  \big(\pi(\boldsymbol{\Theta})_{i, j, l},, \pi(\boldsymbol{\Delta}_{i, j, l}) \big)^{\top} \big\|_{2} \leq \varepsilon.
\]
Since
\begin{align*}
&   \Big\{ \tilde{f}^{t, k}_{i, j, l} \big( \boldsymbol{\Theta}_{i, j, l} \big) - 
 \tilde{F}^{t, k}_{i, j, l} \big( \boldsymbol{\Theta}_{i, j, l} \big)  \Big\}^2 +  \Big\{\bar{f}^{t, k}_{i, j, l}  \big( \boldsymbol{\Delta}_{i, j, l} \big) - \bar{F}^{t, k}_{i, j, l} \big( \boldsymbol{\Delta}_{i, j, l} \big) \Big\}^2   \nonumber\\
  \leq & \Big\{ \tilde{f}^{t, k}_{i, j, l} \big( \pi(\boldsymbol{\Theta}_{i, j, l}) \big) - 
 \tilde{F}^{t, k}_{i, j, l} \big( \pi( \boldsymbol{\Theta}_{i, j, l}  ) \big)  \Big\}^2 +  \Big\{\bar{f}^{t, k}_{i, j, l}  \big( \pi(\boldsymbol{\Delta}_{i, j, l}) \big) - \bar{F}^{t, k}_{i, j, l} \big( \pi(\boldsymbol{\Delta}_{i, j, l} )\big)\Big\}^2 \nonumber \\
& \hspace{0.5cm} +
\Big\{ \tilde{f}^{t, k}_{i, j, l} \big( \pi(\boldsymbol{\Theta}_{i, j, l} )\big) - 
 \tilde{f}^{t, k}_{i, j, l} \big( \boldsymbol{\Theta}_{i, j, l} \big)  \Big\}^2 +  \Big\{\bar{f}^{t, k}_{i, j, l}  \big( \pi(\boldsymbol{\Delta}_{i, j, l}) \big) - \bar{f}^{t, k}_{i, j, l} \big( \boldsymbol{\Delta}_{i, j, l} \big)\Big\}^2 \nonumber \\
& \hspace{0.5cm} + \Big\{ \tilde{F}^{t, k}_{i, j, l} \big( \pi( \boldsymbol{\Theta}_{i, j, l} ) \big) - 
 \tilde{F}^{t, k}_{i, j, l} \big( \boldsymbol{\Theta}_{i, j, l} \big)  \Big\}^2 +  \Big\{\bar{F}^{t, k}_{i, j, l}  \big( \pi( \boldsymbol{\Delta}_{i, j, l})  \big) - \bar{F}^{t, k}_{i, j, l} \big(\boldsymbol{\Delta}_{i, j, l} \big)\Big\}^2,
\end{align*}
it follows that
\begin{align}\label{eq-step1-main_2}
&\sup_{(\boldsymbol{\Theta}_{i, j, l}, \boldsymbol{\Delta}_{i, j, l})^{\top} \in \mathcal{Q}'}
\Big\{ \tilde{f}^{t, k}_{i, j, l} \big( \boldsymbol{\Theta}_{i, j, l} \big) - 
 \tilde{F}^{t, k}_{i, j, l} \big( \boldsymbol{\Theta}_{i, j, l} \big)  \Big\}^2 +  \Big\{\bar{f}^{t, k}_{i, j, l}  \big( \boldsymbol{\Delta}_{i, j, l} \big) - \bar{F}^{t, k}_{i, j, l} \big( \boldsymbol{\Delta}_{i, j, l} \big) \Big\}^2   \nonumber\\
 \leq &
\sup_{(\boldsymbol{\Theta}_{i, j, l}, \boldsymbol{\Delta}_{i, j, l})^{\top} \in \mathcal{N}} \Big[
\Big\{ \tilde{f}^{t, k}_{i, j, l} \big( \boldsymbol{\Theta}_{i, j, l} \big) - 
 \tilde{F}^{t, k}_{i, j, l} \big( \boldsymbol{\Theta}_{i, j, l} \big)  \Big\}^2 +  \Big\{\bar{f}^{t, k}_{i, j, l}  \big( \boldsymbol{\Delta}_{i, j, l} \big) - \bar{F}^{t, k}_{i, j, l} \big( \boldsymbol{\Delta}_{i, j, l} \big) \Big\}^2 \Big] \nonumber \\
&  \hspace{0.5cm} +
\sup_{(\boldsymbol{\Theta}_{i, j, l}, \boldsymbol{\Delta}_{i, j, l})^{\top} \in \mathcal{Q}'}\Big[
\Big\{ \tilde{f}^{t, k}_{i, j, l} \big( \pi(\boldsymbol{\Theta}_{i, j, l} )\big) - 
 \tilde{f}^{t, k}_{i, j, l} \big( \boldsymbol{\Theta}_{i, j, l} \big)  \Big\}^2 +  \Big\{\bar{f}^{t, k}_{i, j, l}  \big( \pi(\boldsymbol{\Delta}_{i, j, l}) \big) - \bar{f}^{t, k}_{i, j, l} \big( \boldsymbol{\Delta}_{i, j, l} \big)\Big\}^2 \Big] \nonumber\\
&  \hspace{0.5cm} 
+ \sup_{(\boldsymbol{\Theta}_{i, j, l}, \boldsymbol{\Delta}_{i, j, l})^{\top} \in \mathcal{Q}'} \Big[ \Big\{ \tilde{F}^{t, k}_{i, j, l} \big( \pi( \boldsymbol{\Theta}_{i, j, l} ) \big) - 
 \tilde{F}^{t, k}_{i, j, l} \big( \boldsymbol{\Theta}_{i, j, l} \big)  \Big\}^2 +  \Big\{\bar{F}^{t, k}_{i, j, l}  \big( \pi( \boldsymbol{\Delta}_{i, j, l})  \big) - \bar{F}^{t, k}_{i, j, l} \big(\boldsymbol{\Delta}_{i, j, l} \big)\Big\}^2 \Big] \nonumber\\
 = & (I.1) + (I.2) +(I.3).
\end{align}
For term $(I.1)$, by the fact that $(a+b)^2 \leq 2a^2 + 2b^2$, we obtain that 
\begin{align}
(I.1) \leq &\sup_{(\boldsymbol{\Theta}_{i, j, l}, \boldsymbol{\Delta}_{i, j, l})^{\top} \in \mathcal{Q}'}  \Big\{ 2  \boldsymbol{\Theta}_{i, j, l}^{-2}  \big(\mathbf{Y}_{i, j, l}^{t, k}\big)^2 + 2 (1- \boldsymbol{\Theta}_{i, j, l})^{-2} \big(\widetilde{\mathbf{Y}}_{i, j, l}^{t, k}\big)^2 + 2 \boldsymbol{\Delta}_{i, j, l}^{-2} \big(\mathbf{X}_{i, j, l}^{t, k}\big)^2  \nonumber\\
& \hspace{0.5cm} + 2 (1- \boldsymbol{\Delta}_{i, j, l})^{-2}\big(\widetilde{\mathbf{X}}_{i, j, l}^{t, k}  \big)^{2} \Big\}  \nonumber\\
\leq & \sup_{(\boldsymbol{\Theta}_{i, j, l}, \boldsymbol{\Delta}_{i, j, l})^{\top} \in \mathcal{Q}'} \Big[ 2 c_{\min}^{-2} \big\{ \big(\mathbf{Y}_{i, j, l}^{t, k}\big)^2 + \big(\widetilde{\mathbf{Y}}_{i, j, l}^{t, k}\big)^2 +  \big(\mathbf{X}_{i, j, l}^{t, k}\big)^2 + \big(\widetilde{\mathbf{X}}_{i, j, l}^{t, k}  \big)^{2}  \big\} \Big].  \nonumber
\end{align}
By \eqref{def-event-A-prob}, \eqref{eq-varepsilon-net} and a union bound argument, it holds that
\begin{align}\label{eq-step1-sum-1}
 &  \P \bigg\{ 
  (I.1) \geq 8C_0^2c_{\min}^{-2} \max\bigg(  \frac{ \log(n \vee L \vee T)}{k}, 
   \frac{\log^2(n \vee L \vee T) \log^2 (k) \{\log\log (k)\}^2}{k^2}\bigg) \bigg\} \nonumber\\
   \leq &  4 (1+2\epsilon^{-1})^{2}  (n\vee L\vee T)^{-c_0/2}.
\end{align}

For term $(I.2)$, we have that 
\begin{align}\label{eq-step1-sum-2}
    (I.2)= &\sup_{(\boldsymbol{\Theta}_{i, j, l}, \boldsymbol{\Delta}_{i, j, l})^{\top} \in \mathcal{Q}'}  \bigg[
\big(\boldsymbol{\Theta}_{i, j, l} - \pi(\boldsymbol{\Theta}_{i, j, l}) \big) k^{-1}\sum_{u=t-(k-1)}^{t} 
(1 - \mathbf{A}^{u-1}_{i, j, l}) 
\cdot \bigg\{
\frac{1 - \mathbf{A}^{u}_{i, j, l}}{(1-\boldsymbol{\Theta}_{i, j, l})(1-\pi(\boldsymbol{\Theta}_{i, j, l}))}
\nonumber\\
& \hspace{0.5cm}
-\frac{\mathbf{A}^{u}_{i, j, l}}{\boldsymbol{\Theta}_{i, j, l} \pi(\boldsymbol{\Theta}_{i, j, l})}\bigg\} \bigg]^2  + \sup_{(\boldsymbol{\Theta}_{i, j, l}, \boldsymbol{\Delta}_{i, j, l})^{\top} \in \mathcal{Q}'} \bigg[ \big(\boldsymbol{\Delta}_{i, j, l} - \pi({\boldsymbol{\Delta}}_{i, j, l})\big)  k^{-1}\sum_{u=t-(k-1)}^{t}  \mathbf{A}^{u-1}_{i, j, l}
\nonumber\\
& \hspace{0.5cm}
\cdot \bigg\{
\frac{\mathbf{A}^{u}_{i, j, l}}{(1-\boldsymbol{\Delta}_{i, j, l})(1-\pi({\boldsymbol{\Delta}}_{i, j, l}))} + \frac{1 - \mathbf{A}^{u}_{i, j, l}}{\boldsymbol{\Delta}_{i, j, l}\pi({\boldsymbol{\Delta}}_{i, j, l})}  \bigg\}\bigg]^2 \nonumber\\
\leq &  c_{\min}^{-4}\sup_{(\boldsymbol{\Theta}_{i, j, l}, \boldsymbol{\Delta}_{i, j, l})^{\top} \in \mathcal{Q}'} \big\{ \big(\boldsymbol{\Theta}_{i, j, l} - \pi(\boldsymbol{\Theta}_{i, j, l}) \big)^2 + (\boldsymbol{\Delta}_{i, j, l} - \pi({\boldsymbol{\Delta}}_{i, j, l}))^2 \big\}
\leq  c_{\min}^{-4} \epsilon^2.
\end{align}

For term $(I.3)$, we have that 
\begin{align}\label{eq-step1-sum-3}
    (I.3) = & \sup_{(\boldsymbol{\Theta}_{i, j, l}, \boldsymbol{\Delta}_{i, j, l})^{\top} \in \mathcal{Q}'} \bigg[
\big(\boldsymbol{\Theta}_{i, j, l} - \pi(\boldsymbol{\Theta}_{i, j, l}) \big) k^{-1}\sum_{u=t-(k-1)}^{t} 
\cdot \bigg\{
\frac{ \E \big\{ (1 - \mathbf{A}^{u-1}_{i, j, l}) (1 - \mathbf{A}^{u}_{i, j, l} \big\}}{(1-\boldsymbol{\Theta}_{i, j, l})(1-\pi(\boldsymbol{\Theta}_{i, j, l}))}
\nonumber\\
& \hspace{0.5cm}
- \frac{\E \big\{ (1 - \mathbf{A}^{u-1}_{i, j, l}) \mathbf{A}^{u}_{i, j, l} \big\}}{\boldsymbol{\Theta}_{i, j, l} \pi(\boldsymbol{\Theta}_{i, j, l})}\bigg\} \bigg]^2  + \bigg[ \big(\boldsymbol{\Delta}_{i, j, l} - \pi({\boldsymbol{\Delta}}_{i, j, l})\big)  k^{-1}\sum_{u=t-(k-1)}^{t}  
\nonumber\\
& \hspace{0.5cm}
\cdot \bigg\{
\frac{\E \big\{ \mathbf{A}^{u-1}_{i, j, l} \mathbf{A}^{u}_{i, j, l}  \big\} }{(1-\boldsymbol{\Delta}_{i, j, l})(1-\pi({\boldsymbol{\Delta}}_{i, j, l}))} - \frac{ \E \big\{ \mathbf{A}^{u-1}_{i, j, l}(1 - \mathbf{A}^{u}_{i, j, l})}{\boldsymbol{\Delta}_{i, j, l}\pi({\boldsymbol{\Delta}}_{i, j, l})}  \bigg\}\bigg]^2 \nonumber\\
\leq & c_{\min}^{-4}\bigg[
\big(\boldsymbol{\Theta}_{i, j, l} - \pi(\boldsymbol{\Theta}_{i, j, l}) \big) k^{-1}\sum_{u=t-(k-1)}^{t} 
\cdot \Big\{
(1 - \boldsymbol{\Theta}^u_{i, j, l})  
(1-\boldsymbol{\Pi}_{i,j,l}^{u-1}) 
\nonumber\\
& \hspace{0.5cm}
+  \boldsymbol{\Theta}_{i, j, l}^u 
(1-\boldsymbol{\Pi}_{i,j,l}^{u-1})  \Big\}  \bigg]^2  + c^{-4} \bigg[ \big(\boldsymbol{\Delta}_{i, j, l} - \pi({\boldsymbol{\Delta}}_{i, j, l})\big)  k^{-1}\sum_{u=t-(k-1)}^{t}  
\nonumber\\
& \hspace{0.5cm}
\cdot \Big\{
(1 - \boldsymbol{\Delta}_{i, j, l}^u)  
\boldsymbol{\Pi}_{i,j,l}^{u-1}
+  \boldsymbol{\Delta}_{i, j, l}^u  
\boldsymbol{\Pi}_{i,j,l}^{u-1}  \Big\} \bigg]^2
\nonumber\\ 
\leq &  c_{\min}^{-4} \big\{ \big(\boldsymbol{\Theta}_{i, j, l} - \pi(\boldsymbol{\Theta}_{i, j, l}) \big)^2 + (\boldsymbol{\Delta}_{i, j, l} - \pi({\boldsymbol{\Delta}}_{i, j, l}))^2 \big\}
\leq  c_{\min}^{-4} \epsilon^2.
\end{align}
where the second inequality follows from \eqref{eq-pi-upper-lower}. 
Combining \eqref{eq-step1-main_2}, \eqref{eq-step1-sum-1}, \eqref{eq-step1-sum-2} and \eqref{eq-step1-sum-3}, we can derive that 
\begin{align*}
& \P \bigg\{ \sup_{(\boldsymbol{\Theta}_{i, j, l}, \boldsymbol{\Delta}_{i, j, l})^{\top} \in \mathcal{Q}'}
\Big\{ \tilde{f}^{t, k}_{i, j, l} \big( \boldsymbol{\Theta}_{i, j, l} \big) - 
 \tilde{F}^{t, k}_{i, j, l} \big( \boldsymbol{\Theta}_{i, j, l} \big)  \Big\}^2 +  \Big\{\bar{f}^{t, k}_{i, j, l}  \big( \boldsymbol{\Delta}_{i, j, l} \big) - \bar{F}^{t, k}_{i, j, l} \big( \boldsymbol{\Delta}_{i, j, l} \big) \Big\}^2  
 \nonumber\\
& \hspace{0.5cm} \geq 8C_0^2c_{\min}^{-2} \max\bigg(  \frac{ \log(n \vee L \vee T)}{k}, 
   \frac{\log^2(n \vee L \vee T) \log^2 (k) \{\log\log(k)\}^2}{k^2}\bigg)   + 2c_{\min}^{-4}\varepsilon^2   \bigg\}   \nonumber\\
\leq &  4 (1+ 2\epsilon^{-1})^{2}  (n\vee L \vee T)^{-c_0/2}.
\end{align*}
Choosing  $\varepsilon =  k^{-1/2}$,  we have that 
\begin{align*}
   &  \P \bigg\{ 
   \sup_{(\boldsymbol{\Theta}_{i, j, l}, \boldsymbol{\Delta}_{i, j, l})^{\top} \in \mathcal{Q}'}
\Big\{ \tilde{f}^{t, k}_{i, j, l} \big( \boldsymbol{\Theta}_{i, j, l} \big) - 
 \tilde{F}^{t, k}_{i, j, l} \big( \boldsymbol{\Theta}_{i, j, l} \big)  \Big\}^2 +  \Big\{\bar{f}^{t, k}_{i, j, l}  \big( \boldsymbol{\Delta}_{i, j, l} \big) - \bar{F}^{t, k}_{i, j, l} \big( \boldsymbol{\Delta}_{i, j, l} \big) \Big\}^2  
 \nonumber\\
& \hspace{0.5cm}   \geq   C_1 \max\bigg(  \frac{ \log(n \vee L \vee T)}{k}, 
   \frac{\log^2(n \vee L \vee T) \log^2 (k)\log\log ^2(k)}{k^2}\bigg)   \bigg\}     \leq  36 (n \vee L \vee T)^{-c_0/2},
\end{align*}
where $C_1 = 8C_0c^{-2} + 2c_{\min}^{-4}$.  
Define the event
\begin{align}\label{eq-step1-main-event}
\mathcal{B} = &
\bigg\{
   \sup_{(\boldsymbol{\Theta}_{i, j, l}, \boldsymbol{\Delta}_{i, j, l}) \in \mathcal{Q}'} \| \nabla f^{t,k}_{i, j, l} (\boldsymbol{\Theta}_{i, j, l}, \boldsymbol{\Delta}_{i, j, l}) -  \nabla F^{t,k}_{i, j, l}(\boldsymbol{\Theta}_{i, j, l}, \boldsymbol{\Delta}_{i, j, l})\|_2 
  \leq \sqrt{C_1} \max\bigg( \sqrt{ \frac{ \log(n \vee L \vee T)}{k}}, \nonumber\\
&\hspace{0.5cm} \frac{\log(n \vee L \vee T) \log (k)\log\log (k)}{k}\bigg), \quad
  \forall  1 \leq i\leq j \leq n, l \in [L],  t \in [T],  k \in [t]
\bigg\},
\end{align}
By  \eqref{eq-main-step-1} and a union bound argument, we have that 
\begin{align}\label{eq-step1-main-event-prop}
\mathbb{P}(\mathcal{B})
\geq 1 -36\sum_{t=1}^{T}\sum_{k=1}^{t} \sum_{j=1}^n\sum_{i=1}^{j}
\sum_{l=1}^L
(n \vee L\vee T)^{-c_0/2} = 1- (n \vee L\vee T)^{-c_1}, 
\end{align}
where $c_1 > 0$ is a constant. In the following steps, we assume the event $\mathcal{B}$ holds.

\medskip
\noindent
\textbf{Step 1.2.}  In this sub-step, we will prove \eqref{eq-step1-goal} under the the event $\mathcal{B}$ defined in \eqref{eq-step1-main-event}.

For any $ 1 \leq i\leq j \leq n $, $l \in [L]$, $t \in [T]$, $k \in [t]$ and $ (\boldsymbol{\Theta}_{i, j, l}, \boldsymbol{\Delta}_{i, j, l}) \in \mathcal{Q}'$, let the Hessian matrix of $F^{t, k}_{i, j, l}(\boldsymbol{\Theta}_{i, j, l}, \boldsymbol{\Delta}_{i, j, l}) $ be denoted by $H^{t, k}_{i, j, l} (\boldsymbol{\Theta}_{i, j, l}, \boldsymbol{\Delta}_{i, j, l}) \in \R^{2 \times 2 }$. 
It is straightforward to verify that $H^{t, k}_{i, j, l} (\boldsymbol{\Theta}_{i, j, l}, \boldsymbol{\Delta}_{i, j, l})$ is diagonal matrix, and thus its minimum eigenvalue is given by
\begin{align}
 & \lambda_{\min}(H^{t, k}_{i, j, l} (\boldsymbol{\Theta}_{i, j, l}, \boldsymbol{\Delta}_{i, j, l}))     \nonumber\\
=  & \min \bigg\{  
\frac{k^{-1}\sum_{u=t-(k-1)}^{t} \E\{\mathbf{A}^u_{i, j, l}  (1-\mathbf{A}^{u-1}_{i, j, l})\}}{\boldsymbol{\Theta}_{i, j, l}^2} 
+ \frac{ k^{-1}\sum_{u=t-(k-1)}^{t}\E\{(1-\mathbf{A}^u_{i, j, l}) (1-\mathbf{A}^{u-1}_{i, j, l}) \}}{\big(1-\boldsymbol{\Theta}_{i, j, l}\big)^2} , \nonumber\\
 & \hspace{0.5cm}    \frac{ k^{-1} \sum_{u=t-(k-1)}^{t} \E\{(1- \mathbf{A}^u_{i, j, l}) \mathbf{A}^{u-1}_{i, j, l}\}}{ \boldsymbol{\Delta}_{i, j, l}^2}  + \frac{ k^{-1}\sum_{u=t-(k-1)}^{t}\E\{\mathbf{A}^u_{i, j, l} \mathbf{A}^{u-1}_{i, j, l}\}}{(1-\boldsymbol{\Delta}_{i, j, l})^2}  \bigg\}  \nonumber\\
= &  \min \bigg\{ 
 \frac{k^{-1}\sum_{u=t-(k-1)}^{t} \boldsymbol{\Theta}^u_{i, j, l} (1 - \boldsymbol{\Pi}_{i,j,l}^{u-1}) }{\boldsymbol{\Theta}_{i, j, l}^2} + \frac{ k^{-1}\sum_{u=t-(k-1)}^{t} (1-\boldsymbol{\Theta}^u_{i, j, l}) (1-\boldsymbol{\Pi}_{i,j,l}^{u-1})}{(1-\boldsymbol{\Theta}_{i, j, l})^2},   \nonumber\\
 & \hspace{0.5cm}    \frac{ k^{-1} \sum_{u=t-(k-1)}^{t} \boldsymbol{\Delta}_{i, j, l}^{u} \boldsymbol{\Pi}_{i,j,l}^{u-1} }{\boldsymbol{\Delta}_{i, j, l}^2}  + \frac{ k^{-1}\sum_{u=t-(k-1)}^{t}(1- \boldsymbol{\Delta}_{i, j, l}^{u} ) \boldsymbol{\Pi}_{i,j,l}^{u-1} }{(1-\boldsymbol{\Delta}_{i, j, l})^2}  \bigg\} \nonumber\\
\geq  & 2 c_{\min}(1-c_{\min})^{-2} \min    \bigg\{ k^{-1}\sum_{u=t-(k-1)}^{t} (1-\boldsymbol{\Pi}_{i,j,l}^{u-1} ),  k^{-1}\sum_{u=t-(k-1)}^{t}  \boldsymbol{\Pi}_{i,j,l}^{u-1}  \bigg\}\nonumber \\
\geq & 2c_{\min}^2(1-c_{\min})^{-2},
\end{align}
where the first inequality follows from  for any $t \in [T]$,  $(\boldsymbol{\Theta}^t_{i, j, l}, \boldsymbol{\Delta}^t_{i, j, l}) \in \mathcal{Q}'$ and the last inequality follows from \eqref{eq-pi-upper-lower}. Similarly, the largest eigenvalue of $H^{t, k} _{i, j, l}(\boldsymbol{\Theta}_{i, j, l}, \boldsymbol{\Delta}_{i, j, l})$ is
\begin{align}
\lambda_{\max}(H^{t, k}_{i, j, l} (\boldsymbol{\Theta}_{i, j, l}, \boldsymbol{\Delta}_{i, j, l}))   \leq   2 c_{\min}^{-2}(1-c_{\min})^2. \nonumber
\end{align}
Thus,  $ F^{t, k}_{i, j, l}$ is $2 c_{\min}^2(1-c_{\min})^{-2}$-strongly convex and $2c_{\min}^{-2}(1-c_{\min})^2$-smooth  on $ \mathcal{Q}'$.
By Part 3 of Lemma 5.1 in \cite{huang2023stability}  and the event $\mathcal{B}$ defined in \eqref{eq-step1-main-event}, it holds that
under the event $\mathcal{B}$, for any  $ 1 \leq i \leq  j \leq n$, $l \in [L]$, $t \in [T]$ and $k \in [t]$, 
\[ 
f^{t,k}_{i, j, l}  \mbox{ and } F^{t,k}_{i, j, l} \mbox{ are } \big( \log(2), \psi(k) \big)\mbox{-close},
\]
where $ \psi(k) = C_{\psi}\log(n \vee L \vee T) \log(T)\log\log(T) /k $, with a constant $C_{\psi} > 0$.

\medskip
\noindent\textbf{Step 2.} In this step, we will show that for any $ 1 \leq i \leq  j \leq n$, $l \in [L]$ and $t, t' \in [T]$,  $F^{t}_{i, j, l}$ and $F^{t'}_{i, j, l}$  are 
 \[
\Big( \log\{4(1-c_{\min})^4/c_{\min}^4\},  c_{\min}^2(1-c_{\min})^{-2} \big( \vert \boldsymbol{\Theta}^t_{i, j, l}  - \boldsymbol{\Theta}^{t'}_{i, j, l} \vert^2 + \vert \boldsymbol{\Delta}^t_{i, j, l} - \boldsymbol{\Delta}^{t'}_{i, j, l} \vert^2 \big)\Big)\text{-close}.
\]

By \textbf{Step 1.2}, for any $t \in [T]$, $F^{t, k}_{i, j, l}$ is $2 c_{\min}^2(1-c_{\min})^{-2}$-strongly convex and $2(1-c_{\min})^2c_{\min}^{-2}$-smooth. Since for any $t \in [T]$, $F^t_{i, j, l} = F^{t, 1}_{i, j, l}$,  it follows that  $ F^t_{i, j, l}$ is $2 c_{\min}^2(1-c_{\min})^{-2}$-strongly convex and $2c_{\min}^{-2}(1-c_{\min})^2$-smooth  on $ \mathcal{Q}'$.
By Part 4 in Lemma 5.1 in \cite{huang2023stability}, for any $t, t' \in [T]$,  $F_{t}$ and $F_{t'}$  are \[
\Big( \log\{4(1-c_{\min})^4/c_{\min}^4\},  c_{\min}^2(1-c_{\min})^{-2} \big( \vert \boldsymbol{\Theta}^t_{i, j, l}  - \boldsymbol{\Theta}^{t'}_{i, j, l} \vert^2 + \vert \boldsymbol{\Delta}^t_{i, j, l} - \boldsymbol{\Delta}^{t'}_{i, j, l} \vert^2 \big)\Big)\text{-close},
\]
in $\mathcal{Q}'$.

\medskip
\noindent
\textbf{Step 3.} In this step, we show that under the event $\mathcal{B}$, for any $ 1 \leq i \leq  j \leq n$, $l \in [L]$ and $t \in [T_{g}]/ [T_{g-1}]$ with $g \in [G]$, 
\begin{align}
   \big\vert \widehat{\boldsymbol{\Theta}}^{t,\hat{k}_t}_{i, j, l}- \boldsymbol{\Theta}^{t}_{i, j, l} \big) \big\vert^2 + \big\vert \widehat{\boldsymbol{\Delta}}^{t,\hat{k}_t}_{i, j, l}  - \boldsymbol{\Delta}^{t}_{i, j, l}\big \vert^2  
  \leq  C_2\max\bigg\{ \frac{ \log(n \vee L \vee T) \log(T) \log\log(T)}{ t - T_{g-1}},  \big\{V( t - T_{g-1}) \big\}^2  \bigg\}, \nonumber
\end{align}
for some constant $C_2 >0$.

By \textbf{Step 1}, we have that for any $k \in [t]$, under the event $\mathcal{B}$, 
\begin{equation}\label{eq-f-F-close}
f^{t,k}_{i, j, l}  \mbox{ and } F^{t,k}_{i, j, l} \mbox{ are } \big( \log(2), \psi(k) \big)\mbox{-close},
\end{equation}
where $ \psi(k) = C_{\psi} \log(n \vee L \vee T) \log(T)\log\log(T) /k $ with a constant $C_{\psi} > 0$. Let  $\bar{k}_t = \max\{k \in \mathcal{G}^{(t)}\colon k \leq  t - {T_{g -1}}\}$. 
By \textbf{Step 2} and \Cref{def-seg}, we have that 
\[
 F^{t - (\bar{k}_t-1)}_{i, j, l}, \ldots, F^{t}_{i, j, l} \mbox{ are } \Big( \log\{4c_{\min}^{-4}(1-c_{\min})^4\},  \phi(T_{g} - T_{g-1}) \Big)\mbox{-close to } F_{t}
\]
where $ \phi(T_{g} - T_{g-1}) =C_{\phi}\big( V(T_{g} - T_{g-1})\big)^2$ with $C_{\phi} = 6 c_{\min}^2(1-c_{\min})^{-2}  $. Then by Part 7 of Lemma 5.2 in \cite{huang2023stability}, for any $u \in [\bar{k}_t]$ 
\[
F^{t,u}_{i, j, l}  \mbox{ and }  F^{t}_{i, j, l} \mbox{ are } \Big( \log\{4c_{\min}^{-4}(1-c_{\min})^4\},  C_{\phi'}\phi(T_{g} - T_{g-1}) \Big)\mbox{-close},
\]
where $C_{\phi'} = 4c_{\min}^{-4}(1-c_{\min})^4+1$. 

Let $\hat{k}_t$ be defined in \Cref{alg:fast-msbm} with 
\[
\tau(u) =  C_{\tau}\max \Big\{u^{-1} \log(n \vee L \vee T) \log(T) \log\log(T), \big(V(u)\big)^2 \Big\},  \quad \forall u \in [t],
\]
where $C_{\tau}>  0$ is a sufficiently large constant.

\medskip
\noindent\textbf{Step 3.1} In this sub-step, we will first prove that $ \hat{k}_t \geq \bar{k}_t $. To this end, it suffices to show that 
\[
\max_{1 \leq i\leq j \leq n, l \in [L]} \big\{ f^{t,u}_{i, j, l}\big(\widehat{\boldsymbol{\Theta}}^{t,\bar{k}_t}_{i, j, l}, \widehat{\boldsymbol{\Delta}}^{t,\bar{k}_t}_{i, j, l} \big) - \inf_{(\boldsymbol{\Theta}_{i, j, l}, \boldsymbol{\Delta}_{i, j, l}) \in \mathcal{Q}'} f^{t,u}_{i, j, l}(\boldsymbol{\Theta}_{i, j, l}, \boldsymbol{\Delta}_{i, j, l}) \big\} \leq \tau(u),   \forall u \in \{k \in \mathcal{G}^{(t)}: k \leq \bar{k}_t\}. 
\]
For any $1 \leq i \leq j \leq n$ and $l 
\in [L]$, we have that
\[
f^{t,\bar{k}_t}_{i, j, l}\big(\widehat{\boldsymbol{\Theta}}^{t,\bar{k}_t}_{i, j, l}, \widehat{\boldsymbol{\Delta}}^{t,\bar{k}_t}_{i, j, l} \big)- \inf_{(\boldsymbol{\Theta}_{i, j, l}, \boldsymbol{\Delta}_{i, j, l}) \in \mathcal{Q}'} f^{t,\bar{k}_t}_{i, j, l}(\boldsymbol{\Theta}_{i, j, l}, \boldsymbol{\Delta}_{i, j, l}) = 0.
\]
Since $f^{t,\bar{k}_t}_{i, j, l}$ and $F^{t,\bar{k}_t}_{i, j, l}$ are $(\log(2), \psi(\bar{k}_t))$-close, it holds that 
\begin{align*}
& F^{t,\bar{k}_t}_{i, j, l}\big(\widehat{\boldsymbol{\Theta}}^{t,\bar{k}_t}_{i, j, l}, \widehat{\boldsymbol{\Delta}}^{t,\bar{k}_t}_{i, j, l} \big) - \inf_{(\boldsymbol{\Theta}_{i, j, l}, \boldsymbol{\Delta}_{i, j, l}) \in \mathcal{Q}'} F^{t,\bar{k}_t}_{i, j, l}(\boldsymbol{\Theta}_{i, j, l}, \boldsymbol{\Delta}_{i, j, l})  \nonumber\\
\leq &  2 \big\{ f^{t,\bar{k}_t}_{i, j, l}\big(\widehat{\boldsymbol{\Theta}}^{t,\bar{k}_t}_{i, j, l}, \widehat{\boldsymbol{\Delta}}^{t,\bar{k}_t}_{i, j, l} \big) - \inf_{(\boldsymbol{\Theta}_{i, j, l}, \boldsymbol{\Delta}_{i, j, l}) \in \mathcal{Q}'} 
 f^{t,\bar{k}_t}_{i, j, l} (\boldsymbol{\Theta}_{i, j, l}, \boldsymbol{\Delta}_{i, j, l}) + \psi(\bar{k}_t)  \big\} \\
 = & 2 \psi(\bar{k}_t).
\end{align*}
Since  $F^{t,\bar{k}_t}_{i, j, l}$ and  $F^{t}_{i, j, l}$ are $ \big(\log\{4(1-c_{\min})^4c_{\min}^{-4}\}, C_{\phi'}\phi(T_{g} - T_{g-1} \big) $-close, it holds that 
\begin{align}
   &  F^{t}_{i, j, l}\big(\widehat{\boldsymbol{\Theta}}^{t,\bar{k}_t}_{i, j, l}, \widehat{\boldsymbol{\Delta}}^{t,\bar{k}_t}_{i, j, l} \big) - \inf_{(\boldsymbol{\Theta}_{i, j, l}, \boldsymbol{\Delta}_{i, j, l}) \in \mathcal{Q}'} F^{t}_{i, j, l}(\boldsymbol{\Theta}_{i, j, l}, \boldsymbol{\Delta}_{i, j, l}) 
   \nonumber\\
   \leq &4  c_{\min}^{-4}(1-c_{\min})^4 \big\{ F^{t,\bar{k}_t}_{i, j, l}\big(\widehat{\boldsymbol{\Theta}}_{i, j, l}^{t,\bar{k}_t}, \widehat{\boldsymbol{\Delta}}_{i, j, l}^{t,\bar{k}_t} \big)  - \inf_{(\boldsymbol{\Theta}_{i, j, l}, \boldsymbol{\Delta}_{i, j, l}) \in \mathcal{Q}'}  F^{t,\bar{k}_t}_{i, j, l}(\boldsymbol{\Theta}_{i, j, l}, \boldsymbol{\Delta}_{i, j, l}) +C_{\phi'}\phi(T_{g} - T_{g-1})  \big\} \nonumber\\
    \leq & 4   c_{\min}^{-4}(1-c_{\min})^4  \big\{2 \psi(\bar{k}_t) + C_{\phi'} \phi(T_{g} - T_{g-1} ) \big\}. \nonumber
\end{align}
For all $u \in \{k \in \mathcal{G}^{(t)}: k \leq \bar{k}_t\}$, since $F^{t,u}_{i, j, l}$ and  $F^{t}_{i, j, l}$ are  $(\log\{4c_{\min}^{-4}(1-c_{\min})^4\}, C_{\phi'}\phi(T_{g} - T_{g-1} )$-close, it holds that 
\begin{align}
& F^{t,u}_{i, j, l}\big(\widehat{\boldsymbol{\Theta}}^{t,\bar{k}_t}_{i, j, l}, \widehat{\boldsymbol{\Delta}}^{t,\bar{k}_t}_{i, j, l} \big)  - \inf_{(\boldsymbol{\Theta}_{i, j, l}, \boldsymbol{\Delta}_{i, j, l}) \in \mathcal{Q}'}  F^{t,u}_{i, j, l}(\boldsymbol{\Theta}_{i, j, l}, \boldsymbol{\Delta}_{i, j, l}) \nonumber\\
 \leq & 4 c_{\min}^{-4}(1-c_{\min})^4 \big\{  F^{t}_{i, j, l}\big(\widehat{\boldsymbol{\Theta}}^{t,\bar{k}_t}_{i, j, l}, \widehat{\boldsymbol{\Delta}}^{t,\bar{k}_t} \big)  - \inf_{(\boldsymbol{\Theta}_{i, j, l}, \boldsymbol{\Delta}_{i, j, l}) \in \mathcal{Q}}  F^{t}_{i, j, l}(\boldsymbol{\Theta}_{i, j, l}, \boldsymbol{\Delta}_{i, j, l})  +  C_{\phi'}\phi(T_{g} - T_{g-1} ) \big\} \nonumber\\
 \leq  & 4 c_{\min}^{-4}(1-c_{\min})^4 \Big[  4c_{\min}^{-4}(1-c_{\min})^4  \big\{2 \psi(\bar{k}_t) + C_{\phi'} \phi(T_{g} - T_{g-1} ) \big\} + C_{\phi'} \phi(T_{g} - T_{g-1} )  \Big] \nonumber\\
 \leq & C_3 \psi(\bar{k}_t)   + C_4 \phi(T_{g} - T_{g-1} ) \nonumber
\end{align}
where $C_3 = 32 c_{\min}^{-8}(1-c_{\min})^8$  and $C_4 = 4c_{\min}^{-4}(1-c_{\min})^4 \{ 4 c_{\min}^{-4}(1-c_{\min})^4 +1 \}  C_{\phi'}$. 
Since $f^{t, u}_{i, j, l}$ and $F^{t, u}_{i, j, l}$ are $(\log(2), \psi(u))$-close, it holds that 
\begin{align}
&f^{t,u}_{i, j, l}\big(\widehat{\boldsymbol{\Theta}}^{t,\bar{k}_t}_{i, j, l}, \widehat{\boldsymbol{\Delta}}^{t,\bar{k}_t}_{i, j, l} \big)- \inf_{(\boldsymbol{\Theta}_{i, j, l}, \boldsymbol{\Delta}_{i, j, l}) \in \mathcal{Q}'}  f^{t,u}_{i, j, l}(\boldsymbol{\Theta}_{i, j, l}, \boldsymbol{\Delta}_{i, j, l})  \nonumber\\
\leq & 2 \big\{ F^{t,u}_{i, j, l}(\widehat{\boldsymbol{\Theta}}^{t,\bar{k}}_{i, j, l}, \widehat{\boldsymbol{\Delta}}^{t,\bar{k}}_{i, j, l} \big) - \inf_{(\boldsymbol{\Theta}_{i, j, l}, \boldsymbol{\Delta}_{i, j, l}) \in \mathcal{Q}'}  F^{t,u}_{i, j, l}(\boldsymbol{\Theta}_{i, j, l}, \boldsymbol{\Delta}_{i, j, l}) 
+ \psi(u) \big\} \nonumber\\
\leq & 2 C_3 \psi(\bar{k}_t)   + 2 C_4 \phi(T_{g} - T_{g-1} ) +  \psi(u) \leq C_5 \psi(u) + 2 C_4 \phi(T_{g} - T_{g-1} ),\nonumber
\end{align}
where $C_5 = 2C_3 +1$. 
Since $C_{\tau}$ is a sufficiently large constant and $V$ is a non-increasing function,  the last inequality above implies for all $1\leq i \leq j \leq n$, $l \in [L]$ and $u \in \{k \in \mathcal{G}^{(t)}: k \leq \bar{k}_t\}$,
\[
f^{t,u}_{i, j, l}\big(\widehat{\boldsymbol{\Theta}}^{t,\bar{k}_t}_{i, j, l}, \widehat{\boldsymbol{\Delta}}^{t,\bar{k}_t}_{i, j, l} \big)  - \inf_{(\boldsymbol{\Theta}_{i, j, l}, \boldsymbol{\Delta}_{i, j, l})  \in \mathcal{Q}'} f^{t,u}_{i, j, l}(\boldsymbol{\Theta}_{i, j, l}, \boldsymbol{\Delta}_{i, j, l})  \leq \tau(u),
\]
which leads to 
\[
\max_{ 1 \leq i \leq j \leq n, l \in [L]} \big\{f^{t,u}_{i, j, l}\big(\widehat{\boldsymbol{\Theta}}^{t,\bar{k}_t}_{i, j, l}, \widehat{\boldsymbol{\Delta}}^{t,\bar{k}_t}_{i, j, l} \big)  - \inf_{(\boldsymbol{\Theta}_{i, j, l}, \boldsymbol{\Delta}_{i, j, l})  \in \mathcal{Q}'} f^{t,u}_{i, j, l}(\boldsymbol{\Theta}_{i, j, l}, \boldsymbol{\Delta})_{i, j, l} \big\} \leq \tau(u).
\]
This shows that $ \hat{k}_t \geq \bar{k}_t$.

\medskip
\noindent\textbf{Step 3.2.} Since $ \hat{k}_t \geq \bar{k}_t $,  by the definition of $ \hat{k}_t $,
\begin{align}
&\max_{ 1 \leq i \leq j \leq n, l \in [L]} f^{t,\bar{k}_t}_{i, j, l}\big(\widehat{\boldsymbol{\Theta}}^{t,\hat{k}_t }_{i, j, l}, \widehat{\boldsymbol{\Delta}}^{t,\hat{k}_t }_{i, j, l} \big) - \inf_{(\boldsymbol{\Theta}_{i, j, l}, \boldsymbol{\Delta}_{i, j, l})\in \mathcal{Q}'} f^{t,\bar{k}_t}_{i, j, l}(\boldsymbol{\Theta}_{i, j, l}, \boldsymbol{\Delta}_{i, j, l}) 
\leq   \tau(\bar{k}_t). \nonumber
\end{align}
Since for any $ 1 \leq i\leq j \leq n$ and $l \in [L]$, $f^{t, \bar{k}_t}_{i, j, l}$ and $F^{t,\bar{k}_t}_{i, j, l}$ are $(\log(2), \psi(\bar{k}_t))$-close, it holds that 
\begin{align}
&F^{t,\bar{k}_t}_{i, j, l}\big(\widehat{\boldsymbol{\Theta}}^{t,\hat{k}_t}_{i, j, l}, \widehat{\boldsymbol{\Delta}}^{t,\hat{k}_t}_{i, j, l} \big) - \inf_{(\boldsymbol{\Theta}_{i, j, l}, \boldsymbol{\Delta}_{i, j, l})\in \mathcal{Q}'} F^{t,\bar{k}_t}_{i, j, l}(\boldsymbol{\Theta}_{i, j, l}, \boldsymbol{\Delta}_{i, j, l})  \nonumber\\
\leq & 2 \big[ f^{t,\bar{k}_t}_{i, j, l}\big(\widehat{\boldsymbol{\Theta}}^{t,\hat{k}_t}_{i, j, l}, \widehat{\boldsymbol{\Delta}}^{t,\hat{k}_t}_{i, j, l} \big)- \inf_{(\boldsymbol{\Theta}_{i, j, l}, \boldsymbol{\Delta}_{i, j, l})  \in \mathcal{Q}'} f^{t,\bar{k}_t}_{i, j, l}(\boldsymbol{\Theta}_{i, j, l}, \boldsymbol{\Delta}_{i, j, l})  + \psi(\bar{k}_t) \big] \nonumber\\
\leq & 
2\big[ \tau(\bar{k}_t)  + \psi(\bar{k}_t)\big] \leq 4 \tau(\bar{k}_t).  \nonumber
\end{align}
Since  $F^{t,\bar{k}_t}_{i, j, l}$ and  $F^{t}_{i, j, l}$ are $( \log\{4(1-c_{\min})^4/c_{\min}^4\}, C_{\phi'}\phi(T_{g} - T_{g-1} ) $-close, it holds that 
\begin{align}
&  F^{t}_{i, j, l}\big(\widehat{\boldsymbol{\Theta}}^{t,\hat{k}_t}_{i, j, l}, \widehat{\boldsymbol{\Delta}}^{t,\hat{k}_t}_{i, j, l} \big) - \inf_{(\boldsymbol{\Theta}_{i, j, l}, \boldsymbol{\Delta}_{i, j, l})  \in \mathcal{Q}'} F^{t}_{i, j, l}(\boldsymbol{\Theta}_{i, j, l}, \boldsymbol{\Delta}_{i, j, l}) 
\nonumber\\
 \leq & 4 c_{\min}^{-4}(1-c_{\min})^4  \big\{F^{t,\bar{k}_t}_{i, j, l}\big(\widehat{\boldsymbol{\Theta}}^{t,\hat{k}_t}_{i, j, l}, \widehat{\boldsymbol{\Delta}}^{t,\hat{k}_t}_{i, j, l} \big) - \inf_{(\boldsymbol{\Theta}_{i, j, l}, \boldsymbol{\Delta}_{i, j, l})  \in \mathcal{Q}'} F^{t,\bar{k}_t}_{i, j, l}(\boldsymbol{\Theta}_{i, j, l}, \boldsymbol{\Delta}_{i, j, l})   
 \nonumber\\
 & \hspace{0.5cm} + C_{\phi'}\phi(T_{g} - T_{g-1} ) \big\} 
 \nonumber\\
 \leq & 4 c_{\min}^{-4}(1-c_{\min})^4 \{ 4 \tau(\bar{k}_t) + C_{\phi'}\phi(T_{g} - T_{g-1} )\}  \leq 20 c_{\min}^{-4}(1-c_{\min})^4  \tau(\bar{k}_t). \nonumber
\end{align}
Since $\bar{k}_t \geq \max\{t - 2^{\lceil \log_2 (T_{g-1}) \rceil}, 1\}$, we have that  
\begin{align}
& F^{t}_{i, j, l}\big(\widehat{\boldsymbol{\Theta}}^{t,\hat{k}_t}_{i, j, l}, \widehat{\boldsymbol{\Delta}}^{t,\hat{k}_t}_{i, j, l} \big) - \inf_{(\boldsymbol{\Theta}_{i, j, l}, \boldsymbol{\Delta}_{i, j, l}) \in \mathcal{Q}'} F^{t}_{i, j, l}(\boldsymbol{\Theta}_{i, j, l}, \boldsymbol{\Delta}_{i, j, l}) \nonumber\\
\leq &  20 C_{\tau} \frac{(1-c_{\min})^4}{c_{\min}^{4}}  \max\bigg\{ \frac{\log(n \vee L \vee T) \log(T) \log\log(T)}{ \max\{t - 2^{\lceil \log_2 (T_{g-1}) \rceil}, 1\}}, \big[V \big( \max\{t - 2^{\lceil \log_2 (T_{g-1}) \rceil}, 1\} \big) \big]^2 \bigg\}, \nonumber \\
\leq &  20 C_{\tau}'  \frac{(1-c_{\min})^4}{c_{\min}^{4}}  \max\bigg\{ \frac{\log(n \vee L \vee T) \log(T) \log\log(T)}{ t - T_{g-1}}, \big\{V( t - T_{g-1}) \big\}^2 \bigg\}, \nonumber
\end{align} 
where $C_{\tau}' >0$ is an absolute constant and the last inequality follows from  \Cref{ass-bias}$(i)$.  Then, it holds that
\begin{align}
&  F^{t}_{i, j, l} \big(\widehat{\boldsymbol{\Theta}}^{t,\hat{k}_t}_{i, j, l}, \widehat{\boldsymbol{\Delta}}^{t,\hat{k}_t}_{i, j, l} \big)-   F^{t}_{i, j, l}\big( \boldsymbol{\Theta}^{t}_{i, j, l}, \boldsymbol{\Delta}^{t}_{i, j, l} \big)   \nonumber\\
\leq &  20 C_{\tau}'  \frac{(1-c_{\min})^4}{c_{\min}^{4}} \max \bigg\{ \frac{ \log(n \vee L \vee T)\log(T) \log\log(T)}{ t - T_{g-1}},  \big\{V( t - T_{g-1}) \big\}^2  \bigg\}. \nonumber
\end{align}
Note that since  $F^t_{i, j, l}$ is $2c_{\min}^2(1-c_{\min})^{-2}$-strongly convex, then 
\begin{align}
& F^t_{i, j, l} \big(\widehat{\boldsymbol{\Theta}}^{t,\hat{k}_t}_{i, j, l}, \widehat{\boldsymbol{\Delta}}^{t,\hat{k}_t}_{i, j, l} \big) - F^t_{i, j, l} \big( \boldsymbol{\Theta}^{t}_{i, j, l}, \boldsymbol{\Delta}^{t}_{i, j, l} \big)   \nonumber\\
\geq  & \big(\widehat{\boldsymbol{\Theta}}^{t,\hat{k}_t}_{i, j, l}- \boldsymbol{\Theta}^{t}_{i, j, l} \big) \nabla_{\boldsymbol{\Theta}_{i, j, l}} F^t_{i, j, l} \big( \boldsymbol{\Theta}^{t}_{i, j, l}, \boldsymbol{\Delta}^{t}_{i, j, l} \big)  +\big( \widehat{\boldsymbol{\Delta}}^{t,\hat{k}_t}_{i, j, l}  - \boldsymbol{\Delta}^{t}_{i, j, l}\big) \nabla_{\boldsymbol{\Delta}_{i, j, l}} F^t_{i, j, l} \big( \boldsymbol{\Theta}^{t}_{i, j, l}, \boldsymbol{\Delta}^{t}_{i, j, l} \big)     \nonumber\\
& \hspace{0.5cm} +  \frac{2c_{\min}^2}{(1-c_{\min})^2} \Big\{ \big\vert \widehat{\boldsymbol{\Theta}}^{t,\hat{k}_t}_{i, j, l}- \boldsymbol{\Theta}^{t}_{i, j, l} \big) \big\vert^2 + \big\vert \widehat{\boldsymbol{\Delta}}^{t,\hat{k}_t}_{i, j, l}  - \boldsymbol{\Delta}^{t}_{i, j, l}\big \vert^2  \Big\}
\nonumber \\
= & \frac{2c_{\min}^2}{(1-c_{\min})^2}\Big\{ \big\vert \widehat{\boldsymbol{\Theta}}^{t,\hat{k}_t}_{i, j, l}- \boldsymbol{\Theta}^{t}_{i, j, l} \big) \big\vert^2 + \big\vert \widehat{\boldsymbol{\Delta}}^{t,\hat{k}_t}_{i, j, l}  - \boldsymbol{\Delta}^{t}_{i, j, l}\big \vert^2  \Big\}. \nonumber
\end{align}
Thus, we finally have that under the event $\mathcal{B}$, 
\begin{align}
  &  \big\vert \widehat{\boldsymbol{\Theta}}^{t,\hat{k}_t}_{i, j, l}- \boldsymbol{\Theta}^{t}_{i, j, l} \big) \big\vert^2 + \big\vert \widehat{\boldsymbol{\Delta}}^{t,\hat{k}_t}_{i, j, l}  - \boldsymbol{\Delta}^{t}_{i, j, l}\big \vert^2  \nonumber\\
  \leq  & \frac{10C_{\tau}'(1-c_{\min})^6}{c_{\min}^6}  \max\bigg\{ \frac{ \log(n \vee L \vee T) \log(T) \log\log(T)}{ t - T_{g-1}},  \big\{V( t - T_{g-1}) \big\}^2  \bigg\}, \nonumber
\end{align}
which completes the proof.

\end{proof}

\subsubsection{Proof of Proposition \ref{prop-bias-final}}\label{sec-pro-bias}
\begin{proof}[Proof of \Cref{prop-bias-final}]

Without loss of generality, assume $t \in [T_g] \backslash [T_{g-1}]$. Define
\begin{equation}\label{def-bar-k-t}
\bar{k}_t = \max\{k \in \mathcal{G}^{(t)}\colon k \leq  t - {T_{g -1}}\}.
\end{equation}
By \textbf{Step 3.1} in the proof of \Cref{main-theorem-1}, for any $t \in [T]$, under the event $\mathcal{B}$ defined in \eqref{eq-step1-main-event},  we have $\hat{k}_t \geq  \bar{k}_t$.

We now claim that under the event $\mathcal{B}$, for any $ s,u \in [t] \backslash[t-\hat{k}_t]$, 
\begin{align}\label{eq-bias-calim}
 \vert \boldsymbol{\Theta}^{s}_{i, j, l} -  \boldsymbol{\Theta}^{u}_{i, j, l} \vert  \leq C' \max\bigg\{  \sqrt{\frac{ \log(n \vee L \vee T)\log(T) \log\log(T)}{ t - T_{g-1}}},  V( t - T_{g-1} )   \bigg\}
\end{align}
and
\begin{align}\label{eq-bias-calim-delta}
 \vert \boldsymbol{\Delta}^{s}_{i, j, l} -  \boldsymbol{\Delta}^{u}_{i, j, l} \vert  \leq C'  \max\bigg\{  \sqrt{\frac{ \log(n \vee L \vee T)\log(T) \log\log(T)}{ t - T_{g-1}}},  V( t - T_{g-1} )   \bigg\},
\end{align}
where $C' > 0$ is an absolute constant. 
From now on, we assume the event $\mathcal{B}$ holds. 

This proof consists of four steps. In \textbf{Steps 1-3}, we will prove the claim above. In \textbf{Step 4}, by applying \Cref{prop-bias} and the claim above, we will conclude the proof.

\medskip
\noindent
\textbf{Step 1.}  
 Let $\mathcal{Q}' = [c_{\min}, 1-c_{\min}]^{2}$. By \Cref{alg:fast-msbm}, we can derive that 
\[
\max_{ 1 \leq i \leq j \leq n, l \in [L]} \big\{ f^{t,\hat{k}_t } _{i, j, l}\big(\widehat{\boldsymbol{\Theta}}^{t,\hat{k}_t}_{i, j, l}, \widehat{\boldsymbol{\Delta}}^{t,\hat{k}_t}_{i, j, l} \big) - \inf_{(\boldsymbol{\Theta}_{i, j, l}, \boldsymbol{\Delta}_{i, j, l})\in \mathcal{Q}'} f^{t,\hat{k}_t}_{i, j, l}(\boldsymbol{\Theta}_{i, j, l}, \boldsymbol{\Delta}_{i, j, l})  \big\} \leq \tau(\hat{k}_t), 
\]

By \textbf{Step 1} in the proof of \Cref{main-theorem-1},  for any  $ 1 \leq i \leq j \leq n$, $l \in [L]$  and $k \in [t]$, we have that 
\begin{equation}\label{eq-f-F-close-rep}
f^{t,k}_{i, j, l}  \mbox{ and } F^{t,k}_{i, j, l} \mbox{ are } \big( \log(2), \psi(k) \big)\mbox{-close},
\end{equation}
where $ \psi(k) = C_{\psi} \log(n \vee L \vee T) \log(T)\log\log(T) /k $ with a constant $C_{\psi} > 0$. Then  it holds that
\begin{align}\label{eq-bias-end-1}
F^{t,\hat{k}_t }_{i, j, l}\big(\widehat{\boldsymbol{\Theta}}^{t,\hat{k}_t}_{i, j, l}, \widehat{\boldsymbol{\Delta}}^{t,\hat{k}_t}_{i, j, l} \big) - \inf_{(\boldsymbol{\Theta}_{i, j, l}, \boldsymbol{\Delta}_{i, j, l})\in \mathcal{Q}'} F^{t,\hat{k}_t }_{i, j, l}(\boldsymbol{\Theta}_{i, j, l}, \boldsymbol{\Delta}_{i, j, l}) \leq \tau(\hat{k}_t) + \psi(\hat{k}_t).  
\end{align}
Let
\[
\big(\widetilde{\boldsymbol{\Theta}}^{t, \hat{k}_t }_{i, j, l}, \boldsymbol{\widetilde{\Delta}}^{t, \hat{k}_t }_{i, j, l} \big) = \arginf_{(\boldsymbol{\Theta}_{i, j, l}, \boldsymbol{\Delta}_{i, j, l})\in \mathcal{Q}'} F^{t,\hat{k}_t }_{i, j, l}(\boldsymbol{\Theta}_{i, j, l}, \boldsymbol{\Delta}_{i, j, l}).  
\]
Since  $F^{t,\hat{k}_t }_{i, j, l}$ is $2c_{\min}^2(1-c_{\min})^{-2}$-strongly convex, we obtain that 
\begin{align}\label{eq-bias-end-2}
& F^{t,\hat{k}_t }_{i, j, l} \big(\widehat{\boldsymbol{\Theta}}^{t,\hat{k}_t}_{i, j, l}, \widehat{\boldsymbol{\Delta}}^{t,\hat{k}_t}_{i, j, l} \big) -F^{t,\hat{k}_t }_{i, j, l} \big(\widetilde{\boldsymbol{\Theta}}^{t, \hat{k}_t }_{i, j, l}, \boldsymbol{\widetilde{\Delta}}^{t, \hat{k}_t }_{i, j, l} \big)   \nonumber\\
\geq  & \big(\widehat{\boldsymbol{\Theta}}^{t,\hat{k}_t}_{i, j, l}- \widetilde{\boldsymbol{\Theta}}^{t, \hat{k}_t}_{i, j, l} \big) \cdot \nabla_{\boldsymbol{\Theta}_{i, j, l}} F^{t, \hat{k}_t}_{i, j, l} \big( \widetilde{\boldsymbol{\Theta}}^{t, \hat{k}_t}_{i, j, l}, 
\widetilde{\boldsymbol{\Delta}}^{t, \hat{k}_t}_{i, j, l} \big)  +\big( \widehat{\boldsymbol{\Delta}}^{t,\hat{k}_t}_{i, j, l}  - \widetilde{\boldsymbol{\Delta}}^{t, \hat{k}_t}_{i, j, l}\big)  \nonumber\\
& \hspace{0.5cm} \cdot \nabla_{\boldsymbol{\Delta}_{i, j, l}} F^{t, \hat{k}_t}_{i, j, l} \big( \widetilde{\boldsymbol{\Theta}}^{t, \hat{k}_t}_{i, j, l}, \widetilde{\boldsymbol{\Delta}}^{t, \hat{k}_t}_{i, j, l} \big)      +  \frac{2c_{\min}^2}{(1-c_{\min})^2} \Big\{ \big\vert \widehat{\boldsymbol{\Theta}}^{t,\hat{k}_t}_{i, j, l}- \widetilde{\boldsymbol{\Theta}}^{t, \hat{k}_t}_{i, j, l} \big\vert^2 + \big\vert \widehat{\boldsymbol{\Delta}}^{t,\hat{k}_t}_{i, j, l}  - \widetilde{\boldsymbol{\Delta}}^{t, \hat{k}_t}_{i, j, l}\big \vert^2  \Big\}
\nonumber \\
= & \frac{2c_{\min}^2}{(1-c_{\min})^2}\Big\{ \big\vert \widehat{\boldsymbol{\Theta}}^{t,\hat{k}_t}_{i, j, l}- \widetilde{\boldsymbol{\Theta}}^{t, \hat{k}_t}_{i, j, l} \big) \big\vert^2 + \big\vert \widehat{\boldsymbol{\Delta}}^{t,\hat{k}_t}_{i, j, l}  - \widetilde{\boldsymbol{\Delta}}^{t, \hat{k}_t}_{i, j, l}\big \vert^2  \Big\}. 
\end{align}
Combining \eqref{eq-bias-end-1} and \eqref{eq-bias-end-2}, we have that 
\begin{align}\label{eq-bias-final-1}
& \big\vert \widehat{\boldsymbol{\Theta}}^{t,\hat{k}_t}_{i, j, l}-
\widetilde{\boldsymbol{\Theta}}^{t,\hat{k}_t}_{i, j, l}
\big\vert^2 + \big\vert 
\widehat{\boldsymbol{\Delta}}^{t,\hat{k}_t}_{i, j, l}  - \widetilde{\boldsymbol{\Delta}}^{t,\hat{k}_t}_{i, j, l}   \big\vert^2  \nonumber\\
\leq &  \frac{(1-c_{\min})^2}{ 2c_{\min}^2} \big\{ \tau(\hat{k}_t) + \psi(\hat{k}_t) \big\} \nonumber\\
\leq &  C_1  \max \bigg\{ \frac{ \log(n \vee L \vee T) \log(T) \log\log(T)}{ \hat{k}_t}, \big\{V(\hat{k}_t)\big\}^2 \bigg\}\nonumber\\
\leq &   C_1 \max \bigg\{ \frac{ \log(n \vee L \vee T) \log(T) \log\log(T)}{ \bar{k}_t}, \big\{V(\bar{k}_t)\big\}^2 \bigg\}  \nonumber\\ 
\leq &   C_1  \max \bigg\{  \frac{ \log(n \vee L \vee T) \log(T) \log\log(T)}{ \max\{t - 2^{\lceil \log_2 (T_{g-1}) \rceil}, 1\}}, \big[V\big(\max\{t - 2^{\lceil \log_2 (T_{g-1}) \rceil}, 1\} \big)\big]^2 \bigg\}  
\nonumber\\
\leq &   C_2  \max \bigg\{ \frac{ \log(n \vee L \vee T) \log(T) \log\log(T)}{ t -T_{g-1}}, \big\{V(t -T_{g-1})\big\}^2  \bigg\} 
\end{align}
where $C_1, C_2 >0$ are absolute constants, the third inequality follows from $\hat{k}_t \geq \bar{k}_t$ and $V$ is non-increasing, the forth inequality follows from $\bar{k}_t \geq \max\{t - 2^{\lceil \log_2 (T_{g-1}) \rceil}, 1\}$ and $V$ is non-increasing, and the final inequality follows from \Cref{ass-bias}$(i)$.

\medskip
\noindent\textbf{Step 2.} Since $ \hat{k}_t \geq \bar{k}_t $,  by the definition of $ \hat{k}_t $ in \Cref{alg:fast-msbm},
\[ 
\max_{ 1 \leq i \leq j \leq n, l \in [L]} \big\{ f^{t,\bar{k}_t}_{i, j, l} \big(\widehat{\boldsymbol{\Theta}}^{t,\hat{k}_t}_{i, j, l}, \widehat{\boldsymbol{\Delta}}^{t,\hat{k}_t}_{i, j, l} \big) - \inf_{(\boldsymbol{\Theta}_{i, j, l}, \boldsymbol{\Delta}_{i, j, l})\in \mathcal{Q}'} f^{t,\bar{k}_t}_{i, j, l}(\boldsymbol{\Theta}_{i, j, l}, \boldsymbol{\Delta}_{i, j, l}) \big\} \leq \tau(\bar{k}_t).
\]
Since for any $ 1 \leq i\leq j \leq n$ and $l \in [L]$, $f^{t,\bar{k}_t}_{i, j, l}$ and $F^{t,\bar{k}_t}_{i, j, l}$ are $(\log(2), \psi(\bar{k}_t))$-close, it holds that 
\begin{align}
&F^{t,\bar{k}_t}_{i, j, l}\big(\widehat{\boldsymbol{\Theta}}^{t,\hat{k}_t}_{i, j, l}, \widehat{\boldsymbol{\Delta}}^{t,\hat{k}_t}_{i, j, l} \big) - \inf_{(\boldsymbol{\Theta}_{i, j, l}, \boldsymbol{\Delta}_{i, j, l})\in \mathcal{Q}'} F^{t,\bar{k}_t}_{i, j, l}(\boldsymbol{\Theta}_{i, j, l}, \boldsymbol{\Delta}_{i, j, l})  \nonumber\\
\leq & 2 \big\{ f^{t,\bar{k}_t}_{i, j, l}\big(\widehat{\boldsymbol{\Theta}}^{t,\hat{k}_t}_{i, j, l}, \widehat{\boldsymbol{\Delta}}^{t,\hat{k}_t}_{i, j, l} \big)- \inf_{(\boldsymbol{\Theta}_{i, j, l}, \boldsymbol{\Delta}_{i, j, l})  \in \mathcal{Q}'} f^{t,\bar{k}_t}_{i, j, l}(\boldsymbol{\Theta}_{i, j, l}, \boldsymbol{\Delta}_{i, j, l})  + \psi(\bar{k}_t) \big\} \leq 
2\big\{ \tau(\bar{k}_t)  + \psi(\bar{k}_t)\big\} \leq 4 \tau(\bar{k}_t).  \nonumber
\end{align}
Since for any $ u \in [t] \backslash[t-T_g]$, $F^{t,\bar{k}_t}_{i, j, l}$ and  $F^{u}_{i, j, l}$ are $( \log\{4(1-c_{\min})^4/c_{\min}^4\}, C_{\phi'}\phi(T_{g} - T_{g-1} ) $-close, we can derive that 
\begin{align}
&  F^{u}_{i, j, l}\big(\widehat{\boldsymbol{\Theta}}^{t,\hat{k}_t}_{i, j, l}, \widehat{\boldsymbol{\Delta}}^{t,\hat{k}_t}_{i, j, l} \big) - \inf_{(\boldsymbol{\Theta}_{i, j, l}, \boldsymbol{\Delta}_{i, j, l})  \in \mathcal{Q}'} F^{u}_{i, j, l}(\boldsymbol{\Theta}_{i, j, l}, \boldsymbol{\Delta}_{i, j, l}) 
\nonumber\\
 \leq & 4 c_{\min}^{-4}(1-c_{\min})^4  \big\{F^{t,\bar{k}_t}_{i, j, l}\big(\widehat{\boldsymbol{\Theta}}^{t,\hat{k}_t}_{i, j, l}, \widehat{\boldsymbol{\Delta}}^{t,\hat{k}_t}_{i, j, l} \big) - \inf_{(\boldsymbol{\Theta}_{i, j, l}, \boldsymbol{\Delta}_{i, j, l})  \in \mathcal{Q}'} F^{t,\bar{k}_t}_{i, j, l}(\boldsymbol{\Theta}_{i, j, l}, \boldsymbol{\Delta}_{i, j, l})   + C_{\phi'}\phi(T_{g} - T_{g-1} ) \big\} 
 \nonumber\\
 \leq & 4 c_{\min}^{-4}(1-c_{\min})^4 \{ 4 \tau(\bar{k}_t) + C_{\phi'}\phi(T_{g} - T_{g-1} )\}  \leq 20 c_{\min}^{-4}(1-c_{\min})^4  \tau(\bar{k}_t), \nonumber
\end{align}
where the last inequality follows from $\bar{k}_t \leq t - T_{g-1} \leq T_{g} -  T_{g-1}$ and $V$ is non-increasing. 
Since $\bar{k}_t \geq \max\{t - 2^{\lceil \log_2 (T_{g-1}) \rceil}, 1\}$ and $V$ is non-increasing, for any $ u \in [t] \backslash[t-T_g]$, we have that  
\begin{align}
& F^{u}_{i, j, l}\big(\widehat{\boldsymbol{\Theta}}^{t,\hat{k}_t}_{i, j, l}, \widehat{\boldsymbol{\Delta}}^{t,\hat{k}_t}_{i, j, l} \big) - \inf_{(\boldsymbol{\Theta}_{i, j, l}, \boldsymbol{\Delta}_{i, j, l}) \in \mathcal{Q}'} F^{u}_{i, j, l}(\boldsymbol{\Theta}_{i, j, l}, \boldsymbol{\Delta}_{i, j, l}) \nonumber\\
\leq &  20 C_{\tau}  \frac{(1-c_{\min})^4}{c_{\min}^{4}}  \max\bigg\{ \frac{\log(n \vee L \vee T) \log(T) \log\log(T)}{ \max\{t - 2^{\lceil \log_2 (T_{g-1}) \rceil}, 1\}}, \big[ V\big( \max\{t - 2^{\lceil \log_2 (T_{g-1}) \rceil}, 1\} \big) \big]^2  \bigg\} \nonumber\\
\leq &  20 C_{\tau}'  \frac{(1-c_{\min})^4}{c_{\min}^{4}}  \max\bigg\{ \frac{\log(n \vee L \vee T) \log(T) \log\log(T)}{ t - T_{g-1}}, 
 \big\{ V( t - T_{g-1} ) \big\}^2  \bigg\} \nonumber
\end{align} 
where $C_{\tau}' >0$ is an absolute constant and the last inequality follows from \Cref{ass-bias}$(i)$.  Thus, for any $ u \in [t] \backslash[t-T_g]$, it holds that
\begin{align}
 & F^{u}_{i, j, l} \big(\widehat{\boldsymbol{\Theta}}^{t,\hat{k}_t}_{i, j, l}, \widehat{\boldsymbol{\Delta}}^{t,\hat{k}_t}_{i, j, l} \big)-   F^{u}_{i, j, l}\big( \boldsymbol{\Theta}^{u}_{i, j, l}, \boldsymbol{\Delta}^{u}_{i, j, l} \big)   \nonumber\\
 \leq&  20 C_{\tau}'  \frac{(1-c_{\min})^4}{c_{\min}^{4}} \max\bigg\{  \frac{ \log(n \vee L \vee T)\log(T) \log\log(T)}{ t - T_{g-1}}, \big\{ V( t - T_{g-1} ) \big\}^2  \bigg\}. \nonumber
\end{align}
Since  $F^u_{i, j, l}$ is $2c_{\min}^2(1-c_{\min})^{-2}$-strongly convex, then 
\begin{align}
& F^u_{i, j, l} \big(\widehat{\boldsymbol{\Theta}}^{t,\hat{k}_t}_{i, j, l}, \widehat{\boldsymbol{\Delta}}^{t,\hat{k}_t}_{i, j, l} \big) - F^u_{i, j, l} \big( \boldsymbol{\Theta}^{u}_{i, j, l}, \boldsymbol{\Delta}^{u}_{i, j, l} \big)   \nonumber\\
\geq  & \big(\widehat{\boldsymbol{\Theta}}^{t,\hat{k}_t}_{i, j, l}- \boldsymbol{\Theta}^{u}_{i, j, l} \big) \nabla_{\boldsymbol{\Theta}_{i, j, l}} F^u_{i, j, l} \big( \boldsymbol{\Theta}^{u}_{i, j, l}, \boldsymbol{\Delta}^{u}_{i, j, l} \big)  +\big( \widehat{\boldsymbol{\Delta}}^{t,\hat{k}_t}_{i, j, l}  - \boldsymbol{\Delta}^{u}_{i, j, l}\big) \nabla_{\boldsymbol{\Delta}_{i, j, l}} F^u_{i, j, l} \big( \boldsymbol{\Theta}^{u}_{i, j, l}, \boldsymbol{\Delta}^{u}_{i, j, l} \big)     \nonumber\\
& \hspace{0.5cm} +  \frac{2c_{\min}^2}{(1-c_{\min})^2} \Big\{ \big\vert \widehat{\boldsymbol{\Theta}}^{t,\hat{k}_t}_{i, j, l}- \boldsymbol{\Theta}^{u}_{i, j, l} \big) \big\vert^2 + \big\vert \widehat{\boldsymbol{\Delta}}^{t,\hat{k}_t}_{i, j, l}  - \boldsymbol{\Delta}^{u}_{i, j, l}\big \vert^2  \Big\}
\nonumber \\
= & \frac{2c_{\min}^2}{(1-c_{\min})^2}\Big\{ \big\vert \widehat{\boldsymbol{\Theta}}^{t,\hat{k}_t}_{i, j, l}- \boldsymbol{\Theta}^{u}_{i, j, l} \big) \big\vert^2 + \big\vert \widehat{\boldsymbol{\Delta}}^{t,\hat{k}_t}_{i, j, l}  - \boldsymbol{\Delta}^{u}_{i, j, l}\big \vert^2  \Big\}. \nonumber
\end{align}
Therefore,  for any $ u \in [t] \backslash[t-T_g]$, we have that  
\begin{align}\label{eq-bias-final-2}
   &  \big\vert \widehat{\boldsymbol{\Theta}}^{t,\hat{k}_t}_{i, j, l}- \boldsymbol{\Theta}^{u}_{i, j, l}  \big\vert^2 + \big\vert \widehat{\boldsymbol{\Delta}}^{t,\hat{k}_t}_{i, j, l}  - \boldsymbol{\Delta}^{u}_{i, j, l}\big \vert^2  \nonumber\\
  \leq & 10 C_{\tau}'  \frac{(1-c_{\min})^6}{c_{\min}^{6}} \max\bigg\{  \frac{ \log(n \vee L \vee T)\log(T) \log\log(T)}{ t - T_{g-1}}, \big\{ V( t - T_{g-1} ) \big\}^2  \bigg\}.
\end{align}

\medskip
\noindent\textbf{Step 3.} 
By \eqref{eq-bias-final-1} and \eqref{eq-bias-final-2},  under the event $\mathcal{B}$, for any $ u \in [t] \backslash[t-T_g]$, we have that 
\[
   \big\vert  \widehat{\boldsymbol{\Theta}}^{t,\hat{k}_t}_{i, j, l}  - \widetilde{\boldsymbol{\Theta}}^{t,\hat{k}_t}_{i, j, l}  \big\vert^2  +  \big\vert \widehat{\boldsymbol{\Theta}}^{t,\hat{k}_t}_{i, j, l}- \boldsymbol{\Theta}^{u}_{i, j, l}  \big\vert^2 
  \leq  C_3    \max\bigg\{  \frac{ \log(n \vee L \vee T)\log(T) \log\log(T)}{ t - T_{g-1}}, \big\{ V( t - T_{g-1} ) \big\}^2  \bigg\},
\]
and 
\[
   \big\vert  \widehat{\boldsymbol{\Delta}}^{t,\hat{k}_t}_{i, j, l}  - \widetilde{\boldsymbol{\Delta}}^{t,\hat{k}_t}_{i, j, l}  \big\vert^2  +  \big\vert \widehat{\boldsymbol{\Delta}}^{t,\hat{k}_t}_{i, j, l}- \boldsymbol{\Delta}^{u}_{i, j, l}\big\vert^2 
  \leq  C_3    \max\bigg\{  \frac{ \log(n \vee L \vee T)\log(T) \log\log(T)}{ t - T_{g-1}}, \big\{ V( t - T_{g-1} ) \big\}^2  \bigg\},
\]
where $C_3 > 0$ is an absolute constant.
By $(a+b)^2 \leq 2a^2 + 2b^2$, for any $ u \in [t] \backslash[t-T_g]$, we obtain that 
\begin{equation}\label{eq-bias-cases-all-1}
   \big\vert \widetilde{\boldsymbol{\Theta}}^{t,\hat{k}_t}_{i, j, l}  - \boldsymbol{\Theta}^{u}_{i, j, l} \big\vert^2 \leq 2 C_3    \max\bigg\{  \frac{ \log(n \vee L \vee T)\log(T) \log\log(T)}{ t - T_{g-1}}, \big\{ V( t - T_{g-1} ) \big\}^2  \bigg\},
\end{equation}
and
\[
    \big\vert \widetilde{\boldsymbol{\Delta}}^{t,\hat{k}_t}_{i, j, l}  - \boldsymbol{\Delta}^{u}_{i, j, l}\big \vert^2 
  \leq  2C_3   \max\bigg\{  \frac{ \log(n \vee L \vee T)\log(T) \log\log(T)}{ t - T_{g-1}}, \big\{ V( t - T_{g-1} ) \big\}^2  \bigg\}.
\]

Note that we can derive that for any $u \in [t]\backslash [T_{g-1}]$,
\begin{align}
\vert\widetilde{\boldsymbol{\Theta}}^{t,\hat{k}_t}_{i, j, l} -  \boldsymbol{\Theta}^{u}_{i, j, l} \vert 
= & \bigg\vert \frac{\sum_{s=t-\hat{k}_t+1}^{t} \big(1- \boldsymbol{\Pi}^{s-1}_{i, j, l}\big)\boldsymbol{\Theta}^{s}_{i, j, l} }{\sum_{s=t-\hat{k}_t+1}^{t} \ \big(1-\boldsymbol{\Pi}^{s-1}_{i, j, l}\big)}  - \boldsymbol{\Theta}^u_{i, j, l} \bigg\vert\nonumber\\
= & \bigg\vert \frac{\sum_{s=t-\hat{k}_t+1}^{t} \big(1-\boldsymbol{\Pi}^{s-1}_{i, j, l}\big)\big(\boldsymbol{\Theta}^{s}_{i, j, l} -  \boldsymbol{\Theta}^{u}_{i, j, l} \big)}{\sum_{s=t-\hat{k}_t+1}^{t}  \big(1-\boldsymbol{\Pi}^{s-1}_{i, j, l}\big)} \bigg\vert.  \nonumber
\end{align}
Choose $g' \in [G]$ such that  $ T_{g'-1} +1 \leq t-\hat{k}_t+1 \leq T_{g'}$. Then we obtain that
\begin{align}
\vert\widetilde{\boldsymbol{\Theta}}^{t,\hat{k}_t}_{i, j, l} -  \boldsymbol{\Theta}^{u}_{i, j, l} \vert 
= & \Bigg\vert\sum_{g''=g'-1}^{g-1}\frac{ \sum_{s=(T_{g''}+1) \vee (t-\hat{k}_t+1)}^{T_{g''+1} \wedge t } \big(1-\boldsymbol{\Pi}^{s-1}_{i, j, l}\big)\big(\boldsymbol{\Theta}^{s}_{i, j, l} -  \boldsymbol{\Theta}^{u}_{i, j, l} \big)}{\sum_{s=t-\hat{k}_t+1}^{t}  \big(1-\boldsymbol{\Pi}^{s-1}_{i, j, l}\big)} \Bigg\vert. 
\end{align}
Choosing $ u = T_{g-1} + 1$, we have that 
\begin{align}\label{eq-bias-cases-all-2-new}
&
\vert\widetilde{\boldsymbol{\Theta}}^{t,\hat{k}_t}_{i, j, l} -  \boldsymbol{\Theta}^{T_{g-1} + 1}_{i, j, l} \vert  \nonumber\\
 = & \bigg\{\sum_{s=t-\hat{k}_t+1}^{t}  \big(1-\boldsymbol{\Pi}^{s-1}_{i, j, l} \big) \bigg\}^{-1} \bigg\vert \sum_{g''=g'-1}^{g-1}\sum_{s=(T_{g''}+1) \vee (t-\hat{k}_t+1)}^{T_{g''+1} \wedge t } \big(1-\boldsymbol{\Pi}^{s-1}_{i, j, l}\big)  
 \nonumber \\
 & \hspace{0.5cm} \big(\boldsymbol{\Theta}^{_{T_{g''}+1}}_{i, j, l} - \boldsymbol{\Theta}^{ T_{g-1} + 1}_{i, j, l}  + \boldsymbol{\Theta}^{s}_{i, j, l} -  \boldsymbol{\Theta}^{T_{g''}+1}_{i, j, l} \big) \bigg\vert \nonumber\\
  \geq &\frac{ \Big\vert \sum_{g''=g'-1}^{g-1}\sum_{s=(T_{g''}+1) \vee (t-\hat{k}_t+1)}^{T_{g''+1} \wedge t } \big(1-\boldsymbol{\Pi}^{s-1}_{i, j, l}\big) \big(\boldsymbol{\Theta}^{_{T_{g''}+1}}_{i, j, l} - \boldsymbol{\Theta}^{ T_{g-1} + 1}_{i, j, l} \big)\Big\vert }{\sum_{s=t-\hat{k}_t+1}^{t}  \big(1-\boldsymbol{\Pi}^{s-1}_{i, j, l} \big)}   \nonumber \\
 & \hspace{0.5cm} - \frac{\Big\vert \sum_{g''=g'-1}^{g-1}\sum_{s=(T_{g''}+1) \vee (t-\hat{k}_t+1)}^{T_{g''+1} \wedge t }   \big(1-\boldsymbol{\Pi}^{s-1}_{i, j, l}\big) \big( \boldsymbol{\Theta}^{s}_{i, j, l} -  \boldsymbol{\Theta}^{T_{g''}+1}_{i, j, l} \big)\big) \Big\vert}{\sum_{s=t-\hat{k}_t+1}^{t}  \big(1-\boldsymbol{\Pi}^{s-1}_{i, j, l} \big)} \nonumber\\
 = & (I) - (II) 
 \end{align}

\medskip
\noindent
\textbf{Step 3.1.}
In this sub-step,  we consider the term $(II)$ in \eqref{eq-bias-cases-all-2-new}. 
By \eqref{eq-pi-upper-lower} and \Cref{def-seg}, we have that 
\begin{align}\label{eq-bias-cases-all-ii-new}
 (II) \leq &   \frac{(1-c_{\min})}{\hat{k}_t c_{\min}} \sum_{g''=g'-1}^{g-1} \big\{ (T_{g''+1}\wedge t) -T_{g''} \vee (t-\hat{k}_t) \big\}V(T_{g''+1} - T_{g''})
     \nonumber\\ 
     \leq & C_4  V(T_{g}-T_{g-1}) \leq C_4 V(t-T_{g-1}),
\end{align}
where  $C_4 > 0$ is an absolute constant, the second inequality follows from \Cref{ass-bias}$(i)$-$(ii)$, and the final inequality follows from $t - T_{g-1} \leq T_{g} - T_{g-1}$ and $V$ is non-increasing. Combining \eqref{eq-bias-cases-all-1}, \eqref{eq-bias-cases-all-2-new} and \eqref{eq-bias-cases-all-ii-new}, we obtain that 
\begin{equation}\label{eq-bias-cases-all-3-new}
   (I) \leq  (\sqrt{2C_3}+C_4)  \max\bigg\{  \sqrt{\frac{ \log(n \vee L \vee T)\log(T) \log\log(T)}{ t - T_{g-1}}},  V( t - T_{g-1} )   \bigg\}.
\end{equation}

\medskip
\noindent
\textbf{Step 3.2.}
In this sub-step, we analyze term $(I)$ in \eqref{eq-bias-cases-all-2-new}, considering it separately under three cases.

\medskip 
\noindent
\textbf{Case 1:} $\hat{k}_t  \leq   t - T_{g-1} $.
By \Cref{def-seg}, we directly have that for any $s, v \in [t] \backslash [t-\hat{k}_t]$, 
\begin{align}\label{eq-bias-cases-1-new}
  \max \big\{  \vert \boldsymbol{\Theta}^{s}_{i, j, l} -  \boldsymbol{\Theta}^{v}_{i, j, l}  \vert, \vert \boldsymbol{\Delta}^{s}_{i, j, l} -  \boldsymbol{\Delta}^{v}_{i, j, l} \big\}  \vert
    \leq   V(T_g - T_{g-1}) \leq  V(t - T_{g-1}),
\end{align}
where the last inequality follows from $t -T_{g-1} \leq T_{g} - T_{g-1}$ and $V$ is non-increasing.

\medskip 
\noindent
\textbf{Case 2:} $   t - T_{g-1} <  \hat{k}_t \leq   t- T_{g-2} $.  By $\hat{k}_t \geq \bar{k}_t$,  \eqref{def-bar-k-t} and the dynamic grid map $\mathcal {G}^{(t)}$ defined in \eqref{def-dynamic-grid}, we have that $\hat{k}_t \geq C_{\mathrm{low}} (t-T_{g-1})$ for some constant $C_{\mathrm{low}} > 1$. 
Thus, by \eqref{eq-pi-upper-lower},  for any $ s  \in [t] \backslash[t-\hat{k}_t]$,
\begin{align}\label{eq-bias-cases-2-sub-neq}
 (I) \geq   & \frac{c_{\min}}{\hat{k}_t(1-c_{\min})}  \Big[ T_{g-1} -  (t-\hat{k}_t) \} \Big]   \big \vert \boldsymbol{\Theta}^{T_{g-2}+1}_{i, j, l} -  
  \boldsymbol{\Theta}^{T_{g-1} + 1}_{i, j, l} \big \vert  
  \nonumber\\
  \geq & \frac{c_{\min}}{1-c_{\min}}  (1- C_{\mathrm{low}}^{-1})   \big \vert \boldsymbol{\Theta}^{ T_{g-2}+1}_{i, j, l} -   \boldsymbol{\Theta}^{T_{g-1} + 1}_{i, j, l} \big \vert.
\end{align}
Combining \eqref{eq-bias-cases-all-3-new} and \eqref{eq-bias-cases-2-sub-neq}, we obtain that
\begin{align}\label{eq-bias-cases-2-sub2-new}
    \big \vert \boldsymbol{\Theta}^{T_{g-2}+1}_{i, j, l} -   \boldsymbol{\Theta}^{T_{g-1} + 1}_{i, j, l} \big \vert \leq  C_5    \max\bigg\{  \sqrt{\frac{ \log(n \vee L \vee T)\log(T) \log\log(T)}{ t - T_{g-1}}},  V( t - T_{g-1} )   \bigg\},
\end{align}
for some absolute constant $C_5 >0$.
By  \Cref{def-seg} and the triangle inequality, for any $ s,u \in [t] \backslash[t-\hat{k}_t]$, it holds that 
\begin{align}\label{eq-bias-cases-2-new}
   \vert \boldsymbol{\Theta}^{s}_{i, j, l} -  \boldsymbol{\Theta}^{u}_{i, j, l} \vert  \leq C_6   \max\bigg\{  \sqrt{\frac{ \log(n \vee L \vee T)\log(T) \log\log(T)}{ t - T_{g-1}}},  V( t - T_{g-1} )   \bigg\},
\end{align}
where $C_6 > 0$ is an absolute constant

\medskip 
\noindent
\textbf{Case 3:}  $ \hat{k}_t  > t - T_{g-2}$.  By  \Cref{ass-bias}$(ii)$ and  the dynamic grid map $\mathcal {G}^{(t)}$ defined in \eqref{def-dynamic-grid}, we have that 
$c' (T_{g} - T_{g-1}) \leq \hat{k}_t \leq c'' (T_{g} - T_{g-1}) $ for some constants $1 < c' \leq c''$.  Let 
\[
p_{i, j, l} = \sum_{u = g''}^{g-2} (\boldsymbol{\Theta}^{T_{g''+1}+1}_{i, j, l} -\boldsymbol{\Theta}^{T_{g''}+1}_{i, j, l})_+, 
\quad 
n_{i, j, l} = \sum_{u = g''}^{g-2} (\boldsymbol{\Theta}^{T_{g''+1}+1}_{i, j, l} -\boldsymbol{\Theta}^{T_{g''}+1}_{i, j, l})_-.
\]
Based on \Cref{def-local-dom}, we may without loss of generality assume that $\{\boldsymbol{\Theta}_{i, j, l}^{ T_{{g-1}+1}}\}_{g \in [G]}$ is locally dominantly increasing. Thus, for any $ s  \in [t] \backslash[t-\hat{k}_t]$, by \Cref{ass-bias}$(i)$, we have that 
\begin{align}\label{eq-bias-cases-3-sub-new}
  (I) \geq &  \frac{c_{\min}}{\hat{k}_t(1-c_{\min})} 
   C_{7} \hat{k}_t  \big( c_{\min}  p_{i, j, l} - (1-c_{\min})n_{i, j, l} \big)\nonumber\\ 
  \geq &   C_{7}  \frac{c_{\min}}{(1-c_{\min})} 
   \big\{ c_{\min} -C_{\alpha} (1-c_{\min} \big)\big\} p_{i, j, l}, 
\end{align}
where $C_{7} >0$ is an absolute constant. 
Combining  \eqref{eq-bias-cases-all-3-new} and \eqref{eq-bias-cases-3-sub-new},   we obtain that
\begin{align}
     \max_{g'' \in \{g'-1, \ldots g-2\} }  \big \vert \boldsymbol{\Theta}^{T_{g''+1}+1}_{i, j, l} -   \boldsymbol{\Theta}^{T_{g''} + 1}_{i, j, l} \big \vert \leq  C_8\bigg\{  \sqrt{\frac{ \log(n \vee L \vee T)\log(T) \log\log(T)}{ t - T_{g-1}}},  V( t - T_{g-1} )   \bigg\}. \nonumber
\end{align}
for some absolute constant $C_8 > 0$.
By \Cref{def-seg}, \Cref{ass-bias}$(i)$-$(ii)$ and the triangle inequality, for any $ s,u \in [t] \backslash[t-\hat{k}_t]$, it holds that 
\begin{align}\label{eq-bias-cases-3-new}
   \vert \boldsymbol{\Theta}^{s}_{i, j, l} -  \boldsymbol{\Theta}^{u}_{i, j, l} \vert  \leq C_{9} \max\bigg\{  \sqrt{\frac{ \log(n \vee L \vee T)\log(T) \log\log(T)}{ t - T_{g-1}}},  V( t - T_{g-1} )   \bigg\},
\end{align}
where $C_{9} > 0$ is an absolute constant

Combining \eqref{eq-bias-cases-1-new}, \eqref{eq-bias-cases-2-new} and \eqref{eq-bias-cases-3-new}, we establish the claim stated in \eqref{eq-bias-calim}. Similarly, we can prove~\eqref{eq-bias-calim-delta}.

\medskip
\noindent
\textbf{Step 4.} Applying \Cref{prop-bias}, we obtain that
\begin{align}\label{eq-bias-1-start}
\Big\vert\mathbb{E}\big\{\widehat{\boldsymbol{\Theta}}_{i,j,l}^{t,\hat{k}_t} \big\} - \boldsymbol{\Theta}_{i,j, l}^{t, \hat{k}_t} \Big\vert \leq   &  \mathcal{E}^{t,\hat{k}_t}_{i, j, l} + \frac{C_{10}}{\hat{k}_t}  
\leq   \mathcal{E}^{t,\hat{k}_t}_{i, j, l} + \frac{C_{10}}{\max\{t-2^{\lceil \log_2 (T_{g-1}) \rceil}, 1\}} 
\leq   \mathcal{E}^{t,\hat{k}_t}_{i, j, l} + \frac{C^{*}}{ t-T_{g-1} } ,
\end{align}
where $C_{10}, C^*  >0$ are absolute constants and the second inequality follows from $\hat{k}_t \geq \bar{k}_t  \geq \max\{t-2^{\lceil \log_2 (T_{g-1}) \rceil}, 1\}$. 
Note that
\begin{align}\label{eq-bias-1-detail-1}
  \mathcal{E}^{t,\hat{k}_t}_{i, j, l}
\leq  & \frac{1}{2c_{\min}^2}  \max_{s, u \in [t]\backslash[t-\hat{k}_t]} \vert \boldsymbol{\Theta}^{s}_{i, j, l} -  \boldsymbol{\Theta}^{u}_{i, j, l} \vert
\sum_{u' = T_{g'-1} +2}^{t}  \big( \big\vert \boldsymbol{\Theta}^{u'}_{i,j,l} -   \boldsymbol{\Theta}^{u'-1}_{i,j,l}  \big\vert  +  \big\vert  \boldsymbol{\Delta}^{u'}_{i,j,l} - \boldsymbol{\Delta}^{u'-1}_{i,j,l}  \big\vert\big) \nonumber\\
\leq  &  \frac{ C'}{2c_{\min}^2}   \max\bigg\{  \sqrt{\frac{ \log(n \vee L \vee T)\log(T) \log\log(T)}{ t - T_{g-1}}},  V( t - T_{g-1} )   \bigg\} \bigg(
\sum_{u' = T_{g'-1} +2}^{t}  \big\vert \boldsymbol{\Theta}^{ u'}_{i,j,l} -   \boldsymbol{\Theta}^{ u'-1}_{i,j,l}  \big\vert
\nonumber\\
& \hspace{0.5cm} + \sum_{u' = T_{g'-1} +2}^{t}
\big\vert \boldsymbol{\Delta}^{u'}_{i,j,l} -   \boldsymbol{\Delta}^{u' - 1}_{i,j,l}  \big\vert  \bigg) \nonumber\\
=  &  \frac{ C'}{2c_{\min}^2}  \max\bigg\{  \sqrt{\frac{ \log(n \vee L \vee T)\log(T) \log\log(T)}{ t - T_{g-1}}},  V( t - T_{g-1} )   \bigg\}   (IV) ,
\end{align}
where the first inequality follows from \eqref{eq-bias-calim}.

For term $(IV)$ in \eqref{eq-bias-1-detail-1},
for any $g \in \{g'-1, \ldots, g-1\}$, let 
\[
p_{i, j, l}^{g''} = \sum_{u = T_{g''}+1}^{T_{g''+1}-1} (\boldsymbol{\Theta}^{t+1}_{i, j, l} -\boldsymbol{\Theta}^{t}_{i, j, l})_+, 
\quad 
n_{i, j, l}^{g''} = \sum_{u = T_{g''}+1}^{T_{g''+1}-1}    (\boldsymbol{\Theta}^{t+1}_{i, j, l} -\boldsymbol{\Theta}^{t}_{i, j, l})_-, \quad 
\]
and 
\[
\tilde{p}_{i, j, l}^{g''} = \sum_{u = T_{g''}+1}^{T_{g''+1}-1} (\boldsymbol{\Theta}^{t+1}_{i, j, l} -\boldsymbol{\Theta}^{t}_{i, j, l})_+, 
\quad 
\tilde{n}_{i, j, l}^{g''} = \sum_{u = T_{g''}+1}^{T_{g''+1}-1}    (\boldsymbol{\Theta}^{t+1}_{i, j, l} -\boldsymbol{\Theta}^{t}_{i, j, l})_-. \quad 
\]
Based on \Cref{def-local-dom}, for any $g \in \{g'-1, \ldots, g-1\}$, we may without loss of generality assume that $\{\boldsymbol{\Theta}_{i, j, l}^{t}\}_{t \in [T_{g''}] \backslash [T_{g''-1}]}$ is dominantly increasing. Then we can derive that 
\begin{align}\label{eq-bias-1-detail-2}
(IV) \leq &  \sum_{g'' = g'-1}^{g-1} \big( p_{i, j, l}^{g''}+ n_{i, j, l}^{g''} + \tilde{p}_{i, j, l}^{g''}+ \tilde{n}_{i, j, l}^{g''} \big) +  \sum_{g'' = g'}^{g-1}\Big( \big\vert \boldsymbol{\Theta}^{ T_{g''}+1}_{i,j,l} -   \boldsymbol{\Theta}^{ T_{g''}}_{i,j,l}  \big\vert +  \big\vert \boldsymbol{\Delta}^{ T_{g''}+1}_{i,j,l} -   \boldsymbol{\Delta}^{ T_{g''}}_{i,j,l}  \big\vert \Big)
\nonumber\\
\leq &  \frac{1+C_{\alpha}}{1+C_{\alpha}}  \sum_{g'' = g'-1}^{g-1} \big(  p_{i, j, l}^{g''} - n_{i, j, l}^{g''} + \tilde{p}_{i, j, l}^{g''}+ \tilde{n}_{i, j, l}^{g''}\big)
+ \sum_{g'' = g'}^{g-1}\Big( \big\vert \boldsymbol{\Theta}^{ T_{g''}+1}_{i,j,l} -   \boldsymbol{\Theta}^{ T_{g''}}_{i,j,l}  \big\vert +  \big\vert \boldsymbol{\Delta}^{ T_{g''}+1}_{i,j,l} -   \boldsymbol{\Delta}^{ T_{g''}}_{i,j,l}  \big\vert \Big) \nonumber\\
 \leq &  \frac{1+C_{\alpha}}{1+C_{\alpha}}  \sum_{g'' = g'-1}^{g-1} \Big( \big\vert \boldsymbol{\Theta}^{ T_{g''+1}}_{i,j,l} -   \boldsymbol{\Theta}^{ T_{g''}+1}_{i,j,l}  \big\vert +  \big\vert \boldsymbol{\Delta}^{ T_{g''+1}}_{i,j,l} -   \boldsymbol{\Delta}^{ T_{g''}+1}_{i,j,l}  \big\vert 
\Big) + \sum_{g'' = g'}^{g-1}\Big( \big\vert \boldsymbol{\Theta}^{ T_{g''}+1}_{i,j,l} -   \boldsymbol{\Theta}^{ T_{g''}}_{i,j,l}  \big\vert 
\nonumber\\
& \hspace{0.5cm} +  \big\vert \boldsymbol{\Delta}^{ T_{g''}+1}_{i,j,l} -   \boldsymbol{\Delta}^{ T_{g''}}_{i,j,l}  \big\vert \Big) 
\nonumber\\
\leq & C''  \max\bigg\{  \sqrt{\frac{ \log(n \vee L \vee T)\log(T) \log\log(T)}{ t - T_{g-1}}},  V( t - T_{g-1} )   \bigg\}
\end{align}
where  $C'' >0$ is an absolute constant, the second inequality follows from  \Cref{def-local-dom}, and the final inequality follows from \Cref{def-seg}, \eqref{eq-bias-calim}, \eqref{eq-bias-calim-delta} and \Cref{ass-bias}$(ii)$. 

Combining \eqref{eq-bias-1-detail-1} and \eqref{eq-bias-1-detail-2}, we can conclude that under the event $\mathcal{B}$, 
\begin{align}
\Big\vert\mathbb{E} \big\{\widehat{\boldsymbol{\Theta}}_{i,j,l}^{t,\hat{k}_t} \big\} - \boldsymbol{\Theta}_{i,j, l}^{t, \hat{k}_t} \Big\vert \leq  &   \bigg(\frac{C'C''}{2c_{\min}^2} + C^*\bigg) \max\bigg\{ \frac{ \log(n \vee L \vee T)\log(T) \log\log(T)}{ t - T_{g-1}},  \big\{ V( t - T_{g-1} )  \big\}^2 \bigg\} \nonumber
\end{align}
Similarly, we can prove that  under the event $\mathcal{B}$, 
\begin{align}
\Big\vert\mathbb{E} \big\{\widehat{\boldsymbol{\Delta}}_{i,j,l}^{t,\hat{k}_t} \big\} - \boldsymbol{\Delta}_{i,j, l}^{t, \hat{k}_t} \Big\vert \leq  &   \bigg(\frac{C'C''}{2c_{\min}^2} + C^*\bigg) \max\bigg\{ \frac{ \log(n \vee L \vee T)\log(T) \log\log(T)}{ t - T_{g-1}},  \big\{ V( t - T_{g-1} )  \big\}^2 \bigg\}. \nonumber
\end{align}
Finally, by \eqref{eq-step1-main-event-prop}, we conclude the proof

\end{proof}

\subsection{Proofs for Appendix~\ref{app-non-stat-add}}\label{app-non-stat-add-proof}

This section provides the proofs of the results stated in Appendix~\ref{app-non-stat-add}. The proofs are organized as follows: Lemma~\ref{lemma-mixing-nonstat} is proved in Appendix~\ref{sec-app-1-non}, Lemma~\ref{lem-temporal-cov} in Appendix~\ref{sec-app-2-non}, Proposition~\ref{prop-bias} in Appendix~\ref{sec-app-3-non} and Lemma~\ref{lemma-singular-values-non-stat} in Appendix~\ref{sec-app-4-non}.

\subsubsection{Proof of Lemma \ref{lemma-mixing-nonstat}}\label{sec-app-1-non}

\begin{proof}[Proof of \Cref{lemma-mixing-nonstat}]
Fix $ 1 \leq i \leq j \leq n$ and $l \in [L]$. For any $a, b \in \mathbb{N}$ with $ a\leq b$, define the event
\[
\mathcal{R}^{a,b}_{i, j, l} = \bigcap_{t=a}^b \{\mathbf{E}_{i, j, l}^t = 0\}.
\]
Since $\{E_{i, j, l}^t\}_{t \geq 0}$ are independent, we have that
\[
\P ( \mathcal{R}^{a,b}_{i, j, l} ) = \prod_{t=a}^b (1 - \boldsymbol{\Theta}_{i, j, l}^t - \boldsymbol{\Delta}_{i, j, l}^t ).
\]
For any $t, k \in \mathbb{Z}^+$,  let $\mathcal{F}^{0, t}_{i, j, l} = \sigma(\mathbf{A}^0_{i, j, l}, \mathbf{A}^1_{i, j, l}, \cdots, \mathbf{A}^t_{i, j, l})$ and $\mathcal{F}^{t+k, \infty}_{i, j, l}= \sigma( \mathbf{A}^{t+k}_{i, j, l}, \mathbf{A}^{t+k+1}_{i, j, l}, \cdots)$. For any  $\mathcal{A} \in \mathcal{F}^{0,t}_{i, j, l}$ and $\mathcal{B} \in \mathcal{F}^{t+k, \infty}_{i, j, l}$, we have that
\begin{align*}
\mathbb{P}(\mathcal{A} \cap \mathcal{B}) &= \mathbb{P}(\mathcal{A} \cap \mathcal{B} \cap \mathcal{R}^{t+1,t+k}_{i, j, l}) + \mathbb{P} \big(\mathcal{A} \cap \mathcal{B} \cap (\mathcal{R}^{t+1,t+k}_{i, j, l})^c \big) \\
&= \mathbb{P}(\mathcal{A} \cap \mathcal{B} \cap \mathcal{R}^{t+1,t+k}_{i, j, l}) + \mathbb{P}\big(\mathcal{A} \cap \mathcal{B} \mid (\mathcal{R}^{t+1,t+k}_{i, j, l})^c \big)\mathbb{P} \big( (\mathcal{R}^{t+1,t+k}_{i, j, l})^c \big) \\
&= \mathbb{P}(\mathcal{A} \cap \mathcal{B} \cap \mathcal{R}^{t+1,t+k}_{i, j, l}) + \mathbb{P}\big(\mathcal{A} \mid (\mathcal{R}^{t+1,t+k}_{i, j, l})^c\big)\mathbb{P}\big( \mathcal{B} \mid (\mathcal{R}^{t+1,t+k}_{i, j, l})^c\big)\mathbb{P}\big((\mathcal{R}^{t+1,t+k}_{i, j, l})^c\big) \\
&= \mathbb{P}(\mathcal{A} \cap \mathcal{B} \cap \mathcal{R}^{t+1,t+k}_{i, j, l}) + \mathbb{P}\big(\mathcal{A})\mathbb{P} \big(\mathcal{B} \cap (\mathcal{R}^{t+1,t+k}_{i, j, l})^c\big),
\end{align*}
where the third equality follows from the conditional independence of $\mathcal{A}$ and $\mathcal{B}$ given $(\mathcal{R}^{t+1,t+k}_{i, j, l})^c$,  and the last equality follows from the independence between $\mathcal{A}$ and $\mathcal{R}^{t+1,t+k}_{i, j, l}$. Thus, we can derive that 
\begin{align*}
|\mathbb{P}(\mathcal{A} \cap \mathcal{B}) - \mathbb{P}(\mathcal{A})\mathbb{P}(\mathcal{B})| &= |\mathbb{P}(\mathcal{A} \cap \mathcal{B} \cap \mathcal{R}^{t+1,t+k}_{i, j, l}) - \mathbb{P}(\mathcal{A})\mathbb{P}(\mathcal{B} \cap \mathcal{R}^{t+1,t+k}_{i, j, l})| \\
&\leq \mathbb{P}(\mathcal{R}^{t+1,t+k}_{i, j, l}) = \prod_{j=t+1}^{t+k} (1 - \boldsymbol{\Theta}_{i, j, l}^t -\boldsymbol{\Delta}_{i, j, l}^t ) \leq (1-2c_{\min})^k,
\end{align*}
where the first inequality follows from the fact that  $\mathbb{P}(\mathcal{A} \cap \mathcal{B} \cap \mathcal{R}^{t+1,t+k}_{i, j, l})$ and $\mathbb{P}(\mathcal{A})\mathbb{P}(\mathcal{B} \cap \mathcal{R}^{t+1,t+k}_{i, j, l})$ both lie in $\big[0, \mathbb{P}( \mathcal{R}^{t+1,t+k}_{i, j, l})\big]$.

Therefore, for any $k\in \Z^+$,
\[
\alpha(k) = \sup_{t \in \mathbb{Z}^+} \sup_{\mathcal{A} \in \mathcal{F}^{0,t}_{i, j, l}, \mathcal{B} \in \mathcal{F}^{t+k, \infty}_{i, j, l} }|\mathbb{P}(\mathcal{A} \cap \mathcal{B}) - \mathbb{P}(\mathcal{A})\mathbb{P}(\mathcal{B})| \leq (1-2c_{\min})^k.
\]
Since $\alpha(k) \to  0$ as $k \to \infty$, by \Cref{def-mixing}, the process $\{ \mathbf{A}_{i,j,l}^t \}_{t \geq 0}$ is strongly mixing, which completes the proof.

\end{proof}

\subsubsection{Proof of Lemma \ref{lem-temporal-cov}}\label{sec-app-2-non}
\begin{proof}[Proof of \Cref{lem-temporal-cov}]
Fix $1 \leq i \leq j \leq n$ and $l \in [L]$. Without loss of generality, assume $1 \leq r \leq s$.
For $s > r$, define the indicator of no update between $r$ and $s$
\[
N^{r,s} = \prod_{t=r+1}^{s}\mathbbm{1}\{\mathbf{E}^{t}_{i,j,l}=0\}.
\]
Note that
\begin{equation}\label{eq-cov-EN}
\mathbb{E}\{N^{r,s}\}
=\prod_{t=r+1}^{s}(1-\boldsymbol{\Theta}^{t}_{i,j,l}-\boldsymbol{\Delta}^{t}_{i,j,l}).
\end{equation}

Next, define the first update time
\[
\tau =\min\big\{t\in[s]\backslash [r] \colon \mathbf{E}_{i, j, l}^t\in\{\pm1\} \big\},
\]
with the convention that $\tau = \infty$ if the set is empty. Then define
\[
\widetilde{B}
=\begin{cases}
\mathbbm{1}\{\mathbf{E}_{i, j, l}^\tau=1\}, & \tau<\infty,\\
0, & \tau=\infty.
\end{cases}
\]
By recursively applying \Cref{def-armsb}, we obtain that
\[
\mathbf{A}_{i,j,l}^{s}=N^{r,s}\mathbf{A}_{i,j,l}^{r}+(1-N^{r,s})\widetilde{B}.
\]
Since $N^{r,s}$ and $\widetilde{B}$ depend only on $\{\mathbf{E}_{i,j,l}^t\}_{t=r+1}^s$ and are thus independent of $\mathbf{A}_{i,j,l}^r$, we 
compute the covariance
\begin{align}\label{eq-cov-calc}
\mathrm{Cov} \big(\mathbf{A}_{i,j,l}^{r},\mathbf{A}_{i,j,l}^{s} \big)
= & \mathbb{E}\big\{ \mathbf{A}_{i,j,l}^{r}\mathbf{A}_{i,j,l}^{s}\big\}-\mathbb{E}\big\{\mathbf{A}_{i,j,l}^{r}\} \mathbb{E}\big\{\mathbf{A}_{i,j,l}^{s}\big\}  \nonumber\\
= &  \mathbb{E}\big\{ N^{r,s} (\mathbf{A}_{i,j,l}^{r})^2 + (1-N^{r,s})\widetilde{B} \mathbf{A}_{i,j,l}^{r} \big\}-\mathbb{E}\{\mathbf{A}_{i,j,l}^{r}\} \mathbb{E}\big\{ N^{r,s}\mathbf{A}_{i,j,l}^{r}+(1-N^{r,s})\widetilde{B} \big\}  \nonumber\\
= &  \mathbb{E} \{ N^{r,s} \} \mathbb{E}\big\{(\mathbf{A}_{i,j,l}^{r})^2\big\} + 
 \mathbb{E} \big\{(1-N^{r,s})\widetilde{B} \big\}\mathbb{E} \big\{ \mathbf{A}_{i,j,l}^{r} \big\} -  \big[ \mathbb{E}\{\mathbf{A}_{i,j,l}^{r}\} \big]^2 \mathbb{E}\{ N^{r,s}\} \nonumber\\
&  \hspace{0.5cm} - \mathbb{E} \big\{ \mathbf{A}_{i,j,l}^{r} \big\} \mathbb{E} \big\{(1-N^{r,s})\widetilde{B} \big\}  \nonumber\\
= &  \mathbb{E} \{ N^{r,s} \} \mathrm{Var}(\mathbf{A}_{i,j,l}^{r}),
\end{align}
Since $\mathbf{A}_{i,j,l}^{r}$ is a Bernoulli random variable, we have that 
\begin{equation}\label{eq-cov-var}
\mathrm{Var}(\mathbf{A}_{i,j,l}^{r})\leq\frac{1}{4}.
\end{equation}
If $r = s$, the bound is immediate
\[
0 \leq \mathrm{Cov}(\mathbf{A}_{i,j,l}^{r},\mathbf{A}_{i,j,l}^{s})
\leq \frac{1}{4}  =  \frac{1}{4} (1-2c_{\min})^{s-r};
\]
and if $r < s$, combining \eqref{eq-cov-EN}, \eqref{eq-cov-calc} and \eqref{eq-cov-var},  we have that 
\[
0 \leq \mathrm{Cov}(\mathbf{A}_{i,j,l}^{r},\mathbf{A}_{i,j,l}^{s})
\leq \frac{1}{4}\prod_{t=r+1}^{s}(1-\boldsymbol{\Theta}^{t}_{i,j,l}-\boldsymbol{\Delta}^{t}_{i,j,l})
\leq  \frac{1}{4} (1-2c_{\min})^{s-r}.
\]
These conclude the proof.
\end{proof}

\subsubsection{Proof of Proposition \ref{prop-bias}}
\label{sec-app-3-non}
\begin{proof}[Proof of \Cref{prop-bias}]
Fix $(i,j,l)$ with $1 \leq i \leq j \leq n$ and $l \in [L]$, and consider a look-back window of length $k$ ending at time $t$, i.e.~observations $\{\mathbf{A}^{s}\}_{s=t-k+1}^{t}$. Recall that the estimators are defined as
\[
\widehat{\boldsymbol{\Theta}}_{i,j,l}^{t,k} = \frac{ k^{-1}\sum_{s=t-k+1}^{t} \mathbf{A}_{i,j,l}^{s} (1 - \mathbf{A}_{i,j,l}^{s-1})}{ 1- k^{-1}\sum_{s=t-k+1}^{t}  \mathbf{A}_{i,j,l}^{s-1}}, \quad \widehat{\boldsymbol{\Delta}}_{i,j,l}^{t,k} = \frac{k^{-1}\sum_{s=t-k+1}^{t} ( 1 - \mathbf{A}_{i,j,l}^{s}) \mathbf{A}_{i,j,l}^{s-1}}{k^{-1} \sum_{s=t-k+1}^{t} \mathbf{A}_{i,j,l}^{s-1}}.
\]
For any $t \in [T]$ and $s \in [t] \cup \{0\}$, write
\[
\boldsymbol{\Pi}_{i,j,l}^{s} = \mathbb{E}\{\mathbf{A}_{i,j,l}^{s}\} \quad \mbox{and} \quad
\boldsymbol{\Pi}_{i,j,l}^{t,k} = \frac{1}{k} \sum_{s=t-k+1}^t \boldsymbol{\Pi}_{i,j,l}^{s-1}.
\]

Define the event
\begin{align}\label{eq-bias-event}
\mathcal{I}_{t,k} = \bigg\{ \bigg\vert  \frac{1}{k} \sum_{s=t-k+1}^t \mathbf{A}_{i,j,l}^{s-1} - \boldsymbol{\Pi}_{i,j,l}^{t,k} \bigg\vert  \le \frac{1 - \boldsymbol{\Pi}_{i,j,l}^{t,k}}{2} \bigg\}.
\end{align}
Expanding $(1- k^{-1}\sum_{s=t-k+1}^{t}  \mathbf{A}_{i,j,l}^{s-1})^{-1}$ around $(1 -\boldsymbol{\Pi}_{i,j,l}^{t,k})^{-1}$ yields
\begin{align}\label{eq-bias-all}
& \mathbb{E} \big\{\widehat{\boldsymbol{\Theta}}_{i,j,l}^{t,k}  \big\} \nonumber\\
= & \frac{1}{k} \mathbb{E}\bigg\{  \sum_{s=t-k+1}^{t} \frac{\mathbf{A}_{i,j,l}^{s} (1 - \mathbf{A}_{i,j,l}^{s-1}) }{ 1- k^{-1}\sum_{s=t-k+1}^{t}  \mathbf{A}_{i,j,l}^{s-1}} \bigg\} \nonumber\\
= & \frac{1}{k} \mathbb{E} \bigg[\sum_{s=t-k+1}^{t} \mathbf{A}_{i,j,l}^{s} (1 - \mathbf{A}_{i,j,l}^{s-1}) 
\bigg\{
\frac{1}{1 -\boldsymbol{\Pi}_{i,j,l}^{t,k}} 
+ \frac{k^{-1} \sum_{s=t-k+1}^{t} \mathbf{A}_{i,j,l}^{s-1} - \boldsymbol{\Pi}_{i,j,l}^{t,k}}{(1 - \boldsymbol{\Pi}_{i,j,l}^{t,k} )^2} \nonumber\\
& \hspace{0.5cm} + \sum_{u=2}^{\infty} \frac{\big(k^{-1} \sum_{s=t-k+1}^{t} \mathbf{A}_{i,j,l}^{s-1} - \boldsymbol{\Pi}_{i,j,l}^{t,k}\big)^u}{\big(1 - \boldsymbol{\Pi}_{i,j,l}^{t,k}\big)^{u+1}}
\bigg\}  \bigg] \nonumber\\
= & \frac{1}{k} \mathbb{E} \bigg\{\sum_{s=t-k+1}^{t} \frac{\mathbf{A}_{i,j,l}^{s} (1 - \mathbf{A}_{i,j,l}^{s-1})}{1 -\boldsymbol{\Pi}_{i,j,l}^{t,k}} \bigg\} 
+  \frac{1}{k} \mathbb{E} \bigg\{ \sum_{s=t-k+1}^{t} \mathbf{A}_{i,j,l}^{s} (1 - \mathbf{A}_{i,j,l}^{s-1}) \frac{k^{-1} \sum_{s=t-k+1}^{t} \mathbf{A}_{i,j,l}^{s-1} - \boldsymbol{\Pi}_{i,j,l}^{t,k}}{(1 - \boldsymbol{\Pi}_{i,j,l}^{t,k} )^2} \bigg\} \nonumber\\
& \hspace{0.5cm} +  \frac{1}{k} \mathbb{E} \bigg\{\sum_{s=t-k+1}^{t} \mathbf{A}_{i,j,l}^{s} (1 - \mathbf{A}_{i,j,l}^{s-1}) \sum_{u=2}^{\infty} \frac{\big(k^{-1} \sum_{s=t-k+1}^{t} \mathbf{A}_{i,j,l}^{s-1} - \boldsymbol{\Pi}_{i,j,l}^{t,k}\big)^u}{\big(1 - \boldsymbol{\Pi}_{i,j,l}^{t,k}\big)^{u+1}}
\bigg\}  \nonumber\\
= & \frac{1}{k} \mathbb{E} \bigg\{\sum_{s=t-k+1}^{t} \frac{\mathbf{A}_{i,j,l}^{s} (1 - \mathbf{A}_{i,j,l}^{s-1})}{1 -\boldsymbol{\Pi}_{i,j,l}^{t,k}} \bigg\} 
+  \frac{1}{k} \mathbb{E} \bigg\{ \sum_{s=t-k+1}^{t} \mathbf{A}_{i,j,l}^{s} (1 - \mathbf{A}_{i,j,l}^{s-1}) \frac{k^{-1} \sum_{s=t-k+1}^{t} \mathbf{A}_{i,j,l}^{s-1} - \boldsymbol{\Pi}_{i,j,l}^{t,k}}{(1 - \boldsymbol{\Pi}_{i,j,l}^{t,k} )^2} \bigg\} \nonumber\\
& \hspace{0.5cm} 
+  \frac{1}{k}\mathbb{E}  \bigg[ \bigg\{\sum_{s=t-k+1}^{t} \mathbf{A}_{i,j,l}^{s} (1 - \mathbf{A}_{i,j,l}^{s-1}) \sum_{u=2}^{\infty} \frac{\big(k^{-1} \sum_{s=t-k+1}^{t} \mathbf{A}_{i,j,l}^{s-1} - \boldsymbol{\Pi}_{i,j,l}^{t,k}\big)^u}{\big(1 - \boldsymbol{\Pi}_{i,j,l}^{t,k}\big)^{u+1}} \bigg\}  \mathbbm{1}_{\mathcal{I}_{t,k}}
 \bigg]
 \nonumber\\
& \hspace{0.5cm} +   \frac{1}{k} \mathbb{E}  \bigg[ \bigg\{\sum_{s=t-k+1}^{t} \mathbf{A}_{i,j,l}^{s} (1 - \mathbf{A}_{i,j,l}^{s-1}) \sum_{u=2}^{\infty} \frac{\big(k^{-1} \sum_{s=t-k+1}^{t} \mathbf{A}_{i,j,l}^{s-1} - \boldsymbol{\Pi}_{i,j,l}^{t,k}\big)^u}{\big(1 - \boldsymbol{\Pi}_{i,j,l}^{t,k}\big)^{u+1}}
\bigg\}  \mathbbm{1}_{\mathcal{I}_{t,k}^c}\bigg] \nonumber\\
= & (I)  + (II)  +  (III.1) + (III.2).
\end{align}
In the following steps, we address the above terms separately.

\medskip
\noindent
\textbf{Step 1.}
In this step, we analyze term $(I)$ in \eqref{eq-bias-all}. First, note that
\begin{align}\label{eq-bias-I-beg}
  (I) = &  \frac{1}{k}  \sum_{s=t-k+1}^{t} \boldsymbol{\Theta}^{s}_{i, j, l}(1-\boldsymbol{\Pi}_{i,j,l}^{s-1})
\frac{1}{1 -\boldsymbol{\Pi}_{i,j,l}^{t,k}}  \nonumber\\
 =  & \boldsymbol{\Theta}^{t, k}_{i, j, l} + \frac{1}{k^2} \sum_{s, u=t-k+1}^{t} (\boldsymbol{\Theta}^{s}_{i, j, l} - \boldsymbol{\Theta}^{u}_{i, j, l}) (1-\boldsymbol{\Pi}_{i,j,l}^{s-1})
\frac{1}{1 -\boldsymbol{\Pi}_{i,j,l}^{t,k}}  \nonumber\\
= & \boldsymbol{\Theta}_{i,j, l}^{t, k} +  \mathcal{E}^{t,k}_{i, j, l}
\end{align}
If $k =1$, then we can directly have that 
$\mathcal{E}^{t,k}_{i, j, l} = 0$.
For any $v \in [t] \backslash [t-k]$, we have that 
\begin{equation}\label{eq-bias-I-spec}
    \mathcal{E}^{t,k}_{i, j, l} =   \frac{1}{k^2} \sum_{s, u=t-k+1}^{t} (\boldsymbol{\Theta}^{s}_{i, j, l} - \boldsymbol{\Theta}^{u}_{i, j, l}) (\boldsymbol{\Pi}_{i,j,l}^{v-1}-\boldsymbol{\Pi}_{i,j,l}^{s-1})
\frac{1}{1 -\boldsymbol{\Pi}_{i,j,l}^{t,k}}.
\end{equation}

For any $\pi, \theta, \delta \in [c_{\min}, 1-c_{\min}]$, define
\[
F_{\pi}(\theta,\delta) = 
\theta+ (1-\theta-\delta)\pi.
\]
Then for any $\pi \in [c_{\min}, 1-c_{\min}]$, we have that 
\[
\frac{\partial F_{\pi}}{\partial \theta}=1-\pi,
\quad \mbox{and} \quad
\frac{\partial F_{\pi}}{\partial \delta}=-\pi.
\]
Consequently, it holds that
\[
\|\nabla F_{\pi}\|_{\infty} \leq  1-c_{\min} \leq C_{\pi},
\]
where the last inequality follows from $C_{\pi} \geq 1$.
Then we have 
\[
\boldsymbol{\Pi}^{u}_{i,j,l}=F_{\boldsymbol{\Pi}^{u-1}_{i,j,l}}\big(\boldsymbol{\Theta}^{u}_{i,j,l}, \boldsymbol{\Delta}^{u}_{i,j,l}\big)
\quad \mbox{and} \quad 
\boldsymbol{\boldsymbol{\Pi}}^{u-1}_{i,j,l}=F_{\boldsymbol{\Pi}^{u-2}_{i,j,l}}\big(\boldsymbol{\Theta}^{u-1}_{i,j,l},\boldsymbol{\Delta}^{u-1}_{i,j,l}\big).
\]

We claim that under \Cref{ass-bias}$(iii)$, for any $u \geq T_g+2$, 
\begin{equation}\label{eq-pi-theta-delta}
\bigl|\boldsymbol{\Pi}^{u}_{i,j,l}-\boldsymbol{\Pi}^{u-1}_{i,j,l}\big| \leq C_{\pi}\sum_{v=T_g+2}^u (1-2c_{\min})^{u- v} \big( \big\vert \boldsymbol{\Theta}^{v}_{i,j,l} -   \boldsymbol{\Theta}^{v-1}_{i,j,l}  \big\vert  +  \big\vert  \boldsymbol{\Delta}^{v}_{i,j,l} - \boldsymbol{\Delta}^{v-1}_{i,j,l}  \big\vert  \big).
\end{equation}
When $u = T_g+2$, inequality \eqref{eq-pi-theta-delta} follows directly from \Cref{ass-bias}$(iii)$.
Assume now that for some  $u \geq T_g+3$,
\[
\bigl|\boldsymbol{\Pi}^{u-1}_{i,j,l}-\boldsymbol{\Pi}^{u-2}_{i,j,l}\big| \leq C_{\pi} \sum_{v=T_g+2}^{u-1}(1-2c_{\min})^{u- v - 1} \big( \big\vert \boldsymbol{\Theta}^{v}_{i,j,l} -   \boldsymbol{\Theta}^{v-1}_{i,j,l}  \big\vert  +  \big\vert  \boldsymbol{\Delta}^{v}_{i,j,l} - \boldsymbol{\Delta}^{v-1}_{i,j,l}  \big\vert  \big).
\]
We then have that 
\begin{align}
& \big|\boldsymbol{\Pi}^{u}_{i,j,l}-\boldsymbol{\Pi}^{u-1}_{i,j,l}\big| \nonumber\\
= & \big|F_{\boldsymbol{\Pi}^{u-1}_{i,j,l}}\big(\boldsymbol{\Theta}^{u}_{i,j,l}, \boldsymbol{\Delta}^{u}_{i,j,l}\big)
      -F_{\boldsymbol{\Pi}^{u-2}_{i,j,l}}\big(\boldsymbol{\Theta}^{u-1}_{i,j,l},\boldsymbol{\Delta}^{u-1}_{i,j,l}\big)\big| \nonumber\\
\leq& \big|F_{\boldsymbol{\Pi}^{u-1}_{i,j,l}}\big(\boldsymbol{\Theta}^{u}_{i,j,l}, \boldsymbol{\Delta}^{u}_{i,j,l}\big)
      -F_{\boldsymbol{\Pi}^{u-1}_{i,j,l}}\big(\boldsymbol{\Theta}^{u-1}_{i,j,l},\boldsymbol{\Delta}^{u-1}_{i,j,l}\big)\big|  + \big|F_{\boldsymbol{\Pi}^{u-1}_{i,j,l}}\big(\boldsymbol{\Theta}^{u-1}_{i,j,l}, \boldsymbol{\Delta}^{u-1}_{i,j,l}\big)
      -F_{\boldsymbol{\Pi}^{u-2}_{i,j,l}}\big(\boldsymbol{\Theta}^{u-1}_{i,j,l},\boldsymbol{\Delta}^{u-1}_{i,j,l}\big)\big| \nonumber\\      
\leq &    \big\|\nabla  F_{\boldsymbol{\Pi}^{u-1}_{i,j,l}} \big\|_{\infty} \Big( \big\vert \boldsymbol{\Theta}^{u}_{i,j,l} -   \boldsymbol{\Theta}^{u-1}_{i,j,l}  \big\vert  +  \big\vert  \boldsymbol{\Delta}^{u}_{i,j,l} - \boldsymbol{\Delta}^{u-1}_{i,j,l}  \big\vert  \Big)
+ \Big \vert \big( 1-\boldsymbol{\Theta}^{u-1}_{i,j,l} - \boldsymbol{\Delta}^{u-1}_{i,j,l}\big)  \big( \boldsymbol{\Pi}^{u-1}_{i,j,l}-\boldsymbol{\Pi}^{u-2}_{i,j,l} \big)  \Big \vert \nonumber\\
\leq &C_{\pi} \big\{ \big\vert \boldsymbol{\Theta}^{u}_{i,j,l} -   \boldsymbol{\Theta}^{u-1}_{i,j,l}  \big\vert  +  \big\vert  \boldsymbol{\Delta}^{u}_{i,j,l} - \boldsymbol{\Delta}^{u-1}_{i,j,l}  \big\vert  \big\} + C_{\pi} \sum_{v=T_g+2}^{u-1} (1-2c_{\min})^{u- v } \big( \big\vert \boldsymbol{\Theta}^{v}_{i,j,l} -   \boldsymbol{\Theta}^{v-1}_{i,j,l}  \big\vert  +  \big\vert  \boldsymbol{\Delta}^{v}_{i,j,l} - \boldsymbol{\Delta}^{v-1}_{i,j,l}  \big\vert  \big) \nonumber\\
\leq & C_{\pi} \sum_{v=T_g+2}^u (1-2c_{\min})^{u- v} \big( \big\vert \boldsymbol{\Theta}^{v}_{i,j,l} -   \boldsymbol{\Theta}^{v-1}_{i,j,l}  \big\vert  +  \big\vert  \boldsymbol{\Delta}^{v}_{i,j,l} - \boldsymbol{\Delta}^{v-1}_{i,j,l}  \big\vert  \big), \nonumber
\end{align}
which proves the claim stated in \eqref{eq-pi-theta-delta}. 

Then for any $v \neq s$ and $v, s \geq T_g +2$, we have that 
\begin{align}\label{eq-pi-s-v}
    \bigl|\boldsymbol{\Pi}^{v-1}_{i,j,l}-\boldsymbol{\Pi}^{s-1}_{i,j,l}\big| 
    \leq &  
    \sum_{v' = v \wedge s }^{v \vee s -1 }\bigl|\boldsymbol{\Pi}^{v'}_{i,j,l}-\boldsymbol{\Pi}^{v'-1}_{i,j,l}\big|  \nonumber\\
    \leq & \sum_{v' = v \wedge s}^{v \vee s -1 } C_{\pi} \sum_{u' =T_g+2}^{v'} (1-2c_{\min})^{v'- u'}  \big( \big\vert \boldsymbol{\Theta}^{u'}_{i,j,l} -   \boldsymbol{\Theta}^{u'-1}_{i,j,l}  \big\vert  +  \big\vert  \boldsymbol{\Delta}^{u'}_{i,j,l} - \boldsymbol{\Delta}^{u'-1}_{i,j,l}  \big\vert  \big)   \nonumber\\
    \leq & \frac{C_{\pi}}{2c_{\min}}\sum_{v' = T_g+2}^{v \vee s -1 }  \big( \big\vert \boldsymbol{\Theta}^{v'}_{i,j,l} -   \boldsymbol{\Theta}^{v'-1}_{i,j,l}  \big\vert  +  \big\vert  \boldsymbol{\Delta}^{v'}_{i,j,l} - \boldsymbol{\Delta}^{v'-1}_{i,j,l}  \big\vert  \big),
\end{align}
where the last inequality follows from summing the geometric series $\sum_{a=0}^{\infty} (1 - 2c_{\min})^a \leq 1 / (2c_{\min})$ for $0 < c_{\min} < 1/2$.

Combining  \eqref{eq-bias-I-spec} and \eqref{eq-pi-s-v},  we can derive that  if $t -k  \geq T_g +1$ for some $g \in [G]$, then
\begin{align}\label{eq-bias-I.1}
    \big\vert   \mathcal{E}^{t,k}_{i, j, l}  \big\vert 
    \leq & \frac{1}{1 -\boldsymbol{\Pi}_{i,j,l}^{t,k}} \frac{C_{\pi}}{2c_{\min}} \max_{s, u \in [t]\backslash[t-k]} \vert \boldsymbol{\Theta}^{s}_{i, j, l} -  \boldsymbol{\Theta}^{u}_{i, j, l}  \vert  \sum_{v' = T_g+2}^{t-1}  \big( \big\vert \boldsymbol{\Theta}^{v'}_{i,j,l} -   \boldsymbol{\Theta}^{v'-1}_{i,j,l}  \big\vert  +  \big\vert  \boldsymbol{\Delta}^{v'}_{i,j,l} - \boldsymbol{\Delta}^{v'-1}_{i,j,l}  \big\vert  \big) \nonumber\\
    \leq &  \frac{C_{\pi}}{2c_{\min}^2} \max_{s, u \in [t]\backslash[t-k]} \vert \boldsymbol{\Theta}^{s}_{i, j, l} -  \boldsymbol{\Theta}^{u}_{i, j, l}  \vert  \sum_{v' = T_g+2}^{t-1}  \big( \big\vert \boldsymbol{\Theta}^{v'}_{i,j,l} -   \boldsymbol{\Theta}^{v'-1}_{i,j,l}  \big\vert  +  \big\vert  \boldsymbol{\Delta}^{v'}_{i,j,l} - \boldsymbol{\Delta}^{v'-1}_{i,j,l}  \big\vert  \big).
\end{align}
where the last inequality follows from \eqref{eq-pi-upper-lower}. 

Combining  \eqref{eq-bias-I-beg}  and \eqref{eq-bias-I.1}, we can conclude that
\begin{align}\label{eq-bias-I}
  (I) = &    \boldsymbol{\Theta}_{i,j, l}^{t, k} + \mathcal{E}^{t,k}_{i, j, l}  ,  
\end{align}
where if $k =1$, 
\begin{equation}\label{eq-bias-I-pro-1}
   \mathcal{E}^{t,k}_{i, j, l} = 0,
\end{equation}
and if $t -k  \geq T_g +1$ for some $g \in [G]$, then
\begin{align}\label{eq-bias-I-pro-2}
\vert \mathcal{E}^{t,k}_{i, j, l}\vert  \leq \frac{C_{\pi}}{2c_{\min}^2} \max_{s, u \in [t]\backslash[t-k]} \vert \boldsymbol{\Theta}^{s}_{i, j, l} -  \boldsymbol{\Theta}^{u}_{i, j, l} \vert \sum_{v' = T_g+2}^{t-1}  \big( \big\vert \boldsymbol{\Theta}^{v'}_{i,j,l} -   \boldsymbol{\Theta}^{v'-1}_{i,j,l}  \big\vert  +  \big\vert  \boldsymbol{\Delta}^{v'}_{i,j,l} - \boldsymbol{\Delta}^{v'-1}_{i,j,l}  \big\vert  \big). 
\end{align}

\medskip
\noindent
\textbf{Step 2.}  In this step,  we focus on term $(III.1)$ in \eqref{eq-bias-all}. By Taylor series with Lagrange remainder, there exist a random scalar  $r_{i,j, l}^{t,k}$ between $k^{-1} \sum_{s=t-k+1}^{t} \mathbf{A}_{i,j,l}^{s-1}$ and $\boldsymbol{\Pi}_{i,j,l}^{t,k}$ such that
\begin{align}\label{eq-bias-III-lower}
(III.1)  =  \frac{1}{k} \mathbb{E}  \bigg[ \sum_{s=t-k+1}^{t} \mathbf{A}_{i,j,l}^{s} (1 - \mathbf{A}_{i,j,l}^{s-1})
\frac{\big(k^{-1} \sum_{s=t-k+1}^{t} \mathbf{A}_{i,j,l}^{s-1} - \boldsymbol{\Pi}_{i,j,l}^{t,k}\big)^2}{\big(1 - r_{i,j, l}^{t,k}\big)^{3}}
 \mathbbm{1}_{\mathcal{I}_{t,k}}\bigg]\geq  0.
\end{align}
On the other hand, due to the event $\mathcal{I}_{t, k}$ defined in \eqref{eq-bias-event}, we can derive that
\begin{align*}
\sum_{u=2}^{\infty} \frac{\big| k^{-1} \sum_{s=t-k+1}^{t} \mathbf{A}_{i,j,l}^{s-1} - \boldsymbol{\Pi}_{i,j,l}^{t,k} \big|^u}{\big(1 - \boldsymbol{\Pi}_{i,j,l}^{t,k}\big)^{u+1}}  \mathbbm{1}_{\mathcal{I}_{t, k}}
&\leq 
\bigg( k^{-1} \sum_{s=t-k+1}^{t} \mathbf{A}_{i,j,l}^{s-1} - \boldsymbol{\Pi}_{i,j,l}^{t,k} \bigg)^2
\sum_{u=0}^{\infty} \frac{1}{(1 - \boldsymbol{\Pi}_{i,j,l}^{t,k})^{3} 2^u} \\
&= \bigg(  k^{-1} \sum_{s=t-k+1}^{t} \mathbf{A}_{i,j,l}^{s-1} - \boldsymbol{\Pi}_{i,j,l}^{t,k} \bigg)^2
\cdot \frac{2}{(1 - \boldsymbol{\Pi}_{i,j,l}^{t,k})^3}.
\end{align*}
We obtain  that 
\begin{align}\label{eq-bias-III-upper}
(III.1)  \leq \frac{1}{k} \mathrm{Var}\bigg( \frac{1}{\sqrt{k}} \sum_{s=t-k+1}^t \mathbf{A}_{i,j,l}^{s-1} \bigg) \cdot \frac{2}{(1 - \boldsymbol{\Pi}_{i,j,l}^{t,k})^3}.
\end{align}
Then we have that
\begin{align}\label{eq-bias-III-var}
\mathrm{Var}\bigg( \frac{1}{\sqrt{k}} \sum_{s = t - k + 1}^{t} \mathbf{A}_{i,j,l}^{s - 1} \bigg) 
= & \frac{1}{k} \sum_{r, s = t - k + 1}^{t} \mathrm{Cov}\big(\mathbf{A}_{i,j,l}^{r - 1}, \mathbf{A}_{i,j,l}^{s - 1} \big) 
\leq  \frac{1}{k} \sum_{r, s = t - k + 1}^{t}  \frac{1}{4} (1 - 2c_{\min})^{\vert r-s \vert}  \nonumber\\
= & \frac{1}{4k} \sum_{h=-(k-1)}^{k-1} (k - |h|) (1 - 2c_{\min})^{|h|}
=  \frac{1}{4k}\bigg\{ k + 2 \sum_{h=1}^{k-1} (k - h) (1 - 2c_{\min})^h \bigg\} \nonumber\\
= & \frac{1}{4k} \bigg[ k + \frac{(1 - 2c_{\min}) \{ 1 - (k-1) (1 - 2c_{\min})^{k} + k(1 - 2c_{\min})^{k+1} \}}{(2c_{\min})^2} \bigg] \nonumber\\
\leq &  \frac{1}{4} \bigg\{ 1 + \frac{2(1 - 2c_{\min})}{k(2c_{\min})^2} \bigg\},
\end{align}
where the first inequality follows from \Cref{lem-temporal-cov}.
Combining \eqref{eq-bias-III-lower}, \eqref{eq-bias-III-upper} and \eqref{eq-bias-III-var}, we can derive that 
\begin{align}\label{eq-bias-III}
0  \leq (III.1)  \leq    
\frac{1}{4k} \bigg\{ 1 + \frac{2(1 - 2c_{\min})}{k(2c_{\min})^2} \bigg\}  \frac{2}{c_{\min}^3} \leq k^{-1} C_0, 
\end{align}
for some absolute constant $C_0 >0$.

\medskip
\noindent
\textbf{Step 3.} In this step, we consider term $(II)$ in \eqref{eq-bias-all}. 
By Cauchy--Schwarz inequality, it holds that
\begin{align}
\vert  (II) \vert = & \bigg\vert \frac{1}{\big(1 - \boldsymbol{\Pi}_{i,j,l}^{t,k}\big)^2} \mathrm{Cov}\bigg( \frac{1}{k} \sum_{s} \mathbf{A}_{i,j,l}^{s}(1 - \mathbf{A}_{i,j,l}^{s-1}),  \frac{1}{k} \sum_{s=t-k+1}^{t} \mathbf{A}_{i,j,l}^{s-1} \bigg)\bigg\vert \nonumber\\
\leq & \frac{1}{\big(1 - \boldsymbol{\Pi}_{i,j,l}^{t,k}\big)^2} \sqrt{\mathrm{Var}\bigg( \frac{1}{k} \sum \mathbf{A}_{i,j,l}^{s}(1 - \mathbf{A}_{i,j,l}^{s-1}) \bigg)} \cdot \sqrt{\mathrm{Var} \bigg(\frac{1}{k} \sum_{s=t-k+1}^{t} \mathbf{A}_{i,j,l}^{s-1}\bigg)}. \nonumber
\end{align}
By \eqref{eq-bias-III-var} in \textbf{Step 2}, we obtain that 
\[
\mathrm{Var} \bigg(\frac{1}{k} \sum_{s=t-k+1}^{t} \mathbf{A}_{i,j,l}^{s-1}\bigg)\leq \frac{C_1'}{k}, 
\]
for some absolute constant $C_1' >0$.
Similarly, we obtain that 
\[
\mathrm{Var}\bigg( \frac{1}{k} \sum \mathbf{A}_{i,j,l}^{s}(1 - \mathbf{A}_{i,j,l}^{s-1}) \bigg) \le \frac{C_2'}{k}.
\] 
Hence, by \eqref{eq-pi-upper-lower}, it holds that
\begin{equation}\label{eq-bias-II}
    \vert (II)  \vert \leq   \frac{1}{c_{\min}^2}  \cdot \sqrt{\frac{C_1' C_2'}{k^2}} \leq  k^{-1}C_1, 
\end{equation}
for some absolute constant $C_1 >0$.

\medskip
\noindent
\textbf{Step 4.}  Analogous to the expansion in \eqref{eq-bias-all}, we obtain that 
\begin{align}
(III.2)  =   & \mathbb{E} \bigg[  \bigg\{ \widehat{\boldsymbol{\Theta}}_{i,j,l}^{t,k}  - \frac{1}{k} \sum_{s=t-k+1}^{t} \frac{\mathbf{A}_{i,j,l}^{s} (1 - \mathbf{A}_{i,j,l}^{s-1})}{1 -\boldsymbol{\Pi}_{i,j,l}^{t,k}} \nonumber\\
& \hspace{0.5cm} 
-  \frac{1}{k} \sum_{s=t-k+1}^{t} \mathbf{A}_{i,j,l}^{s} (1 - \mathbf{A}_{i,j,l}^{s-1}) \frac{k^{-1} \sum_{s=t-k+1}^{t} \mathbf{A}_{i,j,l}^{s-1} - \boldsymbol{\Pi}_{i,j,l}^{t,k}}{(1 - \boldsymbol{\Pi}_{i,j,l}^{t,k} )^2}  \bigg\} \mathbbm{1}_{\mathcal{I}_{t,k}^c}  \bigg].  
 \nonumber\\
\end{align}
Note that 
$\widehat{\boldsymbol{\Theta}}_{i,j,l}^{t,k} \in [0, 1]$. 
By \eqref{eq-pi-upper-lower}, we have that 
\[
\frac{1}{k} \sum_{s=t-k+1}^{t} \frac{\mathbf{A}_{i,j,l}^{s} (1 - \mathbf{A}_{i,j,l}^{s-1})}{1 -\boldsymbol{\Pi}_{i,j,l}^{t,k}}  \in [0, c_{\min}^{-1}]
\]
and 
\[
\frac{1}{k} \sum_{s=t-k+1}^{t} \mathbf{A}_{i,j,l}^{s} (1 - \mathbf{A}_{i,j,l}^{s-1}) \frac{k^{-1} \sum_{s=t-k+1}^{t} \mathbf{A}_{i,j,l}^{s-1} - \boldsymbol{\Pi}_{i,j,l}^{t,k}}{(1 - \boldsymbol{\Pi}_{i,j,l}^{t,k} )^2} \in [-c_{\min}^{-2}, c_{\min}^{-2}]. 
\]
Thus, 
\[
 \vert (III.2)  \vert 
\leq \big({c_{\min}^{-2}} +c_{\min}^{-1}\big) \P\{\mathcal{I}_{t,k}^c\}. 
\]
If $ 1 \leq k  \leq 3$, since $\frac{1}{k} \sum_{s=t-k+1}^t \mathbf{A}_{i,j,l}^{s-1} \in [0, 1]$, we have that 
\[
\mathbb{P} \{\mathcal{I}_{t,k}^c\} \leq   2 \exp \bigg\{ -\frac{ \big(1- \boldsymbol{\Pi}_{i,j,l}^{t,k}\big)^2 }{2} \bigg\}
\leq  2 \exp \big\{ - c_{\min}^2/2 \big\} \leq C_2 k^{-1}
\]
for some absolute constant $C_2 >0$. 

If $ k  \geq 4$, applying the Bernstein inequality for $\alpha$-mixing sequences \citep[e.g.~Theorem 1 in][]{merlevede2009bernstein} together with \eqref{eq-pi-upper-lower} yields 
\[
\mathbb{P} \{\mathcal{I}_{t,k}^c\} \leq  \exp\bigg\{- \frac{ c_{0} k }{\log (k)\log\log (k)} \bigg\} \leq \frac{C_3}{k}
\]
for some absolute constant $C_3 >0$. Hence,  it holds that 
\begin{align}\label{eq-bias-event-prob}
\vert (III.2) \vert \leq & \frac{C_4}{k},
\end{align}
for some absolute constant $C_4 >0$

 \medskip
\noindent
\textbf{Step 5.} Combining \eqref{eq-bias-all}, \eqref{eq-bias-I}, \eqref{eq-bias-I-pro-1}, \eqref{eq-bias-I-pro-2}, \eqref{eq-bias-III}, \eqref{eq-bias-II} and \eqref{eq-bias-event-prob}, we have that
\begin{align}
\Big\vert\mathbb{E}\big\{\widehat{\boldsymbol{\Theta}}_{i,j,l}^{t,k} \big\} - \boldsymbol{\Theta}_{i,j, l}^{t, k} \Big\vert \leq &  \vert \mathcal{E}_{i, j, l}^{t, k}\vert   + \frac{C_0+C_1+C_4}{k}, \nonumber
\end{align}
 where, if $k =1$, 
\[
   \mathcal{E}^{t,k}_{i, j, l} = 0;
\]
and if $t -k  \geq T_g +1$ for some $g \in [G]$, then
\begin{align}
\vert \mathcal{E}^{t,k}_{i, j, l}\vert  \leq \frac{C_{\pi}}{2c_{\min}^2} \max_{s, u \in [t]\backslash[t-k]} \vert \boldsymbol{\Theta}^{s}_{i, j, l} -  \boldsymbol{\Theta}^{u}_{i, j, l} \vert \sum_{v' = T_g+2}^{t-1}  \big( \big\vert \boldsymbol{\Theta}^{v'}_{i,j,l} -   \boldsymbol{\Theta}^{v'-1}_{i,j,l}  \big\vert  +  \big\vert  \boldsymbol{\Delta}^{v'}_{i,j,l} - \boldsymbol{\Delta}^{v'-1}_{i,j,l}  \big\vert  \big). \nonumber
\end{align}
This completes the proof of \eqref{eq-bias-1}.
Similarly, we can prove  \eqref{eq-bias-2}.

\end{proof}

\subsubsection{Proof of Lemma \ref{lemma-singular-values-non-stat}}\label{sec-app-4-non}

\begin{proof}[Proof of \Cref{lemma-singular-values-non-stat}]

Each row of $Z$ contains exactly one $1$ and each column $k \in[K]$ contains $s_k$ ones, so
\[
Z^\top Z = \mathrm{diag}(s_1,\dots,s_K).
\]
Therefore $\mathrm{rank}(Z) = K$, $\sigma_1(Z) = \sqrt{s_{\max}}$ and 
$\sigma_K(Z) = \sqrt{s_{\min}}$.

For any $t \in [T]$ and $k \in [t]$, it holds that 
\[
  \mathbf{\Omega}^{t, k} = \mathbf{Q}^{t, k} \times_1 Z \times_2 Z.
\]
By the definition and basic properties of tensor matricisation \citep[e.g. Lemma 4 in ][]{zhang2018tensor}, we see that 
\[
    \mathcal{M}_1(\mathbf{\Omega}^{t, k}) = Z \mathcal{M}_1 (\mathbf{Q}^{t, k})(Z \otimes I_{L})^{\top}.
\]
Fix  $t \in [T]$ and $k \in [t]$. We aim to show that $\mathrm{rank}(\mathcal{M}_1(\mathbf{\Omega}^{t, k})) = \mathrm{rank}(\mathcal{M}_1(\mathbf{Q}^{t, k}))$.  By Lemma S3 in the Supplement of \cite{wang2025multilayer}, it suffices to verify that $Z \otimes I_{L} \in \mathbb{R}^{nL \times KL}$ has full column rank, i.e.~$\mathrm{rank}(Z \otimes I_{L}) = KL$. 
Using the property of Kronecker products, we obtain $\mathrm{rank}(Z \otimes I_{L}) = \mathrm{rank}(Z) \mathrm{rank}( I_{L}) = LK$. Thus $\mathcal{M}_1(\boldsymbol{\Omega}^{t,k})$ and $\mathcal{M}_1(\mathbf{Q}^{t,k})$ have the same rank.

Next, we have that
\[
    \|\mathcal{M}_1(\mathbf{\Omega}^{t, k})\| \leq \|Z\| \|\mathcal{M}_1(\mathbf{Q}^{t, k})\| \|Z \otimes I_L\| = \sigma_1^2(Z) \| \mathcal{M}_1(\mathbf{Q}^{t, k}) \| \leq s_{\max} \|\mathbf{Q}^{t, k}\|.
\]
It follows from Lemma S4 in the Supplement of \cite{wang2025multilayer}  that 
\[
    \sigma_{\min} (\mathcal{M}_1(\mathbf{\Omega}^{t, k})) \geq \sigma_{\min} (Z) \sigma_{\min}(\mathcal{M}_1(\mathbf{Q}^{t, k})) \sigma_{\min} (Z \otimes I_L) = \sigma_K^2(Z)  \sigma_{K} (\mathcal{M}_1(\mathbf{Q}^{t, k})) \geq s_{\min} \sigma_{\min} (\mathbf{Q}^{t, k}).
\]
Applying the same arguments, we obtain for any $s \in [3]$ that 
\[
\mathrm{rank} \big( \mathcal{M}_s(\mathbf{\Omega}^{t, k} \big)
= \mathrm{rank}(\mathcal{M}_s(\mathbf{Q}^{t, k}))
\]
and
\[ s_{\min} \sigma_{\min} (\mathbf{Q}^{t, k}) \leq  \sigma_{\min}\big(\mathcal{M}_s(\boldsymbol{\Omega}^{t, k})\big)  \leq \big\| \mathcal{M}_s(\boldsymbol{\Omega}^{t, k}) \big\| \leq   s_{\max} \| \mathbf{Q}^{t, k}\|.
\]  
We complete the proof.

\end{proof}

\end{document}